\documentclass[11pt]{article}

\usepackage[margin=1in]{geometry}
\usepackage{amsmath,amssymb,amsthm}
\usepackage{graphicx}
\usepackage{booktabs}
\usepackage{array}
\usepackage{multirow}
\usepackage{xcolor}
\usepackage{enumitem}
\usepackage{microtype}
\usepackage{url}
\usepackage[colorlinks=true,allcolors=blue]{hyperref}
\usepackage{natbib}
\usepackage{seqsplit}
\usepackage{ragged2e}
\usepackage[section]{placeins}
\input{styles/paper-figures.sty}

\graphicspath{{figs/}}

\newtheorem*{definition}{Definition}
\newtheorem*{remark}{Remark}
\newtheorem{theorem}{Theorem}[section]
\newtheorem{proposition}[theorem]{Proposition}
\newtheorem{lemma}[theorem]{Lemma}
\newtheorem{corollary}[theorem]{Corollary}
\newtheorem{fact}[theorem]{Fact}
\numberwithin{equation}{section}

\newcommand{\dhat}{\widehat{d}}
\newcommand{\E}{\mathbb{E}}
\DeclareMathOperator{\Var}{Var}
\DeclareMathOperator{\Cov}{Cov}
\DeclareMathOperator{\tr}{tr}
\newcommand{\one}{\mathbf{1}}

\newcommand{\given}{\,|\,}

\title{IBLTs Measure Before They Decode: Self-Sizing Set Reconciliation for Database Consistency Verification}
\author{
  Min Wu\textsuperscript{1,*} \quad
  Ji Qi\textsuperscript{2,*} \quad
  Zhengsheng Ye\textsuperscript{3} \quad
  Chengdui Luo\textsuperscript{2}  \\
  Shudong Lu\textsuperscript{2} \quad
  Zhengyang Wei\textsuperscript{3} \\[4pt]
  \small \textsuperscript{1}Hangzhou Dianzi University \\
  \small \textsuperscript{2}China Mobile (Suzhou) Software Technology Co., Ltd.  \\
  \small \textsuperscript{3}Nine Chapters (Zhejiang) Technology Co., Ltd.  \\ [2pt]
  {\small\texttt{wumin@hdu.edu.cn}} \quad
  \\[2pt]
  \small \textsuperscript{*}Equal contribution
}
\date{\today}

\begin{document}

\maketitle

\begin{abstract}
Cross-system data replication pipelines cannot confirm end-to-end
consistency from the local guarantees of each hop, so the two endpoints
must be compared directly on a periodic basis. Once the rows of a fixed
snapshot are normalized into fingerprints, the task reduces to finding
the symmetric difference of the two sets. Set reconciliation encodes both
sets with an Invertible Bloom Lookup Table (IBLT), whose communication
grows only with the difference cardinality~$d$ and is independent of table
size, which suits large tables. The capacity of an IBLT, however, must be
fixed while $d$ is still unknown. Across 41{,}603 production
reconciliations over 90 days, nonzero $d$ spans about seven orders of
magnitude, and no reliable empirical constant exists.

We show that the count array of an IBLT has already measured $d$ before
decoding. The measurement is in-band: it is carried by the recovery sketch
itself and adds no bytes dedicated to estimation. A mapping-aware theorem
carries the same construction to Irregular, Rateless, and MET IBLTs. The
protocol reads the
estimate only after a decoding failure; we prove that the
failure-conditioned lower quantile bounds the risk of underestimation,
which gives the second-round capacity a configurable success-probability
guarantee. On this basis we build a self-sizing protocol: the first round
attempts recovery with a small sketch and stops there on success; on
failure it reads $d$ and sets the second-round capacity, completing
reconciliation in at most two rounds. Against a controlled oracle,
communication is 1.29--1.47 times that of a scheme given $d$ in advance;
production workload characterization, relational-database replay, and a
cross-city KV deployment confirm the end-to-end mechanism. In production
operation on an Oracle--MySQL link, all completed runs succeeded within two
rounds, over 90\% of them on the 1-RTT fast path with a single 16\,KB sketch. Measurement and
recovery share one IBLT and need no separate estimation structure.
\end{abstract}

\section{Introduction}
\label{sec:intro}

Set reconciliation is the problem in which two parties hold sets $A$ and $B$
and wish to recover the symmetric difference $A\triangle B$ using
communication far smaller than $|A|+|B|$.  The Invertible Bloom Lookup Table
(IBLT)~\cite{goodrich2011invertible} is the standard tool for this problem.

This paper targets database reconciliation---verifying and repairing
consistency between database instances.  Even when both sides run the
same engine, the synchronization channel itself (replication logs or
change-data capture) is precisely what is being verified; it cannot
serve as its own oracle.  Snapshot-level full-table scans are
therefore the verification path that always applies, whether the engines are
homogeneous or heterogeneous.

Each scan reduces every record to a short fingerprint derived from its
primary key and a normalized row value; identical records cancel when
the two sketches are subtracted, while inserts, deletes, and updates
leave fingerprints in $A\triangle B$.

The scan itself dominates cost: computing and transmitting fingerprints
over hundreds of millions of rows far exceeds the time the
reconciliation algorithm spends.  The hardest parameter to fix in
advance is the
difference cardinality~$d$, which ranges from zero to tens of millions
across the same table at different times.  Guessing too small forces a
rescan of the entire table; guessing too large wastes network bandwidth
and memory after the scan cost has already been incurred.
\S\ref{sec:system} reports a production workload from NineData Cloud:
the measured $d$ spans seven orders of magnitude, and the workload is
nothing like a ``mostly equal'' regime---$d=0$ in only 56.4\% of the
41{,}603 runs, and among the tables that dominate reconciliation cost the
fraction with differences reaches 98.7\%.  The failure-then-measure path
is therefore the cost-dominant regime.  A cross-city
Redis/Pika deployment on China Mobile checks cross-data-model behavior under
similarly wide-ranging~$d$.  Deployed in NineData's production
reconciliation service, self-sizing completed every run within two rounds,
91.7\% of them on the 1-RTT fast path with a single 16\,KB sketch.

An IBLT is a hash-table sketch that supports linear subtraction.  Both sides
hash every element into $k$ out of $M$ cells; each cell maintains a count,
an invertible key aggregate (\emph{keySum}), and a verification digest
(\emph{checkSum}).  After the two sketches are subtracted cell by cell,
common elements cancel, and the resulting difference table encodes exactly
$A\triangle B$.  A cell's count records the net number of difference
elements mapped to it: $+1$ for elements in $A\setminus B$, $-1$ for those
in $B\setminus A$.

When a cell's count is $\pm 1$ and its aggregate fields verify, the decoder
identifies the lone difference element mapped to it and removes that element
from all $k$ cells it touches, often exposing new singletons in a cascade
called \emph{peeling}.  Peeling succeeds only when the capacity $M$ is
roughly $1.2$--$1.4\times$ the difference cardinality~$d$; below this
threshold the success probability drops sharply at the 2-core phase
transition of the underlying random hypergraph.  A failed IBLT's cells
still hold their counts and aggregates.

Three existing lines of work address unknown $d$, each paying the cost in a
different dimension:

\begin{enumerate}[itemsep=0pt,topsep=2pt]
  \item \textbf{Pre-estimation structures} (Strata / ToW / min-wise): exchange
  a dedicated summary to estimate~$d$, then configure the main IBLT from the
  estimate.  The summary is an independent structure whose bytes are dedicated solely to estimation.
  \item \textbf{Blind doubling} (exponential backoff): try a conservative
  capacity and double on failure until decoding succeeds.  The approach is
  simple and needs no extra structure, but the cost concentrates in rounds:
  in databases each extra round often means another full-table scan, so
  round cost far exceeds byte cost.
  \item \textbf{Rateless IBLT}~\cite{yang2024rateless}: the sender keeps
  producing coded symbols and the receiver appends them until decoding
  succeeds, sidestepping capacity presetting by construction.  Its design
  point assumes the sender can incrementally produce symbols on demand;
  the incremental encoder maintains $O(N)$ per-row mapping state, whereas
  a Plain IBLT encoder is stateless and its memory scales with the table
  dimension $M$, not the row count~$N$.  In the scan-dominated,
  request/response regime of database reconciliation this cost is
  prohibitive: on a 609M-row production snapshot the incremental endpoint
  runs ${\sim}48\times$ slower and peaks at ${\sim}93$\,GiB memory, against
  a few tens of MiB for self-sizing's endpoint
  (Appendix~\ref{app:rateless-comparison}).
\end{enumerate}

The industrial alternative is \textbf{Merkle-style localization}:
partition the key space into ranges, exchange checksums level by level, and
drill down to dirty blocks.  Its cost depends on how many ranges contain
differences: clustered differences are cheap; scattered ones can degrade to
shipping the whole table.  This is the end-to-end comparison baseline of
\S\ref{sec:system}; its engineering practice and cost model are discussed in
\S\ref{sec:related}.

The core observation of this paper is that \textbf{an IBLT has already
completed a measurement before it attempts to decode}.  Conventionally, when
decoding fails the whole sketch is discarded and retried; the transmitted
counts go to waste.  In fact, the pre-peeling count array $C$ is a
mapping-aware second-moment observation: centering each count about the
count-array average and summing the squares aggregates the fluctuation
energy that
difference elements leave in the sketch.  Every difference element, whether
from side $A$ or side $B$, contributes roughly $k$ to this energy; dividing
by the mapping-determined normalization constant yields an
\textbf{exactly unbiased} estimate $\dhat$ of $d$.  Its variance has the
closed form $2d(d-1)/(M-1)$ and is independent of the number of hashes $k$;
for fixed $M$ there is also a complete asymptotic distribution (a
chi-square limit) that yields confidence intervals.  The protocol can
therefore observe itself and size itself.  If first-round decoding
fails, it reads $\dhat$ and a confidence upper bound from the same
sketch, computes a sufficient second-round capacity directly, and makes
one fresh decoding attempt.  The main path thus requires at most two
round trips.  Bounding the number of rounds makes the interaction
component of end-to-end latency predictable, while scan and transfer
costs remain workload-dependent,
especially over
cross-datacenter links where each round trip costs tens to hundreds of
milliseconds (\S\ref{sec:protocol}).  This path adds no estimation-only
fields or bytes: the measurement comes entirely from the count channel that
the Plain IBLT has already transmitted.  The estimator belongs to the $F_2$
second-moment estimation tradition;
\S\ref{sec:related} gives the point-by-point comparison with
AMS~\cite{alon1996spacecomplexity,alon1999spacecomplexity} and
Guo--Li~\cite{guo2013counting}.

Intuitively, \emph{failed IBLTs are not wasted}; technically, \emph{IBLTs
measure before they decode}; the system consequence is that a reconciliation
protocol can \emph{size itself} from the same sketch.

The paper follows one contribution chain.
\S\ref{sec:count-measurement} develops the exact measurement theory;
\S\ref{sec:mapping-aware} generalizes it to a mapping-aware family;
\S\ref{sec:protocol} builds the self-sizing protocol;
\S\ref{sec:system} validates the chain with production data and
deployments; \S\ref{sec:related} and \S\ref{sec:scope} position the work
and state its boundaries.

\paragraph{Contributions.}
\begin{enumerate}[itemsep=0pt,topsep=2pt]
  \item \textbf{Measurement theory for the Plain count channel}
  (\S\ref{sec:count-measurement}).  We instantiate the second-moment method
  exactly for the signed $k$-subset mapping of an IBLT: the estimator is
  exactly unbiased under arbitrary sign patterns; its variance has the closed
  form $\Var(\dhat)=2d(d-1)/(M-1)$ independent of $k$ and sign balance; and,
  for a fixed table dimension, it has a chi-square limit.  Because the
  protocol reads the estimate only after a first-round failure, we further
  establish failure-conditioned guarantees---an exact fixed-sign projection
  with exact cancellation under i.i.d. random signs, inheritance of the
  chi-square law when failure becomes typical, and distribution-free
  lower-tail bounds---providing the
  statistical interface~$q_\delta$ for self-sizing.
  \item \textbf{Family-level mapping-aware construction}
  (\S\ref{sec:mapping-aware}).  The mapping-aware theorem
  $\E[C^\top Q C]=\sum_j d_j\,\tr(Q\Sigma_j)$ \eqref{eq:master} gives the
  construction criterion for an unbiased estimator on the count channel of an
  IBLT variant---necessary and sufficient among positive-semidefinite
  quadratic forms that hold for all sign compositions
  (Appendix~\ref{app:thm6}); the
  design rule is to center within the smallest mapping-homogeneous block and
  then normalize by the true mapping covariance.  From the mapping mean and
  covariance of each variant we construct adapters for Irregular, Rateless,
  and MET, and unify the estimate interface on the failure path
  (Corollary~\ref{cor:variant-failure}).
  \item \textbf{Self-sizing protocol} (\S\ref{sec:protocol}).
  First-round IBLT $\rightarrow$ $\dhat$ $\rightarrow$ second-round capacity,
  at most two round trips, with no doubling probes.  Proposition~\ref{prop:capacity}
  composes the statistical interface $q_\delta$ with the decoder operating
  point $\beta$, giving a general capacity formula that accounts for
  recovered elements and a second-round success-probability guarantee.
  Proposition~\ref{prop:joint} guarantees that jointly peeling the first-round
  residual core never decreases the success probability.  Compared with blind
  doubling and Strata-first, self-sizing keeps the round count at 2 (or 1 when
  the first round already decodes) at roughly $1.3$--$1.7\times$ the byte cost
  of an oracle across the whole range of $d$.
  \item \textbf{System evaluation and production difference profiles}
  (\S\ref{sec:system}).  We implement a cross-engine database sidecar---a lightweight
  agent process co-located with each database instance---and
  verify the first-round-success and
  fail$\rightarrow$measure$\rightarrow$resize paths against ground truth.
  A 90-day profile of 41{,}603 table-level reconciliation runs at NineData
  shows that difference cardinality spans seven orders of magnitude, that
  cost concentrates in large tables, and that differences cluster into
  key-rank segments with significant within-segment scatter.  Production-shape
  replay gives the cost crossover between IBLT and Merkle-style
  localization, and a cross-city Redis/Pika deployment on China Mobile
  checks cross-data-model behavior and the applicable regime.  In
  production operation on an Oracle--MySQL link, all 434 completed runs
  decoded within two rounds, and every second-round estimate stayed above
  the calibrated lower quantile.
\end{enumerate}

\section{Plain IBLT Count Measurement}
\label{sec:count-measurement}

This section develops the measurement theory of the Plain IBLT count
channel in three steps: an exactly unbiased estimator
(Theorem~\ref{thm:unbiased}), its exact variance
(Theorem~\ref{thm:variance}), and a fixed-dimension chi-square limit
(Theorem~\ref{thm:chi2}).  Because the protocol reads the estimate only
after a decoding failure, \S\ref{sec:failure-cond} adds
failure-conditioned mean and distributional guarantees
(Theorem~\ref{thm:cond-mean} and Proposition~\ref{prop:cond-chi2}), and
\S\ref{sec:quantile-interface} consolidates them into the single quantity the
protocol needs: the lower quantile~$q_\delta$.  All results depend only on
the pre-peeling count vector.

\subsection{Model and estimator}
\label{sec:model}

We use $M\ge 2$ cells and $1\le k<M$; difference elements are indexed
$x\in\{1,\dots,d\}$.  If element $x$ appears only in $A$ (i.e., $x\in
A\setminus B$), we write $s_x=+1$; if only in $B$, $s_x=-1$.  Thus
$d_+=|A\setminus B|$, $d_-=|B\setminus A|$, and
$d=d_++d_-=|A\triangle B|$ under the \emph{unit-weight multiset}
interpretation: a modified row contributes its old and new fingerprints once
each.  For a fixed sign composition with $d>0$, let
$\theta_-=d_-/d$ and define the net directional fraction
$\theta_\Delta=1-2\theta_-=(d_+-d_-)/d$.  The condition $k<M$ excludes the
degenerate case $\gamma=0$, where the count vector has only the constant
direction and $d$ is not identifiable.

Each element independently maps to $k$ \emph{distinct} cells: $S_x\subseteq
[M]$ with $|S_x|=k$, and indicator vector $a_x\in\{0,1\}^M$.  All
expectations and variances below are taken over the mapping randomness,
modeled as ideal random hashing.

The difference count vector before peeling is
$C=\sum_x s_x a_x$.  Let $I_M$ be the $M\times M$ identity and
$\one_M=[1,\dots,1]^\top$ the all-ones column vector, and define the
centering matrix
\[
  Q = I_M - \frac{1}{M}\one_M\one_M^\top.
\]
The \emph{centered energy} $T$ aggregates the per-cell fluctuation into a
single scalar.  $Q$ is the orthogonal projector that removes the constant
component: $Q\one_M=0$ and $Q^2=Q$.  With $\bar C = M^{-1}\sum_i C_i$, we
have
$QC = C - \bar C\one_M$, and therefore
\begin{equation}
\label{eq:centered-energy}
  T = C^\top Q C
  = \sum_i (C_i - \bar C)^2
  = \sum_i C_i^2 - \frac{\bigl(\sum_i C_i\bigr)^2}{M}.
\end{equation}
Thus $T=\lVert QC\rVert^2$ is the sum of squared deviations of the per-cell
counts from the count-array mean.  As shown in Theorem~\ref{thm:unbiased}, cross
terms cancel in expectation, so $\E[T]=d\gamma$.

The normalization constant is
\begin{equation}
\label{eq:gamma}
  \gamma=k(1-k/M),
\end{equation}
and the estimator is
\begin{equation}
\label{eq:estimator}
  \dhat = \frac{T}{\gamma}.
\end{equation}
It is computed before peeling and does not depend on the residual core or the
decoding outcome.  Computationally, $T$ requires one linear pass over the
counts (accumulating $\sum_i C_i^2$ and $\sum_i C_i$): $O(M)$ time and
$O(1)$ extra space, with no additional structure.

The full symbol table used across the paper, including the family adapters,
confidence intervals, and experiment-specific notation, appears in
Appendix~\ref{app:notation}.

\subsection{Exact measurement theory}
\label{sec:exact-theory}

\begin{theorem}[Exact unbiasedness, arbitrary sign patterns]
\label{thm:unbiased}
Let $M\ge 2$ and $1\le k<M$.  Conditioned on an arbitrary fixed sign vector
$s=(s_1,\dots,s_d)\in\{\pm1\}^d$ (the fully one-sided case $d_-=0$
included), if $S_1,\dots,S_d$ are mutually independent and each is uniform
over all $k$-subsets of $[M]$, then over the mapping randomness
\[
  \E[T \given s] = d\gamma,
  \qquad
  \E[\dhat \given s] = d.
\]
The conclusion holds exactly for finite $M,d,k$; it does not rely on any
asymptotic approximation.
\end{theorem}

\noindent\textit{Proof.}  See Appendix~\ref{app:thm1} (exact unbiasedness
proof).

Intuitively, $Q$ projects $C$ onto the subspace with the count-array mean removed:
if the signs are imbalanced ($d_+\ne d_-$), every component of $C$ carries a
common drift of roughly $k\theta_\Delta d/M$; taking $\sum_i C_i^2$ directly would
square this drift and fold it into the sum, whereas the construction of $Q$ removes it
exactly.

Sign balance does not affect the expectation: the estimator is centered
at~$d$ regardless of how many differences come from each side, so the
protocol can read the same count channel for capacity sizing without a
separate cardinality sketch.

\begin{theorem}[Exact variance]
\label{thm:variance}
Under the independent, uniform, without-replacement $k$-subset mapping model,
\begin{equation}
\label{eq:variance}
  \Var(\dhat) = \frac{2d(d-1)}{M-1},
  \qquad
  \operatorname{Std}\!\left(\frac{\dhat}{d}\right) \approx \sqrt{\frac{2}{M-1}}.
\end{equation}
The variance of the normalized estimator over the random mappings is
independent of $k$ and the sign balance for any fixed sign pattern.
\end{theorem}

\noindent\textit{Proof.}  See Appendix~\ref{app:thm2} (exact variance proof).

Increasing $M$ shrinks the relative error as $M^{-1/2}$; $k$ still
affects the unnormalized $T$, the peeling threshold, and higher-order
finite-sample behavior, but not the normalized estimator's first two
moments.

The relative standard deviation is defined for $d\ge 1$; for $d=0$ we have
$C=0$ and $\dhat=0$ deterministically.

\begin{corollary}[Relative standard deviation]
\label{cor:rsd}
For $d\ge 1$,
\begin{equation}
\label{eq:rsd}
  \operatorname{RSD}(\dhat)
  = \frac{\operatorname{Std}(\dhat)}{\E[\dhat]}
  = \sqrt{\frac{2(d-1)}{d(M-1)}}
  \approx \sqrt{\frac{2}{M-1}}.
\end{equation}
Equivalently,
\[
  \E\!\left[\frac{\dhat}{d}-1\right]=0,
  \qquad
  \operatorname{Std}\!\left(\frac{\dhat}{d}-1\right)
  = \sqrt{\frac{2(d-1)}{d(M-1)}}
  = O\!\left((M-1)^{-1/2}\right).
\]
\end{corollary}

This gives a direct capacity--precision trade-off: for example, RSD is about
$8.9\%$ at $M=256$ and about $4.4\%$ at $M=1024$.  Thus even a first-round
sketch too small to peel still acts as a fixed-precision measurement summary
of $M$ count cells; roughly quadrupling $M$ halves the typical relative
error.

\subsection{Distributional characterization}
\label{sec:distribution}

Protocol configuration needs to control the probability of
underestimation, which requires the distribution of~$\dhat/d$, not just
its first two moments.  This subsection gives the unconditional
asymptotic distribution; \S\ref{sec:failure-cond} handles failure
conditioning.

The mean and variance are exact, but two moments do not determine tail
probabilities or confidence intervals.  To obtain the full asymptotic
distribution, fix the IBLT dimension $M$ and write the centered count vector
$QC$ as a sum of $d$ independent bounded mapping vectors: the multivariate
central limit theorem first gives a Gaussian limit for $QC/\sqrt{d}$, and the
squared norm then converts it to a chi-square limit.

The chi-square form matters for protocol configuration.  For small~$M$
the distribution of $\dhat/d$ is markedly right-skewed; a symmetric
normal approximation systematically under-predicts the lower quantile,
so the planned capacity errs on the conservative side.  The chi-square
limit captures this asymmetry; the quantitative gap is verified in
\S\ref{sec:exp-validation}.

From a systems perspective: with $M$ fixed by configuration, $\dhat/d$
converges to a stable distribution as differences grow, converting
``about $8.9\%$ typical error'' into a configurable tail probability.

\begin{theorem}[Fixed-$M$ chi-square limit, asymptotic confidence interval]
\label{thm:chi2}
Fix $M,k$.  As $d\to\infty$,
\begin{equation}
\label{eq:chi2}
  \frac{\dhat}{d} \Longrightarrow \frac{\chi^2_{M-1}}{M-1}.
\end{equation}
\end{theorem}

\noindent\textit{Proof.}  See Appendix~\ref{app:thm3} (chi-square limit
proof).

Let $\chi^2_{\nu,p}$ denote the $p$-quantile of the chi-squared distribution
with $\nu$ degrees of freedom.  Theorem~\ref{thm:chi2} gives the
unconditional asymptotic $1-\delta$ confidence interval
\begin{equation}
\label{eq:chi2-ci}
  \left[
    \frac{\nu\,\dhat}{\chi^2_{\nu,1-\delta/2}},
    \frac{\nu\,\dhat}{\chi^2_{\nu,\delta/2}}
  \right],
\end{equation}
with $\nu=M-1$.  For large $M$, $\chi^2_\nu/\nu$ is approximately
$\mathcal N(1,2/\nu)$; for example, its $\delta$-level lower quantile is
about $1-z_{1-\delta}\sqrt{2/(M-1)}$, where $z_{1-\delta}$ is the
$(1-\delta)$-quantile of the standard normal.

The receiver needs only the sketch dimension $M$ to convert the same pre-peeling
counts into an asymptotic confidence interval for $d$ using $\chi^2_{M-1}$
quantiles: no distribution fitting, no calibrated model, and no extra
estimation summary.  The interval is asymptotic in $d\to\infty$ for
fixed~$M$; \S\ref{sec:exp-validation} checks its finite-sample accuracy.

A finite-sample strengthening (Proposition~\ref{prop:kolmogorov}, supplementary
material) shows that, for fixed $M$, the distributional error of the chi-square
approximation shrinks as $O_M(d^{-1/2})$ via a convex-set Berry--Esseen
bound; explicit constants and the full proof are in the supplementary
material.

\subsection{Validity under failure conditioning}
\label{sec:failure-cond}

The results above characterize the estimator under the unconditional mapping
distribution, but the protocol reads the estimate only after peeling fails.
This section verifies that the selection bias does not invalidate the
measurement and provides the failure-conditioned guarantees the protocol
requires.

There is a subtle selection-bias issue: the samples the protocol observes are
restricted to the subset of IBLTs that happen not to decode.  Failed IBLT
sketches may have more crowded hash structures than typical sketches, so unconditional
unbiasedness does not automatically imply unbiasedness after failure.  We
first identify the mapping information selected by failure and then show that
an arbitrary fixed sign composition enters the conditional mean through one
scalar projection.  I.i.d. random signs cancel that projection exactly after
sign averaging.  The failure-conditioned asymptotic law and the
distribution-free finite-sample lower-tail bound likewise allow every sign
composition fixed independently of the mapping randomness, including fully
one-sided differences.

\begin{fact}[Failure measurability]
\label{fact:failure-measurable}
View the mappings of the difference elements as a hypergraph---cells are
vertices, and each difference element forms a hyperedge connecting the $k$
cells to which it maps.  In the idealized decoder, a cell can trigger peeling if and
only if it currently touches exactly one undecoded difference element;
graph-theoretically, this is a vertex of degree 1.  Whether the element comes
from $A\setminus B$ or $B\setminus A$ only affects whether its count reads
$+1$ or $-1$; it does not change how many elements the cell touches.

Repeated peeling is therefore equivalent to repeatedly deleting degree-1
vertices and their incident hyperedges.  If a subgraph remains in which every
vertex has degree at least 2 (a 2-core), the decoder cannot continue; hence
the failure event $F$ depends only on the hash mappings
$\{S_x\}_{x=1}^d$, not on which side the differences come from.  We assume
the checksum never collides: checksum collisions are excluded from
the idealized model, and their probability is controlled separately by the
checksum width.
\end{fact}

\begin{theorem}[Failure-conditioned sign projection]
\label{thm:cond-mean}
Let $d\ge2$ and $p_F=\Pr[F]>0$.  For distinct elements $x\ne y$, write
$W_{xy}=a_x^\top Qa_y$.  Conditional exchangeability of the mappings given
$F$ makes
\[
  c_F=\gamma^{-1}\E[W_{xy}\given F]
\]
the same for every pair $x\ne y$ and independent of the fixed sign pattern.
Then the following statements hold.

\emph{Fixed-sign projection.}  For any fixed sign vector with
$\theta_\Delta=1-2\theta_-=(d_+-d_-)/d$,
\begin{equation}
\label{eq:sign-projection}
  \frac{\E[\dhat\given F]}{d}-1
  =c_F\bigl(d\theta_\Delta^2-1\bigr).
\end{equation}

\emph{Random-sign cancellation.}  If the signs $\{s_x\}$ are i.i.d.\
uniform on $\{\pm1\}$ and independent of the mappings, then averaging over
both signs and mappings gives the exact identity
\[
  \E[\dhat\given F]=d.
\]

\emph{Composition-aware envelope.}  For every fixed sign vector,
\begin{equation}
\label{eq:composition-aware-bound}
  \left|\frac{\E[\dhat\given F]}{d}-1\right|
  \le
  \left|\frac{d\theta_\Delta^2-1}{d-1}\right|
  \sqrt{\frac{2(d-1)}{d(M-1)}}
  \sqrt{\frac{1-p_F}{p_F}}.
\end{equation}
\end{theorem}

\noindent\textit{Proof.}  See Appendix~\ref{app:thm4} (failure-conditioned
validity proof).

For fixed signs, Eq.~\eqref{eq:sign-projection} isolates the entire
conditional-mean effect in the product of a mapping-dependent coefficient
$c_F$ and the composition factor $d\theta_\Delta^2-1$.  I.i.d. random signs
make that factor zero on average, yielding exact conditional unbiasedness.
An exactly balanced fixed composition has shift $-c_F$, only
$1/(d-1)$ of the fully one-sided shift in magnitude; under i.i.d. random
signs, $\E_s[\theta_\Delta^2]=1/d$, so the projection cancels exactly after
sign averaging.
For deterministic compositions, Eq.~\eqref{eq:composition-aware-bound}
controls the remaining projection and vanishes with
$\sqrt{(1-p_F)/p_F}$ as failure becomes typical.  Since
$|d\theta_\Delta^2-1|\le d-1$, the previous sign-agnostic
Cauchy--Schwarz envelope remains an immediate corollary.

For $d=1$, the estimator is deterministically $\dhat=1$.  In the idealized
decoder the sole element is always peelable, so the failure event has
probability zero and no failure-conditioned distribution arises.

Equation~\eqref{eq:sign-projection} fixes no sign for~$c_F$, and both signs
occur.  The next corollary records a family with $c_F<0$, together with a
finite-parameter range in which conditioning on failure has no effect.

\begin{corollary}[Envelope, sign, and certain failure]
\label{cor:cf-regimes}
For every $d\ge2$ and $p_F>0$,
\begin{equation}
\label{eq:cf-envelope}
  |c_F|
  \le
  \sqrt{\frac{2(1-p_F)}
  {p_F\,d(d-1)(M-1)}}.
\end{equation}
Moreover:

\emph{(i) Negative coefficients occur.}
For $d=3$ and $k=M-2$,
\begin{equation}
\label{eq:cf-dense-counterexample}
  c_F=
  \frac{2(-M^3+15M^2-62M+80)}
       {M^4-14M^3+85M^2-200M+160},
\end{equation}
which is negative for every $M\ge10$; for example,
$(M,k,d)=(10,8,3)$ gives $c_F=-4/133$.

\emph{(ii) Certain failure.}
If $d\ge M-k+2$, no mapping can be peeled completely, so $p_F=1$ and
$c_F=0$.
\end{corollary}

\noindent\textit{Proof.}  See Appendix~\ref{app:thm4}.

The sign of $c_F$ has a direct reading in terms of cell loads.  Summing
$H_{xy}$ over all pairs of elements counts, for every cell, the pairs of
elements that share it.  Because the degrees of all cells always add up
to~$dk$, a larger total overlap means a more uneven spread of elements over
cells, and $c_F>0$ exactly when failure favors uneven spreads.  Two failure
mechanisms act in opposite directions.  A few heavily shared cells can trap
a small group of elements, which raises the overlap.  When every covered
cell holds at least two elements, peeling has no degree-1 cell to start from,
and an even spread makes this more likely.  In part~(i) each element misses
only two cells, so the second mechanism prevails.

Control of the conditional mean alone does not determine a conditional
quantile: two distributions can have the same mean and very different lower
tails.  The next proposition supplies the distribution-level bridge.  It
uses the fact that failure becomes typical in deep overload to transfer the
entire chi-square limit of Theorem~\ref{thm:chi2} to the failed instances.

\begin{proposition}[Chi-square inheritance when failure is typical]
\label{prop:cond-chi2}
Fix $M,k$, let $d\ge1$, and fix any sign vector $s\in\{\pm1\}^d$.
Let $F_d$ be the peeling failure event with $d$ difference elements, and define
\begin{equation}
\label{eq:cond-eps}
  \varepsilon_{M,d}
  =d\,k\left(1-\frac{k}{M}\right)^{d-1}.
\end{equation}
Whenever $\Pr[F_d]>0$,
\begin{equation}
\label{eq:cond-transfer}
  \sup_t\left|
    \Pr\!\left[\dhat/d\le t\given F_d\right]
    -\Pr\!\left[\dhat/d\le t\right]
  \right|
  \le 1-\Pr[F_d]
  \le \min\{1,\varepsilon_{M,d}\}.
\end{equation}
Consequently, as $d\to\infty$ with $M,k$ fixed,
\begin{equation}
\label{eq:cond-chi2}
  \left.\frac{\dhat}{d}\;\right|F_d
  \Longrightarrow \frac{\chi^2_{M-1}}{M-1}.
\end{equation}
\end{proposition}

\noindent\textit{Proof.}  See Appendix~\ref{app:thm4}.  A successful
nonempty peeling process must start from at least one degree-1 cell, which
gives $1-\Pr[F_d]\le\varepsilon_{M,d}$.  Conditioning on an event moves the
underlying probability measure by exactly the probability of its complement;
mapping the outcome to $\dhat/d$ can only decrease that distance, giving
Eq.~\eqref{eq:cond-transfer}.  The conclusion then follows by the triangle
inequality and Theorem~\ref{thm:chi2}.

The chi-square shape is therefore operational rather than an artifact of
unconditional sampling: in deep overload, conditioning on failure removes an
asymptotically negligible fraction of mappings, so the same fixed-degree
chi-square law describes the instances on which the protocol reads the
estimate.  Its skewness $\sqrt{8/(M-1)}$ does not vanish with $d$, preserving
the asymmetric lower tail after conditioning.

Notably, Proposition~\ref{prop:cond-chi2} requires no random-sign model.  The
unconditional chi-square law already holds for every fixed sign composition,
and the total-variation transfer does not average over signs.  In
Theorem~\ref{thm:cond-mean}, random signs enter only when the exact
fixed-sign projection is averaged to obtain $\E[\dhat\given F]=d$.

Equation~\eqref{eq:cond-eps} also makes the certified regime explicit.  For
$k=3$ and the five standard sketch dimensions, $\varepsilon_{M,d}$ falls below
$10^{-3}$ once $d/M$ reaches approximately $4.5$--$6.1$.  At the $d/M=8$
grid tier it ranges from $3.4\times10^{-8}$ to $3.7\times10^{-6}$, and at the
production-shaped load $d/M_1\approx182$ it is below $10^{-200}$.  On the
$d/M\le1.6$ grid tiers the bound is vacuous: near the peeling transition,
conditioning genuinely selects mappings and the proposition makes no claim
about the full conditional shape.  The distribution-free bound below and the
failed-only empirical quantiles of \S\ref{sec:exp-validation} cover that
range.  The supplementary material also records a complementary
proportional-growth direction in which failure need not become typical.

\begin{proposition}[Finite-sample failure-conditioned lower-tail bound]
\label{prop:lower-tail}
For any fixed sign pattern, any $p_F>0$, and any $q\in(0,1)$,
\[
  \Pr\!\left[\dhat \le q\,d \given F\right]
  \le
  \frac{2(d-1)}{d(M-1)\,p_F\,(1-q)^2}.
\]
\end{proposition}

The bound uses only the exact mean and variance of
Theorems~\ref{thm:unbiased}--\ref{thm:variance} and does not rely on any
asymptotic distribution; under random signs it can be combined with
Theorem~\ref{thm:cond-mean} to give a slightly tighter Cantelli bound.  The
proof, the explicit $q_\delta$, and the range of $(M,p_F,\delta)$ over which
the inversion returns a positive quantile at all are in the supplementary
material.

\subsection{Quantile interface for self-sizing}
\label{sec:quantile-interface}

The previous subsections characterize the distribution of $\dhat$ and its
failure-conditioned behavior.  This subsection consolidates them into the single
quantity the protocol needs: the failure-conditioned lower quantile
$q_\delta$.  It is the only interface the statistical layer hands to
\S\ref{sec:protocol}---the protocol needs this one number to configure the
second-round capacity, not the full distribution.

\begin{definition}[Failure-conditioned lower quantile]
\label{def:qdelta}
For the failure-conditioned distribution encountered by the protocol, a
lower quantile $q_\delta>0$ satisfies
\begin{equation}
\label{eq:qdelta}
  \Pr[\dhat \ge q_\delta\,d \given F] \ge 1-\delta.
\end{equation}
\end{definition}

Operationally: after a first-round failure, with at most $\delta$ risk of
underestimation, the estimate is at least a $q_\delta$ fraction of the
true value.  For example, $q_{0.01}=0.8$ means that with at least $99\%$
conditional probability, $\dhat\ge 0.8d$, so planning capacity as
$\dhat/0.8$ compensates for the lower-tail underestimation.  In practice $q_\delta<1$:
on the standard $M$ tiers calibrated in \S\ref{sec:exp-validation},
$q_{0.01}\in[0.63,0.95]$; smaller
sketches and heavier lower tails produce smaller $q_\delta$ and therefore
larger capacity multipliers.

We fix $\delta=0.01$ throughout the paper, corresponding to a $99\%$
conditional lower-tail target at runtime.

\emph{Obtaining $q_\delta$.}  Production configuration uses the failed-only empirical
quantile; \S\ref{sec:exp-validation} quantifies its sampling uncertainty from
the observed failed-sample counts.  In deep overload,
Proposition~\ref{prop:cond-chi2} explains why the unconditional chi-square
shape remains predictive after failure selection.  Its transfer term is
finite-sample and already negligible on the $d/M=8$ grid tier, where it is at
most $3.7\times10^{-6}$.  The limitation in obtaining a fully explicit
finite-$d$ chi-square calibration comes from the unconditional approximation:
the worst-case Berry--Esseen constant in Proposition~\ref{prop:kolmogorov} is
too conservative for engineering parameters.
When only a lower bound on
the failure probability $p_F\ge p_{F,\min}$ is known, or a distribution-free
guarantee is required, Proposition~\ref{prop:lower-tail} can be inverted to
obtain an analytic $q_\delta$.  That analytic bound is conservative at deep
lower tails (small $M$, $\delta=0.01$), where the inversion can return a
non-positive value and then certifies no usable quantile; where it is
positive it serves as a distribution-free theoretical floor.  The explicit
formula and its validity range are in the supplementary material.

\emph{Interface to the protocol.}  The statistical layer outputs only
$q_\delta$, converting the point estimate into the planning upper bound
$\dhat/q_\delta$.  Capacity sizing is one-sided: overestimation costs bytes,
underestimation risks a second round that is too small.
\S\ref{sec:protocol}'s Proposition~\ref{prop:capacity} composes this planning
upper bound with the decoder's capacity--success operating point and gives
the second-round capacity formula and the end-to-end success-probability
guarantee.

\subsection{Experimental validation}
\label{sec:exp-validation}

The claims of \S\ref{sec:count-measurement} are exact or asymptotic
statements about random mappings; the simulations check whether they hold in
the parameter range the protocol actually uses.  We organize the validation
around three questions.  First, do the unconditional mean, variance, and
quantiles of $\dhat$ match Theorems~\ref{thm:unbiased}--\ref{thm:chi2}?
Second, the protocol reads only after decoding failure---does this selection
step break the conditional unbiasedness of Theorem~\ref{thm:cond-mean}, and
does the failed-sample tail change shape?  Third, when the two sides are
extremely unbalanced, do the conditional-bias bound of
Theorem~\ref{thm:cond-mean} and the lower-tail bound of
Proposition~\ref{prop:lower-tail} cover the empirical values, leaving a safe
margin for the interface of \S\ref{sec:quantile-interface}?

\emph{Experimental grid.}  Simulations use
\[
  M\in\{64,256,512,1024,4096\},\quad
  k\in\{3,4\},\quad
  d/M\in\{0.4,0.8,1.6,8\},
\]
four fixed sign compositions $\theta_-=d_-/d\in\{0,0.1,0.5,0.9\}$ plus i.i.d.\
random signs: 200 configurations total, $10^6$ independent trials each.  The
four axes target different questions: $M$ controls the relative error in
Theorem~\ref{thm:variance}; $k$ tests whether the $k$-cancellation in the
normalized variance really occurs; $d/M$ extends from near-decodable to deep
overload; and $\theta_-$ covers sign compositions from fully one-sided to
balanced.  Separating $M$ from $d/M$ distinguishes two effects: a small
sketch has greater estimation variance, while the load controls how strongly
conditioning selects the failed subsample.  With $10^6$ trials per configuration, the
empirical quantiles are stable enough to report directly; for failed-only
quantiles, the uncertainty depends on the realized $n_F$ and is quantified
below.

The failure event $F$ is that the mapping hypergraph has a non-empty 2-core,
which depends only on the mapping, consistent with
Fact~\ref{fact:failure-measurable}.  We report tails at $\delta=0.01$; when
distinguishing the all-trial and failed-only samples, we write $q_\delta^U$
and $q_\delta^F$ respectively, where the superscripts denote the
all-trial (unconditional) and failed-only (failure-conditioned) semantics;
the interface $q_\delta$ of \S\ref{sec:quantile-interface} refers to the
latter.

\subsubsection*{All-trial (unconditional) accuracy
(\texorpdfstring{Theorems~\ref{thm:unbiased}--\ref{thm:chi2}}{Theorems 3.1--3.3})}

\begin{table}[!hbtp]
\centering
\caption{Plain estimator: unconditional accuracy (each configuration $10^6$ trials).}
\label{tab:t1a}
\small
\begin{tabular}{rrrrr}
\toprule
$M$ & $d$ & mean($\dhat/d$) & RSD: measured / theory & $q_{0.01}^{U}$: measured / $\chi^2$ / normal \\
\midrule
64    & 102   & 0.9996 & 0.1774 / 0.1773 & 0.6309 / 0.6326 / 0.5875 \\
256   & 410   & 1.0000 & 0.0885 / 0.0885 & 0.8062 / 0.8056 / 0.7942 \\
512   & 819   & 1.0000 & 0.0625 / 0.0625 & 0.8601 / 0.8602 / 0.8545 \\
1024  & 1638  & 0.9999 & 0.0442 / 0.0442 & 0.9001 / 0.9000 / 0.8972 \\
4096  & 6554  & 1.0000 & 0.0221 / 0.0221 & 0.9493 / 0.9493 / 0.9486 \\
\bottomrule
\end{tabular}
\end{table}

Table~\ref{tab:t1a} takes the representative slice $k=3$, $\theta_-=0.5$,
$d=1.6M$ and compares the empirical results with the three theoretical
predictions along the $M$
axis.  The empirical mean of $\dhat/d$ stays within $5\times 10^{-4}$ of~1
at every tier, consistent with exact unbiasedness; the measured RSD matches the
Theorem~\ref{thm:variance} closed form to four significant figures; and the
empirical $1\%$ quantile agrees with the chi-square prediction to thousandths
while the symmetric normal approximation falls systematically low.

\begin{figure}[!hbtp]
\centering
\includegraphics[width=\textwidth]{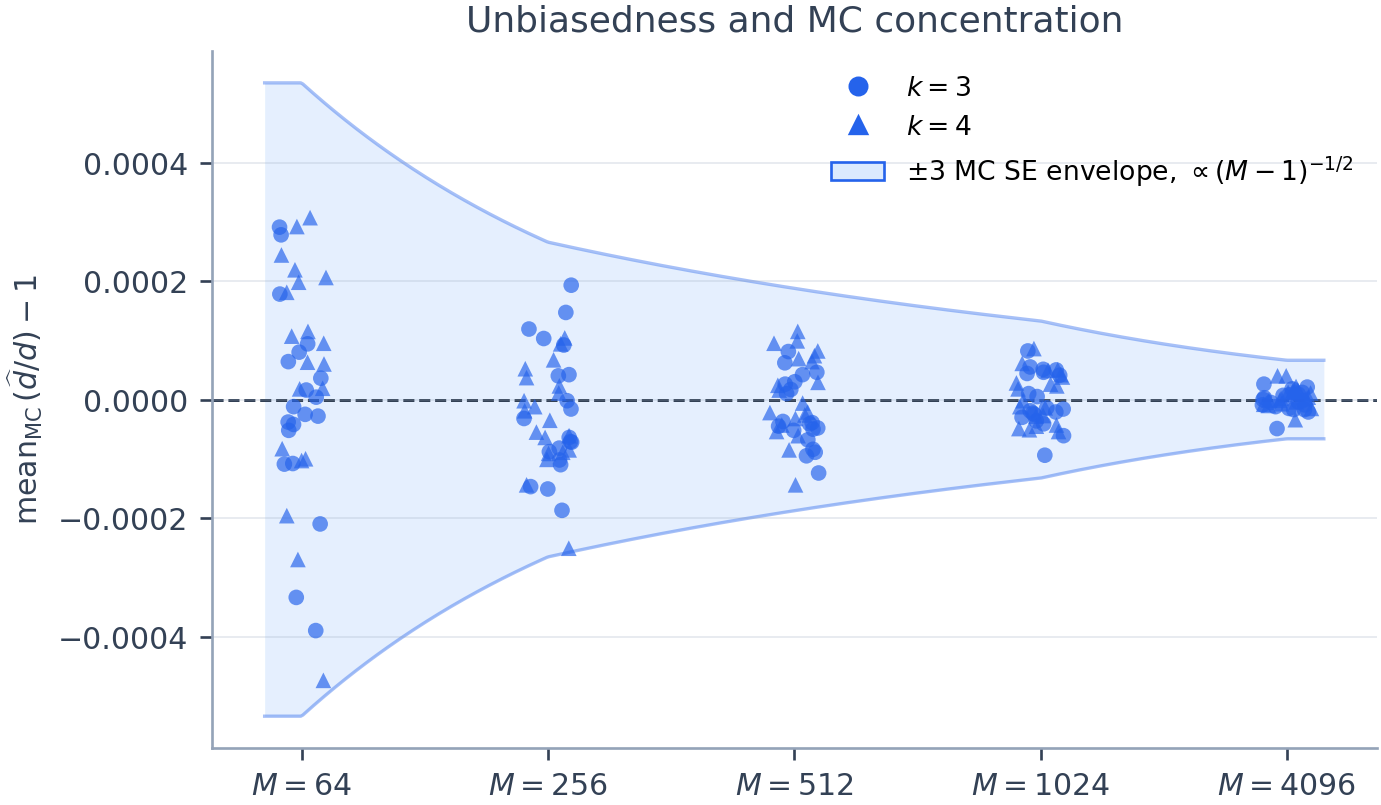}
\caption{Theorem~\ref{thm:unbiased} and Corollary~\ref{cor:rsd}: exact
unbiasedness and the Monte Carlo envelope shrinking with $M$.  Each point is
one of the 200 configurations; the dashed line at 0 is the exact unbiased
value, and the light-blue envelope is $\pm 3\sqrt{2/((M-1)\times10^6)}$,
three Monte Carlo standard errors.}
\label{fig:f3a}
\end{figure}

The mean and variance conclusions hold across the whole grid, not just this
slice: the maximum deviation from unbiasedness is $4.7\times10^{-4}$, and the
measured RSD to closed-form ratio lies in $[0.9976,1.0018]$, with no
systematic drift in $k$, load, or sign composition
(Figures~\ref{fig:f3a}, \ref{fig:f3b} and Table~\ref{tab:d1}).  The
same-configuration RSDs for $k=3$ and $k=4$ coincide tier by tier,
consistent with the
$k$-cancellation of Theorem~\ref{thm:variance}; choosing $k$ for peeling
does not require recalibrating the estimator's mean and RSD.

The last column assesses the accuracy of the analytic quantiles.  The chi-square prediction tracks the empirical quantile to
thousandths at all five tiers; the normal approximation undershoots by 4.3
percentage points at $M=64$, fading to tenths of a point at $M=4096$.
The error is on the conservative side---it inflates capacity rather than
risking under-capacity---so even the cruder approximation is safe, just
wasteful.  The formal protocol still uses the failure-conditioned quantile
of \S\ref{sec:quantile-interface}; this comparison evaluates the analytic
quantile's accuracy.

At $M=64$ a single observation has $18\%$ RSD and a $1\%$ quantile of
about $0.63d$; at $M=1024$ these improve to $4.4\%$ and $0.90d$.  A
smaller first-round sketch preserves unbiasedness while producing a wider
confidence interval.

Table~\ref{tab:t1a} compares a single quantile; Figure~\ref{fig:f1}
compares the whole density along two orthogonal directions, separating the
convergence to the chi-square limit as $d$ grows from the precision loss as
$M$ shrinks.  The top row fixes $M=64$ and increases $d/M$, moving along the
$d\to\infty$ direction of Theorem~\ref{thm:chi2}: the KS distance to the
parameter-free $\chi^2_{M-1}/(M-1)$ density drops from $0.017$ to $0.001$
while the measured RSD stays at $0.175$--$0.178$.  The bottom row fixes
$d=4096$ and shrinks $M$ from 10240 to 512: the KS distance stays below
$0.003$ while the RSD grows from $0.014$ to $0.063$ exactly as
Theorem~\ref{thm:variance} predicts.  Deep overload therefore introduces no
additional estimation bias; the observed precision changes are explained by
the cell budget $M$ alone.

\begin{figure}[!hbtp]
\centering
\includegraphics[width=\textwidth]{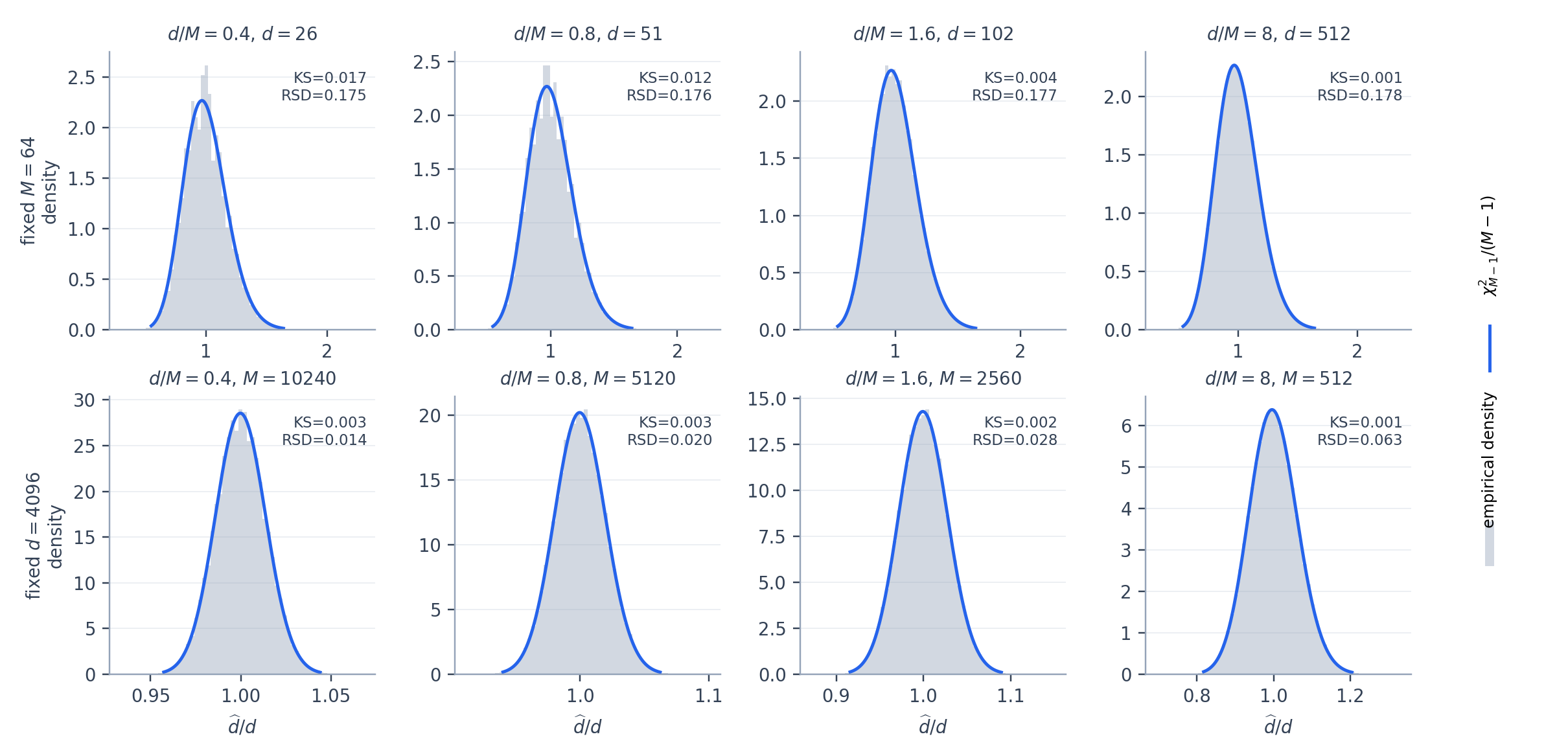}
\caption{Two orthogonal distribution slices.  Top row: fix $M=64$ and
increase $d/M$ (direction of Theorem~\ref{thm:chi2}); KS distance drops from
$0.017$ to $0.001$ while RSD stays at $0.175$--$0.178$.  Bottom row: fix
$d=4096$ and increase $d/M$ by shrinking $M$ from 10240 to 512; KS stays
below $0.003$ while RSD grows from $0.014$ to $0.063$.  Gray histograms are
empirical densities of $\dhat/d$; blue lines are the parameter-free
$\chi^2_{M-1}/(M-1)$ densities.}
\label{fig:f1}
\end{figure}

\subsubsection*{Failure-conditioned measurement
(Theorem~\ref{thm:cond-mean} and Proposition~\ref{prop:cond-chi2})}

\begin{table}[!hbtp]
\centering
\caption{Estimator accuracy on decoding-failed samples (i.i.d.\ random
signs, each configuration $10^6$ trials).}
\label{tab:t1b}
\small
\begin{tabular}{rrrrrl}
\toprule
$M$ & $d$ & $\Pr[F]$ & $n_F$ & $\E[\dhat/d \given F]$ & \shortstack{$q_{0.01}^{F}$ / unconditional\\$\chi^2$ prediction} \\
\midrule
64    & 51    & 0.683 & 683{,}202 & 0.9998 & 0.6308 / 0.6326 \\
256   & 205   & 0.543 & 543{,}451 & 1.0001 & 0.8050 / 0.8056 \\
512   & 410   & 0.422 & 421{,}758 & 1.0001 & 0.8598 / 0.8602 \\
1024  & 819   & 0.257 & 257{,}079 & 1.0001 & 0.8998 / 0.9000 \\
4096  & 3277  & 0.028 & 28{,}336 & 0.9999 & 0.9490 / 0.9493 \\
\bottomrule
\end{tabular}
\end{table}

The second group restricts the sample to first-round failures and checks
whether the part the protocol actually reads is still trustworthy.
Table~\ref{tab:t1b} uses $k=3$, $d=0.8M$, and i.i.d.\ random signs---a
load near the peeling phase transition where both successes and failures
occur.  This is deliberately outside the deep-overload regime certified by
Proposition~\ref{prop:cond-chi2}, and is therefore a stronger empirical
stress test of failure selection.  The conditional means lie in
$[0.9998,1.0001]$, indistinguishable
from~1 and insensitive to failure rates ranging from $3\%$ to $68\%$,
consistent with Theorem~\ref{thm:cond-mean}: the failure event depends only
on the mappings and is independent of the random
signs.  The failure-conditioned $1\%$ quantile $q_{0.01}^F$ differs from
the unconditional chi-square prediction by at most $0.002$, so the two
lower-tail shapes remain close.  Evaluating the
unconditional chi-square $1\%$ cutoff directly on the failed samples gives
conditional tail probabilities $1.042\%$, $1.014\%$, $1.028\%$, $1.016\%$, and
$1.059\%$ across the five $M$ tiers; their binomial standard errors are
$0.012$, $0.014$, $0.016$, $0.020$, and $0.061$ percentage points, respectively.
At $M=64$, the $0.042$-percentage-point excess is detectable at about $3.4$
Monte Carlo standard errors and places the analytic cutoff slightly on the
optimistic side.  It is $4.2\%$ of the one-percentage-point tail budget.
Production configuration uses the failed-only empirical quantile, so this
finite-$d$ discrepancy does not enter capacity planning.  At $M=512$, the
excess is $0.028$ percentage points, or about $1.8$ Monte Carlo standard
errors, representing a modest finite-sample fluctuation.  At $M=256$ and
$M\ge1024$, the deviations are within approximately one Monte Carlo standard
error.  The detrended Q--Q
comparison in Figure~\ref{fig:fcond-qq} extends the same phenomenon to the
whole $1\%$--$99\%$ quantile range.  Collapsing the five ordinary Q--Q
panels onto the common difference axis
$q_{\mathrm{failed}}(p)-q_{\mathrm{uncond}}(p)$ makes the small departures
visible: the four larger $M$ tiers remain tightly concentrated around zero in
the full quantile comparison, while the $M=512$ tail shows the modest
finite-sample fluctuation quantified above and $M=64$ shows the pronounced
finite-$d$ tail deviation.

\begin{figure}[!hbtp]
\centering
\includegraphics[width=\textwidth]{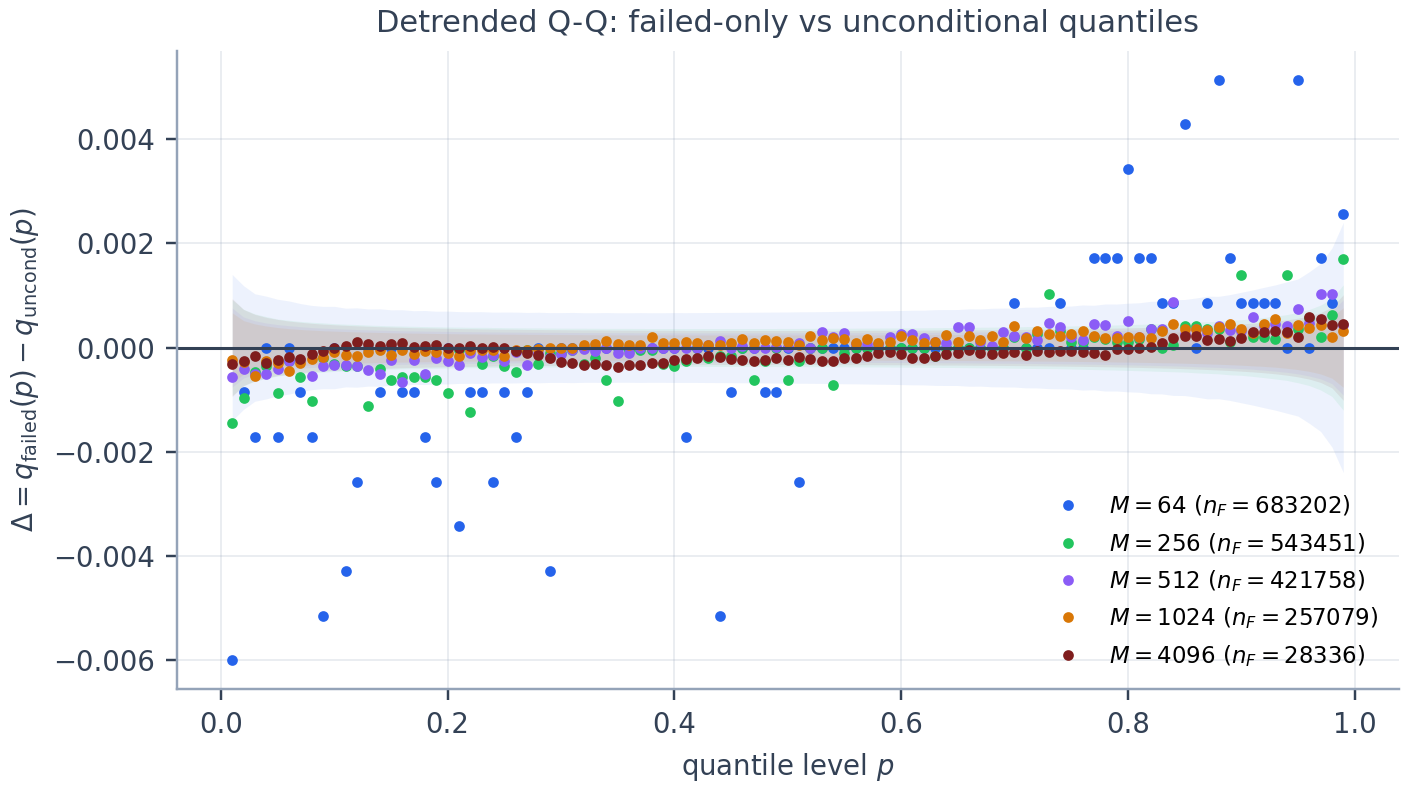}
\caption{Detrended failure-conditioned versus unconditional quantiles for the
transition-regime configurations of Table~\ref{tab:t1b}.  Each series plots
$q_{\mathrm{failed}}(p)-q_{\mathrm{uncond}}(p)$ over the $1\%$--$99\%$
range; zero means that conditioning leaves the quantile unchanged.  Shaded
bands show the pointwise $\pm1.96$ standard-error range under equal
distributions.}
\label{fig:fcond-qq}
\end{figure}

Over the full random-sign grid (not just this slice), for the
non-degenerate configurations ($0<\widehat p_F<1$ and at least $10^3$
failed samples), the maximum absolute deviation of the conditional mean is
$0.19\%$, consistent with Monte Carlo standard error.  The transition-regime
tail comparison is empirical evidence beyond Proposition~\ref{prop:cond-chi2}'s
certified range.  The failed-sample counts in Table~\ref{tab:t1b} range from
$2.8\times10^4$ to $6.8\times10^5$; the one-sided $95\%$ order-statistic
lower confidence bound for the empirical quantile differs by at most $0.002$
(Table~\ref{tab:d3}).  The certified tier provides the complementary check:
at $d/M=8$, all $10^6$ trials failed in each of the five random-sign $k=3$
configurations.  This agrees with the analytic success-probability bound
$\varepsilon_{M,d}\le3.7\times10^{-6}$.

\subsubsection*{Fixed-sign projection and robustness
(Theorem~\ref{thm:cond-mean} and Proposition~\ref{prop:lower-tail})}

Unbalanced difference compositions are common in practice (e.g., inserts on
one side only).  Then the failure event and the sign composition act
through the exact affine relation in Eq.~\eqref{eq:sign-projection}.  We test
this relation on a sweep with $k=3$, $M\in\{256,1024,4096\}$, five
transition-region loads per $M$, and 11 fixed values of $\theta_-$: 165
configurations and 126.5 million trials in total.  For each of the 15
$(M,d)$ curves, a constrained one-parameter fit of conditional-mean bias to
$d\theta_\Delta^2-1$ estimates its common $c_F$.  The minimum and mean
$R^2$ are $0.9948$ and $0.9990$, respectively.

\begin{figure}[!hbtp]
\centering
\includegraphics[width=\textwidth]{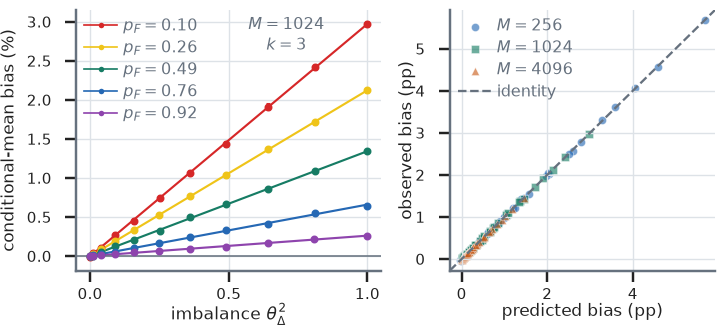}
\caption{Exact failure-conditioned sign projection.  Left: for $M=1024$,
each transition-region load traces an affine conditional-mean bias as a
function of squared sign imbalance $\theta_\Delta^2$.  Right: leave-one-out
predictions versus observed bias across all 165  configurations; each
point is predicted from the other ten sign compositions at the same
$(M,k,d)$.}
\label{fig:sign-projection}
\end{figure}

The leave-one-out audit in Figure~\ref{fig:sign-projection} has RMSE
$1.08\times10^{-4}$ and maximum absolute residual $3.36\times10^{-4}$.
All fitted $c_F$ values are positive in this $k=3$ grid, although
Corollary~\ref{cor:cf-regimes}(i) shows that the sign is not universal.
The experiment concerns
the conditional mean only and does not extend the squared-imbalance relation
to the full failure-conditioned distribution.

The largest observed bias is $+5.69\%$, at the fully one-sided endpoint with
$M=256$ and $p_F=0.139$; it remains below that configuration's $8.8\%$ RSD.
At the same $M$, increasing the load until $p_F=0.932$ reduces the one-sided
bias to $+0.45\%$, consistent with the vanishing envelope as failure becomes
typical.  The balanced fixed-sign endpoints remain near zero at the much
smaller $-c_F$ scale predicted by Eq.~\eqref{eq:sign-projection}.

The earlier 52-configuration audit remains a check of the sign-agnostic
Cauchy--Schwarz envelope implied by Theorem~\ref{thm:cond-mean}; none of its
non-degenerate fixed-sign configurations violates that envelope
(Table~\ref{tab:d4}), and the largest empirical deviation is 31\% of the
corresponding bound.  The appendices additionally verify
Proposition~\ref{prop:lower-tail}'s lower-tail bound on 148 configurations
(Table~\ref{tab:d5}), again with no empirical violation and a minimum coverage
margin of 0.17.  The absolute-deviation and lower-tail bounds provide the
formal guarantees the protocol uses.

Together, the three groups show that the unconditional moments and quantiles,
the failure-conditioned sign projection and lower tail, and the fixed-sign
bounds all agree with the theory developed in this section, and
no systematic underestimation appears that would break second-round capacity
sizing.  The interface of \S\ref{sec:quantile-interface} can therefore use
empirical $q_\delta$, and the capacity multipliers and protocol-layer success
rates derived from it are reported in Appendix~\ref{app:protocol-details}.

\section{Mapping-Aware Constructions: The Unifying Theorem and Four
Variant Adapters}
\label{sec:mapping-aware}

The results of \S\ref{sec:count-measurement} rest on the Plain $k$-regular
uniform mapping.  Other members of the IBLT family change the mapping
statistics: Irregular introduces random degrees, the coded-symbol
sequence of Rateless IBLTs makes hit probabilities decay with position, and MET
partitions cells into types with heterogeneous hit rates.  The natural
question is whether the measurement capability of the count channel is a
coincidence of $k$-regular mappings or a general property of count-bearing
linear sketches.

An \emph{adapter} specifies how to center and normalize the count array
of a particular IBLT variant so that the estimator remains unbiased.
It does not change the IBLT storage format; different variants need
different adapters because their hashing rules produce different mapping
statistics.

The section proceeds in three steps: a mapping-aware construction theorem
giving sufficient conditions for the estimator to exist
(Theorem~\ref{thm:master}); concrete adapters for Irregular, Rateless,
and MET (\S\ref{sec:irregular}--\S\ref{sec:met}), which together with
Plain's adapter from \S\ref{sec:count-measurement} make up the four
adapters of the section title; and a unified failure-conditioning analysis
(Corollary~\ref{cor:variant-failure}, \S\ref{sec:variant-failure-cond}).

\begin{theorem}[Mapping-aware construction / adapter theorem]
\label{thm:master}
Elements map independently; the hit vector $a_x$ of a type-$j$ element has
mean $\mu_j$ and covariance $\Sigma_j$.  Take a symmetric matrix $Q$ with
$Q\mu_j=0$ for every type.  Then
\begin{equation}
\label{eq:master}
  \E[C^\top Q C] = \sum_j d_j\,\tr(Q\Sigma_j).
\end{equation}
If there is a single type with $\gamma=\tr(Q\Sigma)\ne 0$, this reduces to
$\E[C^\top Q C]=d\gamma$, and the estimator
\[
  \dhat = \frac{C^\top Q C}{\tr(Q\Sigma)}
\]
is unbiased.
\end{theorem}

Here $\mu_j$ describes where a type-$j$ element writes its count on average,
and $\Sigma_j$ describes how these writes fluctuate around the average
pattern.  $Q$ deletes the stable average write pattern, and
$\tr(Q\Sigma_j)$ converts the remaining centered energy back into ``number of
elements.''  The core of Theorem~\ref{thm:master} is that, as long as these
two steps agree with the true mapping statistics, count energy continues to
measure the difference cardinality.

\noindent\textit{Proof.}  See Appendix~\ref{app:thm6}.

Measurement capability depends only on the mapping mean $\mu_j$ and
covariance $\Sigma_j$, so adapters can be constructed systematically from
these two quantities.  When specialized to Plain's $k$-regular mapping,
$\tr(Q\Sigma)$ reduces to $\gamma=k(1-k/M)$, recovering
\S\ref{sec:count-measurement}.  The full expectation decomposition showing
why $Q$ must remove the true mean is in Appendix~\ref{app:thm6}.

\emph{Design rule.}  Constructing a correct adapter reduces to two steps:
\textbf{first center within the smallest mapping-homogeneous block}
(removing the drift, $Q\mu_j=0$), \textbf{then normalize by the true mapping
covariance} ($\gamma=\tr(Q\Sigma)$).  For Plain the block is the full cell array;
Irregular keeps the global centering but uses the first and second
moments of the degree; Rateless defines the baseline from the per-symbol hit
probabilities $q_i$; MET centers per cell type.

\subsection{Adapter overview and the Irregular adapter}
\label{sec:irregular}

Following the design rule, the four variants' adapters are derived from their
own mapping statistics (Table~\ref{tab:adapter-overview}):

\begin{table}[!hbtp]
\centering
\caption{Adapter overview: what baseline each variant centers out and how it
normalizes.}
\label{tab:adapter-overview}
\footnotesize
\setlength{\tabcolsep}{3pt}
\begin{tabular}{@{}p{1.8cm}p{3.8cm}p{3.0cm}p{5.0cm}p{1.6cm}@{}}
\toprule
variant & mapping statistics & $Q$ (centering) & $\gamma=\tr(Q\Sigma)$ & derived in \\
\midrule
Plain $k$-regular & uniform $k$-subset, $\mu=(k/M)\one_M$ & global centering & $k(1-k/M)$, exact & \S\ref{sec:count-measurement} \\
Irregular & random degree $D$, uniform positions given degree & global centering & $\E[D]-\E[D^2]/M$ & \S\ref{sec:irregular} \\
Rateless & per-symbol hit probability $q_i$ decreasing & per-symbol centering $Y_i=C_i-q_iC_0$ & per-symbol $q_i(1-q_i)$, $\tr(Q\Sigma)=1$ & \S\ref{sec:rateless} \\
MET & hit rates heterogeneous across cell types & center \emph{per cell type} & $\gamma_{tj}=a_{tj}(m_t-a_{tj})(r_t-1)/[m_t(m_t-1)]$; $r_t=m_t$: $a_{tj}-a_{tj}^2/m_t$ & \S\ref{sec:met} \\
\bottomrule
\end{tabular}
\end{table}

Irregular is the first instance of the design rule: the centering block is
identical to Plain, and only the normalization constant changes.  It is also
the shortest construction among the three new variants; we give the unbiased
adapter first, then characterize the exact variance introduced by the random
degree.

\begin{proposition}[Irregular degree second-moment adapter]
\label{prop:irregular}
In an Irregular IBLT, each element independently draws a random degree
$D_x\sim\Lambda$, where $\Lambda$ is given by its generating polynomial
$\Lambda(x)=\sum_i\lambda_i x^i$ with $\lambda_i=\Pr[D_x=i]$, and $D_x\le M$
almost surely.  Given $D_x$, the element chooses $D_x$ distinct positions
uniformly without replacement among the $M$ cells.  With the Plain global
centering matrix $Q=I_M-\one_M\one_M^\top/M$, let
\begin{equation}
\label{eq:gamma-irr}
  \gamma_{\mathrm{irr}} = \E[D] - \frac{\E[D^2]}{M}.
\end{equation}
If $\gamma_{\mathrm{irr}}\ne 0$, then
$\dhat_{\mathrm{irr}}=C^\top Q C/\gamma_{\mathrm{irr}}$ is unbiased:
$\E[\dhat_{\mathrm{irr}}]=d$.
\end{proposition}

\noindent\textit{Proof.}  See Appendix~\ref{app:irregular}.

\begin{remark}
Compared with Plain, the centering block is still the full IBLT cell array, but
the normalization constant needs both the first and second moments of the
degree.  A natural but incorrect shortcut is to treat an Irregular IBLT sketch
as a regular IBLT sketch with $k=\E[D]$, plugging in
$\gamma_{\mathrm{mean}}=\E[D]-\E[D]^2/M$.  This drops the $\Var(D)/M$ term and
induces a constant multiplicative bias
$\gamma_{\mathrm{irr}}/\gamma_{\mathrm{mean}}<1$---systematic
underestimation---which vanishes as $M\to\infty$ but is significant for
small $M$.  Proposition~\ref{prop:irregular} eliminates this pitfall by
introducing the second moment.
\end{remark}

\begin{corollary}[Exact variance of the Irregular adapter]
\label{cor:irregular-var}
Under the conditions of Proposition~\ref{prop:irregular}, write the per-element
self-energy $e_x=D_x-D_x^2/M$, whose variance $\sigma_E^2=\Var(e_x)$ depends
only on the first four moments of the degree distribution $\Lambda$.  Then,
for any fixed sign pattern,
\begin{equation}
\label{eq:irr-variance}
  \Var\!\left(\frac{\dhat_{\mathrm{irr}}}{d}\right)
  = \frac{2(d-1)}{d(M-1)}
  + \frac{\sigma_E^2}{\gamma_{\mathrm{irr}}^2\,d}.
\end{equation}
\end{corollary}

The first term equals the Plain exact variance \eqref{eq:variance} (random
overlap between different elements); the second term comes from the
Irregular-specific random degree self-energy.  When $D$ is constant,
$\sigma_E^2=0$ and the formula reduces to Theorem~\ref{thm:variance}.  The
absolute-variance form, the expansion of $\sigma_E^2$ in terms of
$\E[D^2],\dots,\E[D^4]$, and the full proof are in
Appendix~\ref{app:irregular}.

\begin{table}[!hbtp]
\centering
\caption{Irregular adapter validation ($d=16384$, 20{,}000 trials per
configuration; RMSE averaged over seven sign compositions).}
\label{tab:irregular}
\small
\begin{tabular}{rrrr}
\toprule
$M$ & mean$(\dhat_{\mathrm{irr}}/d)$ range & measured RMSE (avg) & Cor.~\ref{cor:irregular-var} theory RMSE \\
\midrule
18   & $0.9962$--$1.0002$ & $0.34283$ & $0.34300$ \\
50   & $0.9981$--$1.0012$ & $0.20239$ & $0.20212$ \\
1024 & $0.9998$--$1.0002$ & $0.04488$ & $0.04497$ \\
\bottomrule
\end{tabular}
\end{table}

\emph{Validation.}  We use the degree distribution of the original Irregular
IBLT paper~\cite{lazaro2021irregular},
$\Lambda(x)=0.15x^2+0.725x^3+0.125x^{18}$, and fix $d=16384$.  The
main text reports three of the eight swept values of $M$: $M=18$
(an intentionally extreme configuration in which $12.5\%$ of elements hit
every cell, and the mean-degree plug-in underestimates by $40.3\%$, while
the second-moment normalization stays centered), $M=50$ (the second-moment
correction is
about $12\%$), and $M=1024$ (a realistic sketch dimension, correction $0.5\%$).
Each tier covers seven sign compositions
$\theta_-=d_-/d\in\{0,0.1,0.25,0.5,0.75,0.9,1\}$ with 20{,}000 trials each;
Corollary~\ref{cor:irregular-var} predicts a single theoretical RMSE per
$(M,d,\Lambda)$ triple, independent of sign balance.

The relative differences between measured and closed-form RMSE are at most
$0.22\%$ (Table~\ref{tab:irregular}), and the means show no systematic drift
with $\theta_-$.  The
results validate both structures: the second-moment normalization of
Proposition~\ref{prop:irregular} keeps the estimate centered at $d$, and
Corollary~\ref{cor:irregular-var} explains all second-order fluctuation of a
single estimate without fitting any variance parameter.  The same closed form
holds on a decode-scale batch at $d=1024$ (measured relative RMSE
$4.6\%$--$6.0\%$, closed form $5.5\%$ at $M=1024$).

\subsection{Rateless prefix analytic adapter}
\label{sec:rateless}

Rateless IBLT does not fix a sketch length in advance; the sender appends
coded cells.  Let $m$ be the number of coded cells received so far (the
prefix length); cell 0 is hit by every mapping path and serves only as the
centering baseline, so the usable measurement components number $r=m-1$.
This gives its adapter a property Plain and Irregular do not have: every new
batch of received cells linearly increases the measurement degrees of
freedom, and the relative variance decays as $O(1/m)$---the more cells are
received, the more accurate the estimate of $d$.  The same estimation formula
therefore supports measuring as cells arrive, without committing to a
capacity in advance.  The $O(1/m)$ rate is confirmed by
Proposition~\ref{prop:rateless-ci} and the simulations at the end of this
subsection.

Each Rateless mapping path jumps forward from cell 0.  Following the
reference sampler of the original paper~\cite{yang2024rateless}, the step
size is
\[
  D_j=\left\lceil\left(j+\frac32\right)\left(U^{-1/2}-1\right)\right\rceil,
  \qquad U\sim\operatorname{Uniform}(0,1),
\]
and a path at cell $j$ jumps to $v=j+D_j$.  All lemmas and propositions below
are derived analytically from this formula, so the adapter requires no
offline parameter fitting.

The question is whether each fixed prefix (the first $m$ coded cells) yields
an unbiased $\dhat_m$, and how its precision is characterized as $m$ grows.
Cell 0 is hit by every path, so $C_0$ records the net difference count of the
two sides; the formulas use it to remove the different average hit baseline
of each cell.  The jump rule integrates to a closed-form transition kernel
$K_{jv}$ (Appendix~\ref{app:rateless}).

Constructing the adapter requires the mapping mean (Lemma~\ref{lem:rateless-hit})
and covariance (Lemma~\ref{lem:rateless-pairwise}).  Pairwise independence
ensures diagonal covariance, which is critical: without it the receiver would
need to update a pseudo-inverse on every appended cell.

\begin{lemma}[Closed-form hit probability]
\label{lem:rateless-hit}
Let $B_i$ indicate whether one path hits cell $i$, and let
$q_i=\Pr[B_i=1]$.  Then
\[
  q_0=1,\qquad q_i=\frac{8(i+1)}{(2i+3)^2},\quad i\ge1.
\]
\end{lemma}

\noindent\textit{Proof.}  See Appendix~\ref{app:rateless}.

The closed form lets the receiver compute centering coefficients directly
from the prefix index, with no offline sampling or parameter fitting; its
agreement with the official Rateless sampler is verified in the simulations
below.

\begin{lemma}[Pairwise independence]
\label{lem:rateless-pairwise}
In the fresh-uniform model (each jump draws an independent $U$), for any
$u<v$,
\[
  \Pr[B_v=1 \given B_u=1]=q_v,
\]
so for $0<i<j$ we have $\Pr[B_i=1,B_j=1]=q_iq_j$: the hit indicators inside
a prefix are pairwise independent.
\end{lemma}

\noindent\textit{Proof.}  See Appendix~\ref{app:rateless}.

\begin{proposition}[Analytic estimator]
\label{prop:rateless-estimator}
Let $m\ge 2$ and suppose prefix cells $0,1,\dots,m-1$ have been received.
Set $Y_i=C_i-q_iC_0$ for $i=1,\dots,m-1$.  Then
\begin{equation}
\label{eq:rateless-est}
  \dhat = \frac{1}{m-1}\sum_{i=1}^{m-1}
  \frac{Y_i^2}{q_i(1-q_i)}
\end{equation}
is exactly unbiased: $\E[\dhat]=d$.
\end{proposition}

\noindent\textit{Proof.}  See Appendix~\ref{app:rateless}.

Thus a Rateless prefix offers a sequence of estimates that improve as cells
arrive: each new batch of cells adds a new normalized energy
component.  As the effective degrees of freedom $r=m-1$ grow, $\dhat_m$
concentrates on the true value $d$.

\begin{proposition}[Asymptotic confidence interval]
\label{prop:rateless-ci}
For a fixed prefix length $m\ge 2$, as $d\to\infty$,
\begin{equation}
\label{eq:rateless-ci}
  \dhat/d \Longrightarrow \chi^2_r/r,\qquad r=m-1,
\end{equation}
and the asymptotic $1-\rho$ confidence interval is
\[
  \left[\frac{r\,\dhat}{\chi^2_{r,1-\rho/2}},\
  \frac{r\,\dhat}{\chi^2_{r,\rho/2}}\right].
\]
\end{proposition}

\noindent\textit{Proof.}  See Appendix~\ref{app:rateless}.

Only pairwise independence is needed: the central limit theorem is applied
across elements, each element contributing an independent path, and the
limiting Gaussian is determined uniquely by the covariance matrix; pairwise
moments suffice.  $\Sigma$ is diagonal, so the limiting coordinates are
automatically independent.  Mutual independence is neither assumed nor needed.

This is a fixed-prefix, large-$d$ asymptotic interval, the same semantics as
Theorem~\ref{thm:chi2} for Plain: the receiver chooses a prefix length and
constructs a single confidence interval at that point; finite $(d,m)$
quantiles are verified experimentally.  Repeated inspection would require an
optional-stopping correction, which we leave for future work.

\emph{Implementation boundary.}  Lemmas~\ref{lem:rateless-hit}--\ref{lem:rateless-pairwise}
and Propositions~\ref{prop:rateless-estimator}--\ref{prop:rateless-ci} assume
each jump draws a fresh independent $U$, while the reference implementation
derives the entire path's jump distances from one 64-bit PRNG, making
consecutive $U$'s deterministic functions of the same seed.  The analytic
results apply to the fresh-uniform model; we separately test whether the
pinned reference implementation reproduces the required hit probabilities and
pairwise moments.

\subsubsection*{Rateless simulation validation}

Validation has two layers: first, whether the reference sampler matches the
analytic mapping statistics of
Lemmas~\ref{lem:rateless-hit}--\ref{lem:rateless-pairwise}; second,
whether the estimator remains centered and the intervals achieve near-nominal
coverage at the tested finite $(d,m)$ values.

\emph{Mapping statistics.}  The empirical hit probabilities and pairwise
moments agree with the analytic independence prediction of
Lemmas~\ref{lem:rateless-hit}--\ref{lem:rateless-pairwise}, for both the
analytic kernel and the official pinned sampler ($10^9$ paths each);
details are in Appendix~\ref{app:rateless}.

\emph{Point estimates and intervals.}  The finite-sample experiment fixes
$d=16384$ and takes prefixes $m\in\{64,256,1024\}$, covering seven difference
fractions $\{0,0.1,0.25,0.5,0.75,0.9,1\}$; each configuration runs 100{,}000
trials, $2{,}100{,}000$ observations in total.  The table
(Table~\ref{tab:rateless}) reports min--max
over the seven sign compositions; $r=m-1$.

\begin{table}[!hbtp]
\centering
\caption{Rateless prefix estimation ($d=16384$, 100{,}000 trials per
configuration; intervals are min--max over seven sign compositions).}
\label{tab:rateless}
\footnotesize
\begin{tabular}{rrrrrrrr}
\toprule
prefix $m$ & $r$ & mean($\dhat/d$) & RMSE($\dhat/d$) & $\sqrt{2/r}$ &
90\% cov. & 95\% cov. & 99\% cov. \\
\midrule
64   & 63   & 0.9993--1.0004 & 0.1775--0.1785 & 0.178 & 0.8993--0.9016 & 0.9490--0.9508 & 0.9897--0.9909 \\
256  & 255  & 0.9996--1.0008 & 0.0885--0.0888 & 0.089 & 0.8987--0.9009 & 0.9491--0.9505 & 0.9895--0.9903 \\
1024 & 1023 & 0.9999--1.0003 & 0.0442--0.0444 & 0.044 & 0.8968--0.8992 & 0.9487--0.9506 & 0.9895--0.9901 \\
\bottomrule
\end{tabular}
\end{table}

Point estimates are centered at~$1$ regardless of sign composition;
quadrupling the prefix length halves the RMSE, tracking the $\sqrt{2/r}$
scale of Proposition~\ref{prop:rateless-ci} tier by tier (the point
estimates and RMSE against the theory line are in
Figure~\ref{fig:r1}, Appendix~\ref{app:rateless}).
\begin{figure}[!hbtp]
\centering
\includegraphics[width=0.8\textwidth]{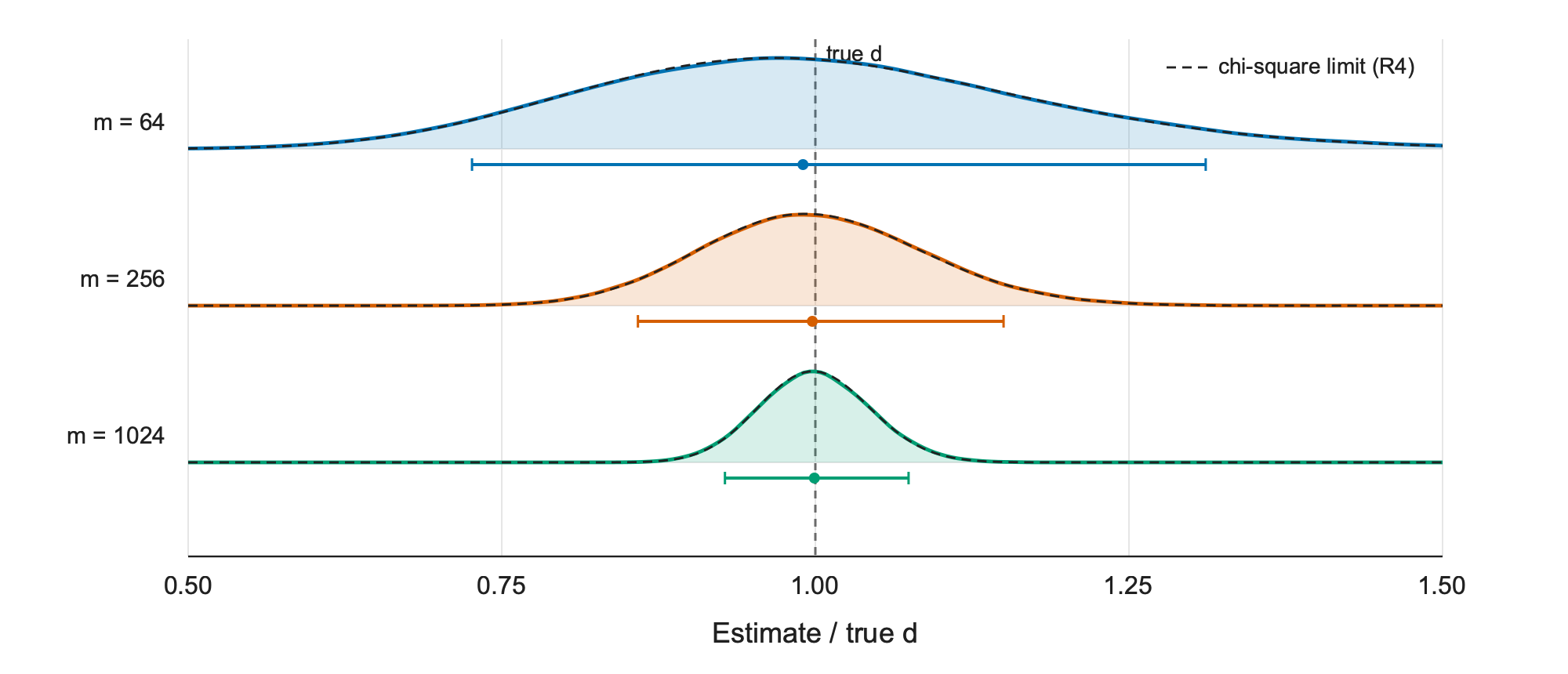}
\caption{Empirical distribution of $\dhat_m/d$ (solid) versus the chi-square
limit $\chi^2_r/r$ (dashed), with the central 90\% interval marked.  The
distribution visibly right-skews at $m=64$ and approaches
symmetry as $m$ grows to 1024.}
\label{fig:r2}
\end{figure}
Figure~\ref{fig:r2} juxtaposes the full
empirical distribution with the Proposition~\ref{prop:rateless-ci}
chi-square limit.

Across all 63 prefix--sign-composition--confidence-level cells, coverage is
within 0.32 percentage points of nominal.  The finite-sample remainder is
most visible at the smallest $d/r$ tier: for $m=1024$, all seven 90\%
coverage values lie 0.08--0.32 percentage points below nominal (0.17 points
on average), whereas the maximum absolute deviations over all confidence
levels are 0.16 and 0.13 points at $m=64$ and $256$, respectively.  A
fresh-uniform control reproduces the $m=1024$ gap, excluding the pinned PRNG
recursion as its source and identifying it with the finite-$(d,m)$ remainder
of the asymptotic approximation; Appendix~\ref{app:rateless} gives the
full comparison.

Both layers pass: the reference sampler reproduces the analytic mapping
statistics, and, with no fitted parameters, the purely analytic estimator
remains centered and the intervals achieve near-nominal coverage at the
tested finite $(d,m)$ values.

\begin{remark}[Implementation pitfall]
The average-of-ratios statistic \eqref{eq:rateless-est} and the
chi-square interval \eqref{eq:rateless-ci} must be used as a pair.
Replacing the former with a ratio-of-sums keeps the point estimate
unbiased but breaks the interval: the empirical coverage of the
nominal-90\% interval drops to about 50\%.  Details are in
Appendix~\ref{app:rateless}.
\end{remark}

\subsection{MET partial-type adapter}
\label{sec:met}

MET-IBLT~\cite{lazaro2023ratecompatible} partitions cells into types with heterogeneous hit rates and keys
into data types.  Following the design rule, its adapter centers
\emph{per cell type}: within type $t$, first
center over the $r_t$ of $m_t$ cells received so far, then normalize by that
type's mapping covariance.  Since transmission is a cell stream, the receiver
often holds only part of a type's cells, so two results handle the
normalization of partial types and the identifiability boundary.

\begin{lemma}[Partial-type correction]
\label{lem:met-partial}
For a partial type (only $r_t<m_t$ of $m_t$ cells received), the finite-prefix
normalization constant of data type $j$ in cell type $t$ is
\begin{equation}
\label{eq:met-gamma}
  \gamma_r(a)=\frac{a(m-a)(r-1)}{m(m-1)}.
\end{equation}
\end{lemma}
Here $\gamma_1(a)=0$ and $\gamma_m(a)=a-a^2/m$, so MET can participate
in measurement on a partial prefix without waiting for a complete-type
boundary.

The current prefix admits a linear unbiased estimator of the total $F_2$
exactly when $\one_J\in\operatorname{rowspan}(\Gamma)$
(Proposition~\ref{prop:met-identifiability}, Appendix~\ref{app:met-partial}).
This criterion defines the identifiability gate.  If the gate fails, the
receiver continues collecting cells or falls back without producing an
estimate.

Full statements, boundary checks, identifiability discussion, and two-scale
simulation validation (partial-type local calibration and end-to-end
incremental estimation) are in Appendix~\ref{app:met-partial}; proofs are in
Appendix~\ref{app:met-proofs}.

\subsection{Failure conditioning for variants
(Corollary~\ref{cor:variant-failure})}
\label{sec:variant-failure-cond}

The protocol of \S\ref{sec:protocol} reads the count array only after a
decoding failure, whereas the adapters of
\S\ref{sec:irregular}--\S\ref{sec:met} are proven unbiased unconditionally.
Because failure is a function of the mapping, conditioning on this selection
event does not preserve the mean in general.  This subsection extends the
failure-conditioned analysis of \S\ref{sec:failure-cond} to all four
variants.

The key difference from Plain is that a variant's per-element self-energy may
be random.  The Plain per-element
self-energy $a_x^\top Qa_x=k-k^2/M$ is deterministic, so conditioned on the
mapping the estimator is unbiased for every realization and selection cannot
move its mean.  A variant's self-energy is random---Irregular's
fluctuates with the degree $D_x$, Rateless's with where the mapping path
lands in the prefix, MET's partial type with the hypergeometric hit count---and
it is a function of the mapping, correlated with decodability.  Selection
bias therefore splits into two channels: the sign channel inherits the Plain
argument and is zero, while the mapping channel is new to variants and needs
an explicit bound.

\begin{corollary}[Failure conditioning for variants]
\label{cor:variant-failure}
Assume the conditions of Theorem~\ref{thm:master}: symmetric $Q$ with
$Q\mu_j=0$ for all types, estimator $\dhat=C^\top QC/\gamma$ with per-element
self-energy $e_x=a_x^\top Qa_x$ and $\gamma=\E[e_x]$; element mappings are
independent.  Parts (i)--(iii) treat this single-type estimator; part (iv)
carries them to the multi-type MET estimator.  Let the idealized peeling
decoder's failure event $F$ be a measurable function of the mappings
$\{a_x\}$, let the signs $\{s_x\}$ be i.i.d.\ uniform on $\{\pm1\}$ and
independent of the mappings, and let $p_F=\Pr[F]>0$.  Then:

\emph{(i) The sign-selection bias is exactly zero:}
\[
  \E[\dhat \given F]
  = \frac{1}{\gamma}\,\E\!\Bigl[\sum_x e_x\;\Big|\;F\Bigr].
\]

\emph{(ii) Exact unbiasedness when self-energies are deterministic:} if every
$e_x$ is constant almost surely (Plain's $k$-regular; the MET counterpart is
in (iv)), then $\E[\dhat \given F]=d$ exactly.

\emph{(iii) Bias bound when self-energies are random:} in general, with
$\sigma_E^2=\Var(e_x)$,
\begin{equation}
\label{eq:cor41-bound}
  \bigl|\E[\dhat \given F]-d\bigr|
  \le \frac{\sqrt{d}\,\sigma_E}{\gamma}\sqrt{\frac{1-p_F}{p_F}},
  \quad\text{i.e.,}\quad
  \frac{\bigl|\E[\dhat \given F]-d\bigr|}{d}
  \le \frac{\sigma_E}{\gamma\sqrt{d}}\sqrt{\frac{1-p_F}{p_F}}.
\end{equation}

\emph{(iv) Several data types (MET).}  Take the MET estimator
$\dhat=w^\top E$ of \S\ref{sec:met}, with the received prefix fixed, so that
the weights $w$ and the normalization matrix $\Gamma$ are determined by the
prefix and satisfy $w^\top\Gamma=\one_J^\top$.  Its per-element self-energy is
the weighted combination
$e_x=\sum_t w_t\,a_{x,t}^\top Q_{r_t}a_{x,t}$ over cell types, and
$w^\top\Gamma=\one_J^\top$ says exactly that $\E[e_x]=1$ for an element of
\emph{every} data type.  Then (i) holds with $\gamma=1$; if all received
cell types are complete, every $e_x$ equals $1$ almost surely and
$\E[\dhat \given F]=d$ exactly; and with partial types present,
\begin{equation}
\label{eq:cor41-bound-met}
  \bigl|\E[\dhat \given F]-d\bigr|
  \le \sqrt{\sum_j d_j\,\sigma_{E,j}^2}\,\sqrt{\frac{1-p_F}{p_F}},
\end{equation}
where $\sigma_{E,j}^2=\Var(e_x)$ for an element of data type $j$ and $d_j$ is
that type's difference count.
\end{corollary}

\noindent\textit{Proof.}  See Appendix~\ref{app:variant-failure}.

Taken together, Corollary~\ref{cor:variant-failure} separates the two
sources of conditional deviation.  Deterministic self-energy gives exact
conditional unbiasedness, while random self-energy produces a
relative-deviation bound of order $O(1/\sqrt d)$ that vanishes as
$p_F\to1$.  The same decomposition extends to heterogeneous MET data types
once the MET weights normalize each element's expected self-energy to one,
with $d\sigma_E^2$ replaced by the sum $\sum_j d_j\sigma_{E,j}^2$ over data
types.

\begin{figure}[!hbtp]
\centering
\includegraphics[width=0.48\textwidth]{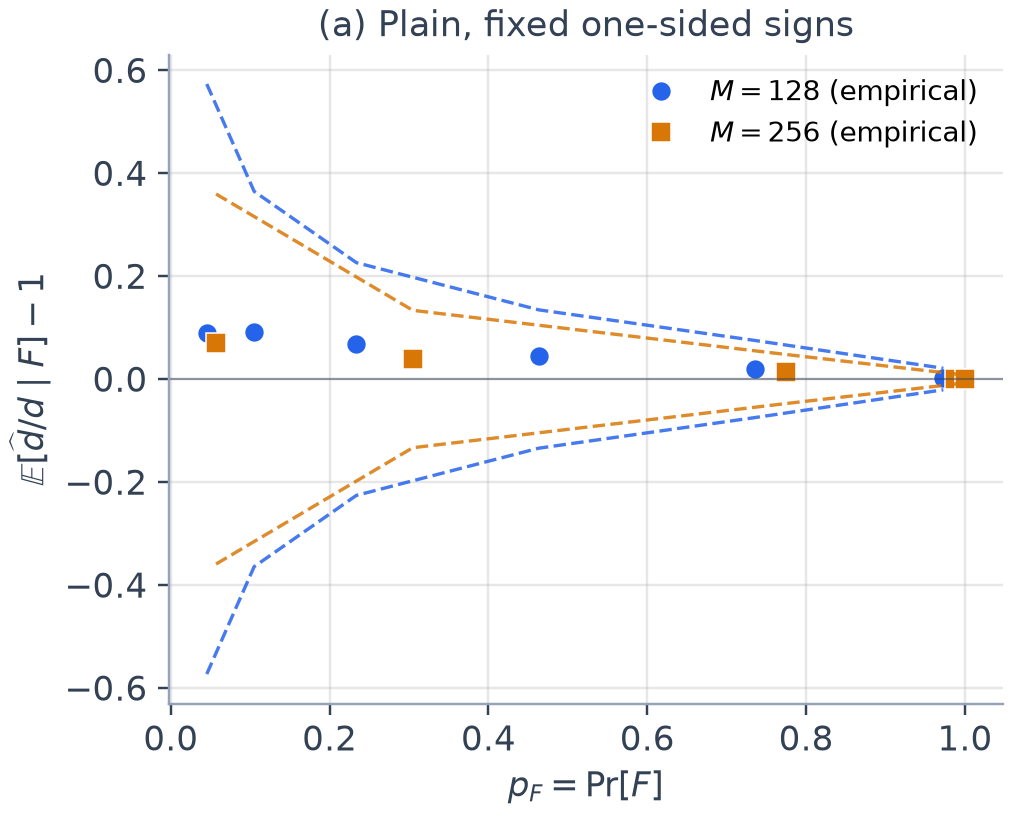}\hfill
\includegraphics[width=0.48\textwidth]{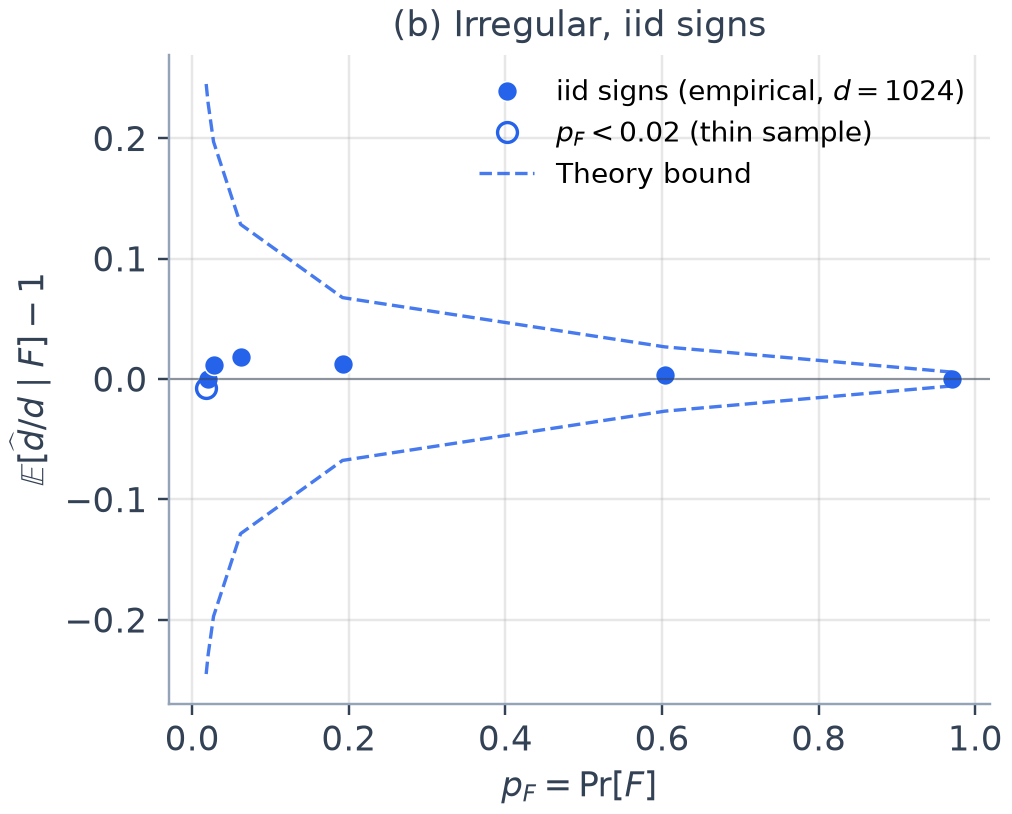}
\caption{Failure-conditioned mean deviation versus analytic bounds.  Left:
Plain with fixed one-sided signs, overlaid with the
Theorem~\ref{thm:cond-mean} fixed-sign envelope; right: Irregular with
i.i.d.\ signs, overlaid with Corollary~\ref{cor:variant-failure}(iii)'s bound
$\pm(\sigma_E/(\gamma\sqrt d))\sqrt{(1-p_F)/p_F}$;
hollow points are sparsely sampled tiers with $p_F<0.02$.  All adequately
sampled points remain within their respective bounds, and the deviations
shrink as $p_F\to 1$.}
\label{fig:e25}
\end{figure}

\emph{Simulation and validation.}  Each variant's $\sigma_E$ is given
explicitly by its mapping statistics; the instantiations are in
Appendix~\ref{app:variant-failure}.  The Irregular instance is checked
directly at decode scale ($d=1024$, the $\Lambda$ of
\S\ref{sec:irregular}, with $M$ covering the transition region from
$p_F=1$ to $p_F=0.018$): in the transition region one sees a positive
conditional-mean deviation (at $M=1150$, $p_F=0.192$,
mean$(\dhat/d \given F)=1.0127$, about 14 Monte Carlo standard errors), so
the mapping-selection channel of (iii) is real; all adequately sampled points
with $0<p_F<1$ stay inside the bound by a wide margin, and the deviation
shrinks to the noise floor as $p_F\to 1$
(Figure~\ref{fig:e25}).

The two panels separate the two channels just discussed.  On the left, the
Plain self-energy is deterministic so the mapping channel is always zero and
all visible deviation comes from the cross terms introduced by fixed
one-sided
signs; on the right, Irregular uses i.i.d.\ signs so the sign channel is
exactly zero by (i) and what remains is the random self-energy's mapping
channel.

The four adapters therefore share a common failure-path interface,
$\{\dhat,\mathrm{CI}\}$.  The corollary establishes conditional-mean
validity; Plain and Rateless provide analytic interval constructions, while
Irregular and MET use failed-only empirical quantiles.

\section{Self-Sizing Protocol}
\label{sec:protocol}

\S\ref{sec:count-measurement}--\S\ref{sec:mapping-aware} established the
estimator itself: the pre-peeling count vector of a failed IBLT sketch
provides an unbiased $\dhat$ together with a failure-conditioned lower-quantile
guarantee.  This section uses that measurement to construct a two-round
protocol.  The sender first transmits a fixed-capacity IBLT sketch.  If
decoding fails, the receiver estimates the difference cardinality from the
same sketch, computes the second-round capacity in one step, and requests one
fresh sketch.  First-round success gives a 1-RTT fast path, while the main
path completes within 2 RTTs.  The same sketch transmitted first
simultaneously carries decoding, measurement, and second-round capacity
sizing.

\begin{figure*}[!hbtp]
\centering
\includegraphics[width=\textwidth]{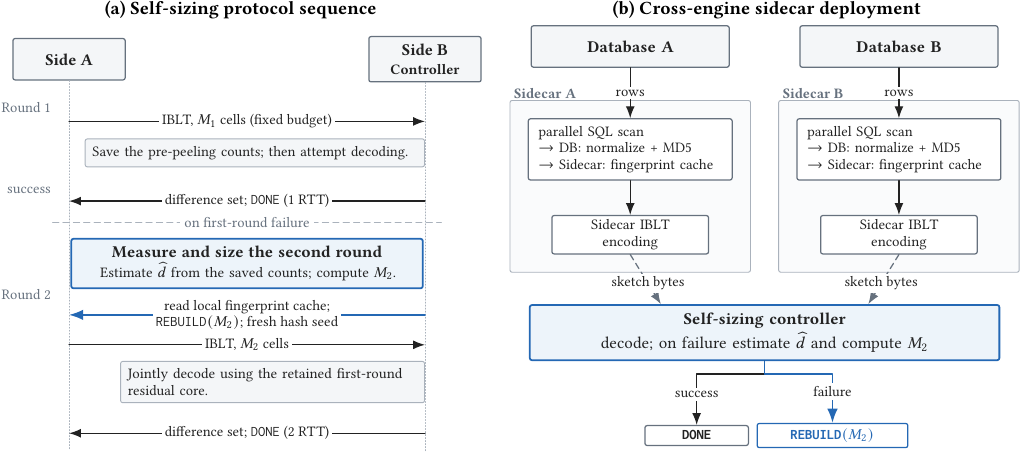}
\caption{Protocol and deployment overview.  (a) The 1-RTT fast path and the
fail$\rightarrow$measure$\rightarrow$size path.  (b) Each side scans in its
local order, while Sidecar B hosts the controller.}
\label{fig:overview}
\end{figure*}

Figure~\ref{fig:overview} combines the protocol flow with its deployment
boundary.  Panel~(a) shows the 1-RTT fast path and the failure-triggered
second round; Side B hosts the decoder and capacity controller.  Panel~(b)
shows the database-external sidecars: each launches parallel SQL scans, the
local database normalizes rows and computes fingerprints, and the sidecar
materializes a local fingerprint cache while encoding the IBLT.  After a
first-round failure, that cache supplies the fresh second-round sketch without
another database scan.

\emph{Why fix the main path at 2 RTTs.}  In cross-datacenter deployments
the round-trip cost is non-trivial (the field deployment of
\S\ref{sec:system} measures 25--30\,ms RTT on a bandwidth-limited cross-city
link), so additional RTTs have a clearly visible impact on performance.  Blind doubling, streaming protocols, and
Merkle-style localization all have round counts that grow with the
differences; self-sizing bounds this component at 2 RTTs.  A fixed 2-RTT
main path improves latency predictability---the
measurements of \S\ref{sec:prod-shape} show that the latency contributed by
the two protocol rounds remains small under batching---although scan and transfer costs
remain workload-dependent.  The principal tradeoff is the additional
second-round capacity introduced by the lower-quantile multiplier $1/q_\delta$
(\S\ref{sec:capacity-formula}).

\subsection{Two-round mechanism and protocol state}
\label{sec:two-round}

The protocol first sends a fixed-capacity IBLT $M_1$; the receiver subtracts
the two sides, saves the count vector \emph{before} peeling, computes $\dhat$
from it, and only then attempts decoding.

Saving the pre-peeling counts is a hard requirement: the order cannot be
reversed.  $\dhat$ must estimate the full difference cardinality $d$, which
is determined by the count vector before peeling and is independent of how
far peeling later proceeds.  Once the decoder runs, every recovered element
subtracts its contribution from the cells it touches, the count array is
rewritten, and the value read afterwards corresponds to the residual-core
size, not $d$.

The second round uses a fresh hash independent of the first; it must
re-encode the same set of elements.  These two conditions---fixed input
view and independent mapping---are listed together in
\S\ref{sec:boundaries} and realized by the system in
\S\ref{sec:system-prereq}.

\subsection{Capacity formula and success guarantee}
\label{sec:capacity-formula}

Given $\dhat$, a capacity multiplier is needed to convert the estimate into a
second-round capacity.
\S\ref{sec:quantile-interface} defines the failure-conditioned lower quantile
$q_\delta>0$ with
\[
  \Pr[\dhat\ge q_\delta d \given F]\ge 1-\delta.
\]
In practice $q_\delta<1$ (the calibrated $q_{0.01}$ tiers of
\S\ref{sec:exp-validation} range over $0.63$--$0.95$, with $1-\delta=99\%$
failure-conditioned coverage), so with conditional probability at least
$1-\delta$, $\dhat/q_\delta$ upper-bounds $d$.  The multiplier therefore
accounts for lower-tail underestimation when setting the second-round
capacity.

IBLT capacity and peeling success are related by a curve determined by the
number of hashes, the decoder, and finite-length effects; existing work gives
its threshold and listing-failure analysis~\cite{goodrich2011invertible,mizrahi2023listing}.
We denote a calibrated operating point on this curve by $\beta$ and compose
$q_\delta$ with $\beta$ once in Proposition~\ref{prop:capacity}.  Concrete
finite-length configuration is checked by the failed-only grid experiments of
\S\ref{sec:success-calibration} and Appendix~\ref{app:protocol-details}.

We write $ok_2$ for the event that the second-round decoder fully recovers the
currently unknown elements, and $d_{\mathrm{rec}}$ for the number of
difference elements exactly recovered by the first-round peeling---each
element's key and sign known, an exact integer; when first-round state is not
reused, $d_{\mathrm{rec}}=0$.  The second round must recover the residual
count $d-d_{\mathrm{rec}}$; we assume it stays within the declared operating
range (if it is zero, nothing is left to recover and $ok_2$ holds with
certainty).

\begin{proposition}[Protocol guarantee: second-round capacity formula and
success probability]
\label{prop:capacity}
Fix the number of hashes, the decoder, and the declared finite-length
operating range.  Let $q_\delta>0$ satisfy the failure-conditioned
lower-quantile guarantee \eqref{eq:qdelta} uniformly in the operating range,
where $F$ is the first-round failure event and the second round uses a fresh
hash independent of the first.  Let $\beta$ be a capacity operating point of
the fresh-hash decoder: for every $n\ge1$ in the operating range, whenever the
capacity is at least $\lceil\beta n\rceil$, the decoder's failure probability
on $n$ unknown elements is at most $\eta_{\mathrm{dec}}(\beta)$, whatever the
first round observed.  If the second-round sketch is sent with
\begin{equation}
\label{eq:capacity}
  M_2
  =
  \left\lceil
    \beta\max\!\left\{0,\;\frac{\dhat}{q_\delta}-d_{\mathrm{rec}}\right\}
  \right\rceil,
\end{equation}
and the receiver first removes the $d_{\mathrm{rec}}$ known elements locally
from the received difference table, then
\begin{equation}
\label{eq:success-bound}
  \Pr[ok_2 \given F]
  \ge (1-\delta)\bigl(1-\eta_{\mathrm{dec}}(\beta)\bigr)
  \ge 1-\delta-\eta_{\mathrm{dec}}(\beta).
\end{equation}
\end{proposition}

\begin{proof}
The argument is conditional on the first round, because $M_2$ is computed from
the first-round estimate and is therefore not independent of it.

Let $H_1$ denote the first-round history---the first-round mapping, counts,
$\dhat$, the peeled elements and $d_{\mathrm{rec}}$---and consider any $H_1$
on which the failure event $F$ occurs together with the safe-estimate event
$\{\dhat\ge q_\delta d\}$.  On such a history, $\dhat$, $d_{\mathrm{rec}}$,
and hence $M_2$ are all determined, and $\dhat\ge q_\delta d$ gives
$d\le\dhat/q_\delta$; since $d_{\mathrm{rec}}\le d$,
\[
  n \;=\; d-d_{\mathrm{rec}}
  \;\le\;
  \frac{\dhat}{q_\delta}-d_{\mathrm{rec}},
  \qquad\text{so}\qquad
  M_2\ge\lceil\beta n\rceil .
\]
The capacity carried by the second round is thus deterministically sufficient
for the number of elements it must recover, whatever that history was.

The sender still generates the second-round sketch over the full set; after the
receiver removes the $d_{\mathrm{rec}}$ known elements, the IBLT sketch encodes
the same $n$ unknown elements, and its hash is drawn fresh and independent of
$H_1$.  The decoder operating point therefore applies history by history: for
each such $H_1$, $\Pr[ok_2 \given H_1]\ge1-\eta_{\mathrm{dec}}(\beta)$ when
$n\ge1$, and $\Pr[ok_2 \given H_1]=1$ when $n=0$.  Taking the conditional
expectation over the first-round histories inside the safe-estimate event and
using $\Pr[\dhat\ge q_\delta d \given F]\ge1-\delta$,
\[
  \Pr[ok_2 \given F]
  \ge
  \bigl(1-\eta_{\mathrm{dec}}(\beta)\bigr)\,
  \Pr[\dhat\ge q_\delta d \given F]
  \ge (1-\delta)\bigl(1-\eta_{\mathrm{dec}}(\beta)\bigr)
  \ge1-\delta-\eta_{\mathrm{dec}}(\beta).
\]
\end{proof}

\noindent\textit{What this guarantee says.}  The success probability
decomposes into two separate risk budgets: $\delta$ (the estimate
falls below the true value) and $\eta_{\mathrm{dec}}(\beta)$ (the
capacity suffices but the decoder still fails due to finite-length
effects).  The fresh second-round hash makes the decoder guarantee
apply uniformly over all qualifying first-round histories: whatever the
first round observed, a sketch sized above the operating point still fails
with probability at most $\eta_{\mathrm{dec}}(\beta)$.  Combining the two
risks by a union bound gives
$\Pr[\neg ok_2 \given F]\le\delta+\eta_{\mathrm{dec}}(\beta)$.  The
statistical layer provides $q_\delta$, the coding layer provides
$\beta$, and the protocol composes them into $M_2$ in one step---no
iterative probing.

\emph{Deployed configuration.}  We set
$d_{\mathrm{rec}}=0$ in the capacity formula, sizing the second round as
$M_2=\lceil\alpha\dhat\rceil$ with $\alpha=\beta/q_\delta$, while still
enabling joint peeling of the first-round residual core
(Proposition~\ref{prop:joint}).  Setting $d_{\mathrm{rec}}=0$ overpays
bytes only in proportion to how much the first round peeled before
failing, and a first-round failure that still recovers a substantial
fraction of $d$ occurs only in a narrow band just above the peeling
threshold.  The cost-dominant two-round runs sit far into overload---$d/M_1$
in the tens to hundreds (\S\ref{sec:prod-shape})---where first-round peeling
recovers almost nothing, $d_{\mathrm{rec}}\approx0$, and the deduction term
vanishes.  Deducting $d_{\mathrm{rec}}$ is therefore a sound, optional
refinement whose guarantee is already contained in
Proposition~\ref{prop:capacity}; it pays off when second-round runs cluster
near the peeling threshold rather than in deep overload.

\emph{Second-round capacity cost.}  Since $\E[\dhat \given F]=d$ under the
random-sign model (Theorem~\ref{thm:cond-mean}), the expected second-round
capacity is about $1/q_\delta$ times the ideal $\lceil\beta d\rceil$---the
cell overhead of lower-tail safety amplification; the complete cost,
including the first-round $M_1$ and occasional recovery steps, is accounted
for in \S\ref{sec:regret}.  The main protocol's second decoding attempt ends
within 2 RTTs; runs in which the second-round decoder still fails enter
\texttt{FALLBACK} of \S\ref{sec:boundaries}.

\subsection{Joint peeling and first-round residual state}
\label{sec:joint-peeling}

The capacity formula of Proposition~\ref{prop:capacity} already deducts the
$d_{\mathrm{rec}}$ elements recovered in the first round.  A first-round
failure leaves one more usable state: the residual core left when peeling
stopped, encoding the remaining $d-d_{\mathrm{rec}}$ unknown elements along
with their keySum and checkSum residuals.  This residual core is data the
first round has already paid for and can participate in the second round's
joint decoding.

\begin{proposition}[Monotonicity of joint peeling]
\label{prop:joint}
Suppose the second-round sketch and the first-round residual core encode the
same set of unknown elements.  Under the idealized decoder, for each concrete
pair of sketch instances, success of peeling the second-round sketch alone
implies success of joint peeling.  Therefore the joint decoding success
probability is at least that of decoding the second-round sketch alone.
\end{proposition}

\begin{proof}
Replay a successful peeling sequence of the second-round sketch in the joint
decoder.  Each recovered element is removed from both the first-round
residual core and the second-round sketch; these extra deletions cannot make a
cell that was pure in the original second-round sequence impure.  The sequence
therefore executes completely.
\end{proof}

Joint peeling can therefore be enabled by default: it adds no transmitted
bytes and preserves the Proposition~\ref{prop:capacity} success guarantee.
The $M_2$ in the main text is configured by the standalone-decoding capacity
coefficient $\beta$ needed for the second-round sketch to decode
independently; the extra capability of the first-round residual core remains
as a runtime margin.  Reducing the second-round capacity further using the
residual core would require separate asymptotic analysis and finite-length
calibration.

The three uses of a failed first-round sketch thus have a clear division
(Table~\ref{tab:failed-state}):

\begin{table}[!hbtp]
\centering
\caption{The three uses of a failed first-round sketch.}
\label{tab:failed-state}
\small
\begin{tabular}{lll}
\toprule
first-round state left behind & protocol use & guarantee \\
\midrule
pre-peeling count vector & compute $\dhat$ and planning bound $\dhat/q_\delta$ & Proposition~\ref{prop:capacity} \\
recovered $d_{\mathrm{rec}}$ elements & exact deduction from the capacity formula & Proposition~\ref{prop:capacity} \\
first-round residual core & joint peeling with the fresh second-round sketch & Proposition~\ref{prop:joint} \\
\bottomrule
\end{tabular}
\end{table}

\subsection{Core protocol flow}
\label{sec:protocol-flow}

Figure~\ref{fig:overview}(a) puts the Proposition~\ref{prop:capacity}
capacity formula and the Proposition~\ref{prop:joint} state reuse into one
exchange sequence; the steps are enumerated below.

\begin{enumerate}[itemsep=2pt,topsep=2pt]
  \item $A\to B$: send the first-round sketch of capacity $M_1$.  $B$
  subtracts, saves the count vector before peeling, then attempts decoding.
  On first-round success, output the full difference set and enter
  \texttt{DONE}.
  \item On first-round failure, $B$ computes $\dhat$ from the saved counts,
  records the exactly recovered element count $d_{\mathrm{rec}}$, and may keep
  the first-round residual core.  If the measurement state is valid, compute
  the second-round capacity by \eqref{eq:capacity}.
  \item If $M_2$ does not exceed the configured budget, $B\to A$ requests a
  second-round sketch of capacity $M_2$ with a fresh hash seed; $A\to B$
  returns it.  $B$ subtracts, then first removes the signed elements already
  recovered in round one from the second-round difference table.
  \item $B$ peels the second-round sketch; if the first-round residual core
  was kept, it runs the joint peeling of
  Proposition~\ref{prop:joint} at the same time.  On full recovery enter
  \texttt{DONE}; with unknown elements remaining, enter \texttt{FALLBACK}.
  If the measurement state is invalid or $M_2$ exceeds the budget, exit the
  main path before sending the second-round request.
\end{enumerate}

First-round reuse has two independent knobs.  Deducting $d_{\mathrm{rec}}$
controls the size of the transmitted second-round sketch; running joint
peeling of the residual core controls only local decoding and adds no bytes.
The default deployment takes $d_{\mathrm{rec}}=0$---so
$M_2=\lceil\alpha\dhat\rceil$ with $\alpha=\beta/q_\delta$---while keeping
joint peeling on; an implementation that retains no decoder state at all
drops both.  The Proposition~\ref{prop:capacity} success guarantee holds in
every case, since $d_{\mathrm{rec}}=0$ is the formula's conservative
setting.

\subsection{Configuration and operating boundaries}
\label{sec:boundaries}

The core flow depends on four pre-configured pieces of information: the risk
budget $\delta$, the first-round capacity $M_1$, the failure-conditioned
lower quantile $q_\delta$ \eqref{eq:qdelta}, and the decoder operating point
$\beta$.  $\delta$ sets the lower-tail target; $M_1$ sets the fixed bytes,
the fast-path coverage, and the estimation precision
$\operatorname{RSD}(\dhat)\approx\sqrt{2/M_1}$ (see \eqref{eq:rsd}); $q_\delta$
and $\beta$ enter the second-round capacity formula \eqref{eq:capacity}.
Without first-round state reuse the two merge into the single multiplier
$\alpha=\beta/q_\delta$.  We take $\delta=0.01$ throughout; the bytes,
precision, $q_\delta$, and recommended multipliers for the five standard $M_1$
tiers are in Appendix~\ref{app:protocol-details}.

\emph{Fixed input view and independent mappings.}  The two rounds must encode
the same elements, and the second-round hash seed must be independent of the
first-round mapping.  The former guarantees that the first-round recovery
state and the second-round difference table describe the same difference
set; the latter guarantees that the nominal $\beta$ curve can be invoked
conditioned on the first-round failure.  How a concrete system maintains this
input view is an implementation prerequisite (see \S\ref{sec:system-prereq}).

\emph{Over-budget exit.}  Self-sizing computes the second-round capacity
before requesting the second-round sketch, so it can compare against a natural
baseline: sending the entire fingerprint set for full verification.
When the computed second-round sketch exceeds this baseline, the protocol
emits \texttt{FALLBACK} directly.  A deployment may also fix a
per-verification byte budget as an earlier exit.

\emph{Representation validity.}  The risk budget $\delta$ covers statistical
underestimation; representation errors such as count wraparound, which
destroys the centered energy $T$ \eqref{eq:centered-energy}, are handled
separately.  Building and subtracting therefore require checked arithmetic or
a sufficiently wide internal type, with the overflow status propagated to the
measurement interface.  On wraparound or when no applicable $q_\delta$ exists,
the interface outputs \texttt{REJECT} and the caller chooses a recovery path.

The protocol exposes four outcomes:
\[
  \texttt{DONE} \;/\; \texttt{REBUILD}(M) \;/\; \texttt{FALLBACK}(path)
  \;/\; \texttt{REJECT}(reason).
\]
\texttt{DONE} means the difference set has been fully recovered;
\texttt{REBUILD($M$)} requests a second-round sketch of capacity $M$ with a
fresh hash; \texttt{FALLBACK} covers over-budget or second-round-decoding
failure; \texttt{REJECT} means the current measurement state cannot provide a
valid capacity estimate.  \texttt{DONE} and \texttt{REBUILD($M$)} drive the
1-RTT fast path and the 2-RTT self-sizing main path; the other two provide
explicit exits for statistical risk, representation errors, and resource
boundaries.

\subsection{Success-rate calibration}
\label{sec:success-calibration}

This subsection evaluates whether the finite-length implementation achieves
the success target of Proposition~\ref{prop:capacity}; bytes and rounds are
evaluated separately in \S\ref{sec:regret}.

The simulation sweeps the first-round capacity and the second-round
multiplier:
\[
  M_1\in\{64,256,512,1024,4096\},
  \qquad
  \alpha\in\{1.6,1.8,2.0\}.
\]
Two success rates must be distinguished.  Overall success counts both
first-round and second-round decodes; it is inflated by easy trials
that decode in round one.  Failed-only second-round success
$\Pr[ok_2 \given F]$ is the conditional probability of
Proposition~\ref{prop:capacity}---the guarantee that must be realized.

\begin{figure}[!hbtp]
\centering
\includegraphics[width=0.9\textwidth]{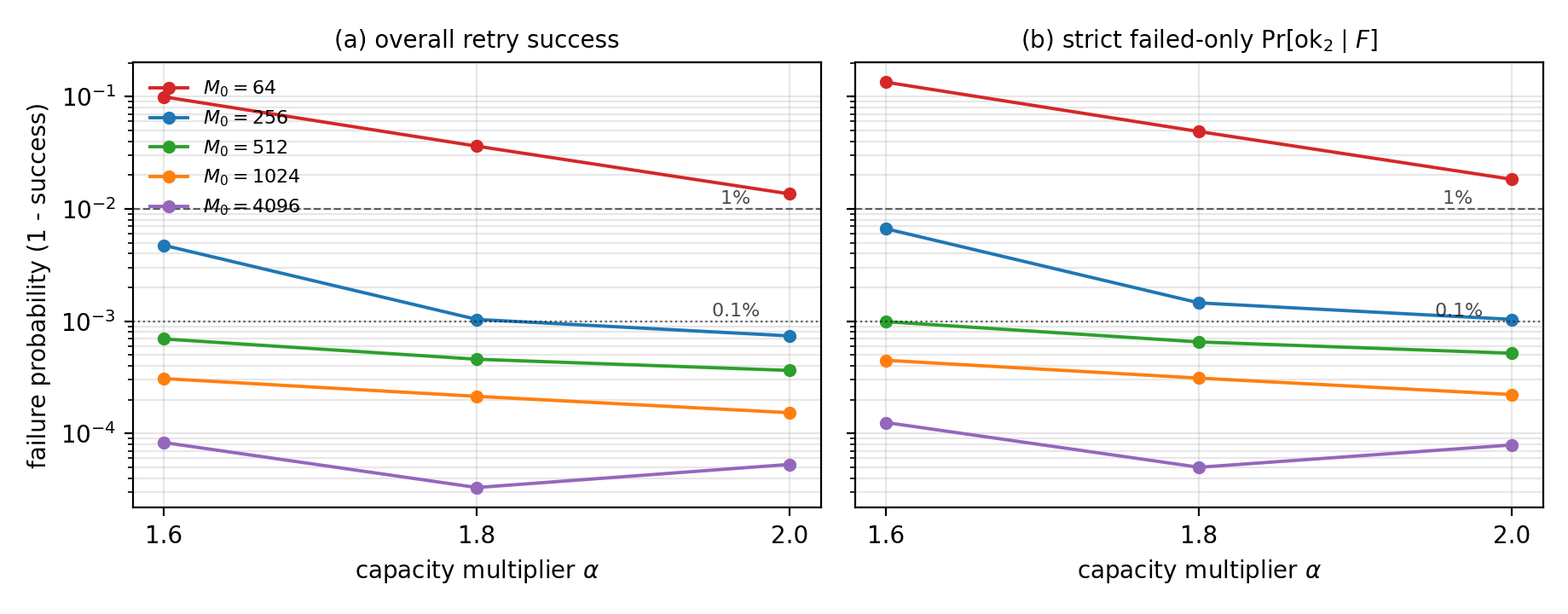}
\caption{Second-round success rate versus capacity multiplier and
first-round capacity.  Left: overall success; right: failed-only second-round
success $\Pr[ok_2 \given F]$; the y-axis is failure probability on a log
scale, with the $1\%$ design target and the $0.1\%$ reference line marked.
For $M_1\ge512$, $\alpha=1.6$ already reaches above $99.9\%$ under both
semantics; the compact $M_1=64$ tier is shown separately as a stress test.}
\label{fig:e13}
\end{figure}

At the smallest swept multiplier $\alpha=1.6$, the failed-only second-round
success $\Pr[ok_2 \given F]$ already reaches $99.901\%$ at $M_1=512$
(Figure~\ref{fig:e13}), comparable to the overall success rate, showing that
the failed-only result is not inflated by easy cases that decode in the first
round.  The compact $M_1=64$ tier is intentionally under-provisioned: its
failed-only rates are $86.589\%$, $95.123\%$, and $98.168\%$ for
$\alpha=1.6,1.8,2.0$, respectively, so it is reported as a stress profile
rather than pooled into the deployment-grade aggregate.  The tightest of the four
tiers is $M_1=256$: at the same multiplier $\Pr[ok_2 \given F]=99.336\%$, and
raising $\alpha$ to $1.8$ increases it to $99.855\%$.  Across the
deployment-grade tiers $M_1\ge256$, capacities configured with a $99\%$
lower-tail probability achieve actual second-round success rates around
$99.9\%$ at the measured operating points, so the design target is realized
with a small margin.  The complete $M_1\times\alpha$
numbers are in Table~\ref{tab:c2} (Appendix~\ref{app:protocol-details});
Table~\ref{tab:c1} derives the recommended multipliers independently from
each tier's $q_\delta$; together the two tables carry configuration and
validation.

\subsection{Protocol cost under unknown difference cardinality: regret versus
oracle}
\label{sec:regret}

\S\ref{sec:success-calibration} asked how reliable the protocol is for a
given $M_1$ and $\alpha$.  This section asks a different question: when the
true $d$ is unknown in advance, how much extra do the various methods pay to
find a capacity that decodes?

Consider an \emph{oracle} that already knows the true $d$ and therefore
configures exactly the smallest decodable capacity in one try: no probing, no
waste.  It is the byte and round lower bound.  Any realistic method must
estimate or probe; the extra bytes and rounds it spends over the oracle are
the price of not knowing $d$, which we call \emph{regret}.  We report byte
cost as a ratio to the oracle and interaction cost as the mean number of
rounds; a byte ratio of 1 and a round count of 1 match the oracle baseline.

We compare four practical methods.  \emph{Self-sizing} is this section's
protocol: a fixed first-round sketch, and on failure, one second-round
capacity computed from its counts; the cost is mainly the extra bytes of the
$q_\delta$ safety amplification.  We consider both second-round-only decoding
(\emph{Independent R2}) and joint decoding with the first-round residual core
(``+ joint'').  The two baselines are \emph{blind doubling}, which does no
estimation, tries a conservative capacity, and doubles on failure until it
decodes---rounds grow logarithmically with $d$---and \emph{Strata-first},
which first sends a stratified summary to estimate $d$ and then configures a
main IBLT, paying a fixed summary cost.  Self-sizing counts
both the first-round sketch and the second-round capacity (with
$d_{\mathrm{rec}}$ not deducted); Strata-first deducts the elements recovered
during estimation and never double-charges.

The main conclusion is that self-sizing stays close to the oracle across the
whole range of $d$.  The baselines incur different costs: blind doubling uses
additional rounds as $d$ grows, while Strata-first pays a fixed summary cost
that dominates at small~$d$.

\begin{figure}[!hbtp]
\centering
\includegraphics[width=0.95\textwidth]{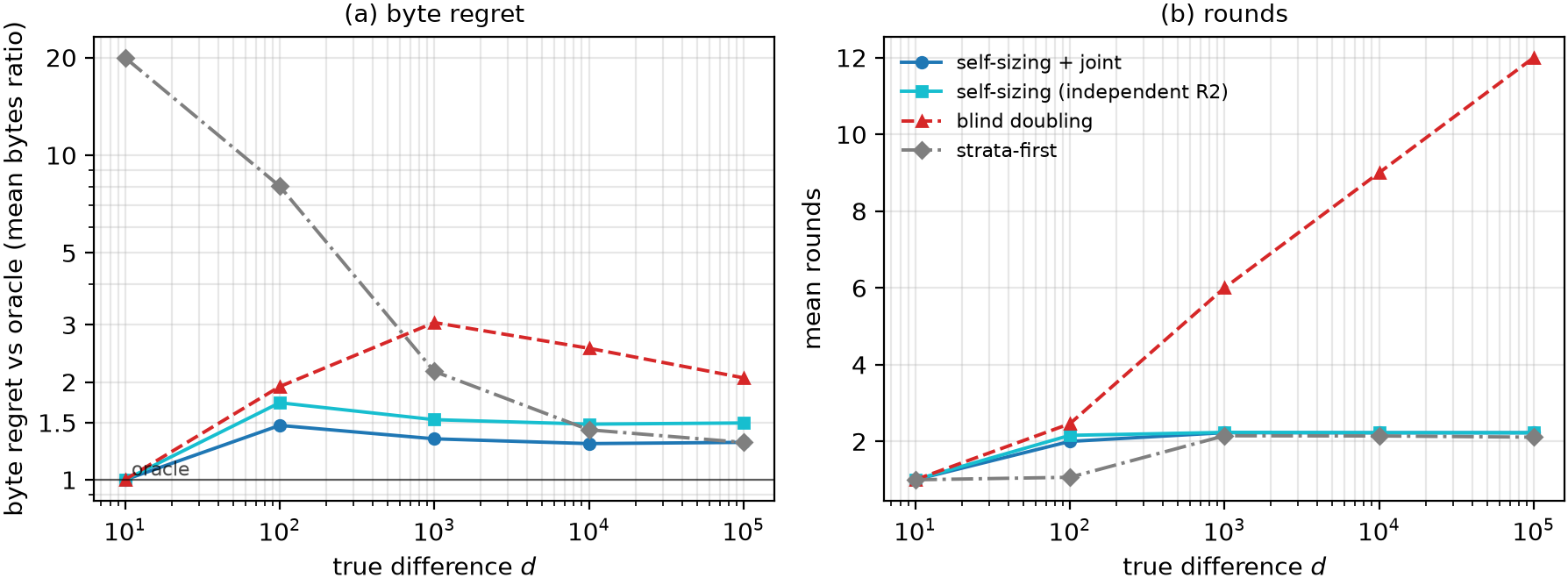}
\caption{Regret versus oracle for four methods ($M_1=64$, $\theta_-=0.5$).
Left: mean bytes relative to the oracle (log y-axis); right: mean rounds to
full recovery; the x-axis is the true $d$ on a log scale.  Self-sizing tracks
the oracle within $1.3$--$1.7\times$ bytes at a stable round count; blind
doubling pays logarithmically growing rounds; Strata-first pays a fixed
summary at small $d$.}
\label{fig:e21}
\end{figure}

Figure~\ref{fig:e21} evaluates the lowest-fixed-cost configuration, with
first-round capacity $M_1=64$ (2\,KB).  Self-sizing's byte regret is $1.3$--$1.7\times$ the
oracle throughout, with the round count flat at about 2 (1 for the small-$d$
end where the first round decodes directly); blind doubling's bytes are
$2.0$--$3.1\times$ with rounds growing logarithmically in $d$ (on average 12
rounds at $d=10^5$); Strata-first's fixed 40\,KB summary cost makes its
byte regret $8$--$20\times$ at small $d$.  Its one-round advantage is
concentrated in the narrow window $d\le100$; for $d\ge10^3$ its byte regret
falls to $1.35$--$2.15\times$ and the round advantage disappears.  The full
$M_1\times d$ crossover phase diagram, expected-byte ledger, and
cost-accounting assumptions are in Appendix~\ref{app:protocol-details}
(\S\ref{app:param-sensitivity}).
We additionally evaluate the deployment operating point
($M_1=512$ with the recommended multiplier $\alpha=1.82$): byte regret
stays $1.44$--$1.79\times$ over the oracle and every trial completes in
exactly two rounds, with the third-round retry tail of the $M_1=64$ stress
profile absent (Appendix~\ref{app:protocol-details},
\S\ref{app:param-sensitivity}).

\noindent\textit{Strata-first in the small-$d$ regime.}  Strata's one-round
advantage is confined to the small-$d$ window $d\le100$; for $M_1\ge256$,
self-sizing's first round already covers this window and Strata's round
advantage disappears, while Strata still carries 40\,KB of layered summary
versus our 2\,KB first-round sketch (the layered recovery mechanism is
described in Appendix~\ref{app:protocol-details}).

\section{System Evaluation and Applicability}
\label{sec:system}

\S\ref{sec:count-measurement} gave the statistical characterization of
estimating the difference cardinality from the pre-peeling counts of a Plain
IBLT, and \S\ref{sec:protocol} connected that measurement to a two-round
self-sizing protocol.  This section examines the three questions the protocol
faces inside a database system: how large are the difference cardinalities in
production, and how are the differences distributed over the primary-key
space; how do these characteristics affect the end-to-end cost of
self-sizing IBLT versus Merkle-style localization; and which scenarios remain
covered when the data model and network conditions change.  It closes with
the protocol's own production operation.

The evidence comes from four kinds of material.  First, we perform a
systematic statistical study of table-level reconciliation records from
NineData from a 90-day production window, examining table sizes, difference
cardinalities, running times, and the distribution of differences in the
primary-key space.  Second, we rebuild controlled data from the structure and
statistical characteristics of production tables and compare the two methods'
end-to-end cost in a cross-engine environment.  Third, we conduct an actual
cross-city deployment on China Mobile's Redis/Pika, examining a
hash-bucketed key space and a real cross-city link.  Fourth, we report
self-sizing's production operation on NineData's reconciliation service.
The first analysis gives the production workload profile; the other three
contain actual runs with result verification.

The 90-day logs were recorded by NineData's Merkle-style service before
self-sizing was deployed, so the rounds and capacities we derive from them
substitute each run's true~$d$ into the capacity model of
\S\ref{sec:protocol}.  All other rounds, transfer volumes, wall-clock times,
and estimates in this section are measured from actual execution: in the
cross-engine replay, in the China Mobile deployment, and in the production
operation of \S\ref{sec:prod-operation}.

\subsection{System implementation, protocol state, and difference semantics}
\label{sec:implementation}

The comparisons below rest on one shared system definition.  The protocol
mechanism and its theoretical guarantees are in \S\ref{sec:protocol}; this
subsection only describes how the system carries the protocol, in three
steps: how database records become set elements, which difference semantics
are used, and how the first-round IBLT proceeds under different outcomes.
The concrete performance comparisons are in \S\ref{sec:prod-distribution}--\S\ref{sec:cross-city}.

\subsubsection{Deployment form and row fingerprints}

The system runs as a sidecar outside the databases
(Figure~\ref{fig:overview}(b)).  Each side reads
local records, normalizes type representations across engines, and
generates a row fingerprint from the primary key and business columns.
Only fixed-length fingerprints cross the network; confirmed differences
are looked up by key.

The local path is implemented as a streaming pipeline: workers scan rows,
normalize values, and add fingerprints to the first-round IBLT concurrently,
so scanning and sketch construction overlap.  This concurrency is local to
each sidecar and does not add a network message or change the sketch
semantics.

In cross-datacenter deployments this consideration is amplified: inter-DC
bandwidth is scarce, so minimizing sketch size and round count is the
primary design target, and the system experiments report transfer
volume, round counts, and end-to-end time over cross-DC links rather
than only local database compute time.

Cross-engine fingerprint alignment requires consistent handling of
nulls, numbers, times, strings, encodings, and case across databases;
the concrete type rules and hash choices are in
Appendix~\ref{app:e1} (the full specification is in the supplementary
material).  The implementation uses a 56-bit
truncated fingerprint, wide enough that the residual risk---a colliding pair
that straddles the difference and cancels---is about $d_Ad_B/2^{56}$ per run,
with $d_A,d_B$ the two sides' difference counts; at the largest replay
cardinality (P1, $d_{\mathrm{ms}}=93{,}130$) this is about $3\times10^{-8}$
per run---seven to eight orders of magnitude below one---and proportionally
smaller for the smaller-difference profiles; the width is a
free parameter of the deployment and is accounted for in
Appendix~\ref{app:e1}.

Fingerprints are written into a Plain IBLT as set elements.  After the two
sketches are subtracted, identical records cancel; inserts, deletes, and
updates leave signed difference elements---a modified record contributes its
old and new fingerprints.  The sketch difference cardinality reported by the
system therefore uses the unit-weight multiset semantics of
\S\ref{sec:model}; the recovered fingerprints are then used to look up and
confirm the original records.

\subsubsection{Difference semantics and protocol state}

The sidecar does more than emit sketches.  During the first-round scan it
saves the local record fingerprints and the lookup information needed; if
first-round decoding fails, the system reuses this already-read state to
build the second-round sketch without rescanning the table; the retained
fingerprints are re-bucketed
concurrently across workers, and the per-cell XOR updates are independent.
This avoids a repeated scan and
guarantees that both rounds compute over the same fixed fingerprint view:
even if the database changes between the two rounds, the protocol knows
exactly which batch of data it is comparing.  Here ``snapshot'' refers to
the fixed input view saved by the sidecar; whether the database itself
provides a transactional snapshot is a deployment decision.

In cross-datacenter deployments, the sidecars sit next to their local
databases: database scanning, fingerprint computation, and second-round
re-bucketing stay local, while the inter-DC link carries only sketches,
control messages, and the final difference results.

Two difference semantics are used: $d_{\mathrm{ms}}$ counts multiset
elements in the sketch (a modified record's old and new fingerprints are two
elements), and $d_{\mathrm{row}}$ counts business rows changed.  Capacity
and transfer analyses use $d_{\mathrm{ms}}$; $d_{\mathrm{row}}$ is for
business scale.

A verification proceeds through the four outcomes of \S\ref{sec:boundaries}:
\texttt{DONE} on first-round success; \texttt{REBUILD} on
fail$\rightarrow$measure$\rightarrow$resize with joint peeling; and
\texttt{FALLBACK} / \texttt{REJECT} as protection exits for resource
boundaries and representation errors.  All controlled runs in this paper
complete within budget and cover \texttt{DONE} and \texttt{REBUILD};
\texttt{FALLBACK} occurs only in production operation
(\S\ref{sec:prod-operation}), and \texttt{REJECT} is not observed.

\subsubsection{Deployment requirements and cross-engine fingerprint alignment}

Deployment must determine the protocol's operating range, not each table's
true difference cardinality.  The first-round capacity $M_1$ fixes the bytes
and the estimation precision; the risk level $\delta$ and its calibrated lower
quantile $q_\delta$ set the statistical margin, while the decoder operating
point $\beta$ and the resulting multiplier $\alpha=\beta/q_\delta$ convert the
estimate into the second-round sketch capacity; the resource budget decides
when to enter \texttt{FALLBACK} (\S\ref{sec:boundaries}).  The true $d$ is
obtained at
runtime from the count measurement, so the system needs neither a per-table
$d$ nor prior knowledge of which key ranges the differences fall into.
Correctness requires both sides to process the same fixed input view
completely and without duplication; scan concurrency is a local tuning knob.

The two-round same-input-view and independent-seed requirements of
\S\ref{sec:boundaries} are satisfied on the system side by the sidecar's
fixed fingerprint view and seed derivation.

The end-to-end relational comparison runs between Oracle 11g and
PolarDB/MySQL 8.0, covering single-column and three-column composite primary
keys (\S\ref{sec:prod-shape}); the cross-city Redis/Pika deployment tests a
different data model and a real cross-city link (\S\ref{sec:cross-city}).

\subsubsection{Why encoding stays in the sidecar rather than the database}

A natural question is whether encoding can be pushed into the database and
expressed in SQL, avoiding out-of-database materialization.  We implemented
and measured this on MySQL.  Expanding each row into its $k$ buckets
replicates the fingerprint expression into every branch, and neither MySQL
nor PostgreSQL performs common-subexpression elimination while both inline
CTEs and derived tables by default, so the per-row hash is re-evaluated once
per branch.  This single mechanism explains the whole spread of the
measurement: under the direct form the path is roughly an order of magnitude
slower than the sidecar (about $11\times$ on MySQL, $4\times$ on
PostgreSQL), while forcing one-time CTE materialization
(\texttt{derived\_merge=off}), which evaluates each row's hash exactly once,
brings it to $1.42\times$ on our 609M-row production-shape table (886~s
versus 626~s for the sidecar; an explicit-temp-table variant measures
$1.79\times$).  All of these SQL variants produce sketches identical to the
sidecar cell by cell.

These numbers, however, understate why in-database encoding is not our
deployment form.  Every way of moving encoding into the database carries a
permission or portability cost that a reconciliation service cannot assume
across hundreds of customer tables.  A native aggregate or
extension---for example a PostgreSQL C aggregate, which can match or beat an
in-engine checksum baseline with parallel workers---must be installed on each
instance, which managed or hosted databases frequently forbid, and still
needs a thin adapter per engine.  Expressing the encoding in SQL avoids a
plugin but is not permission-free: the explicit-temp-table form needs
\texttt{CREATE TEMPORARY TABLES}, a privilege granted per database or per
table, so a service reconciling hundreds of tables requires a separate DBA
authorization per table---an approval and communication burden that scales
with the table count and made this route impractical.  The single-query form
avoids that grant (it needs only \texttt{SELECT} and session-level settings
and is compatible with a read-only account) and is the faster of the two, yet
it is still slower than the sidecar.  Because the SQL path must also be
re-ported per engine and is sensitive to optimizer, session, and deployment
settings, we keep one cross-engine encoding implementation in the sidecar;
its cost is that large-table fingerprints must be pulled out of the
database, which is part of the shared scan floor both methods pay.

\subsection{Production difference distribution}
\label{sec:prod-distribution}

\S\ref{sec:implementation} gave the implementation and evidence semantics.
This subsection examines the actual inputs of the system in production: we
statistically analyze 90 days (April 29--July 28, 2026) of table-level
reconciliation records, giving the distributions of difference cardinality,
table size, and running cost, and take one complete sharded reconciliation
run, covering seven table--task combinations, for the position of differences
in the primary-key rank space.  Each reconciliation of a table counts as one run: 41{,}603
table-level runs in total; runs where both sides are empty, where there is
no data to compare, or where the table has no primary key are excluded.
In production the two ends of a reconciliation usually sit in two machine
rooms in different cities (the source and target form a cross-cloud pair)---the
concrete setting behind the cross-datacenter design rationale of
\S\ref{sec:implementation}.

Small tables are usually equal and inexpensive, whereas cumulative
reconciliation time is concentrated in large tables with differences.  The
following subsections quantify difference frequency (\S\ref{sec:prod-eq}),
the span of difference cardinalities (\S\ref{sec:prod-span}), the
distribution of cost (\S\ref{sec:prod-cost}), and the position of differences
in the key space (\S\ref{sec:prod-rank}).

\subsubsection{The fraction of equal reconciliations drops sharply with table
size}
\label{sec:prod-eq}

Of the 41{,}603 runs, $d=0$ in 56.4\%, already below the ``mostly equal''
expectation; stratified by table size the drop is more pronounced:

\begin{table}[!hbtp]
\centering
\caption{Reconciliation runs by larger-side row count: counts, $d=0$ fraction,
and share of cumulative table-level time.}
\label{tab:prod-strat}
\small
\begin{tabular}{lrrr}
\toprule
larger-side rows & runs & $d=0$ fraction & cumulative table-level time \\
\midrule
$<10^3$     & 22{,}970 & 53.7\% & 0.7\% \\
$10^3$--$10^6$ & 13{,}812 & 66.5\% & 8.6\% \\
$10^6$--$10^7$ & 2{,}398  & 53.4\% & 13.4\% \\
$10^7$--$10^8$ & 1{,}969  & 33.2\% & 48.4\% \\
$>10^8$     & 454    & 1.3\%  & 29.0\% \\
\bottomrule
\end{tabular}
\end{table}

Only 1.3\% of runs above $10^8$ rows are equal, i.e., 98.7\% have
differences (Table~\ref{tab:prod-strat}); from the perspective of periodic
jobs the picture is more
extreme---all five periodic sequences above $10^8$ rows showed differences
within the 90 days.  The table segments that dominate cost almost always take
the difference-recovery path.  The equal fast path still matters, but it sits
on the many small tables that take seconds each.

Although 89\% of runs complete in a single first-round interaction, they
contain less than 0.2\% of all observed differences; the remaining runs
dominate difference recovery and motivate the failure-conditioned measurement
path.  The detailed breakdown is
in Appendix~\ref{app:e2}.

Sorting by cumulative time, the fraction with differences rises sharply
toward the expensive end: 43.6\% overall but 95.1\% in the most expensive
0.1\% of runs.

The logs do not identify the root cause of each difference, because NineData
operates the synchronization and verification service without access to the
customers' application semantics.  They nevertheless show that these
differences are both common and persistent: long-standing differences carry
96.4\% of the positive-difference time (Appendix~\ref{app:e2}), so the
cost-dominant tables take the recovery path almost every run.  The recurring
sources---active--active write conflicts, intermittent CDC/Kafka faults, and
residual cross-engine normalization mismatches---are discussed in
Appendix~\ref{app:e2}.

A large measured~$d$ can therefore also flag cross-engine normalization
bugs or pipeline faults that would otherwise go unnoticed, giving operations
and development a lightweight diagnostic signal independent of the
reconciliation task.

\subsubsection{Difference cardinality spans seven orders of magnitude with no
gap}
\label{sec:prod-span}

\begin{figure}[!hbtp]
\centering
\includegraphics[width=0.48\textwidth]{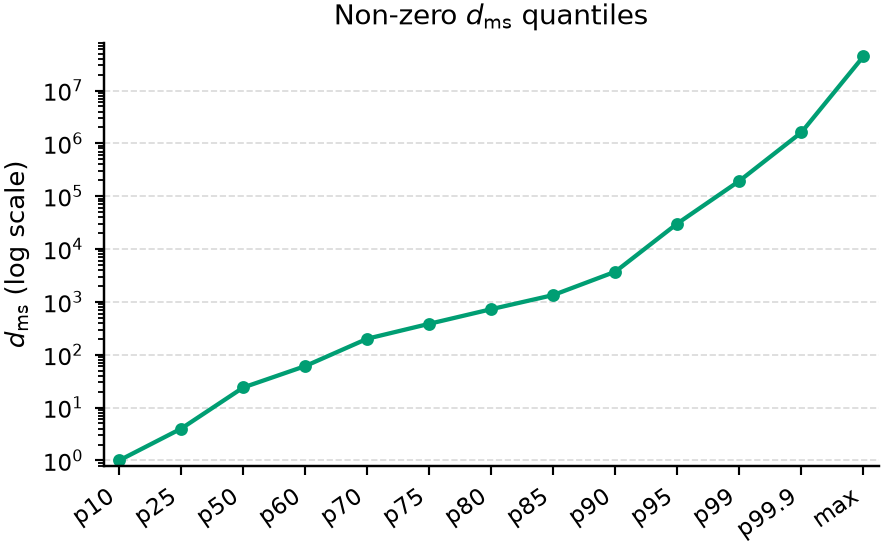}
\hfill
\includegraphics[width=0.48\textwidth]{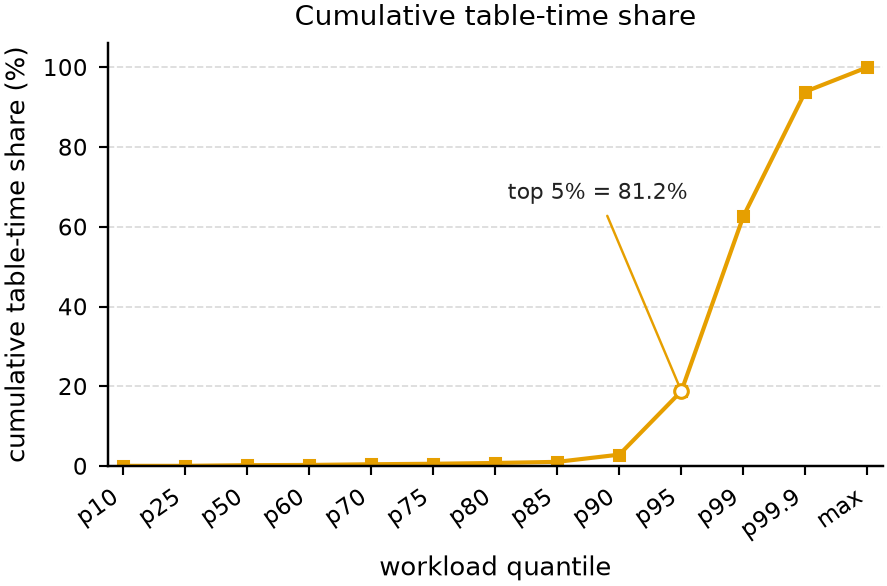}
\caption{Production difference profile.  Left: quantiles of non-zero
$d_{\mathrm{ms}}$ (log scale); right: cumulative share of table-level time.}
\label{fig:s62ab}
\end{figure}

Non-zero difference cardinality spans a wide range: the median $d$ is 24,
p99 is 194{,}698, and the maximum is $4.41\times10^7$---seven
orders of magnitude (Figure~\ref{fig:s62ab}, left; the full quantile ladder
is in Table~\ref{tab:e2-quantiles} of Appendix~\ref{app:e2}).  From here on,
``difference cardinality'' refers to the multiset semantics $d_{\mathrm{ms}}$
of \S\ref{sec:implementation}, written $d$.

The distribution of $d$ climbs smoothly from a median of 24 to tens of
millions, with no natural breakpoint for fixed capacity tiers; table size and
running time both show pronounced steps between p85 and p90, but $d$ shows
none.  The smooth distribution of $d$ therefore makes per-run capacity sizing
necessary: no preset capacity can cover both ends---configuring for the median
difference cardinality forces repeated doubling on large tables, while
configuring for large tables wastes bytes and build time on the majority small
tables.

\subsubsection{Reconciliation cost concentrates in a very narrow tail}
\label{sec:prod-cost}

The time distribution is extremely top-heavy: the median run takes 3
seconds, the 99th percentile nearly 3 hours, and the maximum 16.2 hours.
Most of the increase is concentrated within the five-percentile interval
around p90.

Cost concentration follows (Figure~\ref{fig:s62ab}, right): the cheapest
90\% of runs account for under 3\% of cumulative table-level time, while
the most expensive 5\% carry 81\%.  Within the extreme tail (the most
expensive 0.1\%), the median single-run time is 5.17 hours and the median
row count is $6.0\times10^8$.  The shape repeats within a single batch:
the slowest table alone accounts for 22.4\% of that batch's cumulative
table-level time, constraining the batch completion time.

The median table size is only 533 rows---half of the runs compare tables
of a few hundred rows---but scan volume, like time, concentrates on a few
very large tables: the cheapest 90\% of runs account for only 1.25\% of
cumulative scan volume.  The two ends of this size distribution correspond
to two different cost mechanisms: small tables are dominated by fixed
per-run costs---connection establishment, job dispatch, scheduling, and
result logging---while large tables are dominated by the scan itself.
Reconciliation cost must therefore be time-weighted, and after weighting it
concentrates on a few large tables.  The full quantile ladder is in
Table~\ref{tab:e2-size-quantiles} of Appendix~\ref{app:e2}.

For reference, in-service Merkle-style reconciliation is about 6--8\% slower
with differences than without, measured over the 37 periodic sequences with
sufficient runs; the marginal cost of recovering differences is small.

\subsubsection{Differences cluster in the key space with within-segment
scatter}
\label{sec:prod-rank}

\begin{figure}[!hbtp]
\centering
\includegraphics[width=0.7\textwidth]{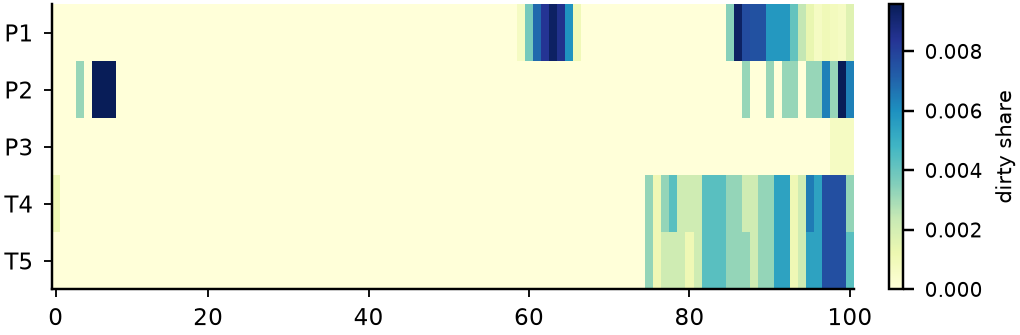}
\caption{Dirty-block distribution over the key-rank space for representative
table--task combinations: dirty blocks concentrate in a few key-rank segments.}
\label{fig:s62c}
\end{figure}

The sharded record comes from one complete sharded run on July 29--30, 2026,
covering the candidate tables of the three replay profiles
(\S\ref{sec:profiles}).  After dividing the primary-key rank space into 100
equal buckets, dirty blocks fall into only 3--24 of them---pronounced
key-rank-segment clustering, with the per-table rates and segment statistics in
Table~\ref{tab:e2-profiles} of Appendix~\ref{app:e2}.  At the same time,
the clustered segments contain many gaps, so production differences take the
form of clustering within a few key-rank segments with further scatter inside the
segments (Figure~\ref{fig:s62c}).

This segment-concentrated pattern is consistent with a partial worker outage
or an application-side failure that affects the key range the failed worker
serves, producing a contiguous dirty segment rather than randomly scattered
differences (the inferred fault sources are in Appendix~\ref{app:e2}).

The sharded log records only dirty-block membership, not per-block row
counts; row-level Jaccard and lookup amplification are outside the current
evidence scope.

\subsection{End-to-end comparison under production-shape replay}
\label{sec:prod-shape}

\S\ref{sec:prod-distribution} gave three key dimensions of the production
input: table size, difference cardinality, and difference location.  This
subsection realizes all three in one Oracle 11g--PolarDB/MySQL 8.0 comparison
environment and compares the end-to-end cost of self-sizing IBLT against the
partner's in-service Merkle-style localization.  With ample bandwidth both
methods are bound by the same full-table scan and tie; the question worth
asking is which method's cost remains bounded and predictable once the
network becomes the constraint.  Self-sizing reads the difference
cardinality before paying any recovery cost, so its transfer volume is pinned
by $\dhat$ wherever the differences fall; Merkle-style's cost is set by the
geometry of the differences and degrades to shipping the whole table on
adversarial inputs.  The replay below shows both faces at once: the tie under
ample bandwidth, and the widening gap under the bandwidth constraint an
actual cross-city link imposes.

The comparison experiments cannot be run directly on customer production
databases or replay customer data, so the experiments reconstruct controlled
datasets from the production schema, scale, primary-key shape, and difference
direction.
We refer to these controlled experiments as production-shape replay.  In this
deployment the Oracle side is the
online transactional system of record, while the MySQL side is a replica
serving analytical queries.  Row-fingerprint comparison does not depend on
any secondary index, and neither method's cost is affected by them, so the
replay keeps only the primary/unique indexes required by the comparison
(Appendix~\ref{app:e3}).

The two methods' configuration spaces are asymmetric.  The IBLT uses the
same $M_1=512$ first-round configuration as \S\ref{sec:implementation},
with the second-round capacity $M_2=\lceil\alpha\dhat\rceil$ of
\eqref{eq:capacity} (the $d_{\mathrm{rec}}=0$ setting adopted throughout,
\S\ref{sec:protocol}); no manual $d$ input and no per-table parameters are
needed.  The comparison method is tuned over its parameter surface---chunk
sizes, boundary positions, refineFanout, refinement granularity, and the
collation and index-plan sensitivity of composite-key sorting---fitting each
configuration before timing.  The two methods also scan differently:
localization reads each side's table in primary-key order to checksum the
shared logical ranges, while IBLT reads the table once in storage order;
\S\ref{sec:cross-city} develops the cross-engine cost argument, and here the
replay already favors localization by using the optimized boundary
generation of Appendix~\ref{app:e3}.  All formal runs are verified item by item against SQL truth, IBLT recovery
results, and Merkle lookup results.

\emph{Optimized Merkle-style baseline.}  In service, the boundary stage of
Merkle-style localization---a serial \texttt{ORDER BY} that generates the
shared chunk boundaries---dominates the run: two in-service composite-key
tables spend 30 and 36 minutes on boundary generation alone in tasks of about
two hours (composite-key boundaries are generated serially block by block:
block~$K$'s upper bound depends on block~$K{-}1$ and cannot be parallelized,
and multi-column \texttt{AND}/\texttt{OR} predicates resist efficient
optimizer plans, so even medium-sized tables take tens of minutes;
\S\ref{app:e3-impl} of Appendix~\ref{app:e3}).  The replay replaces this
serial boundary-generation step with \texttt{ROWID}/\texttt{NTILE} sampling,
reducing that phase from 30--36 minutes in the observed service runs to 5--17
seconds; the reported replay times therefore exclude much of the
boundary-generation overhead observed in service.  Self-sizing IBLT has no
boundary stage: its elements are order-independent,
and each side reads once in storage order with no boundaries and no sort.

\subsubsection{Choosing the three benchmark tables}
\label{sec:profiles}

The profile of \S\ref{sec:prod-distribution} fixes the load range the
comparison must cover: the cost body concentrates on a few large tables (the
most expensive 5\% carry 81.2\% of cumulative table-level time), while the
run-count body sits in the low-difference range.  From the logs we choose
three candidate tables---P1/P2/P3---covering three positions of the long tail
as the replay benchmarks for end-to-end comparison.  The three tables are
rebuilt locally as controlled data according to their schema and statistical
shape.

\begin{table}[!hbtp]
\centering
\caption{The three production profiles: scale and primary key,
$d_{\mathrm{row}}/d_{\mathrm{ms}}$, and their role in the production shape.
The per-direction (source-only / modified / target-only) breakdown is in
Table~\ref{tab:e2-profiles} of Appendix~\ref{app:e2}.}
\label{tab:profiles}
\footnotesize
\begin{tabular}{@{}p{0.9cm}p{5.8cm}p{2.2cm}p{3.6cm}@{}}
\toprule
profile & two-side scale and primary key & $d_{\mathrm{row}}$ / $d_{\mathrm{ms}}$ & production-shape role \\
\midrule
P1 & 609{,}671{,}577 / 609{,}668{,}227 rows; single-column primary key; 43 common business value columns
    & 48{,}240 / 93{,}130 & large table, wide rows, wide dirty region \\
P2 & 31{,}034{,}507 / 31{,}034{,}459 rows; three-column composite primary key; 31 common business value columns
    & 720 / 1{,}392 & medium composite key, multi-band differences \\
P3 & 154{,}356{,}834 / 154{,}356{,}830 rows; three-column composite primary key; 22 common business value columns
    & 18 / 18 & large composite key, tiny local difference \\
\bottomrule
\end{tabular}
\end{table}

The three profiles span three positions of the production cost
distribution.  P1 is the cost body: a 600M-row table with about 93{,}000
multiset differences, close to the most expensive 0.1\% slice.  P2 is a
mid-range case: 30M rows with a composite key and about 1{,}400 differences.
P3 represents the large-table, tiny-difference regime in which localization
is expected to perform particularly well: 150M rows but only 18
differences, where chunk summaries barely need to
drill.  Table~\ref{tab:profiles} gives the full specifications.

\subsubsection{The three benchmark tables and actual run paths}

Table~\ref{tab:profiles} summarizes the three formal profiles.  The two
difference semantics $d_{\mathrm{row}}$ and $d_{\mathrm{ms}}$ follow the
definitions of \S\ref{sec:implementation}; the IBLT \texttt{plus}/\texttt{minus}
reported in the formal runs is consistent with that definition.

The three profiles cover the protocol's two exits.  P3's 18 differences fall
within first-round decodability, and all 12 runs complete directly on
$M_1=512$, taking the 1-RTT fast path; the $\dhat$ read in those runs ranges
from 16.7 to 19.4, consistent with the true value 18, showing the estimate is
observable on the fast path as well.  P1 and P2 exceed the first-round
capacity, and their path is first-round failure, reading the pre-peeling
counts, setting the second-round capacity by $\dhat$, rehashing with a fresh
seed, and joint peeling; both
profiles perform exactly one capacity setting, with no second capacity
probe.

These runs provide a production-scale consistency check for the estimator.
P1's overload ratio $d/M_1\approx182$ leaves almost no first-round peeling;
across 12 runs, the sample mean of $\dhat$ is 93{,}382.4 against the true value
93{,}130.  For P2, the corresponding values are 1{,}353.8 and 1{,}392.  The
resulting second-round capacities range from 1.64--1.97 times the true
difference cardinality for P1 and 1.64--1.92 for P2.

For correctness, all three profiles' recovery results match the injected
ground truth item by item (the per-profile \texttt{plus}/\texttt{minus}
records and the locally rebuilt physical table sizes are in
Appendix~\ref{app:e3}).

\subsubsection{End-to-end cost and its composition}

\begin{table}[!hbtp]
\centering
\caption{End-to-end times (seconds) for the three production profiles,
formatted self-sizing IBLT / Merkle-style localization, each at the better of
W16 and W32.}
\label{tab:e2e}
\small
\begin{tabular}{lccc}
\toprule
network & P1 & P2 & P3 \\
\midrule
22.5\,ms / 100\,Mbps & 636.6 / 652.4 & 37.3 / 38.9 & 155.8 / 155.6 \\
22.5\,ms / 10\,Mbps  & 657.4 / 1020.3 & 37.3 / 46.8 & 149.5 / 137.5 \\
100\,ms / 10\,Mbps   & 654.4 / 1011.2 & 37.0 / 46.6 & 151.7 / 137.6 \\
\bottomrule
\end{tabular}
\end{table}

Table~\ref{tab:e2e} gives the end-to-end results under a unified semantics:
\texttt{chunkSize=10000}, three network tiers, two repetitions per configuration, each
method at the better of the W16 and W32 concurrency settings (16 and 32
concurrent workers).

Before reading these numbers, look at their composition.  Both methods'
end-to-end time is dominated by the same floor cost---reading the whole
table and computing a fingerprint per row; on the cost-body P1 this step is
more than nine tenths of self-sizing IBLT's end-to-end time, while the
algorithmic phases---subtraction, peeling, and recovery---together account
for only single-digit percentages: P1's phase stack (Figure~\ref{fig:e3-phases})
shows the two methods first paying the same full-table scan/checksum floor,
after which the divergence falls entirely on Merkle-style's dirty-block
drill-down, with the per-phase times in Table~\ref{tab:e3-iblt-phases} and
Table~\ref{tab:e3-merkle-phases} (both in Appendix~\ref{app:e3}).  This scan
cost is unavoidable for any full-table reconciliation and both methods pay
it.  The gaps compared below are therefore the extra each method pays above
the same scan floor: self-sizing adds almost nothing above the floor.

\begin{figure}[!hbtp]
\centering
\includegraphics[width=\textwidth]{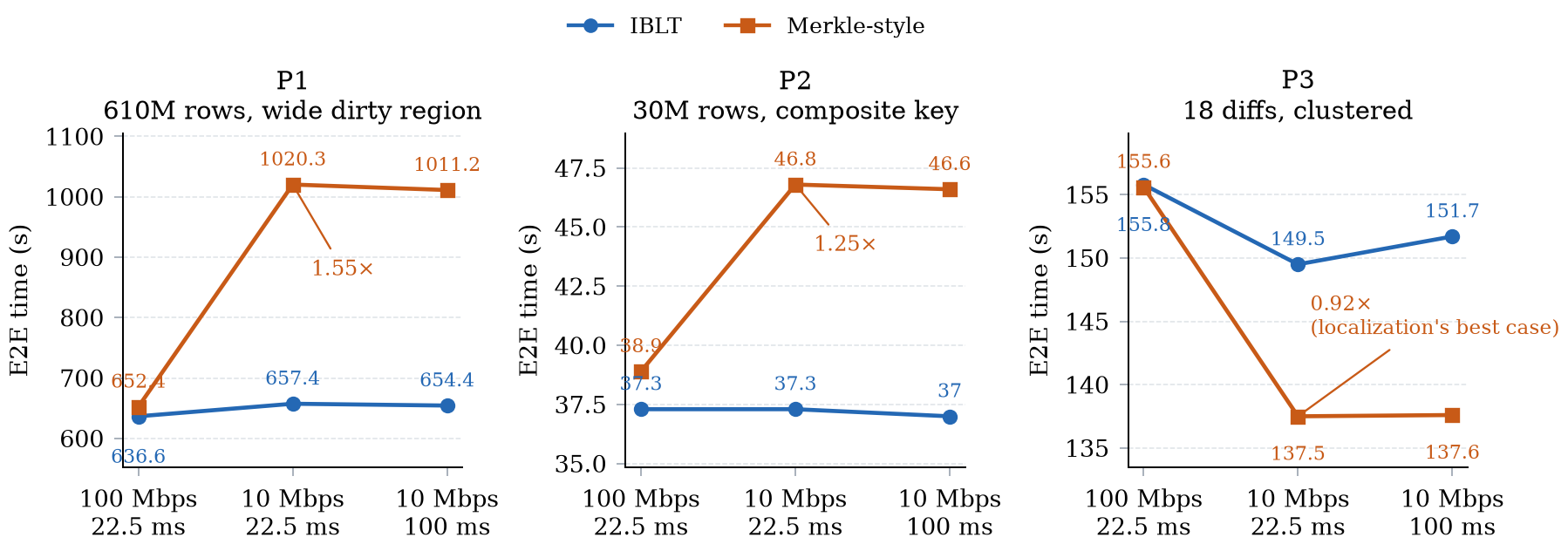}
\caption{P1/P2/P3 production-shape replay end-to-end times against network
tier (each subplot independently scaled; x-axis is the network tier).  With
ample bandwidth the two methods nearly coincide; under limited bandwidth
Merkle-style localization increases by 1.55$\times$ on P1 and 1.25$\times$
on P2, while the IBLT times remain within the observed run-to-run variation.
On P3, where 18 differences occupy a narrow key range, Merkle-style
localization is approximately 12\,s faster.}
\label{fig:s63}
\end{figure}

The comparison outcome is set by two factors: bandwidth decides whether a
gap appears, and the distribution of differences in the key space decides
its direction.  At 100\,Mbps the gaps on all three profiles fall within the
run-to-run noise (2.5\%, 4.3\%, 0.1\%; Figure~\ref{fig:s63})---when the
network is not the bottleneck, both methods converge to the common full-table
scan cost.
The gap appears at 10\,Mbps, which is also the actual operating range of the
cross-city link in \S\ref{sec:cross-city}: on the cost-body P1 self-sizing
leads by 1.55$\times$ (657\,s versus 1020\,s), P2 by 1.25$\times$ (37\,s
versus 47\,s), while P3's 18 differences concentrate in a tiny key range and
Merkle-style overtakes by about 12\,s (149.5\,s versus 137.5\,s).  The
direction and magnitude of the gap are primarily determined by how the
differences are distributed over the key
space, exactly the quantity self-sizing's cost does not depend on.

Transfer volume explains this direction.  For a 610M-row, 165\,GB P1 table,
the self-sizing first round sends only 16{,}400\,B per side and the
second-round sketch with about 164{,}000 cells (about 5.2\,MB), so the total
carried over the link is more than four orders of magnitude smaller than the
table itself, entirely determined by the measured $\dhat$, independent of
where the differences scatter in the key space.  Merkle-style partitions
P1's key space into 60{,}294 logical ranges, of which 677 are judged dirty;
refinement must re-read those ranges over the link in chunks of 10{,}000
rows, and at the 10\,Mbps tier this drill-down step alone adds about 389\,s,
the entire source of P1's gap.  One method's transfer volume is set by how
many differences exist, the other's by where they fall.

Placing the three profiles back into the production profile of
\S\ref{sec:prod-distribution}, the two methods' relative positions fall out
along two axes: run count and cumulative time.  By run count, the majority of
production tasks lie in P3's
low-difference range---the median difference cardinality is 24 and more than
half the comparisons have no difference; that is localization's advantage
region and also self-sizing's 1-RTT fast-path region, and the absolute gaps
there are only about ten seconds.  By cumulative time, the cost body lies on
P1-like large tables---the most expensive 5\% of tasks carry 81.2\% of
table-level time; that is the region where the gap widens to 1.55$\times$
under limited bandwidth.  This gap is measured after stripping out the
in-service Merkle-style boundary cost, whose serial \texttt{ORDER BY}
boundary generation alone takes 30--36 minutes on composite-key tables.

The conclusion of this section is therefore at the protocol level: a
capacity fixed by one measurement suffices, in the measured 22.5--100\,ms
RTT and 10--100\,Mbps bandwidth range, to keep the end-to-end cost within ten
percent of the in-service method in the low-difference region and 35\%
lower on cost-body tables under limited bandwidth---all without knowing $d$ in
advance and without per-table chunk and refinement tuning.  The production
workload contains both regions, and at comparison start one cannot know which
side the current table falls on; that is the value of run-specific capacity
measurement.

\subsubsection{The two methods' worst cases are asymmetric}

The worst case exhibits the maximum asymmetry between the two methods.  When P1's
differences are scattered---deliberately insert one differing row per about ten thousand
rows, so the chunk summaries lose filtering power---Merkle-style
localization's drill-down turns from re-reading a few dirty ranges into
re-shipping almost the whole table over the link: end-to-end time rises 278\%
at 100\,Mbps, to more than five hours at 10\,Mbps/22.5\,ms (a single stress
run of 18{,}502\,s, +1189\%), and the run was interrupted before completion
at 10\,Mbps/100\,ms due to
connection interruption (Figure~\ref{fig:p1extra}).  The same input raises
self-sizing IBLT's end-to-end time by only 2\%--9\%, because its transfer
volume is pinned by the measured $\dhat$, independent of how many ranges the
differences fall into (the complete matrix is in Appendix~\ref{app:e3}).

\begin{figure}[!hbtp]
\centering
\includegraphics[width=\textwidth]{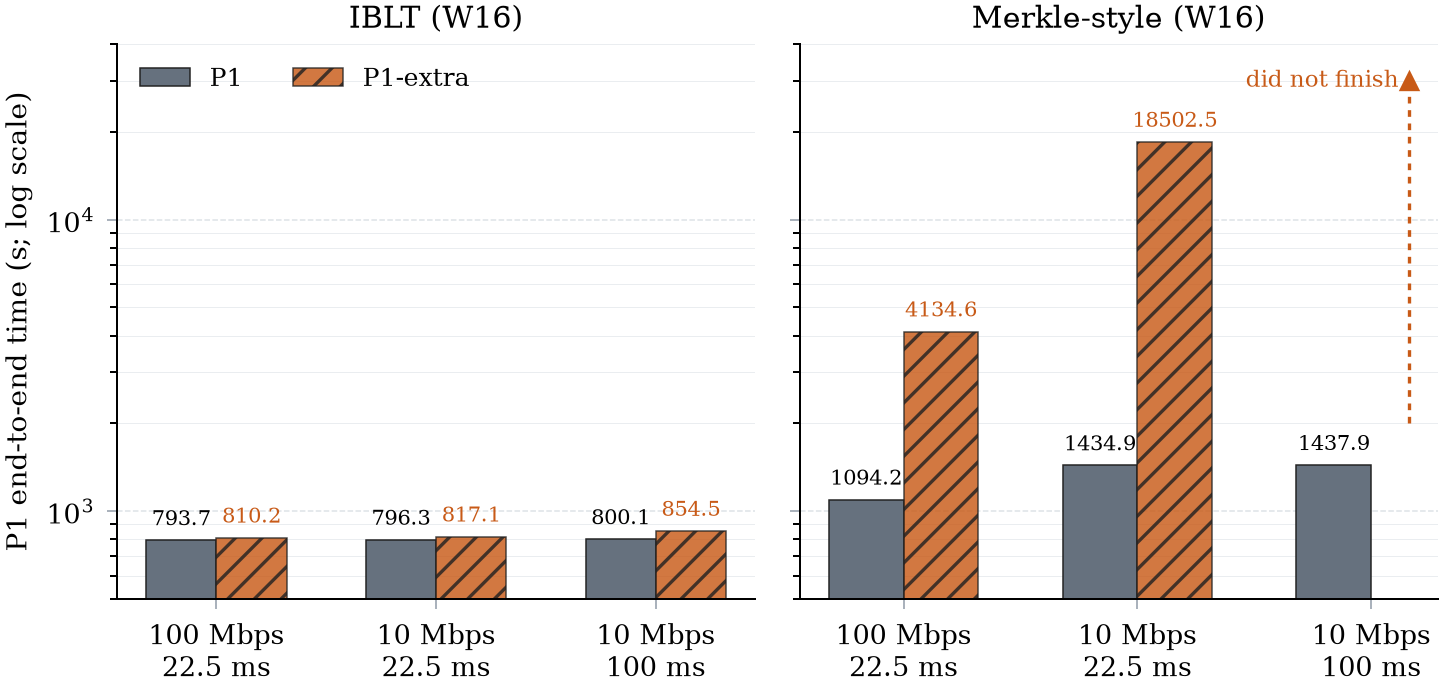}
\caption{Worst-case stress: P1 versus P1-extra with scattered differences
(log y-axis, all W16).  IBLT rises only 2\%--9\%, transfer volume pinned by
$\dhat$; Merkle-style localization's cost increases by 278\% at 100\,Mbps,
to more than five hours at 10\,Mbps/22.5\,ms (a single stress run of
18{,}502\,s, +1189\%), and the run was interrupted before completion at
10\,Mbps/100\,ms.  The complete matrix is in
Table~\ref{tab:e3-extra} (Appendix~\ref{app:e3}).}
\label{fig:p1extra}
\end{figure}

This asymmetry is structural, rooted in measurement-before-decoding.
Self-sizing barely moves under scattered differences because its transfer
volume is pinned by the measured $\dhat$, independent of how many ranges the
differences scatter into, so all IBLT runs in this stress set finish within
budget at only 2\%--9\% extra, without triggering \texttt{FALLBACK}.  The
same measurement provides an additional layer of protection:
when $\dhat$ shows the second-round capacity over budget, self-sizing rejects
the over-budget capacity before transmission and enters \texttt{FALLBACK};
Merkle-style has no such prior signal and must finish the summary and
refinement before it knows how many dirty ranges there are, by which time the
incurred cost has already exceeded five hours.  The evaluated Merkle-style implementation
requires a primary or unique key to define shared ordered ranges, whereas the
IBLT implementation is order-independent.  The estimate allows self-sizing
to reject an over-budget second-round capacity before transmitting the second
sketch and to enter \texttt{FALLBACK} explicitly.

\subsection{Cross-city KV field deployment}
\label{sec:cross-city}

The material in the previous two subsections is all relational: the
production profile of \S\ref{sec:prod-distribution} partitions by
primary-key ranges, and the end-to-end comparison of
\S\ref{sec:prod-shape} emulates cross-city latency and bandwidth with relays
on a lab cluster.  The cross-city Redis deployment on China Mobile
extends the evaluation in four respects: it runs on a real cross-city
production link; the
key space is determined by hash bucketing, so the geometry of differences is
completely different from a database's primary-key order; the operational
scenarios come from real operations, including migration interruption, a
heterogeneous target, cascading lag, and a dual-active partition; and in one
group of runs
the data under verification is continuously changing online, making the
end-to-end duration itself correctness-relevant.

The runtime environment, link, and Redis/Pika instances are production
infrastructure; the dataset is constructed to match production scale and
key-shape statistics, since production keys carry business content and
cannot be used directly.

\subsubsection{Deployment configuration, correctness, and the shared scan
cost}

The cross-city link is about 10\,Mbps with 28.8\,ms RTT.  The baseline data
has about 4.89 million keys, about 99.999\% of them STRING, averaging about
416\,B/key.  Merkle-style localization and self-sizing IBLT run on the
Redis$\rightarrow$Redis and Redis$\rightarrow$Pika data paths respectively,
both reading the same data and verifying item by item against the known
fault-injected ground truth.

There are six scenario groups, denoted G1--G6, covering a gradient from fixed
ground truth to online change.  G1 runs on Redis$\rightarrow$Redis with
$d=0/1$; G2 injects 100/500/937 differences, providing verifiable fixed
truth; G3 runs on Redis$\rightarrow$Pika with the three tiers plus a burst of
about $4\times10^4$ differences accompanied by link interruption; G4 keeps a
replication lag on the B$\rightarrow$C edge of Redis A$\rightarrow$B$\rightarrow$C;
G5 writes 1{,}000 unique keys on each side; G6 deliberately keeps the data
changing online (continuous LRU eviction at the target) while redis-shake
keeps synchronizing, to observe the snapshot window.  Injected differences
remain constant during comparison; synchronization catches up before
reconciliation starts, so the observed differences are the injected truth.

All six groups recover exactly the differences they should: the fixed
scenarios (G1, G2, G4, G5, and G3's injected tiers) match the injected ground
truth item by item with no false negatives or false positives, and G3's burst
and G6 recover the differences observed during the runs
(Table~\ref{tab:g1g6}).

A single-threaded full scan of about 4.89 million keys takes about 28.1\,s,
with a range of only 6.7\% across the six groups, independent of $d$.  This
is the scan cost both methods share and the body of self-sizing IBLT's
end-to-end time: scanning accounts for 93.7\%--99.1\% of IBLT's end-to-end
time across the five tiers of G1--G2.  As $d$ grows from 1 to 937---crossing
the one-round to two-round boundary---end-to-end time barely moves (28.7 to
29.6\,s), confirming that local scan dominates.  The two methods differ only
above this cost, in cross-city payload, dirty-range lookups, and whether a
rescan is needed.

Across the six runs, $\dhat$ deviates from truth by $-2.6\%$ to $+8.1\%$,
commensurate with the spread expected at $M_1=512$; every run that triggered a
second round set capacity once from the estimate and recovered completely, with
no second expansion.

\subsubsection{Three representative difference ranges}

These runs spread over three ranges of difference scale; the per-tier
end-to-end numbers are in Table~\ref{tab:g1g6}.  At the fixed-cost end (G1,
$d=0/1$) both methods remain close to the approximately 28\,s scan baseline
with nearly coincident
end-to-end times, and Merkle-style is slightly faster with no dirty buckets.
The medium range (G2) starts the split: Merkle-style's payload grows linearly
with the dirty-bucket count to about 105\,MB and its end-to-end rises to
about 121\,s, while IBLT stays flat at about 29--30\,s---after a first-round
failure it reuses the already-read fingerprints from the first-round scan,
sets capacity once, and recovers, re-bucketing without rescanning at an extra
cost of only about 1\,s.  In the high-difference G3 scenario (a burst of
about $4\times10^4$),
Merkle-style covers all
4{,}096 buckets with about 519\,MB of payload and about 512--523\,s
end-to-end, while IBLT finishes in about 80\,s.  This high-difference case
runs on the Redis$\rightarrow$Pika path: Pika's disk \texttt{SCAN} raises the shared scan
floor to about 73--80\,s, but the capacity signal and second-round behavior
remain consistent with Redis$\rightarrow$Redis, showing that the capacity
estimate does not depend on both ends using the same database engine.

\begin{table}[!hbtp]
\centering
\caption{Summary of the six scenario groups G1--G6.  G3's burst and G6's
observed differences are the values read from the two ends during the run, not fixed
ground truth.}
\label{tab:g1g6}
\footnotesize
\begin{tabular}{@{}p{1.1cm}p{2.6cm}p{2.8cm}p{2.8cm}p{3.6cm}@{}}
\toprule
group & differences & Merkle-style end-to-end & IBLT end-to-end & takeaway \\
\midrule
G1 & 0 / 1 & 28.3 / 28.3\,s & 29.5 / 28.7\,s & parity; Merkle slightly faster \\
G2 & 100 / 500 / 937 & 39.9 / 81.2 / 121.0\,s & 29.0 / 30.0 / 29.6\,s & IBLT more stable \\
G3 & $\approx 4\times10^4$ (burst) & $\approx 512$--$523$\,s & $\approx 80.1$\,s & gap from whole-bucket shipping \\
G4 & 100 / 500 / 937 (B$\rightarrow$C lag) & 40.5 / 81.9 / 122.5\,s & 30.2--31.0\,s & capacity set once under replication lag \\
G5 & 2{,}000 (bidirectional) & $\approx 202.6$\,s & $\approx 30.6$\,s, residual=0 & bidirectional recovered in one run \\
G6 & 1{,}234 / 1{,}163 & $\approx 159.5$\,s & $\approx 34.5$\,s & narrower observation window \\
\bottomrule
\end{tabular}
\end{table}

\begin{figure}[!hbtp]
\centering
\includegraphics[width=0.95\textwidth]{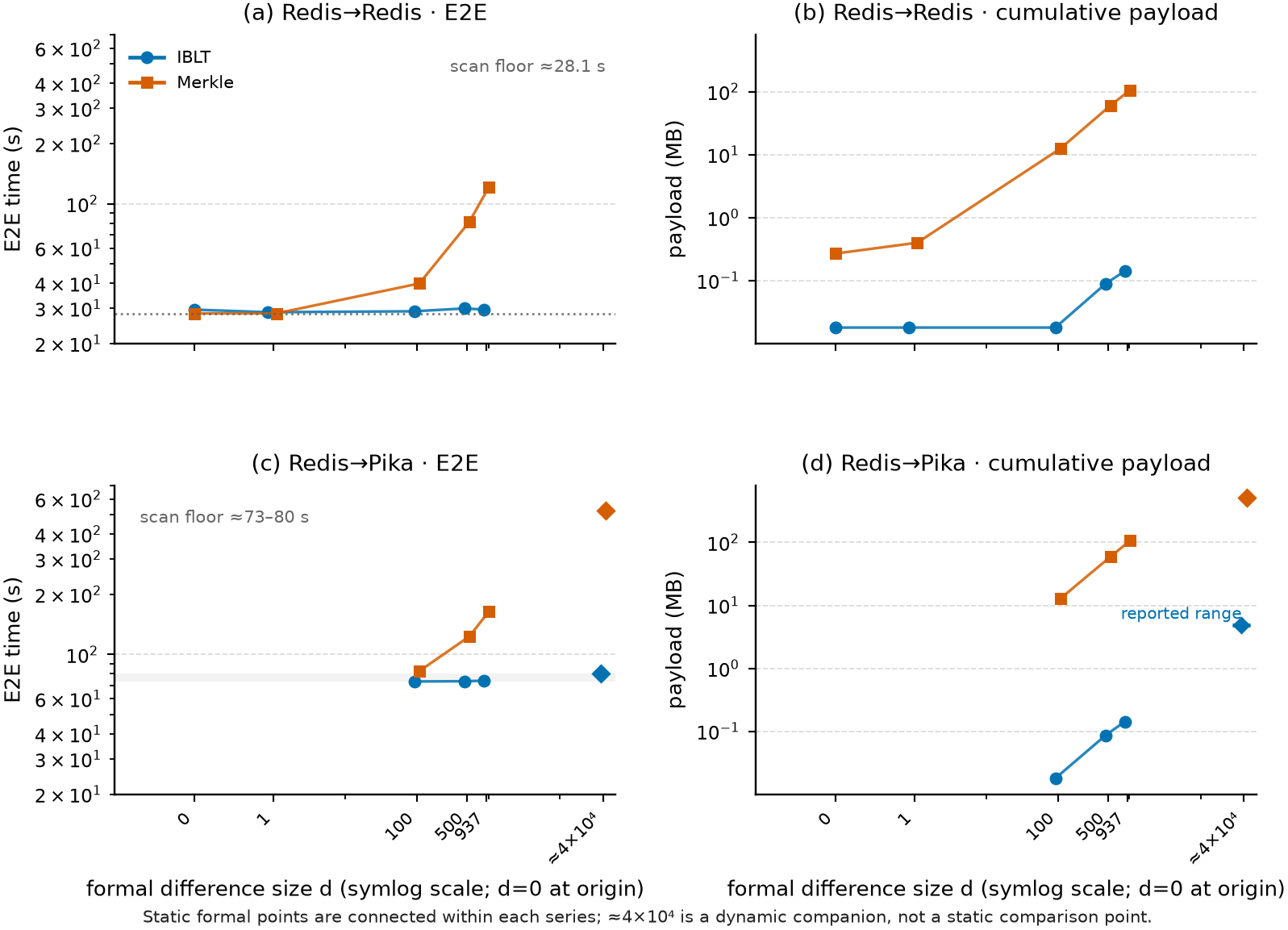}
\caption{End-to-end time and cumulative payload in the China Mobile
Redis$\rightarrow$Redis and Redis$\rightarrow$Pika deployment; the x-axis is
difference scale on a symlog scale with a zero point, the scan cost is marked
as a gray line, and the dynamic companion is marked separately.}
\label{fig:s64}
\end{figure}

Figure~\ref{fig:s64} juxtaposes the two data paths' scan cost, end-to-end
time, and cumulative control payload.  Fixed-truth points are connected by
difference scale; the dynamic companion around $4\times10^4$ uses a separate
marker, preserving the evidence boundary against the fixed-truth comparison.

G4 and G5's cascading-lag and bidirectional runs set capacity once and
recovered fully (Table~\ref{tab:g1g6}).  The dual-active configuration
uses different key prefixes; value conflicts for the same key and business
arbitration are not covered.

G6 deliberately keeps data changing online to observe behavior under change
(Table~\ref{tab:g1g6}): the observed time spans differ by about 4.6$\times$
between Merkle-style and IBLT.  When data changes online, the end-to-end
duration is the
span over which the data is observed; the shorter the path, the closer the
view to a single instant, so a short window is a correctness-relevant
benefit.  This single-shot run is used only to characterize the observation
window under continuous online change, with no fixed ground truth; the
snapshot semantics and observation-window
boundaries are discussed in \S\ref{sec:scope} and Appendix~\ref{app:e4}.

\subsubsection{Key-space structure determines localization cost}

The distribution of differences over buckets is determined by hash bucketing,
independent of how the differences were produced.  Of the five non-zero field
points, four were constructed as ``sequential tail missing'' (differences
appear as a missing run of contiguous keys); G6 was produced by real LRU
random eviction.  Both shapes scatter uniformly over the 4{,}096 buckets once
mapped, and the expected number of dirty buckets (buckets containing at least
one difference key) under uniform hashing is approximately
$4096(1-e^{-d/4096})$: at very small $d$, almost one difference per bucket,
and as $d$ grows the same bucket starts receiving multiple differences.  The
injection controls only the difference count and shape, not the dirty-bucket
count directly; the measured dirty-bucket counts are consistent with the
uniform-hashing model across the tested scenarios, including the LRU-eviction
scenario (G6).  This
directly determines Merkle-style localization cost: it checksums at bucket
granularity, and once a bucket is judged dirty, since it cannot know which
key inside differs, it ships the whole bucket's details over the link---a
bucket holds about $N/B$ keys on average, usually only one of which differs,
so each difference pays roughly the transfer cost of one whole bucket, and
payload and end-to-end time both grow approximately linearly with the
difference scale (the payload and end-to-end-time models and fitted
coefficients are in Appendix~\ref{app:e4}).

Merkle-style behaves completely differently in relational databases.
NineData's tables are chunked by primary-key order, and business differences
often concentrate in a few chunks: of 100 key-rank equal-width buckets only
3--24 contain dirty blocks, at a dirty-block rate of 0.20\%--10.99\%, so only
these blocks need to be shipped back and the cost is low.  KV's hash bucketing
breaks this premise: differences scatter uniformly and each difference pays a
whole bucket.  The two key spaces give opposite orderings: in this deployment
the ordering changes around $d\approx5$--$10$
(difference rate on the order of $10^{-6}$), with Merkle-style still faster
when differences are extremely
few and clustered (P3 of
\S\ref{sec:prod-shape}: 137.5\,s versus 149.5\,s).  The
number of dirty ranges---which depends on both difference cardinality and
placement---is the primary determinant of localization cost.

Localization carries a second premise: the two sides must read a shared
logical order cheaply.  Whether an ordered read is cheap depends on whether
the engine clusters by that order; in cross-engine pairs at least one
side---often both---reads in a logical order different from its storage
order, and when the two sides partition differently (a hash-sharded source
and a range-partitioned target, say), each side's ordered scan loses
sequential locality.  IBLT does not need this premise: its elements are
order-independent, so each side reads its local data once in storage order
with no aligned boundaries and no sort, turning a localization precondition
into a sequential full-table read.

\subsection{Production operation}
\label{sec:prod-operation}

The preceding subsections replay production shapes or deploy on production
infrastructure with constructed data.  Since September 2026, self-sizing
has also run in NineData's production reconciliation service on an
Oracle--MySQL replication link, on the same table tasks as the existing
Merkle-style service.  Operators use $M_1=512$ and $\alpha=2.0$.  Each run
below is one reconciliation of a source--target table pair.  The two
services use different scheduling and concurrency settings, and the
databases carry shared online load, so their wall-clock times are not a
controlled comparison.  Operator inspection found self-sizing run times in
line with expectations and generally somewhat shorter than those of the
Merkle-style service on large tables; below we report outcomes, estimation
accuracy, and the composition of run time.

\begin{table}[!htbp]
\centering
\caption{Production runs by completion path.}
\label{tab:prod-operation}
\small
\begin{tabular}{@{}lrrrrr@{}}
\toprule
path & runs & tables & rows (max) & $d$ & $\dhat/d$ \\
\midrule
1-RTT fast path & 398 & 31 & 14.3\,M & 0--385 & -- \\
second round & 36 & 5 & 31.3\,M & 446--13{,}310 & 0.868--1.004 \\
fallback & 3 & 3 & -- & -- & -- \\
\bottomrule
\end{tabular}
\end{table}

All 434 completed runs decoded successfully, and operator spot checks
confirmed that their results matched those of the Merkle-style service
(Table~\ref{tab:prod-operation}).  Of these runs, 398 (91.7\%) finished on the
1-RTT fast path, close to the 89\% that the trace-driven analysis of
\S\ref{sec:prod-distribution} predicts for $M_1=512$, and the remaining 36
finished in the second round.  About 45\% of completed runs confirmed $d=0$
with a single 16\,KB sketch.  The largest difference decoded in the first
round was 385 and the smallest that required a second round was 446, in line
with the peeling capacity of a 512-cell sketch.  The other three runs fell
back: two hit SQL errors during scanning before any sketch was built, and one
exited because the estimated second-round sketch would cost more than
transferring the keys directly.

Across the 36 second-round runs, $\dhat/d$ stays within 0.868--1.004, above
the failed-only lower quantile $q_{0.01}^F=0.8598$ at $M=512$
(Table~\ref{tab:t1b}), so every $M_2$ is at least $1.73\,d$, well above the
$k=3$ peeling threshold of about 1.22.  The upper end comes from a heavily
one-sided run on a 10-million-row table, with 14 source-only and 568
target-only rows, where $\dhat=584.5$ against $d=582$.  On a
0.53-million-row table, $d$ grew from 11{,}556 to 13{,}310 during the
deployment, and the second round followed this growth without any manual
capacity adjustment.

The fixed configuration also adapted in the other direction.  On one
day, two tables required a second round at $d=446$ and $d=582$; their
next runs completed in the first round at $d=362$ and $d=23$, with no
capacity reconfiguration.  The initial
deployment day ran with $M_1=20{,}000$, a 625\,KB sketch; 12 of its 13
tables had $d\le362$, below the largest difference later decoded in a
512-cell first round.  Counting one sketch per round at 32 bytes per
cell, the 434 completed runs sent 25.0\,MB of sketches in total, whereas
a fixed capacity of 20{,}000 cells would have sent 277.8\,MB, about
11 times as much.

On the fast path the estimate is not needed but is still recorded, which
allows the theory of \S\ref{sec:exact-theory} to be checked on production
data.  Runs with $d=0$ or $d=1$ (194 and 59 runs) return $\dhat=d$ exactly,
as the zero variance of Theorem~\ref{thm:variance} predicts.  
Among first-round runs with $d\ge3$, the resulting 53 combinations give a
mean $\dhat/d$ of 1.003 and a standard deviation of 6.7\%, on the scale of
the 5--6\% relative spread that Theorem~\ref{thm:variance} gives for
these~$d$.  Fully one-sided differences with $d\ge3$ appear in 29 runs on
12 tables; their 14 combinations give a mean of 0.996 and a standard
deviation of 4.9\%, so the estimate stays on the same scale under the most
extreme sign imbalance.

\begin{table}[!htbp]
\centering
\caption{Production run time by stage (medians unless noted).}
\label{tab:prod-timing}
\small
\begin{tabular}{@{}lrrrrrr@{}}
\toprule
path & runs & E2E & E2E (p90) & scan share & re-bucket + peel & key resolution \\
\midrule
1-RTT fast path & 398 & 5.0\,s & 33.6\,s & 98.8\% & $<$1\,ms & 3\,ms \\
second round & 36 & 15.6\,s & 184.5\,s & 67.9\% & 0.24\,s & 0.15\,s \\
\bottomrule
\end{tabular}
\end{table}

Table~\ref{tab:prod-timing} reports end-to-end time and selected stages.
As in the controlled deployment of \S\ref{sec:cross-city}, the full-table
scan, which reads every row, computes its fingerprint, builds the
first-round sketch, and writes the fingerprint to the local cache,
dominates run time.  On the fast path it accounts for a median of 98.8\%
of end-to-end time, and peeling a 512-cell sketch never took more than
16\,ms.  Key resolution maps recovered keys back to primary keys with a
pass over the local cache.  A second round adds re-bucketing, peeling,
and a larger recheck, but no rescan.  On the 31-million-row table, which
took the second round in all fourteen of its runs, the scan took a median of
123\,s, while re-bucketing the cached fingerprints into about 4{,}200
cells and peeling took 3.4\,s and key resolution took 14.2\,s.
End-to-end times also include database rechecks, whose SQL was not
specifically optimized for each database engine.

As in G6 of \S\ref{sec:cross-city}, rows may change between scan and
recheck.  On a 6-million-row table, the sketches recovered 370 keys with
mismatched values ($d=740$); at recheck none of them still differed, and
a run on the same table the next day found $d=0$.  Of the 434
completed runs, 240 produced candidate keys.  Candidates were filtered
out only when both sides held the key, in 20 runs (12 with every candidate
filtered); every one-sided candidate was reported as a source-only or
target-only row.  The outputs total 144{,}830 value mismatches, 598
source-only rows, and 2{,}285 target-only rows; these are cumulative over
runs, so a difference that persists is counted once per run.  The
sketch-level $d$, and hence $\dhat$, counts every difference visible at
scan time, which is what the second-round capacity must cover.

\section{Related Work}
\label{sec:related}

Prior work raises two questions: whether difference estimation requires
bytes beyond the decoding sketch, and how a protocol sets capacity when the
difference cardinality is unknown.  We first review set-reconciliation
methods and their treatment of unknown cardinality
(\S\ref{sec:rw-method}--\S\ref{sec:rw-approaches}), then compare the closest
counter-based estimators and prior uses of failed IBLTs
(\S\ref{sec:rw-measurement}--\S\ref{sec:rw-failure}), and finally discuss
consistency verification in database and storage systems
(\S\ref{sec:rw-systems}).

\subsection{Set reconciliation and IBLTs}
\label{sec:rw-method}

The classic algebraic solution to set reconciliation is
CPISync~\cite{minsky2003cpisync}: it recovers the symmetric difference by
rational-function interpolation of a characteristic polynomial, with
communication near the information-theoretic lower bound, at the cost of
interpolation and root-finding whose cost grows with the difference cardinality.
IBLT~\cite{goodrich2011invertible} takes a different route: it uses a
hash-based sketch that supports linear subtraction, the two sketches cancel
common elements cell by cell, and the remaining difference sketch is listed
element by element via peeling.  Difference
Digest~\cite{eppstein2011whatsdifference} builds a complete reconciliation
protocol on top of IBLT using a two-stage workflow: estimate the difference
cardinality, then configure the IBLT capacity
from the estimate.

Later work improves this structure from different directions.  PBS uses a
parity bitmap sketch to lower space and computation~\cite{gong2021pbs}.
Graphene combines a Bloom filter with an IBLT for block
propagation~\cite{ozisik2019graphene}.  Erlay uses set reconciliation as a
bandwidth optimization for Bitcoin transaction
relay~\cite{naumenko2019erlay}.  Multi-party generalizations appear
in~\cite{mitzenmacher2018multiparty}.

In these IBLT-based protocols, the capacity $M$ is configured externally,
using either a separate estimator or prior knowledge.
Where the difference cardinality comes from is left to a separate structure
or operator judgment.  This paper internalizes that question: the decoding
structure itself provides the estimate.

\subsection{Approaches to unknown difference cardinality}
\label{sec:rw-approaches}

Four families sit at different positions on the two dimensions of
``estimation-only bytes'' and ``capacity-sizing round trips.''
Pre-estimation summaries incur the fixed cost of a separate estimator;
blind retry avoids estimation-only bytes but requires additional rounds;
Rateless methods avoid preset capacity through incremental transmission;
self-sizing uses the existing decoding sketch for estimation and completes
its main path within two rounds.  We elaborate on these two dimensions
below.

\emph{Pre-estimation summaries outside the decoding sketch.}  Difference
Digest uses the Strata Estimator and min-wise sketches~\cite{eppstein2011whatsdifference,broder1998minwise,broder2000minwise}.
Strata is a layered set of small IBLTs whose successful strata recover true
difference elements; when all layers succeed it can complete recovery
directly in the small-$d$ range.  In general, each layer's cells decode
independently, and the protocol configures the subsequent main IBLT from the
estimate of the stopping layer.  Min-wise and PBS's Tug-of-War counters are
estimation-only summaries~\cite{gong2021pbs}, the latter's statistical basis
being the AMS $F_2$ estimation framework~\cite{alon1996spacecomplexity,alon1999spacecomplexity}.
Graphene relies on a sender-known set size and an analytically controllable
Bloom-filter false-positive rate~\cite{ozisik2019graphene}.  These methods
let the estimation phase be tuned independently by budget; the corresponding
cost is a structure and transmission separate from the main decoding sketch.

\emph{Blind retry with variable-capacity sketches.}  This is the more common
engineering default: try a conservative capacity, and on failure upgrade and
retry along a predetermined sequence of candidate sizes.  Gentili's thesis
recommends repeated doubling for small-file
synchronization~\cite{gentili2015setrecon}.  The industrial resizable-IBLT
patent family upgrades along candidate sizes after a decoding failure and
regenerates snapshots~\cite{nochlin2023patent,nochlin2022patentapp280,nochlin2022patentapp279}.
These approaches need no extra structure; the cost concentrates in rounds.
Each failed IBLT sketch's counts are not reused; under scan-dominated cost, each
additional round entails another full-table read or a cached fingerprint
replay.

\emph{Streaming transmission that sidesteps presetting.}  These methods avoid
capacity presetting through incremental transmission.
Rate-compatible IBLT~\cite{lazaro2023ratecompatible} and Irregular
IBLT~\cite{lazaro2021irregular} improve tolerance to capacity mismatch
through incremental transmission and degree design.  Rateless
IBLT~\cite{yang2024rateless} has the sender continuously produce coded
symbols while the receiver appends until decoding succeeds, eliminating the
capacity-presetting step by construction.  CertainSync~\cite{keniagin2025certainsync}
and ConflictSync~\cite{gomes2026conflictsync} proceed along similar lines.  The
listing-guarantee line~\cite{mizrahi2023listing} characterizes how many
elements one IBLT sketch can list under deterministic listing conditions.

\emph{Measurement within the decoding sketch.}
Before decoding is even attempted on the first-round IBLT, the pre-peeling
count vector is already a mapping-aware second-moment observation; after
decoding fails, the protocol reads $\dhat$ off the same IBLT sketch and
computes the second-round capacity in one step.  No separate measurement structure is
needed: the first decoding attempt adds no measurement fields or transmitted
bytes.  For incremental constructions the estimate improves as cells arrive
(\S\ref{sec:rateless}--\S\ref{sec:met}); what this paper uses is a fixed
first-round capacity, which simultaneously fixes the byte cost and the
estimation precision.

The main difference from Rateless IBLT lies in the deployment model: Rateless assumes
incremental symbol production, while databases use request/response with
scan-dominated cost.  This protocol sends a fixed first-round sketch and
requests one sufficient batch from the estimate.
Appendix~\ref{app:rateless-comparison} reports a head-to-head replay on
a 609M-row production snapshot: the incremental Rateless encoder is
${\sim}48\times$ slower and peaks at ${\sim}93$\,GiB endpoint memory, against
a few tens of MiB for self-sizing's endpoint, showing that incremental
symbol generation is poorly matched to the tested scan-dominated database
setting.

\subsection{Estimating difference cardinality from counters}
\label{sec:rw-measurement}

The statistical foundation of second-moment estimation is AMS~\cite{alon1996spacecomplexity,alon1999spacecomplexity}.
Building on the AMS $F_2$ tradition, we derive the estimator from an IBLT
count vector already transmitted for decoding and integrate it with
failure-conditioned capacity sizing.  The closest quantitative comparison is
PBS~\cite{gong2021pbs}:
it uses an independent Tug-of-War sketch in set reconciliation, where a
single squared observation is unbiased for $d$ with variance $2d(d-1)$, and
averaging $\ell$ independent observations brings the variance down to
$2d(d-1)/\ell$.  With $\ell=M-1$ independent observations, ToW's variance
formula reduces to the same expression as this paper's
\eqref{eq:variance}---the per-observation variance is identical; the two
estimators differ only in the carrier.

In its evaluation, PBS uses 128 ToW observations and applies a
$1.38\times$ empirical multiplier when sizing the IBLT.  The multiplier is
an empirical calibration for the evaluated configurations; PBS does not turn
it into a finite-sample, distribution-free bound on the under-capacity
probability (or an equivalent coverage guarantee).  Our analysis makes this
safety interface explicit through a failure-conditioned lower quantile.

The difference is in the carrier and the protocol cost.  PBS's counters are a
separate summary outside the decoding structure; this paper's $M-1$ effective
directions come from the count vector of the IBLT already transmitted for
reversible decoding.  As a pure estimator, ToW is cheaper: each counter is
about 4\,B, while a full IBLT cell is about 32\,B.  The protocol-level
advantage is that the same
first-round sketch adds zero estimation-only bytes, requires no extra round
trip for estimation, verifies equality when $d=0$, and directly decodes small
differences.  The head-to-head comparison under equal
budget is the four-method regret-versus-oracle result of
\S\ref{sec:regret}.

Existing finite Strata configurations use empirical calibration for capacity
sizing.  Our analysis instead provides the IBLT count estimator with exact
mean and variance, a fixed-dimension chi-square limit, and a
failure-conditioned lower-tail capacity interface.

The closest work in terms of input type is the Guo--Li CBF set
reconciliation~\cite{guo2013counting} (CBF continues the Bloom-filter and
Summary-Cache line~\cite{bloom1970spacetime,fan2000summarycache}).  The two
share one point of contact: after the two sides subtract, both read the
difference cardinality from a signed counter array; the estimation method,
operating range, and protocol goals are otherwise separate.  Guo--Li infer
$d$ from the fraction of zero cells and the ratio of positive to negative
cells in a $\{-,0,+\}$ ternary reduction, requiring the counter-array length
to be the same order as $d$ ($m=2d$--$8d$) and reporting relative-error
experiments; we keep the full integer counts and read the centered magnitude
\eqref{eq:centered-energy}, retaining exact second moments, a chi-square
limit, and finite-sample lower tails in deep overload $d/M_1\gg1$, so the
first-round capacity can be far smaller than $d$.  Structurally, CBF has no
reversible decoder and cannot recover difference elements, so there is no
peeling-failure event; the question of reading the same IBLT sketch after a
first-round failure and setting the second-round capacity from it lies
outside Guo--Li's structure.

The chi-square limit of Theorem~\ref{thm:chi2} has a classical ancestor:
when $k=1$ and all difference directions agree, our test statistic reduces
to Pearson's $X^2$ under an equiprobable multinomial.
Theorem~\ref{thm:chi2} generalizes this to signed, without-replacement
$k$-subset mappings; the derivation and algebraic relation are in
Appendix~\ref{app:thm3}~\cite{pearson1900criterion,cochran1934distribution}.

Statistical inference for sketch estimators---confidence intervals for
AMS/$F_2$~\cite{rusu2007statistical,rusu2008sketches}, random-projection
sketches~\cite{ahfock2021sketching}, and distribution-free conformal
intervals~\cite{sesia2022conformal}---is an independent literature; none
targets the count channel of a reversible IBLT or handles estimates
conditioned on decoding failure.

\subsection{Information in failed IBLT sketches and peeling theory}
\label{sec:rw-failure}

Prior work has also examined information retained after IBLT decoding failure.
Gentili asked whether the number of keys already listed by a failed listing
could decide the growth factor, and noted that in deep overload the signal
from singletons and listed keys almost disappears; his actual sync protocol
uses a separate Strata Estimator, returning to repeated doubling for
small files~\cite{gentili2015setrecon}.  Partial-extraction and finite-length
failure analysis characterize how many elements can be peeled and how likely
failure is~\cite{mizrahi2023listing,gadurek2026breaking}, using the output of
peeling or the residual-core event.  The 2-core phase-transition
theory~\cite{jiang2014parallel,jiang2016parallel,molloy2005cores} gives the
event characterization ``failure iff the 2-core is non-empty,'' studying
thresholds, failure probability, and residual-core size.  The industrial
patent family above also stops at the failure event itself: failure triggers
the next candidate size, and the public text does not use the count energy
of the failed IBLT sketch~\cite{nochlin2023patent}.

We read the full count vector before peeling rewrites it.  This is why deep
overload does not limit us: where singletons disappear, the count magnitude
still retains the scale of $d$.  The new problem this creates is the selection bias
of reading only after failure.  We handle it in two steps.  First, the
failure event of idealized peeling is measurable with respect to the mapping
hypergraph, so conditional exchangeability gives an exact fixed-sign
projection through the squared net directional fraction $\theta_\Delta^2$;
i.i.d. random signs cancel this projection on average and yield
$\E[\dhat \given F]=d$ exactly.  Second, the same projection gives an explicit
composition-aware conditional-bias bound, Proposition~\ref{prop:cond-chi2} transfers the full
chi-square law to failed instances when failure becomes typical, and
Proposition~\ref{prop:lower-tail} gives a distribution-free
failure-conditioned lower tail.  The Rademacher cross-term
cancellation and conditional Cauchy--Schwarz are standard tools; what this
paper does is connect failure measurability, count energy, and the timing of
the protocol's read into one complete guarantee.  Together, these results
provide a failure-conditioned guarantee for estimating $d$ from the
pre-peeling count vector.

\subsection{Consistency verification in database and storage systems}
\label{sec:rw-systems}

Database and storage systems commonly localize inconsistencies by
partitioning the key space into ranges and comparing checksums.  For example,
pt-table-checksum partitions tables by primary-key range, compares block
checksums, and examines inconsistent blocks at row
granularity~\cite{percona2026pttablechecksum}.  AWS DMS data validation
follows a similar chunked-comparison path~\cite{aws2026dmsvalidation}.  Merkle
trees provide a hierarchical form of this localization
strategy~\cite{merkle1988digital}: Dynamo and
Cassandra anti-entropy repair start from a root summary and drill down level
by level to dirty leaf blocks~\cite{decandia2007dynamo,lakshman2010cassandra}.
We call this family Merkle-style localization.

Its cost is driven by the number of dirty blocks, the chunk size, and the
tree depth, and is therefore sensitive to where differences fall.
Self-sizing IBLT's cost is one $O(N)$ scan plus $O(d)$ sketch and lookups,
independent of the spatial location of differences.  The cost crossover, the
respective advantage regions, and the two representative regimes---clustered
production differences and near-uniform difference locations in the Redis
deployment---are given with measurements and a cost model in
\S\ref{sec:cross-city} and Appendix~\ref{app:e4}.

\section{Scope and Limitations}
\label{sec:scope}

This section states the conditions under which the paper's conclusions hold,
how precise they can be, and the reach of the system evidence.
\S\ref{sec:model-assumptions} is the theorem-level model and channel
assumptions; \S\ref{sec:precision} is the precision and configuration
boundaries; \S\ref{sec:system-prereq} is the system prerequisites and
evidence scope.

\subsection{Model and channel assumptions}
\label{sec:model-assumptions}

\emph{Unit weight and general frequency semantics.}  The $d$ and $\dhat$ in
this paper count fingerprints in the symmetric difference under the
unit-weight model.  A plain insert or delete contributes one element; an
update produces two independently mapped fingerprints (old and new).  When
the same key has frequency difference $|\Delta\mathrm{freq}_x|\ge2$, the
multiple copies share the same mapping and the element-independence model no
longer applies; the target becomes $F_2=\sum_x w_x^2$.  The expectation
identities of Theorem~\ref{thm:unbiased} and~\ref{thm:master} still hold
under this generalization; the exact variance and confidence intervals need
separate derivation for the weighted model.

\emph{Mapping model.}  All conclusions rest on one mapping model: elements
are independent, each picks $k$ distinct cells uniformly, and hashing is
modeled as ideal random; the implementation matches this (the hash family
enforces $k$ distinct positions).  The fixed-sign projection of
Theorem~\ref{thm:cond-mean} requires the sign composition to be fixed
independently of the mapping randomness; its exact random-sign cancellation
additionally assumes i.i.d. uniform signs.  The model assumes that keys are not chosen adversarially
against a deterministic mapping.  Adversarial settings require keyed hashing
and a corresponding adversarial model.

\emph{Channel requirements.}  Measurement needs a signed, unclipped,
losslessly represented count channel.  The measurement method therefore
requires counters and cannot be applied to structures that retain only
occupancy bits or parity.  The counter width must cover
the absolute per-cell counts of each side's sketch; integer overflow or
saturation there directly destroys the second moment.  The measurement
signal lives in the centered fluctuations, of scale $\sqrt{dk/M}$, far
below the counts themselves---at the deep-overload tier of
\S\ref{sec:precision} ($d/M=1111$, $k=3$) each side's per-cell counts
reach about $1.7\times10^4$ while $\sqrt{dk/M}\approx58$.  Such
representation failures are not within Proposition~\ref{prop:capacity}'s $\delta$ budget: $\delta$
covers only underestimation from statistical fluctuation.  The implementation
must therefore propagate an overflow flag into the estimator status and output
\texttt{REJECT} (\S\ref{sec:boundaries}).
Within the non-overflow range,
deep overload does not saturate the second-moment measurement: on real-machine
data, at $d/M$ as high as $1111$, the error of $\dhat$ stays within
$\pm2.2\%$.  The practical representation boundary is set by the counter
width and $d/M$ together.

\subsection{Precision, configuration, and operating range}
\label{sec:precision}

A first-round IBLT at $M_1=64$ has relative standard deviation about 18\%
and a $1\%$ quantile around $0.63d$; this spread requires a relatively large
capacity multiplier $\alpha$ for high-confidence sizing.  The configuration
targets low fixed cost; changing $M_1$ adjusts the width without affecting
the center---the unconditional center is exact by
Theorem~\ref{thm:unbiased}.  Under failure conditioning,
Theorem~\ref{thm:cond-mean} gives the exact fixed-sign shift as a function of
$\theta_\Delta^2$, exact cancellation under i.i.d. random signs, and a
composition-aware deviation bound
(\S\ref{sec:failure-cond}--\S\ref{sec:exp-validation}).
In deep overload, Proposition~\ref{prop:cond-chi2} additionally transfers the
full fixed-$M$ chi-square law to the failed instances.  Near the peeling
transition, the paper uses failed-only empirical calibration for the
conditional shape.

\emph{Sources of $q_\delta$.}  We fix $\delta=0.01$ throughout: after a
first-round failure, lower-tail underestimation is controlled with $99\%$
conditional probability; the protocol receives the corresponding $q_{0.01}$
from the statistical layer and sets the second-round capacity by
$\alpha\ge\beta/q_{0.01}$.  Plain's $q_{0.01}$ can be obtained in two ways.
The default configuration uses the failed-only empirical quantiles validated
in \S\ref{sec:exp-validation} over the declared $(M,k)$ and load range, with
sampling guarantees.  When a
distribution-free guarantee is needed, the analytic value derived by
Proposition~\ref{prop:lower-tail} from the exact second moments and a failure
probability lower bound can be substituted directly into the same capacity
formula; that value is conservative at deep lower tails and may imply a
larger capacity multiplier.  The empirical value gives an efficient
configuration inside the operating range; the analytic value gives a
configuration that does not depend on empirical tail shape.  Both correspond
to the same end-to-end guarantee interface.

The portability boundary of this interface is set by the mapping statistics
and configuration together.  Under the mapping model of
\S\ref{sec:model} and within the parameter tiers validated in this paper,
the corresponding Plain configuration applies directly; changing the
hash mapping, sketch structure, or operating range requires re-obtaining the
$q_{0.01}$ for that configuration, either from an applicable analytic result
or from failed-only calibration with sampling guarantees.

Next, asymptotic conclusions and finite-sample certificates.  The chi-square
confidence interval of Theorem~\ref{thm:chi2} is an asymptotic statement for
fixed $M$ as $d\to\infty$; Proposition~\ref{prop:cond-chi2} carries the same
limit to failed instances as $p_F\to1$.  The operating range uses the empirically
calibrated quantiles of \S\ref{sec:exp-validation} instead, with a
finite-sample certificate from the supplementary material
(Proposition~\ref{prop:kolmogorov}); its explicit constant is a worst-case
asymptotic certificate and is too conservative for parameter calibration.

The four adapters share the mapping-aware construction of
Theorem~\ref{thm:master} and failure-conditioned point-estimate semantics.
Interval guarantees differ: Plain has exact second moments and a
finite-sample lower-tail bound; Rateless has an asymptotic interval for a
fixed prefix; Irregular and MET use empirical quantiles.  This boundary is
at the distributional-guarantee level and does not affect the family-level
conclusion of unconditional unbiasedness.

\subsection{System prerequisites and evidence scope}
\label{sec:system-prereq}

First, second-round data access.  The protocol requires a fresh hash seed for
the second round; what changes is the element-to-cell mapping.  The sender
can rehash the row-fingerprint stream formed during the first-round scan
without re-reading the underlying table.  Caching and replaying row
fingerprints is a common implementation pattern in database verification and
the preferred path for scan-expensive deployments.  If the second round is
instead built by re-reading the underlying table, it must reuse the
first-round logical snapshot; rescanning the current table under online
change would change the encoded element set and fall outside the
assumptions of the two-round mechanism and joint decoding of
\S\ref{sec:protocol}.  The end-to-end wall clocks of
\S\ref{sec:prod-shape}--\S\ref{sec:cross-city} include the data-access cost
actually used by the tested implementations.  Pushing encoding back into the
database---per-engine SQL tuning, native operators, or UDFs---could remove
the cost of pulling large-table fingerprints out of the database, but each
form is gated by extension-install or per-table grant permissions that
managed deployments may not grant; we leave this to future work.

Second, observation time and snapshot semantics.  The protocol's input is
the row-fingerprint stream read by both sidecars within their scan windows;
the exact-recovery guarantees hold relative to these two input sets.
Building a sketch takes non-zero time; batched queries in relational
databases, full-store traversal in KV stores, and paged enumeration in object
stores correspond to a single instant only when the underlying store provides
a fixed transactional snapshot or equivalent mechanism.  Plain Redis
\texttt{SCAN} is an example lacking snapshot semantics.  If business writes
or replication catch-up continue during the scan, the two sides' scans, and
even two verification paths executed sequentially, may cover slightly
different time windows, so a run may not admit a unique instantaneous ground
truth; a small amount of state change inside the window appears as
candidate-set differences.  Controlled experiments stop replication and
verify against known ground truth.  Online experiments report the running
view within the scan window and eliminate stale differences via lookup
re-verification (\S\ref{sec:cross-city}, G6).  Production deployments that
require same-instant semantics need external mechanisms---write quiescence,
MVCC, storage snapshots, or a common replication position---to establish a
consistent input; this protocol does not provide one.  A shorter
verification path reduces the observation window for drift but does not
replace a consistent snapshot.

Third, beyond budget.  The first-round estimate only provides a computable
\texttt{FALLBACK} signal (\S\ref{sec:boundaries}).  What happens beyond the
boundary---manual inspection, full snapshot, partition comparison, or direct
rebuild---is an operational decision outside this paper.

The system evidence addresses three questions: what the real workload looks
like (the 90-day production profile of \S\ref{sec:prod-distribution}), how
much the two methods differ in cost (\S\ref{sec:prod-shape}), and whether
the conclusions hold under a different data model (\S\ref{sec:cross-city}).
Together they cover workload distribution, cost mechanism, and
cross-environment reproduction.  The rounds and bytes from the production
logs are trace-driven analysis using the true $d$; online decisions use the
$\dhat$ read by the protocol.

\section{Conclusion}
\label{sec:conclusion}

The capacity puzzle of set reconciliation comes from a seemingly
contradictory fact: the sketch size must match the difference cardinality,
which is precisely the quantity neither party knows.  This paper's answer is
to let the decoding structure measure itself---an IBLT has already completed
an observation before it attempts to decode, and failed IBLT sketches and
incomplete prefixes are all paid-for mapping-aware measurements.

The statistical layer characterizes this count channel.
The estimator is exactly unbiased with closed-form variance, admits a
chi-square confidence interval, and has an exact failure-conditioned
sign-composition projection with random-sign cancellation.  In deep overload,
the same chi-square law governs the failed IBLT sketches the protocol reads.
The
mapping-aware theorem \eqref{eq:master} generalizes this construction
to Irregular, Rateless, and MET IBLTs, unifying the estimate interface
across the family.

The protocol layer turns the estimate into one capacity setting: after a
first-round failure the second-round capacity is computed in one step, the
interaction closes within 2 RTTs whatever the difference cardinality turns out
to be, the second-round success probability is \eqref{eq:success-bound}, and
joint peeling of the first-round residual core adds no transmitted bytes.
Bounding the protocol at two rounds makes its interaction cost predictable,
while scan, hashing, and transfer costs remain determined by the workload,
sketch capacity, and network link.  Against blind doubling and
Strata-first, the byte-regret comparison shows self-sizing holding the round
count at about 2 at $1.3$--$1.7\times$ the oracle's bytes, while blind
doubling's rounds grow logarithmically with the difference cardinality and
Strata-first's fixed summary is expensive at the small-difference end.  The
systems layer validates this chain: a 90-day NineData production profile
shows difference cardinality spanning seven orders of magnitude, with cost
concentrated in large tables that almost invariably carry differences;
in end-to-end comparison, self-sizing matches
Merkle-style localization within ten percent at low differences and leads
clearly under bandwidth constraints; and in the cross-city Redis/Pika
deployment, IBLT's wall-clock barely moves as $d$ grows from 1 to 937
while localization's payload grows linearly.  In production operation, all
completed runs decoded within two rounds, over 90\% on the 1-RTT fast path,
and the fast-path estimation errors stayed on the scale the theory
predicts, including on fully one-sided differences.

The current theory and evaluation focus on scan-expensive request/response
deployments with signed, unclipped count channels.  The variance and
confidence intervals of the weighted frequency model, further reduction of
the joint capacity, optimal stopping for Rateless sequential estimation, and
an impossibility characterization for unsigned channels remain future work.

\emph{IBLTs measure before they decode; failed IBLTs are not wasted.}

\section*{Data and Code Availability}
\addcontentsline{toc}{section}{Data and Code Availability}

The actively maintained source repository is available at
\url{https://github.com/whitewum/self-sizing}. Its releases are archived at
Zenodo under the DOI
\href{https://doi.org/10.5281/zenodo.22198044}{10.5281/zenodo.22198044},
which resolves to the latest version.

\section*{Acknowledgments}
\addcontentsline{toc}{section}{Acknowledgments}

We thank Byungwoong Yoo (Independent Researcher) for identifying the calibration mismatch between the theoretical fixed-degree estimator and the mapper implementation.

\bibliography{references}

\appendix
\numberwithin{figure}{section}
\numberwithin{table}{section}
\section{Complete Notation}
\label{app:notation}

This appendix collects the symbols used across the paper.  The main text
keeps the minimal set needed for the contribution chain: \S\ref{sec:model}
states the core symbols of the measurement theory, and local symbols are
defined where they appear in individual theorems.  Unless stated otherwise,
vectors are column vectors, randomness comes from element mappings (and, when
explicitly declared, random signs), and $\E[\cdot]$, $\Var(\cdot)$,
$\Cov(\cdot)$, $\tr(\cdot)$ denote expectation, variance, covariance, and
matrix trace.

\subsection*{Plain IBLT and the core estimator}

\begin{center}
\small
\begin{tabular}{@{}>{\raggedright\arraybackslash}p{3.2cm}>{\raggedright\arraybackslash}p{12.4cm}@{}}
\toprule
symbol & meaning \\
\midrule
$A,B$ & the two sets in reconciliation; difference direction fixed as $A-B$ \\
$M$ & number of cells in a Plain IBLT; $M_1$ first-round capacity, $M_2$ guaranteed independent second-round capacity \\
$k$ & number of distinct cells each element maps to, $1\le k<M$ \\
$d$ & cardinality of the symmetric difference in the unit-weight set model, $d=|A\triangle B|$; a modified row contributes its old and new fingerprints once each \\
$F_2,\widehat F_2$ & second moment of the general frequency model and its estimate: with $w_x=\operatorname{freq}_A(x)-\operatorname{freq}_B(x)$, $F_2=\sum_x w_x^2$; in the plain set model every difference element has $w_x\in\{\pm1\}$, so $F_2=d$ and $\widehat F_2=\dhat$ (see \S\ref{sec:model-assumptions}) \\
$d_+,d_-$ & element counts unique to each side: $d_+=|A\setminus B|$, $d_-=|B\setminus A|$, so $d=d_++d_-$ \\
$\theta_-$ & fraction of elements from $B\setminus A$ when the composition is preset;  $\theta_-=d_-/d$ \\
$\theta_\Delta$ & net directional fraction of the fixed sign composition, $\theta_\Delta=1-2\theta_-=(d_+-d_-)/d$; $\theta_\Delta^2\in[0,1]$ measures one-sidedness \\
$x$ & difference element index, $x\in\{1,\dots,d\}$ \\
$s_x$ & side marker of element $x$: $+1$ for $x\in A\setminus B$, $-1$ for $x\in B\setminus A$ \\
$S_x$ & the set of cells hit by element $x$, $S_x\subseteq[M]$, $|S_x|=k$ \\
$a_x$ & the $M$-dimensional $0$-$1$ indicator vector of $S_x$ \\
$C$, $C_i$ & pre-peeling difference count vector $C=\sum_x s_x a_x$; $C_i$ the count of cell $i$ \\
$I_M,\one_M$ & $M\times M$ identity and the $M$-dimensional all-ones vector \\
$Q$ & centering matrix; Plain: $Q=I_M-\one_M\one_M^\top/M$ \\
$\bar C$ & count-array mean, $\bar C=M^{-1}\sum_i C_i$ \\
$T$ & centered energy, $T=C^\top Q C=\sum_i(C_i-\bar C)^2$ \\
$\gamma$ & per-element normalized self-energy; Plain: $\gamma=k(1-k/M)$ \\
$\dhat$ & difference cardinality estimate; Plain: $\dhat=T/\gamma$ \\
$H_{xy},W_{xy}$ & overlap count of two elements' hit sets $H_{xy}=|S_x\cap S_y|$ and its centered version $W_{xy}=H_{xy}-k^2/M$ (used in the variance proof, Appendix~\ref{app:thm2}) \\
$\sigma^2$ & per-element centered-mapping covariance scale on $\one_M^\perp$; $\sigma^2=k(M-k)/[M(M-1)]$ \\
\bottomrule
\end{tabular}
\end{center}

\subsection*{Family-level mappings and adapters}

\begin{center}
\small
\begin{tabular}{@{}>{\raggedright\arraybackslash}p{3.2cm}>{\raggedright\arraybackslash}p{12.4cm}@{}}
\toprule
symbol & meaning \\
\midrule
$j$ & element or data-type index; scope given by the local section \\
$d_j$ & number of difference elements of type $j$ \\
$\mu_j,\Sigma_j$ & mean and covariance of a type-$j$ element's hit vector \\
$D,D_x$ & Irregular IBLT random degree; $D_x$ the degree of element $x$ \\
$\Lambda,\lambda_i$ & Irregular degree distribution as generating polynomial $\Lambda(x)=\sum_i\lambda_i x^i$, $\lambda_i=\Pr[D_x=i]$ \\
$q_i$ & true hit probability of cell $i$ in a Rateless prefix \\
$B_i, B_i^{(x)}$ & indicator of whether one Rateless mapping path hits cell $i$; superscript $(x)$ for element $x$ \\
$C_0,Y_i$ & Rateless start-cell count and centered count $Y_i=C_i-q_iC_0$ \\
$D_j,K_{jv}$ & Rateless path jump distance from cell $j$, and transition probability from $j$ to $v$ \\
$m$ & Rateless prefix length; in the MET local model, the total number of cells of a cell type \\
$t$ & MET cell-type index \\
$m_t,r_t,a_{tj}$ & MET type-$t$ total cells, currently received cells, and the number of cells chosen by data type $j$ within it \\
$C_t,Q_{r_t},E_t$ & MET type-$t$ received count vector, its local centering matrix, and centered energy \\
$\gamma_{tj}$ & MET partial-type normalization coefficient of data type $j$ in cell type $t$ (Proposition~\ref{prop:met-identifiability}); complete type is $r_t=m_t$ \\
$J,f$ & number of MET data types, and the vector of per-type second moments $f=(F_{2,1},\dots,F_{2,J})^\top$ \\
$E,\Gamma$ & type energy vector $E=(E_t)_t$ and coefficient matrix satisfying $\E[E]=\Gamma f$ \\
$w$ & linear weight recovering total $F_2$ from type energies, $w^\top\Gamma=\one_J^\top$ \\
$r$ & local degrees of freedom or rank; in Proposition~\ref{prop:rateless-ci}, $r=m-1$; in Lemma~\ref{lem:met-partial}, the number of received cells \\
\bottomrule
\end{tabular}
\end{center}

\subsection*{Failure conditioning and the protocol}

\begin{center}
\small
\begin{tabular}{@{}>{\raggedright\arraybackslash}p{3.2cm}>{\raggedright\arraybackslash}p{12.4cm}@{}}
\toprule
symbol & meaning \\
\midrule
$F$ & first-round peeling-failure event \\
$ok_2$ & event that the second-round decoder fully recovers the currently unknown elements \\
$p_F,p_{F,\min}$ & failure probability $p_F=\Pr[F]$ and a lower bound on it in the declared operating range \\
$c_F$ & Plain estimator's normalized failure-conditioned pairwise cross-moment, $c_F=\gamma^{-1}\E[W_{xy}\given F]$ for any $x\ne y$; conditional exchangeability makes it independent of the fixed sign pattern \\
$\varepsilon_{M,d}$ & upper bound $dk(1-k/M)^{d-1}$ on peeling success in Proposition~\ref{prop:cond-chi2} \\
$\delta$ & risk budget for estimating below the safe bound; $\delta=0.01$ throughout, a $99\%$ conditional lower-tail target \\
$q_\delta$ & failure-conditioned relative lower-quantile guarantee: $\Pr[\dhat\ge q_\delta d \given F]\ge1-\delta$ \\
$U,F$ superscripts & quantile sampling semantics: all-trial (unconditional) and failed-only (failure-conditioned), e.g.\ $q_\delta^U$, $q_\delta^F$; no superscript means the protocol interface, the failure-conditioned one \\
$\beta$ & decoder capacity--success operating point; target capacity at least $\lceil\beta d\rceil$ \\
$\beta_{\mathrm{joint}},\beta_{\mathrm{ind}}$ & asymptotic capacity coefficients of joint peeling and of an independent second-round table \\
$\alpha$ & engineering capacity multiplier, $\alpha\ge\beta/q_\delta$ \\
$\eta_{\mathrm{dec}}(\beta)$ & upper bound on the decoder failure probability in the declared operating range once capacity reaches $\lceil\beta d\rceil$; derived from $\beta$, not configured separately \\
$\rho$ & confidence failure probability; confidence level $1-\rho$ \\
\bottomrule
\end{tabular}
\end{center}

The same letter may play different roles in local models, mainly $j$ (type
index), $m$ (Rateless prefix length or MET type size), and $r$ (degrees of
freedom or rank); each theorem restates its scope.  Experiment-specific
notation (e.g., $\widehat p_F$, $n_F$) is defined locally next to the
corresponding tables.

\section{Proof of Theorem~\ref{thm:unbiased} (Exact Unbiasedness)}
\label{app:thm1}

We use the notation of the main text: $C=\sum_x s_xa_x$,
$Q=I_M-\one_M\one_M^\top/M$, $T=C^\top Q C$, and
$\gamma=k(1-k/M)$.

\begin{proof}
Substituting $C=\sum_x s_x a_x$ and expanding bilinearly,
\[
  T=C^\top QC=\sum_{x=1}^d\sum_{y=1}^d s_xs_y\,a_x^\top Qa_y.
\]
From $Q=I_M-\one_M\one_M^\top/M$ and $a_x^\top\one_M=k$ (each element hits
exactly $k$ cells), any pair $x,y$ satisfies
\[
  a_x^\top Qa_y
  =a_x^\top a_y-\frac{(a_x^\top\one_M)(\one_M^\top a_y)}{M}
  =\lvert S_x\cap S_y\rvert-\frac{k^2}{M}.
\]

\emph{Diagonal terms ($x=y$).}  Since the $k$ positions are distinct,
$\lvert S_x\cap S_x\rvert=k$ always, so
\[
  a_x^\top Qa_x=k-\frac{k^2}{M}=\gamma
\]
is a \emph{deterministic} constant: it holds exactly without taking
expectations and is independent of which $k$-cell subset the element
selects.

\emph{Cross terms ($x\ne y$).}  Write $H_{xy}=\lvert S_x\cap S_y\rvert$.
Fix $S_y$; $S_x$ is an independent uniform random $k$-subset, and $H_{xy}$
counts how many of the $k$ cells of $S_x$ land in $S_y$ (the $k$ ``success''
positions among the $M$ cells), so
\[
  H_{xy}\mid S_y\sim\operatorname{Hypergeometric}(M,k,k),
  \qquad
  \E[H_{xy}\mid S_y]=k\cdot\frac{k}{M}=\frac{k^2}{M},
\]
and the conditional distribution does not depend on the particular value of
$S_y$.  By the law of total expectation, $\E[H_{xy}]=k^2/M$, hence
$\E[a_x^\top Qa_y]=\E[H_{xy}]-k^2/M=0$.

\emph{Combining.}
\[
  \E[T \given s]
  =\sum_x s_x^2\,\gamma+\sum_{x\ne y}s_xs_y\cdot0
  =d\gamma,
\]
where the last step uses $s_x^2\equiv1$.  Signs enter the diagonal terms
only squared, so the conclusion holds exactly for every sign pattern, in
particular for the fully one-sided case $d_-=0$.
\end{proof}

\section{Proof of Theorem~\ref{thm:variance} (Exact Variance)}
\label{app:thm2}

We use the notation of the main text and of
Theorem~\ref{thm:unbiased}.

\begin{proof}
By the expansion of Theorem~\ref{thm:unbiased}, the diagonal terms
contribute the deterministic constant $d\gamma$, and all randomness comes
from the cross terms:
\[
  T-d\gamma=\sum_{x\ne y}s_xs_yW_{xy}=2\sum_{x<y}s_xs_yW_{xy},
  \qquad
  W_{xy}=H_{xy}-\frac{k^2}{M}.
\]
Since $\E[W_{xy}]=0$ (the cross terms of
Theorem~\ref{thm:unbiased}), $T-d\gamma$ has zero mean and
$\Var(T)=\E[(T-d\gamma)^2]$.  The signs $s_x$ are fixed constants, so
squaring and taking expectations term by term by linearity gives a double
sum over unordered pairs:
\[
  \Var(T)
  =4\sum_{x<y}\sum_{u<v}s_xs_ys_us_v\,\E[W_{xy}W_{uv}].
\]

\emph{Step 1 (single-pair variance).}
$H_{xy}\sim\operatorname{Hypergeometric}(N{=}M,\,K{=}k,\,n{=}k)$, and the
hypergeometric variance formula
$\Var(H_{xy})=n\frac{K}{N}(1-\frac{K}{N})\frac{N-n}{N-1}$ gives
\[
  \E[W_{xy}^2]=\Var(H_{xy})
  =k\cdot\frac{k}{M}\cdot\frac{M-k}{M}\cdot\frac{M-k}{M-1}
  =\frac{k^2(M-k)^2}{M^2(M-1)}
  =\frac{\gamma^2}{M-1},
\]
where the last equality uses $\gamma=k(M-k)/M$.

\emph{Step 2 (distinct pairs are pairwise uncorrelated).}  We split the
pairs by the number of elements shared between $\{x,y\}$ and $\{u,v\}$.

\begin{itemize}[itemsep=0pt,topsep=2pt]
  \item \emph{No shared element:} $S_x,S_y,S_u,S_v$ are mutually
  independent, so $\E[W_{xy}W_{uv}]=\E[W_{xy}]\,\E[W_{uv}]=0$.
  \item \emph{One shared element} (say $y=v$, $x\ne u$): condition on
  $S_y$.  By the hypergeometric characterization in
  Theorem~\ref{thm:unbiased}, $\E[H_{xy}\mid S_y]=k^2/M$ for every value of
  $S_y$, so $\E[W_{xy}\mid S_y]=0$; and $S_x,S_u$ are conditionally
  independent given $S_y$, so
  \[
    \E[W_{xy}W_{uy}]
    =\E\bigl[\E[W_{xy}\mid S_y]\cdot\E[W_{uy}\mid S_y]\bigr]=0.
  \]
  \item \emph{Two shared elements:} the pairs coincide and fall under Step 1.
\end{itemize}

\emph{Step 3 (combining).}  Only the terms with $\{x,y\}=\{u,v\}$ survive
in the double sum, and for those $(s_xs_y)^2\equiv1$, so
\[
  \Var(T)
  =4\binom d2\Var(W_{xy})
  =\frac{2d(d-1)k^2(M-k)^2}{M^2(M-1)}.
\]

\emph{Step 4 (normalization).}  With $\dhat=T/\gamma$ and
$\gamma^2=k^2(M-k)^2/M^2$, division cancels $k$ and $M-k$ completely:
\[
  \Var(\dhat)=\frac{\Var(T)}{\gamma^2}=\frac{2d(d-1)}{M-1}.
\]
The cancellation of $k$ is not a coincidence: the unnormalized energy
satisfies
\[
  \Var(T)=\frac{2d(d-1)\gamma^2}{M-1},
\]
so the signal scale of $T$ and its noise standard deviation both scale
with $\gamma=k(1-k/M)$; dividing by $\gamma$ removes this common scale
exactly.  The sign pattern likewise enters the surviving variance terms only
as $(s_xs_y)^2=1$, so the normalized variance is independent of sign
balance.
\end{proof}

\begin{remark}[Derivation of Corollary~\ref{cor:rsd}]
From Theorem~\ref{thm:unbiased}'s $\E[\dhat]=d$ and
Theorem~\ref{thm:variance}'s $\Var(\dhat)=2d(d-1)/(M-1)$, for $d\ge1$,
\[
  \operatorname{RSD}(\dhat)
  =\frac{\operatorname{Std}(\dhat)}{\E[\dhat]}
  =\sqrt{\frac{2(d-1)}{d(M-1)}}
  \approx\sqrt{\frac{2}{M-1}};
\]
and since $\E[\dhat/d-1]=0$,
$\operatorname{Std}(\dhat/d-1)=O\!\left((M-1)^{-1/2}\right)$.
\end{remark}

\section{Proof of Theorem~\ref{thm:chi2} (Chi-Square Limit)}
\label{app:thm3}

\begin{proof}
The proof has four steps.

\emph{Step 1 (mean and covariance of a single element's vector).}
$\E[a_x]=(k/M)\one_M$ (each cell is hit with probability $k/M$), so
$\E[Qa_x]=Q\,\E[a_x]=0$.  The covariance is computed entrywise: on the
diagonal
\[
  \Var(a_{x,i})=\frac{k}{M}\left(1-\frac{k}{M}\right);
\]
off the diagonal ($i\ne j$), with
$\Pr[i,j\in S_x]=\frac{k(k-1)}{M(M-1)}$ (two specified positions drawn
without replacement),
\[
  \Cov(a_{x,i},a_{x,j})
  =\frac{k(k-1)}{M(M-1)}-\frac{k^2}{M^2}
  =\frac{k\bigl[(k-1)M-k(M-1)\bigr]}{M^2(M-1)}
  =-\frac{k(M-k)}{M^2(M-1)}.
\]
The two combine exactly into
\[
  \Cov(a_x)=\sigma^2Q,
  \qquad
  \sigma^2=\frac{k(M-k)}{M(M-1)}
\]
(check: $\sigma^2Q$ has diagonal
$\sigma^2(1-\tfrac1M)=\frac{k}{M}(1-\frac{k}{M})$ and off-diagonal
$-\sigma^2/M=-\frac{k(M-k)}{M^2(M-1)}$, as required).  Using $Q^2=Q$,
\[
  \Cov(Qa_x)=Q\,\Cov(a_x)\,Q=\sigma^2Q,
\]
so the covariance is isotropic on the orthogonal complement of $\one_M$ (the
rank-$(M-1)$ centered subspace).

\emph{Step 2 (central limit theorem).}  The vectors
$\{s_xQa_x\}_{x=1}^d$ are mutually independent (i.i.d.\ mappings), have zero
mean, share the covariance $\sigma^2Q$ (fixed signs enter only as
$s_x^2=1$), and each coordinate is bounded ($\lvert(Qa_x)_i\rvert\le1$), so
the Lindeberg condition holds automatically.  Under fixed signs the terms are not identically distributed (the $s_x=-1$
terms are mirror images of the $s_x=+1$ ones); we therefore use the
Lindeberg--Feller triangular-array CLT rather than the i.i.d.\ version.
The CLT gives
\[
  \frac{QC}{\sqrt d}=\frac1{\sqrt d}\sum_x s_xQa_x
  \Longrightarrow
  \mathcal N(0,\sigma^2Q),
\]
an isotropic Gaussian supported on the centered subspace.

\emph{Step 3 (continuous mapping).}  $Q$ is an orthogonal projection
($Q^\top=Q$, $Q^2=Q$), so
$T=C^\top QC=C^\top Q^\top QC=\lVert QC\rVert^2$.  The squared norm is
continuous, and the continuous mapping theorem gives
\[
  \frac{T}{d}=\left\lVert\frac{QC}{\sqrt d}\right\rVert^2
  \Longrightarrow
  \sigma^2\chi^2_{M-1}.
\]
For the degrees of freedom, take an orthonormal basis of the centered
subspace; the limiting Gaussian's coordinates are $M-1$ i.i.d.\
$\mathcal N(0,\sigma^2)$, whose squared norm is $\sigma^2\chi^2_{M-1}$.

\emph{Step 4 (normalization).}
\[
  \frac{\dhat}{d}=\frac{T}{d\gamma}
  \Longrightarrow
  \frac{\sigma^2}{\gamma}\,\chi^2_{M-1}
  =\frac{\chi^2_{M-1}}{M-1},
  \qquad
  \frac{\sigma^2}{\gamma}
  =\frac{k(M-k)/[M(M-1)]}{k(M-k)/M}
  =\frac{1}{M-1}.
\]
Consistency check: the limit $\chi^2_{M-1}/(M-1)$ has variance $2/(M-1)$,
matching the $d\to\infty$ limit of the relative variance in
Theorem~\ref{thm:variance}.
\end{proof}

\begin{remark}[Relation to the Pearson $\chi^2$ statistic]
Take $k=1$ and all difference elements with $s_x=+1$; the all-$-1$ case
flips the count vector and yields the same statistic.  Then
\[
  (C_1,\ldots,C_M)\sim
  \operatorname{Multinomial}\!\left(d;\frac1M,\ldots,\frac1M\right),
\]
and the Pearson goodness-of-fit statistic for the uniform fit is
\[
  X^2
  =\sum_{i=1}^M\frac{(C_i-d/M)^2}{d/M}
  =\frac{M}{d}\,T.
\]
Since in this special case $\gamma=1-1/M=(M-1)/M$, we immediately get
\[
  \frac{\dhat}{d}
  =\frac{T}{d\gamma}
  =\frac{X^2}{M-1}.
\]
Thus the classical Pearson limit $X^2\Rightarrow\chi^2_{M-1}$ under the
equiprobable multinomial is a special case of
Theorem~\ref{thm:chi2}, which extends it to IBLT's signed,
without-replacement $k$-subset mapping.
\end{remark}

\section{Proofs for Failure-Conditioned Validity}
\label{app:thm4}

\subsection*{Exact sign projection and conditional-mean envelope}

\begin{proof}[Proof of Theorem~\ref{thm:cond-mean}]
Under the idealized decoder, peeling clears the IBLT sketch exactly when
the mapping hypergraph's 2-core is empty.  Hence $F$ is a function of the
mappings $\{S_x\}$ and is independent of the signs.  The joint mapping law
is invariant under every permutation of the element labels, and the event
$F$ has the same invariance because relabeling hyperedges does not change
the 2-core.  It follows that, for any two distinct unordered pairs,
\[
  \E[W_{xy}\given F]=\E[W_{uv}\given F].
\]
Thus
\[
  c_F=\gamma^{-1}\E[W_{xy}\given F]
\]
is well defined for any $x\ne y$ and depends on the mapping model and the
failure event, but not on the fixed sign pattern.

The quadratic-form expansion from Appendix~\ref{app:thm1}, together with
the deterministic self-energy $a_x^\top Qa_x=\gamma$, gives
\[
  T=d\gamma+2\sum_{x<y}s_xs_yW_{xy}.
\]
For fixed signs, conditional expectation therefore yields
\begin{align*}
  \frac{\E[\dhat\given F]}{d}-1
  &=\frac{2c_F}{d}\sum_{x<y}s_xs_y \\
  &=\frac{c_F}{d}\left\{\left(\sum_xs_x\right)^2-d\right\}.
\end{align*}
Since
\[
  \sum_xs_x=d_+-d_-=d\theta_\Delta,
\]
this is exactly Eq.~\eqref{eq:sign-projection}.

Now let the signs be i.i.d.\ uniform on $\{\pm1\}$ and independent of the
mappings.  Then
\[
  \E_s\!\left[\left(\sum_xs_x\right)^2\right]
  =\sum_x\E_s[s_x^2]
   +2\sum_{x<y}\E_s[s_xs_y]
  =d.
\]
Averaging the fixed-sign identity over the signs consequently gives
$\E[\dhat\given F]=d$ exactly.

It remains to bound the fixed-sign projection.  Consider the hypothetical
one-sided pattern $s_x\equiv+1$ on the same mapping distribution.  Because
$F$ depends only on the mappings, it has the same $p_F$ and the same $c_F$.
The projection identity becomes
\[
  \left|\frac{\E[\dhat\given F]}{d}-1\right|
  =|c_F|(d-1).
\]
On the other hand, Theorem~\ref{thm:unbiased} and the definition of
conditional expectation give
\[
  \E[\dhat\given F]-d
  =\frac{\Cov(\dhat,\mathbf1_F)}{p_F}.
\]
Cauchy--Schwarz, $\Var(\mathbf1_F)=p_F(1-p_F)$, and
Theorem~\ref{thm:variance} imply
\begin{align*}
  |c_F|(d-1)
  &\le \frac{\operatorname{Std}(\dhat)}{d}
       \sqrt{\frac{1-p_F}{p_F}} \\
  &=\sqrt{\frac{2(d-1)}{d(M-1)}}
       \sqrt{\frac{1-p_F}{p_F}}.
\end{align*}
Multiplying this bound by
$|d\theta_\Delta^2-1|/(d-1)$ and using the exact projection proves
Eq.~\eqref{eq:composition-aware-bound}.  No sign assertion about $c_F$ is
required.
\end{proof}

\begin{proof}[Proof of Corollary~\ref{cor:cf-regimes}]
Apply Theorem~\ref{thm:cond-mean} to the hypothetical fully one-sided sign
pattern.  Its relative conditional bias is $c_F(d-1)$, so
Eq.~\eqref{eq:composition-aware-bound} gives
\[
  |c_F|
  \le \frac{1}{d-1}
       \sqrt{\frac{2(d-1)}{d(M-1)}}
       \sqrt{\frac{1-p_F}{p_F}},
\]
which is Eq.~\eqref{eq:cf-envelope}.

For part~(i), take $d=3$ and condition on $S_1$ and $S_2$.  If $H_{12}=k$,
the first two elements already form a 2-core and every choice of $S_3$
fails.  If $H_{12}<k$, failure occurs when $S_3$ equals $S_1$ or $S_2$, or
when the three distinct sets jointly give every covered cell degree at least
two.  In the latter case $S_3$ must contain $S_1\mathbin{\triangle}S_2$ and
take its remaining $2H_{12}-k$ cells from $S_1\cap S_2$.  Among the
$\binom Mk$ choices of $S_3$, the number that fail is therefore
\[
  \begin{cases}
    \binom Mk, & H_{12}=k,\\[2pt]
    2+\mathbf1_{\{2H_{12}\ge k\}}\binom{H_{12}}{2H_{12}-k}, & H_{12}<k,
  \end{cases}
\]
and $H_{12}$ itself follows the hypergeometric law of the overlap of two
independent uniform $k$-subsets of $[M]$.
For $k=M-2$ only $H_{12}\in\{M-4,M-3,M-2\}$ has positive probability.
Evaluating $\gamma^{-1}\E[W_{12}\given F]$ over these three values gives
Eq.~\eqref{eq:cf-dense-counterexample}.  The denominator is positive for
every $M\ge3$, since after clearing positive factors it is proportional to
$\Pr[F]$.  The numerator equals $-80$ at $M=10$ and is decreasing for
$M\ge10$, so it is negative for every $M\ge10$.

For part~(ii), suppose instead that all $d$ edges peel successfully and order
them by their removal time.  Each of the first $d-1$ removed edges has a
degree-1 witness cell at that time.  Such a witness belongs to no later edge,
and witnesses from different removal steps are therefore distinct.  The last
edge must avoid all $d-1$ witnesses, leaving at most $M-d+1$ available cells.
Since it needs $k$ distinct cells, successful peeling requires
$d\le M-k+1$.  Hence $d\ge M-k+2$ implies $F$ surely.  Conditioning on a
sure event changes no expectation, so
$c_F=\gamma^{-1}\E[W_{xy}]=0$.
\end{proof}

\subsection*{Chi-square inheritance when failure is typical}

We now prove Proposition~\ref{prop:cond-chi2}.  For each cell $i$, let $D_i$
be its degree in the initial mapping hypergraph.  The elements map
independently, and each includes cell $i$ with probability $k/M$, so
\[
  \Pr[D_i=1]
  =d\frac{k}{M}\left(1-\frac{k}{M}\right)^{d-1}.
\]
For $d\ge1$, a successful peeling process must begin with a degree-1 cell.
The union bound over all $M$ cells therefore gives
\[
  1-\Pr[F_d]
  =\Pr[F_d^c]
  \le \Pr[\exists i:D_i=1]
  \le d\,k\left(1-\frac{k}{M}\right)^{d-1}
  =\varepsilon_{M,d}.
\]

It remains to quantify what conditioning can do to the estimator's law.
Write $P_F=P(\mathord\cdot\mid F_d)$ and
$P_{F^c}=P(\mathord\cdot\mid F_d^c)$.  For
$p_F=\Pr[F_d]\in(0,1)$, the mixture identity
\[
  P=p_FP_F+(1-p_F)P_{F^c}
\]
implies
\[
  d_{\mathrm{TV}}(P_F,P)
  =(1-p_F)d_{\mathrm{TV}}(P_F,P_{F^c})
  =1-p_F,
\]
where the last equality follows because $P_F$ and $P_{F^c}$ have disjoint
supports.  The case $p_F=1$ is immediate.  Passing from the underlying
outcome to the scalar $\dhat/d$ cannot increase total variation, hence
\[
  \sup_t\left|
    \Pr[\dhat/d\le t\given F_d]-\Pr[\dhat/d\le t]
  \right|
  \le 1-p_F
  \le \min\{1,\varepsilon_{M,d}\},
\]
which proves Eq.~\eqref{eq:cond-transfer}.

Finally, for fixed $M,k$,
$\varepsilon_{M,d}=dk(1-k/M)^{d-1}\to0$.  The conditional and
unconditional distribution functions are therefore uniformly
indistinguishable, while Theorem~\ref{thm:chi2} gives
$\dhat/d\Rightarrow\chi^2_{M-1}/(M-1)$ unconditionally.  The triangle
inequality completes the proof of Eq.~\eqref{eq:cond-chi2}.

\paragraph{No random-sign assumption for distributional inheritance.}
For every deterministic sign sequence fixed independently of the mapping
randomness, Theorem~\ref{thm:chi2} supplies the unconditional chi-square
limit: multiplying a centered mapping vector by a fixed sign preserves its
zero mean, covariance, boundedness, and independence across elements.  The
total-variation argument above is likewise agnostic to that sign sequence.
Thus Proposition~\ref{prop:cond-chi2} covers arbitrary fixed sign
compositions, including the fully one-sided case.  Random signs are used only
for the exact finite-sample conditional-mean identity in
Theorem~\ref{thm:cond-mean}; no probabilistic sign model is needed for the
failure-conditioned chi-square limit.

\section{Supplementary Results for the Plain Estimator Validation
(\S\ref{sec:exp-validation})}
\label{app:appendix-d}

This appendix reports the full-grid validation results for the Plain IBLT
estimator, complementing the representative slice shown in
\S\ref{sec:exp-validation}.  It follows the same three-part structure as the
main text.  First, we check on the full sample that the estimator's mean,
variance, and interval coverage agree with
Theorems~\ref{thm:unbiased}--\ref{thm:chi2}.  Second, we restrict the
analysis to first-round failures and examine the failure-conditioned
distribution shape, quantile stability, and fixed-sign safety bounds.  Third,
we connect these statistical conclusions to second-round decoding success.
\S\ref{sec:exp-validation} gives the main conclusions; this appendix covers
the full parameter range, quantifies finite-sample error, and reports the
numerical margins of the safety bounds.

Except for the distribution-shape figures---Figures~\ref{fig:f1} and
\ref{fig:d1} below, and main-text Figure~\ref{fig:fcond-qq}---whose captions
state their own semantics, all grid numbers below come from $10^6$
independent trials per cell.

\subsection{Full-grid validation of the unconditional measurement
(\texorpdfstring{Theorems~\ref{thm:unbiased}--\ref{thm:chi2}}{Theorems 3.1--3.3})}

Table~\ref{tab:t1a} in the main text reports only a representative slice.
To rule out that the agreement is coincidental, this subsection verifies the
same
three items on all 200 configurations (5 tiers of $M$ $\times$ 2 tiers of
$k$ $\times$ 5 sign compositions $\times$ 4 load tiers): the variance closed
form (Figure~\ref{fig:f3b}), the worst-case error summarized by sketch dimension
(Table~\ref{tab:d1}), and the convergence of the chi-square approximation
with $d$ and its interval coverage (Figure~\ref{fig:d1},
Table~\ref{tab:d2}).

\begin{figure}[!hbtp]
\centering
\includegraphics[width=0.85\textwidth]{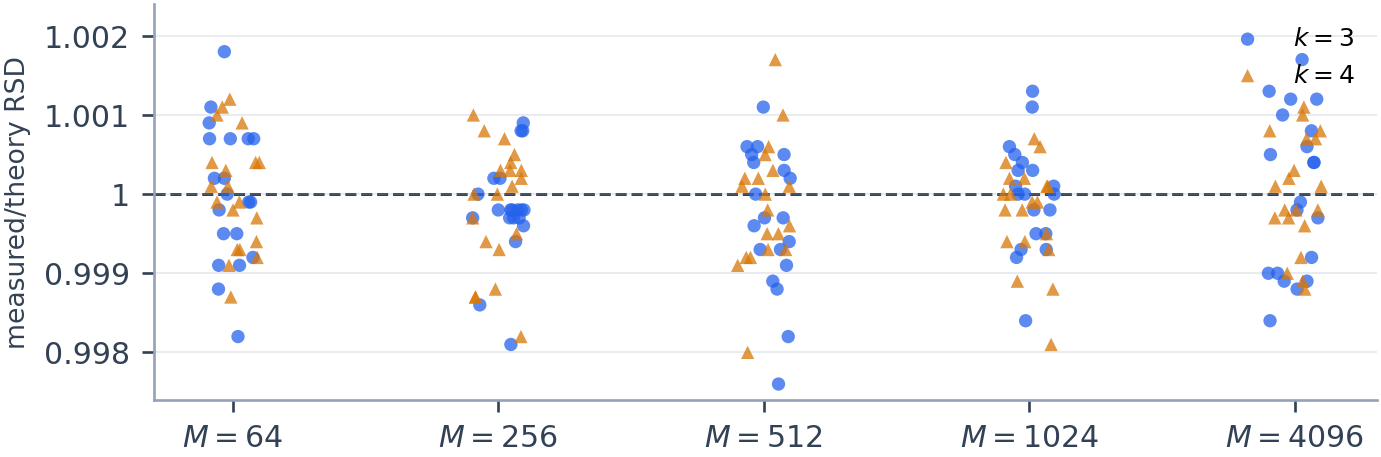}
\caption{Full-grid RSD ratio.  The y-axis is the measured RSD divided by the
Theorem~\ref{thm:variance} prediction; 1 means the variance closed form is
exact.  All 200 points lie in $[0.9976,1.0018]$, with no visible separation
between the $k=3$ and $k=4$ points.}
\label{fig:f3b}
\end{figure}

\begin{table}[!hbtp]
\centering
\caption{Unconditional error summarized by sketch dimension (40 configurations
per tier: 2 $k$ $\times$ 5 sign compositions $\times$ 4 load tiers).}
\label{tab:d1}
\small
\begin{tabular}{lrrrr}
\toprule
$M$ & configurations & $\max |\operatorname{mean}_{\rm MC}(\dhat/d)-1|$ & MC RSD / theory RSD range & \shortstack{$\max \lvert q_{0.01}^{U}(\mathrm{MC})$\\$-q_{0.01}^{U}(\chi^2)\rvert$} \\
\midrule
64   & 40 & 0.00047 & [0.9982, 1.0018] & 0.0257 \\
256  & 40 & 0.00025 & [0.9981, 1.0010] & 0.0066 \\
512  & 40 & 0.00014 & [0.9976, 1.0017] & 0.0039 \\
1024 & 40 & 0.00009 & [0.9981, 1.0013] & 0.0008 \\
4096 & 40 & 0.00005 & [0.9984, 1.0017] & 0.0004 \\
\bottomrule
\end{tabular}
\end{table}

The three quantities behave differently as $M$ increases, consistent with
their respective theoretical guarantees.  Theorems~\ref{thm:unbiased}--\ref{thm:variance}
hold exactly for finite $(M,d,k)$, so the first two items should show only
Monte Carlo fluctuation: the mean error shrinks from 0.00047 to 0.00005 as
$M$ grows, consistent with the standard-error shrinkage, and the RSD ratio
stays near 1 on all five scales.  Theorem~\ref{thm:chi2} needs $d\to\infty$
at fixed $M$, so the third item may show visible error at small $d$: the
full-grid maximum of 0.0257 occurs at $(M,d)=(64,26)$, the smallest-$d$
configuration, consistent with that theorem's asymptotic semantics.

\begin{figure}[!hbtp]
\centering
\includegraphics[width=0.9\textwidth]{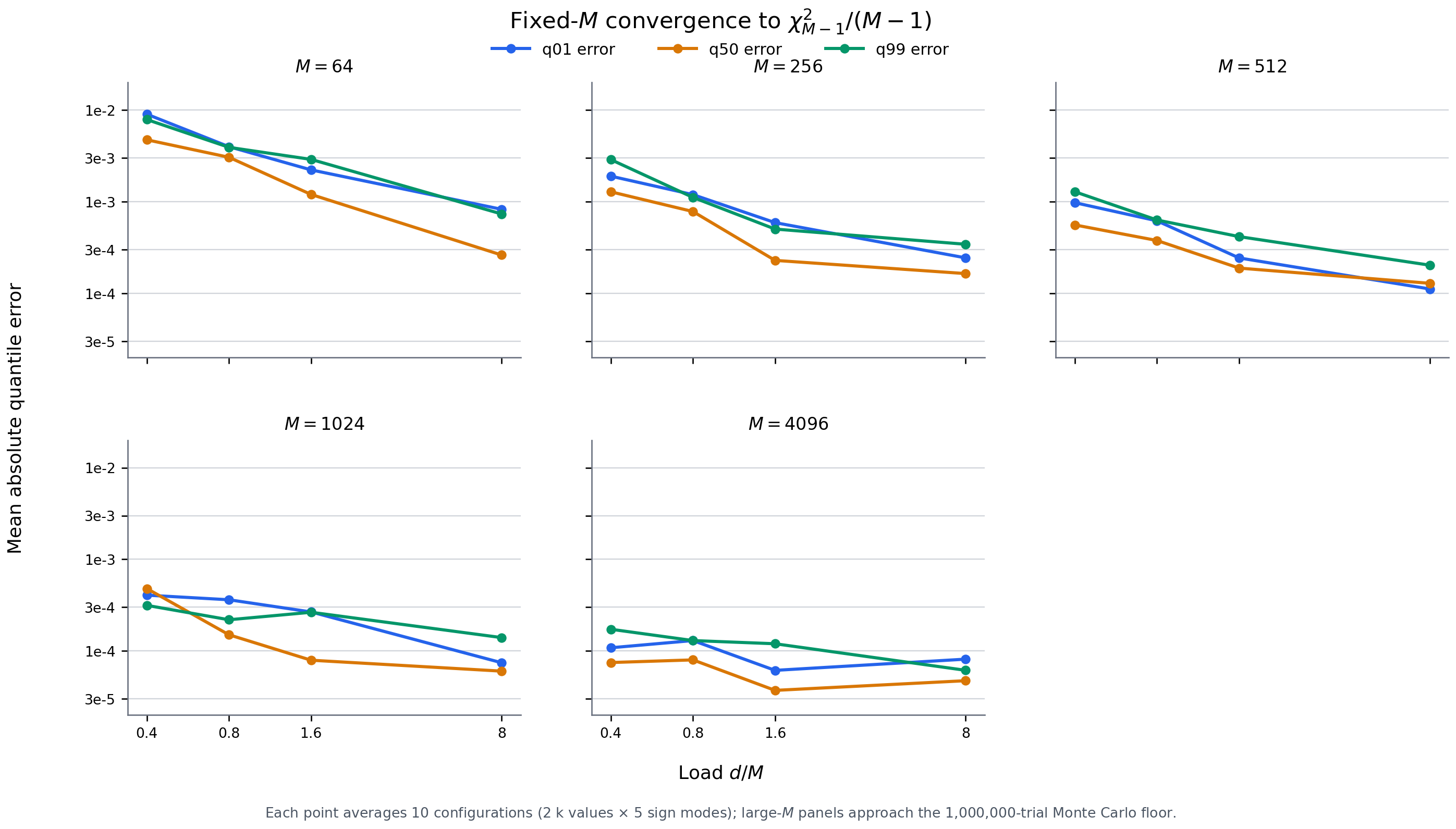}
\caption{Fixed-$M$ chi-square quantile convergence.  Each panel fixes $M$ and
increases $d/M$ from 0.4 to 8, the $d\to\infty$ direction of
Theorem~\ref{thm:chi2}; the y-axis is the mean absolute error between the
empirical quantile and the corresponding $\chi^2_{M-1}/(M-1)$ quantile,
averaged over the two $k$ values and five sign compositions.  The $1\%$
lower tail $q_{0.01}^U$, median $q_{0.50}^U$, and symmetric $99\%$ upper tail
$q_{0.99}^U$ all shrink together from about 0.0040/0.0023/0.0039 at $M=64$
to about 0.0001 at $M=4096$.}
\label{fig:d1}
\end{figure}

\begin{table}[!hbtp]
\centering
\caption{Empirical coverage error of the Theorem~\ref{thm:chi2} asymptotic
confidence intervals (deviation of the empirical coverage from the nominal
level, in percentage points; each $(M,d/M)$ cell aggregates $10^7$ trials).}
\label{tab:d2}
\small
\begin{tabular}{lrrr}
\toprule
$M$ & 90\% interval & 95\% interval & 99\% interval \\
\midrule
64   & 0.478\,pp & 0.495\,pp & 0.121\,pp \\
256  & 0.279\,pp & 0.140\,pp & 0.017\,pp \\
512  & 0.129\,pp & 0.074\,pp & 0.012\,pp \\
1024 & 0.063\,pp & 0.078\,pp & 0.012\,pp \\
4096 & 0.022\,pp & 0.017\,pp & 0.009\,pp \\
\bottomrule
\end{tabular}
\end{table}

For a nominal coverage $1-\alpha$, the interval covers the true $d$ exactly
when
\[
  \chi^2_{M-1,\alpha/2}
  \le (M-1)\frac{\dhat}{d}
  \le \chi^2_{M-1,1-\alpha/2}.
\]
The smallest configuration $(M,d)=(64,26)$ has a coverage error of about
0.5\,pp that does not shrink with more trials, so it is dominated by the
finite-sample bias of the chi-square approximation at small $d$ rather than
Monte Carlo noise; as the sketch dimension increases, all tiers at $M\ge256$ stay below
0.3\,pp and $M=4096$ drops to 0.02\,pp.  For engineering purposes, the
chi-square interval can be used at its nominal level on the sketch dimensions
the protocol actually uses.

\subsection{Failure-conditioned distribution, quantile stability, and
safety bounds}

The protocol reads the estimate only from IBLT instances that fail to decode
in the first round, so the unconditional conclusions of the previous
subsection must be re-checked on the failed subsample: the distribution shape
(main-text Figure~\ref{fig:fcond-qq}), the sampling
stability of the $1\%$ quantile the protocol uses directly
(Table~\ref{tab:d3}), and the coverage of the two safety bounds under fixed
signs (Tables~\ref{tab:d4}--\ref{tab:d5}).

\begin{table}[!hbtp]
\centering
\caption{Sampling stability of the failed-only $1\%$ lower tail (random
signs, $d/M=0.8$).  The empirical $q_{0.01}^F$ differs from its $95\%$
one-sided order-statistic lower confidence bound by at most 0.0017.}
\label{tab:d3}
\small
\begin{tabular}{rrrrrrr}
\toprule
$M$ & $d$ & $\Pr[F]$ & $n_F$ & $q_{0.01}^{F}$ (MC) & 95\% order-statistic lower bound & gap \\
\midrule
64   & 51   & 0.683 & 683{,}202 & 0.6308 & 0.6291 & 0.0017 \\
256  & 205  & 0.543 & 543{,}451 & 0.8050 & 0.8046 & 0.0004 \\
512  & 410  & 0.422 & 421{,}758 & 0.8598 & 0.8589 & 0.0009 \\
1024 & 819  & 0.257 & 257{,}079 & 0.8998 & 0.8992 & 0.0006 \\
4096 & 3277 & 0.028 & 28{,}336  & 0.9490 & 0.9482 & 0.0008 \\
4096$^\ast$ & 3277 & 0.028 & 280{,}482 & 0.9492 & 0.9490 & 0.0002 \\
\bottomrule
\end{tabular}
\\
\footnotesize
$^\ast$ An additional $10^7$ trials on the sparsest tier (balanced signs
$\theta_-=0.5$) increase the number of failed samples to approximately
280{,}000; the result agrees with the random-sign $M=4096$ row's 0.9490.
\end{table}

The main text therefore uses the empirical $q_{0.01}^F$ directly as the
practical input to capacity planning.  The analytic bounds of
Proposition~\ref{prop:lower-tail} and Cantelli do not depend on the
empirical tail shape but become very conservative at lower-tail levels such
as $\delta=0.01$, often yielding a nonpositive lower bound, so they are kept
only as
distribution-free safety floors (see Table~\ref{tab:d5}), not for practical
configuration.

\paragraph{Direct audit of the exact sign projection.}
The dedicated H100 sweep in Figure~\ref{fig:sign-projection} separates the
sign-composition axis more finely than the 1M grid summarized below.  At
$k=3$, it uses 11 fixed values of $\theta_-$ for each of 15 $(M,d)$ curves, with one million
trials per configuration at $M\in\{256,1024\}$ and 300{,}000 at $M=4096$.
This gives 165 configurations and 126.5 million trials.  For each curve, the
weighted one-parameter model
\[
  \E_{\rm MC}[\dhat/d\given F]-1
  =c_F(d\theta_\Delta^2-1)
\]
is constrained to the exact intercept--slope relation of
Theorem~\ref{thm:cond-mean}.  Its minimum and mean $R^2$ are $0.9948$ and
$0.9990$.  Refitting after withholding each point gives 165 out-of-sample
predictions with RMSE $1.08\times10^{-4}$ and maximum absolute residual
$3.36\times10^{-4}$.  This audit targets the conditional-mean identity; the
failure-conditioned quantiles and distributional shape are assessed
separately.

\begin{table}[!hbtp]
\centering
\caption{Sign-agnostic Cauchy--Schwarz envelope implied by
Theorem~\ref{thm:cond-mean}.
The bound utilization $b_{\rm emp}/b_{\rm th}$ is the empirical conditional
deviation divided by the theoretical bound
$\sqrt{2(d-1)/(d(M-1))}\sqrt{(1-\widehat p_F)/\widehat p_F}$.  Each row
reports the maximum deviation and the maximum bound utilization over the
configurations at that $M$, together with the number of violations.}
\label{tab:d4}
\small
\begin{tabular}{lrrrr}
\toprule
$M$ & non-degenerate configs & max $|\E_{\rm MC}[\dhat/d \given F]-1|$ & max utilization $b_{\rm emp}/b_{\rm th}$ & violations \\
\midrule
64   & 16 & 0.099532 & 0.288608 & 0 \\
256  & 16 & 0.031957 & 0.306904 & 0 \\
512  & 12 & 0.022217 & 0.303511 & 0 \\
1024 & 12 & 0.021165 & 0.281562 & 0 \\
4096 & 8  & 0.019042 & 0.146488 & 0 \\
\midrule
total & 64 & --- & 0.306904 & 0 \\
\bottomrule
\end{tabular}
\end{table}

``Non-degenerate'' means $0<\widehat p_F<1$ and $n_F\ge20$: configurations
with $\widehat p_F=0$, $\widehat p_F=1$, or fewer than 20 failed samples are
excluded because they do not support a stable empirical utilization
estimate.  The 64 configurations cover all non-degenerate combinations of
four fixed sign fractions, two $k$ values, four load tiers, and five $M$
tiers; the maximum utilization is only 0.31, i.e., even the largest empirical
deviation is only about one-third of the corresponding theoretical bound.
For the \emph{random-sign} conditional-mean statement of
\S\ref{sec:exp-validation} ($0.19\%$ maximum deviation over the full random
sign grid), the same exclusion uses the stricter threshold $n_F\ge10^3$: the
configurations with fewer than $10^3$ failed samples---sparsest among them
$(M,k,d)=(64,4,26)$ with $n_F=640$ and a conditional-mean deviation of
$1.2\%$, and $(1024,4,410)$ with a single failed sample---are too noisy to
support a stable conditional mean and are left out of that summary.  The two
thresholds target different quantities: the $n_F\ge20$ bound coverage in
Table~\ref{tab:d4} is checked per configuration against a Cauchy--Schwarz
upper bound, whereas the $n_F\ge10^3$ summary reports a maximum observed
deviation.

\begin{table}[!hbtp]
\centering
\caption{Proposition~\ref{prop:lower-tail} fixed-sign lower-tail bound at
$\delta=0.01$.  The coverage margin is the empirical lower quantile minus
the Proposition~\ref{prop:lower-tail} safety bound; for each $M$, the table
reports the minimum margin over all eligible configurations.}
\label{tab:d5}
\small
\begin{tabular}{lrrr}
\toprule
$M$ & configs with $n_F\ge20$ & min coverage margin ($q_{0.01}^{F}$ MC $-q_{0.01}^{\mathrm{Prop5}}$) & violations \\
\midrule
64   & 32 & 1.4039 & 0 \\
256  & 32 & 0.6898 & 0 \\
512  & 28 & 0.4853 & 0 \\
1024 & 28 & 0.3419 & 0 \\
4096 & 28 & 0.1701 & 0 \\
\midrule
total & 148 & 0.1701 & 0 \\
\bottomrule
\end{tabular}
\end{table}

No empirical violation occurs in the 148 configurations.  The margin
narrows as $M$ grows, with the minimum 0.17 at $M=4096$ still positive---this
is consistent with Proposition~\ref{prop:lower-tail}'s conservatism: the
larger the sketch, the closer the empirical tail is to 1 and the smaller the
gap between the distribution-free bound and the actual tail.  This check
certifies coverage per configuration rather than jointly across all
configurations.  The
analytic characterization of the failure-conditioned shape (the total-variation
transfer and the conditional Kolmogorov bound) is in the supplementary
material.

\paragraph{Realization in protocol round two.}
The previous two subsections validate the estimator itself; this paragraph
connects the estimator validation to the protocol layer: the empirical $q_\delta$ of
\S\ref{sec:quantile-interface} is converted through
$\alpha\ge\beta/q_\delta$ into the second-round capacity multiplier, and the
second-round IBLT sketch is sent with capacity $M_2=\lceil\alpha\dhat\rceil$.  The
tier-by-tier derivation of the recommended multipliers and the full
$M_1\times\alpha$ grid of end-to-end success rates are in
Appendix~\ref{app:protocol-details}, Tables~\ref{tab:c1}--\ref{tab:c2}.

\section{Proof of Theorem~\ref{thm:master} (Mapping-Aware Construction)}
\label{app:thm6}

Expand
\[
  \E[C^\top QC]=\sum_{x,y}s_xs_y\,\E[a_x^\top Qa_y].
\]

\emph{Cross terms ($x\ne y$, of types $j$ and $j'$).}  By element
independence and bilinearity of expectation,
\[
  \E[a_x^\top Qa_y]
  =\E[a_x]^\top Q\,\E[a_y]
  =\mu_j^\top Q\mu_{j'}=0,
\]
where the last step uses the condition $Q\mu_{j'}=0$.  For a general
symmetric $Q$, the conditions $Q\mu_j=0$ for all types $j$ provide a concise
sufficient condition for eliminating all cross-mean terms, and every
centering projection used by the adapters in this paper satisfies it.
When $Q\succeq0$, it is also necessary, in the following sense.  Write the full expansion
\[
  \E[C^\top QC]
  =\sum_j d_j\bigl(\tr(Q\Sigma_j)+\mu_j^\top Q\mu_j\bigr)
  +\sum_{x\ne y}s_xs_y\,\mu_{j(x)}^\top Q\mu_{j(y)} .
\]
Suppose $C^\top QC/\gamma$ is unbiased for every fixed sign pattern and every
type mix in which each type contributes at least two elements.  Flipping a
single sign changes only the second sum, so every coefficient
$\mu_j^\top Q\mu_{j'}$ in it must vanish---including the same-type entries
$\mu_j^\top Q\mu_j$.  Those are the diagonal entries of the Gram matrix of
the vectors $Q^{1/2}\mu_j$, so $Q^{1/2}\mu_j=0$ and hence $Q\mu_j=0$ for
every type.  Within the class of positive-semidefinite quadratic forms,
$Q\mu_j=0$ is therefore necessary and sufficient for eliminating all
sign-composition-dependent cross terms.  It is necessary, but not sufficient,
for unbiasedness under arbitrary fixed sign compositions: unbiasedness
additionally requires the diagonal normalization
$\tr(Q\Sigma_j)+\mu_j^\top Q\mu_j=\gamma$ for every type $j$, which reduces to
$\tr(Q\Sigma_j)=\gamma$ once $Q\mu_j=0$.

Requiring unbiasedness only under i.i.d.\ random signs is a weaker demand:
the cross terms then vanish in expectation on their own, and all that remains
is a single normalization condition: $\tr(Q\Sigma_j)+\mu_j^\top Q\mu_j$ must
take the same value $\gamma$ for every type.  The adapters of
\S\ref{sec:mapping-aware} target the stronger guarantee of unbiasedness
under arbitrary fixed sign compositions, which the fixed one-sided
experiments evaluate.

\emph{Diagonal terms ($x=y$, type $j$).}  Using the identity
$a^\top Qa=\tr(Qaa^\top)$ and the cyclicity of the trace,
\[
  \E[a_x^\top Qa_x]
  =\tr\bigl(Q\,\E[a_xa_x^\top]\bigr)
  =\tr\bigl(Q(\Sigma_j+\mu_j\mu_j^\top)\bigr)
  =\tr(Q\Sigma_j)+\mu_j^\top Q\mu_j
  =\tr(Q\Sigma_j).
\]

\emph{Combining.}  Since $s_x^2\equiv1$ and type $j$ has $d_j$ elements,
summing gives $\E[C^\top QC]=\sum_j d_j\,\tr(Q\Sigma_j)$; with a single
type, dividing by $\gamma=\tr(Q\Sigma)$ gives the unbiased estimator.

\subsection*{Consistency check: reduction to the Plain IBLT}

For the Plain IBLT, $\mu=(k/M)\one_M$ and $Q=I_M-\one_M\one_M^\top/M$ satisfies
$Q\mu=0$; by Step 1 of Theorem~\ref{thm:chi2}'s proof,
$\Sigma=\sigma^2Q$ with $\sigma^2=\frac{k(M-k)}{M(M-1)}$, so
\[
  \tr(Q\Sigma)=\sigma^2\tr(Q)=\sigma^2(M-1)=k\left(1-\frac{k}{M}\right)=\gamma,
\]
recovering the normalization constant used in \S\ref{sec:count-measurement}.

\begin{remark}[The two calibration points]
The standard identity
\[
  \E[C^\top QC]
  =\tr\bigl(Q\,\Cov(C)\bigr)
  +\E[C]^\top Q\,\E[C]
\]
shows that an adapter must account for two components simultaneously: $Q$
must remove the true mapping mean, or the remaining baseline will contribute
to the quadratic energy; the resulting energy must then be normalized by
$\tr(Q\Sigma)$ under the true mapping covariance, or it will no longer yield
an unbiased element count.  This is also why the Plain adapter cannot be
applied directly to variants whose mapping mean or covariance has changed.
\end{remark}

\section{Irregular Adapter Proofs and the Degenerate Boundary}
\label{app:irregular}

\subsection*{Proposition~\ref{prop:irregular} (degree second-moment adapter)}

The unconditional mapping mean is $\mu=(\E[D]/M)\one_M$, so $Q\mu=0$ and the
cross-term condition of Theorem~\ref{thm:master} holds.  The diagonal term
needs no distributional assumption: $a_x$ is a $0/1$ vector with $D_x$
ones, and, for every realization,
\[
  a_x^\top Qa_x
  =a_x^\top a_x-\frac{(\one^\top a_x)^2}{M}
  =D_x-\frac{D_x^2}{M},
\]
so taking expectations over $D_x\sim\Lambda$ gives
$\E[a_x^\top Qa_x]=\gamma_{\mathrm{irr}}$.  Substituting into
Theorem~\ref{thm:master} yields $\E[C^\top QC]=d\gamma_{\mathrm{irr}}$.

\subsection*{Corollary~\ref{cor:irregular-var} (exact variance)}

Write
\[
  e_x=a_x^\top Qa_x=D_x-\frac{D_x^2}{M},
  \qquad
  W_{xy}=a_x^\top Qa_y.
\]
Then
\[
  C^\top QC
  =\sum_x e_x+2\sum_{x<y}s_xs_yW_{xy}.
\]
The first term contributes variance $d\sigma_E^2$ from the independent
random degrees.  We now compute the cross terms.  Given $(D_x,D_y)=(r,t)$, let
$H_{xy}=a_x^\top a_y$; then
\[
  H_{xy}\mid(r,t)\sim\operatorname{Hypergeometric}(M,r,t),
  \qquad
  W_{xy}=H_{xy}-\frac{rt}{M}.
\]
Therefore
\[
  \E[W_{xy}\mid r,t]=0
\]
and
\[
  \Var(W_{xy}\mid r,t)
  =\frac{rt(M-r)(M-t)}{M^2(M-1)}.
\]
Since $D_x,D_y$ are i.i.d.\ and
\[
  \E[D(M-D)]=M\gamma_{\mathrm{irr}},
\]
taking expectations yields
\[
  \Var(W_{xy})
  =\frac{\E[D(M-D)]^2}{M^2(M-1)}
  =\frac{\gamma_{\mathrm{irr}}^2}{M-1}.
\]

It remains to account for the covariance terms in the expansion.  If two $W$'s correspond to
disjoint element pairs, they are independent; if they share exactly one
element, say $W_{xy}$ and $W_{xz}$, then conditioned on $a_x$ they are
independent and
\[
  \E[W_{xy}\mid a_x]
  =a_x^\top Q\,\E[a_y]
  =0,
\]
because $\E[a_y]=(\E[D]/M)\one_M$ and $Q\one_M=0$.  The same conditional mean
shows $e_x$ is uncorrelated with any $W_{xy}$ containing $x$; for pairs
sharing no element, independence gives uncorrelatedness.  Hence the
self-energies and cross-energies are uncorrelated, and cross-energies of
distinct unordered pairs are pairwise uncorrelated.

For any fixed sign pattern, $(s_xs_y)^2=1$, so
\[
  \begin{aligned}
  \Var(C^\top QC)
  &=d\sigma_E^2
  +4\binom d2\frac{\gamma_{\mathrm{irr}}^2}{M-1}\\
  &=d\sigma_E^2
  +\frac{2d(d-1)\gamma_{\mathrm{irr}}^2}{M-1}.
  \end{aligned}
\]
Dividing by $\gamma_{\mathrm{irr}}^2$ gives
\[
  \Var(\dhat_{\mathrm{irr}})
  =\frac{2d(d-1)}{M-1}
  +\frac{d\sigma_E^2}{\gamma_{\mathrm{irr}}^2},
\]
and dividing by $d^2$ yields the relative variance in the main text
\[
  \Var\!\left(\frac{\dhat_{\mathrm{irr}}}{d}\right)
  =\frac{2(d-1)}{d(M-1)}
  +\frac{\sigma_E^2}{\gamma_{\mathrm{irr}}^2d}.
\]
Expanding $e_x=D_x-D_x^2/M$ directly, $\sigma_E^2$ depends only on the
first four moments of the degree distribution:
\[
  \sigma_E^2
  =\E[D^2]-\frac{2\E[D^3]}{M}
  +\frac{\E[D^4]}{M^2}
  -\gamma_{\mathrm{irr}}^2.
\]
When $D$ is constant, $\sigma_E^2=0$, and the result reduces to
Theorem~\ref{thm:variance} for the Plain IBLT.

\subsection*{Worked example at the $D=M=18$ boundary}

The $M=18$ row of the validation table in \S\ref{sec:irregular} is a
deliberately extreme configuration that maximizes the difference between
second-moment normalization and mean-degree plug-in.  At $M=18$, $12.5\%$
of elements draw exactly $D=18$, i.e., hit every cell.  Such an
element's mapping vector is the all-ones vector, containing only the
constant direction; the global centering matrix $Q$ exists precisely to
remove the constant direction, so these elements leave no energy after
centering---substituting $D=M=18$ into the self-energy
$D-D^2/M$ also gives zero directly.  In other words, these $12.5\%$ of
elements hit the most cells yet contribute nothing to the centered
energy.

The correct normalization constant
$\gamma_{\mathrm{irr}}=\E[D]-\E[D^2]/M$ automatically accounts for this
zero-energy case through the second moment: the substantial contribution of
degree-18 elements to $\E[D^2]$ reduces $\gamma_{\mathrm{irr}}$ to $2.0792$,
exactly matching the fact that they produce no energy.  The plug-in
with the mean degree, $\gamma_{\mathrm{mean}}=\E[D]-\E[D]^2/M=3.4847$,
overestimates each element's average energy production and introduces a
systematic downward bias by the constant factor
$\gamma_{\mathrm{irr}}/\gamma_{\mathrm{mean}}=0.5967$, i.e., about $40.3\%$
underestimation.  That the $M=18$ row's mean stays near 1 is consistent
with this corrected normalization.

The same row's larger relative RMSE ($\approx0.34$) has a different source.
It corresponds to the first term $2(d-1)/[d(M-1)]$ of
Corollary~\ref{cor:irregular-var}: smaller sketch dimensions produce greater
random overlap between the cell sets of distinct difference elements, so a
single estimate fluctuates more.  In that term $\gamma_{\mathrm{irr}}$ has
been cancelled exactly, and the term's size depends only on $d$ and $M$; all
influence of the
degree distribution is concentrated in the second term
$\sigma_E^2/(\gamma_{\mathrm{irr}}^2d)$.  The two effects are therefore
distinct: the normalization constant determines the center of the estimate,
while its dispersion depends on both the sketch dimension $M$ and the
degree-dependent self-energy term.

\section{Rateless Prefix Adapter Proofs}
\label{app:rateless}

\subsection*{Lemma~\ref{lem:rateless-hit} (closed-form hit probability)}

\emph{Preliminary: the transition kernel.}  For integer $\ell\ge1$, taking
the ceiling does not change the event ($D_j\le\ell$ iff the real quantity in
parentheses is $\le\ell$), so
\[
  \Pr[D_j\le\ell]
  =\Pr\!\left[U^{-1/2}\le1+\frac{\ell}{j+\frac32}\right]
  =\Pr\!\left[U\ge\left(\frac{j+\frac32}{j+\ell+\frac32}\right)^2\right]
  =1-\left(\frac{j+\frac32}{j+\ell+\frac32}\right)^2.
\]
Differencing this cumulative distribution with respect to the integer $\ell$
gives the transition kernel in the main text
\[
  K_{jv}=\Pr[j+D_j=v]
  =\left(j+\frac32\right)^2\left[\frac{1}{(v+\frac12)^2}-\frac{1}{(v+\frac32)^2}\right],
  \qquad v>j.
\]
This subsection and Lemma~\ref{lem:rateless-pairwise} both start from it.

A path hits $v$ iff it hits some $j<v$ and then jumps exactly to $v$;
decomposing by the last hit position before $v$ (mutually exclusive events)
gives the renewal recurrence
\[
  q_v=\sum_{j=0}^{v-1}q_jK_{jv}.
\]
To verify the candidate closed form, observe that $q_j(j+\frac32)^2=2(j+1)$
for $j\ge1$,
and $q_0(\frac32)^2=\frac94$, so
\[
  \sum_{j=0}^{v-1}q_j\left(j+\frac32\right)^2
  =\frac94+\sum_{j=1}^{v-1}2(j+1)
  =\frac94+(v^2+v-2)
  =\left(v+\frac12\right)^2.
\]
Substituting back into the recurrence (the bracket factor of $K_{jv}$ does
not depend on $j$ and can be pulled out):
\[
  \sum_{j=0}^{v-1}q_jK_{jv}
  =\left(v+\frac12\right)^2\left[\frac{1}{(v+\frac12)^2}-\frac{1}{(v+\frac32)^2}\right]
  =\frac{(v+\frac32)^2-(v+\frac12)^2}{(v+\frac32)^2}
  =\frac{2v+2}{(v+\frac32)^2}
  =\frac{8(v+1)}{(2v+3)^2},
\]
which is exactly the candidate closed form at $v$; induction on $v$ completes
the verification.  The base case is
$q_1=K_{01}=1-\bigl(\frac{3/2}{5/2}\bigr)^2=\frac{16}{25}=0.64$.

\subsection*{Lemma~\ref{lem:rateless-pairwise} (pairwise independence)}

Let $h_u(v)$ be the probability of eventually hitting $v$ given that the
path is currently at cell $u$.  Run backward induction on
$u=v-1,v-2,\ldots,0$ with the induction hypothesis $h_w(v)=q_v$ for all
$u<w<v$.  Decompose by the first jump:
\[
  h_u(v)=K_{uv}+\sum_{w=u+1}^{v-1}K_{uw}h_w(v)
  =K_{uv}+q_v\sum_{w=u+1}^{v-1}K_{uw}
  =K_{uv}+q_v\bigl(1-\Pr[D_u\ge v-u]\bigr),
\]
where by the tail-probability formula
$\Pr[D_u\ge v-u]=\bigl(\frac{u+3/2}{v+1/2}\bigr)^2$.  Hence $h_u(v)=q_v$
is equivalent to the identity
\[
  K_{uv}=q_v\left(\frac{u+\frac32}{v+\frac12}\right)^2.
\]
Substituting
$q_v=1-\bigl(\frac{v+1/2}{v+3/2}\bigr)^2=\frac{(v+\frac32)^2-(v+\frac12)^2}{(v+\frac32)^2}$
and expanding directly,
\[
  q_v\left(\frac{u+\frac32}{v+\frac12}\right)^2
  =\left(u+\frac32\right)^2\left[\frac{1}{(v+\frac12)^2}-\frac{1}{(v+\frac32)^2}\right]
  =K_{uv}.
\]
The identity holds; the base case $u=v-1$ is the identity itself (there
$h_{v-1}(v)=K_{v-1,v}$).  Finally, a path's history before hitting $u$ is
independent of its subsequent jumps (every step draws a fresh $U$), so
\[
  \Pr[B_v=1 \given B_u=1]=h_u(v)=q_v=\Pr[B_v=1],
\]
i.e., $B_u$ and $B_v$ are independent and
$\Pr[B_i=1,B_j=1]=q_iq_j$.

\subsection*{Proposition~\ref{prop:rateless-estimator} (analytic estimator)}

Write element $x$'s hit indicators as $B_i^{(x)}$.  Since $q_0=1$,
$C_0=\sum_x s_x$, so
\[
  Y_i=\sum_xs_x\bigl(B_i^{(x)}-q_i\bigr)
\]
is a sum of $d$ mutually independent, zero-mean terms.  Squaring and taking
expectations, the cross terms vanish and the diagonal terms give
\[
  \E[Y_i^2]=\sum_xs_x^2\,\E\bigl[(B_i^{(x)}-q_i)^2\bigr]=d\,q_i(1-q_i)
\]
(with $B_i^{(x)}\sim\operatorname{Bernoulli}(q_i)$).  Therefore each term
$(C_i-q_iC_0)^2/[q_i(1-q_i)]$ is an unbiased estimator of $d$, and so is
their average.

\emph{Generalization (non-diagonal $\Sigma$).}  Unbiasedness uses only the
closed-form $q_i$ and element independence; it does not rely on
diagonality of $\Sigma$.  If a mapping yields an analytically known
non-diagonal covariance $\Sigma$, take $r=\operatorname{rank}(\Sigma)$; then
$\dhat=Y^\top\Sigma^+Y/r$ is still exactly unbiased ($\Sigma^+$ is the
Moore--Penrose pseudo-inverse).  The Rateless prefix's pairwise independence
(Lemma~\ref{lem:rateless-pairwise}) makes
$\Sigma=\operatorname{diag}(q_i(1-q_i))$ exactly diagonal, so the
generalized estimator reduces to the simple average above.

\subsection*{Proposition~\ref{prop:rateless-ci} (asymptotic confidence
interval)}

The elements' hit vectors $(B_i^{(x)})_{i=1}^{m-1}$ are generated i.i.d.\
across $x$ (each path is sampled independently by the same mechanism), so
$Y=\sum_x s_x(B^{(x)}-q)$ is a sum of $d$ mutually independent, zero-mean
vectors whose coordinates are bounded in $[-1,1]$; the Lindeberg condition
holds automatically, and fixed signs enter the covariance only as
$s_x^2=1$.  The single-term covariance is the diagonal matrix
$\Sigma=\operatorname{diag}(q_i(1-q_i))$ by
Lemma~\ref{lem:rateless-pairwise} and the variance calculation in the proof
of Proposition~\ref{prop:rateless-estimator}, and $0<q_i<1$ keeps the diagonal non-degenerate with
$\operatorname{rank}(\Sigma)=r=m-1$.  The multivariate CLT gives
$Y/\sqrt d\Longrightarrow\mathcal N(0,\Sigma)$.  Define the standardized
coordinates $Z_i=Y_i/\sqrt{d\,q_i(1-q_i)}$; then $(Z_i)$ is asymptotically
i.i.d.\ $\mathcal N(0,1)$, and
\[
  \frac{\dhat}{d}=\frac1r\sum_{i=1}^{r}Z_i^2.
\]
The continuous mapping theorem gives $\chi^2_r/r$; the confidence interval
follows by inverting the $\chi^2_r$ quantiles of $\dhat/d$.

\subsection*{Statistic/interval pairing (engineering implementation note)}

Proposition~\ref{prop:rateless-estimator}'s estimator has the shape
``standardize each cell by its own variance, then average''
(average-of-ratios):
\[
  \dhat=\frac1{m-1}\sum_{i=1}^{m-1}\frac{(C_i-q_iC_0)^2}{q_i(1-q_i)}.
\]
There is an alternative form with the same mean---``merge all energies
first, then divide by the merged variance'' (ratio-of-sums):
\[
  \widetilde d=\frac{\sum_i(C_i-q_iC_0)^2}{\sum_iq_i(1-q_i)}.
\]
Both have expectation $d$ and coincide pointwise in the equal-weight case
(all $q_i$ equal), so it is tempting to use the latter as a substitute in
practice.  But on a Rateless prefix with decreasing $q_i$, the two have
different distributions: $\widetilde d$ lets the few early, high-variance
cells dominate the statistic, converging to a weighted chi-square with
unequal weights whose effective degrees of freedom
$(\sum_i w_i)^2/\sum_i w_i^2$ (with $w_i=q_i(1-q_i)$) is far smaller than
$r$ and saturates quickly as the prefix grows, independent of $d$.  Applying
the $\chi^2_r$ interval of
Proposition~\ref{prop:rateless-ci} to $\widetilde d$ reduces the empirical
coverage of the nominal $90\%$ interval to approximately $50\%$.  Correctness
of the mean does not imply correctness of the distribution: the
Proposition~\ref{prop:rateless-ci} confidence interval holds only for the
average-of-ratios form, and implementations must use the two together.

\subsection*{Finite-$(d,m)$ coverage calibration and PRNG implementation
check}

Proposition~\ref{prop:rateless-ci} is a fixed-$m$, $d\to\infty$ limit, so at
finite $d$ the chi-square approximation leaves a calibration residual.  At
$d=16384$, the $m=1024$, $256$, $64$ tiers correspond to $d/r=16$, $64$,
$260$.  The paper-facing experiment contains seven sign compositions, three
prefixes, and three nominal levels---63 cells in total.  Every cell stays
within $0.324$ percentage points of nominal.  The residual is most visible
at the smallest-$d/r$ tier: all seven nominal-90\% cells at $m=1024$ lie
$0.084$--$0.324$ percentage points below nominal, averaging $0.172$ points
below.  Over all three nominal levels, the maximum absolute deviations at
$m=64$ and $256$ are only $0.157$ and $0.130$ points, respectively.  Thus
the complete sign-composition grid retains the finite-$(d,m)$ trend, at a
few tenths of a percentage point rather than the multi-tens-of-points
artifact produced by the historical ratio-of-sums bug.

A separate balanced-sign diagnostic compares the two sampling mechanisms
rather than expanding the sign-composition grid.  Across the two samplers,
three prefixes, and three nominal levels---18 diagnostic cells---coverage
stays within $0.22$ percentage points of nominal: 13 cells sit slightly below
and 5 slightly above.  Its single largest deviation, $0.21$ points at
$(m=64,90\%)$ on the official sampler, is a Monte Carlo fluctuation and is
not used to summarize the 63-cell paper-facing grid.

We next test whether the 64-bit PRNG used by the pinned implementation
contributes to the observed coverage error.  The pinned implementation derives
jump distances from a 64-bit PRNG, while
Proposition~\ref{prop:rateless-ci} assumes a fresh independent $U$ per
step.  If the recursion's correlation entered the second moments, the
deviation would not come from finite $d$ and would not vanish as $d$ grows.
The criterion is direct: run the ideal fresh-uniform kernel on the same
$(d,m)$; the ideal mechanism contains no PRNG, so if it deviates equally,
the deviation can only come from finite $d$.  At $m=1024$ with 100{,}000
trials per sampler:

\begin{table}[!hbtp]
\centering
\caption{Coverage of the ideal fresh-uniform kernel versus the official
pinned sampler at $m=1024$, $d=16384$, 100{,}000 trials per sampler.}
\label{tab:prng-check}
\small
\begin{tabular}{lrr}
\toprule
nominal & ideal fresh-uniform kernel & official pinned sampler \\
\midrule
90\% & 0.8982 & 0.8988 \\
95\% & 0.9486 & 0.9488 \\
\bottomrule
\end{tabular}
\end{table}

The two deviate at the same scale, so the PRNG recursion is not the source
(Table~\ref{tab:prng-check}).  In particular, the fresh-uniform mechanism
contains no pinned-PRNG recursion yet reproduces the systematic
smallest-$d/r$ deficit.  The residual reported in the main text is therefore
the finite-$(d,m)$ calibration error relative to the fixed-$m$, $d\to\infty$
chi-square limit, with ordinary Monte Carlo fluctuation superimposed on
individual cells.

\subsection*{Mapping-statistics verification}

To verify Lemmas~\ref{lem:rateless-hit}--\ref{lem:rateless-pairwise}
empirically, we sample $10^9$ independent mapping paths from both the
official pinned sampler and the fresh-uniform kernel, observing cells
$1,\dots,1023$.  The marginal frequencies match the closed-form $q_i$ within
sampling error in both samplers, and a standardized test over all
$\binom{1023}{2}=522{,}753$ cell pairs shows the joint-frequency deviation
$\widehat{\Pr}(B_i{=}1,B_j{=}1)-q_iq_j$ is pure sampling noise with no
systematic drift, confirming pairwise independence in both.

\begin{figure}[!hbtp]
\centering
\includegraphics[width=0.8\textwidth]{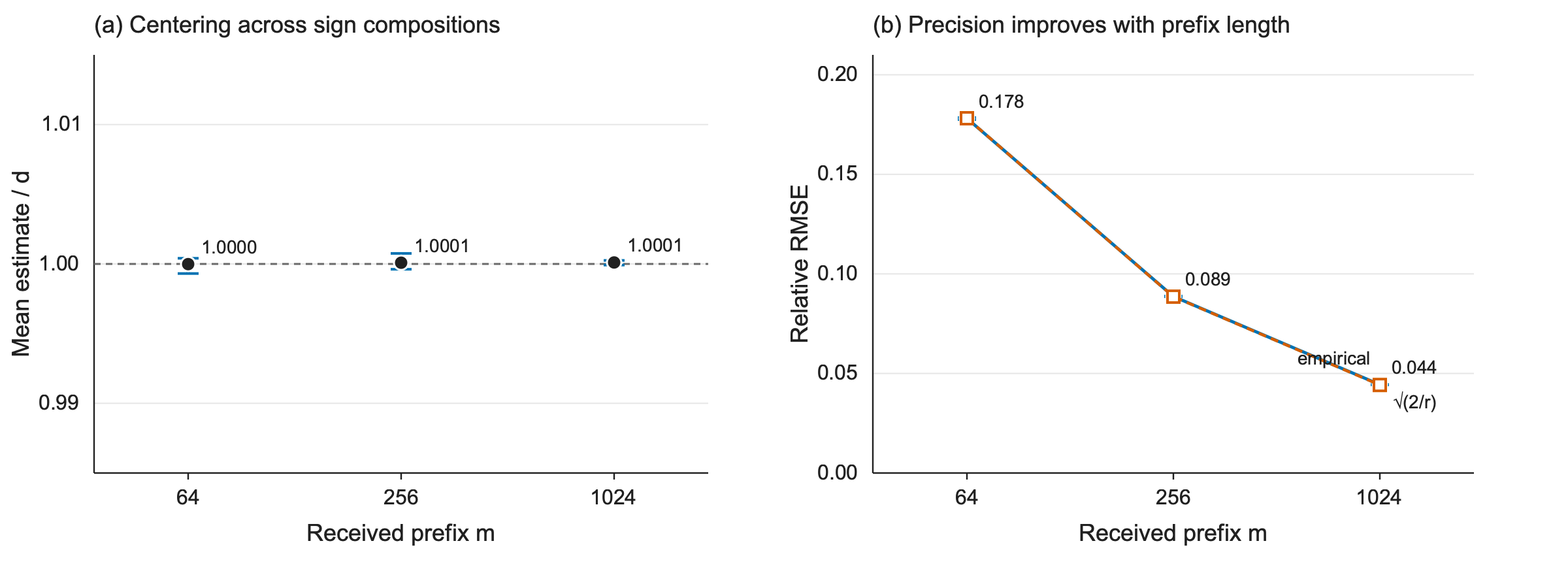}
\caption{Rateless prefix accuracy.  Top: mean($\dhat/d$) (points, vertical
ranges show the minimum and maximum over seven sign compositions); bottom:
RMSE($\dhat/d$) (points) compared with the theoretical curve $\sqrt{2/r}$
with $r=m-1$.}
\label{fig:r1}
\end{figure}

\section{Rateless IBLT: Production-Scale Comparison}
\label{app:rateless-comparison}

To test whether Rateless IBLT is viable under database reconciliation's
scan-dominated cost model, we compare three methods on two production-shaped
snapshots.  All methods share the same row-fingerprint snapshot, so the
comparison isolates the algorithm from the scan.

\begin{enumerate}[itemsep=0pt,topsep=2pt]
  \item \textbf{Self-sizing} (SS): Plain IBLT first round
  $M_1=512$, estimate $\dhat$, second round at $\alpha=1.824$.
  \item \textbf{Rateless fixed-prefix} (RF): a fixed Rateless prefix of
  $M=249{,}555$ symbols, constructed and transmitted in one batch.  The
  prefix is a fixed hard cap chosen generously above the decodable minimum
  (on P2, the minimal decodable prefix is only about $5{,}932$ symbols).
  \item \textbf{Rateless incremental} (RI): the official incremental
  encoder, starting with 1{,}024 symbols and extending until the
  receiver's horizon predictor signals convergence.
\end{enumerate}

Table~\ref{tab:rateless-vs-self-sizing} summarizes both snapshots.  The primary
comparison runs on the P2 snapshot (source $31{,}034{,}507$ rows, target
$31{,}034{,}459$, ground truth $d^+=720$, $d^-=672$), which has stable
repeated runs for all three methods plus a 22.5\,ms RTT point and complete
wire/memory telemetry.  The P1-extra snapshot (${\sim}609$M rows,
$d^+=109{,}208$, $d^-=44{,}890$) serves as the scale stress point.

\begin{table}[!hbtp]
\centering
\caption{Self-sizing versus Rateless on the two snapshots.  Wire bytes are
per endpoint (both directions combined double them).  SS timing on P1-extra
is the range of two runs.}
\label{tab:rateless-vs-self-sizing}
\small
\setlength{\tabcolsep}{4pt}
\begin{tabular}{lllcc}
\toprule
snapshot & rows (src / tgt) & method & wall time (s) & wire bytes / endpoint \\
\midrule
\multirow{3}{*}{P2}
  & \multirow{3}{*}{31{,}034{,}507 / 31{,}034{,}459}
  & SS     & 5.584--5.611 & 91{,}424 \\
  & & RF ($M=249{,}555$) & 25.115--25.158 & 7{,}985{,}760 \\
  & & RI     & 141.393--141.824 & 189{,}824 \\
\midrule
\multirow{3}{*}{P1-extra}
  & \multirow{3}{*}{$\sim$609M / $\sim$609M}
  & SS     & 115.6--119.8 & 8{,}002{,}144 \\
  & & RF ($M=249{,}555$) & 490.4--491.9 & 7{,}985{,}760 \\
  & & RI     & 5{,}715.8--5{,}803.6 & 6{,}960{,}320 \\
\bottomrule
\end{tabular}
\end{table}

\paragraph{Timing.}
On P1-extra, self-sizing completes in ${\sim}120$\,s.  The fixed-prefix method
takes ${\sim}490$\,s despite transmitting a comparable byte volume, because
the endpoint must execute the Rateless mapping over all rows for every
coded symbol batch.  The incremental method takes ${\sim}5{,}800$\,s, about
$48\times$ slower than self-sizing, with endpoint memory peaking at
${\sim}93$\,GiB---three orders of magnitude above self-sizing's endpoint,
which peaks at a few tens of MiB.  Its wire bytes are the lowest on this
snapshot, but the saving is outweighed by the computation and memory cost.

\paragraph{Why.}
The official incremental encoder retains per-source mapping and scheduling
state that is $O(N)$ in size: producing a prefix advances these shared
mapping paths,
with mapping and heap work that grows with the requested prefix.  In
contrast, a Plain IBLT encoder is stateless---each row independently
increments $k$ cells---so memory scales with the table size $M$, not the
row count $N$.  When $N$ is large and the reconciliation runs in a
request/response pattern (one scan per round, not a persistent stream),
the $O(N)$ encoder cost converts Rateless's bandwidth advantage into a
net loss on wall time and memory.

\section{MET Adapter Proofs and the Identifiability Boundary}
\label{app:met}
\label{app:met-proofs}

\subsection*{Lemma~\ref{lem:met-partial} (partial-type correction)}

A key's $a$ positions are chosen uniformly without replacement among the
$m$ cells; the $r$ received cells fix $r$ of these positions, so the hit
count follows a hypergeometric distribution with mean and variance
\[
  \E[L]=\frac{ra}{m},
  \qquad
  \Var(L)=r\cdot\frac{a}{m}\left(1-\frac{a}{m}\right)\cdot\frac{m-r}{m-1}.
\]
Substituting $\E[L^2]=\Var(L)+\E[L]^2$:
\[
  \gamma_r(a)
  =\frac{ra}{m}-\frac{a}{m}\cdot\frac{m-a}{m}\cdot\frac{m-r}{m-1}-\frac{ra^2}{m^2}.
\]
The first and last terms combine to $\frac{a(m-a)}{m^2}\,r$, so
\[
  \begin{aligned}
  \gamma_r(a)
  =\frac{a(m-a)}{m^2}\left[r-\frac{m-r}{m-1}\right]
  &=\frac{a(m-a)}{m^2}\cdot\frac{r(m-1)-(m-r)}{m-1}\\
  &=\frac{a(m-a)}{m^2}\cdot\frac{m(r-1)}{m-1}
  =\frac{a(m-a)(r-1)}{m(m-1)}.
  \end{aligned}
\]

\subsection*{Proposition~\ref{prop:met-identifiability} (identifiability
boundary)}

The expression for $\E[E_t]$ applies Lemma~\ref{lem:met-partial} to each
data type and sums over that type's keys.  Keys are independent, so
cross-term expectations vanish, and the expansion parallels
Theorem~\ref{thm:master}.
\emph{Sufficiency:} $\E[w^\top E]=w^\top\Gamma f=\one_J^\top f=F_2$.
\emph{Necessity:} a linear estimator $w^\top E$ must be unbiased for
\emph{every} non-negative $f$, i.e., $w^\top\Gamma f=\one_J^\top f$ must
hold identically; the non-negative orthant spans $\mathbb R^J$, so
$w^\top\Gamma=\one_J^\top$.

\subsection*{The ill-conditioned general-matrix adapter (unknown $p_j$ path)}

When the assignment fractions $p_j$ are unknown, only
Proposition~\ref{prop:met-identifiability}'s general matrix adapter is
available (take $w^\top\Gamma=\one_J^\top$ and estimate $F_2$ by $w^\top E$);
whether measurement is possible is governed by the identifiability gate.
Passing the gate guarantees only that an unbiased estimate \emph{exists};
it does not guarantee that this linear combination has usable variance.  In
the end-to-end experiments of the partial-type appendix, the first point
just past the gate (entering type 3) is highly ill-conditioned: at $r=2$
the matrix adapter's relative RMSE is $102.9$--$107.9$, and even with type 3
fully received it is still $12.82$--$12.97$.  The relative RMSE becomes
usable only after entering type 5, whose mapping direction adds a new
dimension: at $r/m=1/4$ the RMSE
drops to $0.865$--$0.880$, and with type 5 complete to
$0.459$--$0.467$, with mean intervals $0.993$--$1.006$ and $0.997$--$1.005$
respectively.  The general matrix adapter therefore must check weights or
variance beyond the identifiability gate; the known-$p_j$ scalar path
normalizes more stably on the same prefix, which is why the protocol uses it
by default.  For unbiasedness the two paths agree: the empirical mean
deviations of all identifiable points fall within $3.16$ Monte Carlo
standard errors, with the maximum exactly at the first identifiable point:
its RMSE exceeds $100$, so its Monte Carlo mean is the noisiest of the set,
and that deviation reflects sampling noise under the high variance rather
than bias; all other
points with usable precision agree with the unbiasedness of
Proposition~\ref{prop:met-identifiability}.

\section{MET Partial-Type Adapter: Full Statements and Validation}
\label{app:met-partial}

This appendix gives the complete statements, boundary checks,
identifiability discussion, and two-scale simulation validation of the MET
variant adapter in \S\ref{sec:met}.  The proofs are in
Appendix~\ref{app:met-proofs}.

MET-IBLT~\cite{lazaro2023ratecompatible} partitions cells into cell types and keys into data types: a cell
type is a group of cells sharing the same mapping parameters, and a data
type is a group of keys sharing the same mapping rule; together they
specify how the IBLT cells are partitioned and which cell types each data
type maps to.  Formally, type $t$ has $m_t$ cells, and each key of data
type $j$ chooses $a_{tj}$ cells uniformly without replacement inside type
$t$.

Transmission is a cell stream, so at any moment the receiver typically holds
only \emph{part} of a cell type: the group has $m_t$ cells, but the current
prefix contains only $r_t$ of them, a partial type.  Applying the
complete-type normalization to a partial prefix would overstate the
information contained in the received cells.  Lemma~\ref{lem:met-partial}
provides the corresponding finite-prefix correction.  The two results below
therefore show that a partial type already carries analytically computable
measurement
information, and whether total $F_2$ can be recovered has an exact linear
algebra criterion.

\subsection*{Lemma~\ref{lem:met-partial} (partial-type correction)}

\emph{Boundary checks.}  $\gamma_1(a)=0$: with one received cell there is no
residual degree of freedom after centering within the type;
$\gamma_m(a)=a-a^2/m$: the complete-type correction when all cells of the
type are received.  $a\in\{0,m\}$ also gives $\gamma_r(a)=0$, because that
type provides no centered energy for that data type.  Excluding these
degenerate rows, estimability does not have to wait for a complete-type
boundary: a partial MET prefix can already contribute to the estimate before
the complete type has been received.  The proof is in
Appendix~\ref{app:met-proofs}.

\begin{proposition}[Identifiability boundary]
\label{prop:met-identifiability}
Let there be $J$ data types and let
$f=(F_{2,1},\ldots,F_{2,J})^\top$ be their second moments, with target total
$F_2=\one_J^\top f$.  For each (partially) received cell type $t$ (having
received $r_t$ of $m_t$ cells), let $C_t\in\mathbb R^{r_t}$ be its received
counts, $Q_{r_t}=I_{r_t}-\one\one^\top/r_t$, and define the within-type
centered energy $E_t=C_t^\top Q_{r_t}C_t$.  Then
\[
  \E[E_t]=\sum_j\gamma_{tj}F_{2,j},
  \qquad
  \gamma_{tj}=\frac{a_{tj}(m_t-a_{tj})(r_t-1)}{m_t(m_t-1)},
\]
which stacks to $\E[E]=\Gamma f$.  The total $F_2$ can be linearly and
unbiasedly identified from the current prefix iff
\[
  \one_J\in\operatorname{rowspan}(\Gamma),
\]
i.e., there exists $w$ with $w^\top\Gamma=\one_J^\top$; in that case
$w^\top E$ is an unbiased estimator of $F_2$.
\end{proposition}

\emph{Proof.}  See Proposition~\ref{prop:met-identifiability} in
Appendix~\ref{app:met-proofs}.

The criterion can identify the total $F_2$ without identifying every data
type separately, and known assignment fractions can make identification
possible even earlier.

If $\Gamma$ had full column rank, each data type's second moment $F_{2,j}$
could be identified separately, and the total follows.  But we only need the
total, so the criterion requires only $\one_J$ in the row space of $\Gamma$:
many prefixes whose $\Gamma$ is rank-deficient and cannot separate data
types still suffice to identify the total $F_2$ unbiasedly.

If each key's data type is assigned by an independent hash with known
fractions $p_j$, and this random assignment is included in the expectation,
then $\E[F_{2,j}]=p_jF_2$, so each received cell type's energy coefficient
merges into $\sum_j p_j\gamma_{tj}$, and a single cell type already gives an
unbiased estimate of total $F_2$---the total becomes identifiable from a
shorter prefix than the general matrix criterion requires.

When neither path holds ($\one_J\notin\operatorname{rowspan}(\Gamma)$ and
the fractions $p_j$ unknown), total $F_2$ is generally unidentifiable from
the current prefix; the adapter should return \texttt{UNIDENTIFIABLE} and
either continue receiving cells or enter fallback without producing an
estimate.  ``Identifiable'' here is a control criterion: it asks
only whether the current observations uniquely determine the total
difference cardinality $F_2$ that the controller cares about, and does not
require decomposing $F_2$ back into per-data-type values.  After
Lemma~\ref{lem:met-partial} and
Proposition~\ref{prop:met-identifiability}, the unsolved case on a MET
partial prefix reduces exactly to this one; all other scenarios
have an explicit estimation formula or state criterion.

\subsection*{Simulation validation: partial-type calibration and
incremental estimation}

Validation has two layers.  The first layer validates the normalization
constant $\gamma_r(a)$ of Lemma~\ref{lem:met-partial} independently of the
full MET mapper: it samples the hypergeometric hit count $L$ directly and
compares $\operatorname{mean}(L-L^2/r)$ against $\gamma_r(a)$, locating
potential algebraic or implementation errors in the finite-prefix correction
without involving the mapper and block centering.  The second layer puts the
verified constant back into the full estimation chain and examines three
behaviors of the complete adapter on a real MET mapper: whether the estimate
is unbiased, whether precision improves steadily with the prefix, and
whether the M2 identifiability criterion flips from \texttt{UNIDENTIFIABLE}
to \texttt{IDENTIFIABLE} as expected on real prefixes.  Both experiments use
the three data types of the literature design, with hash assignment
probabilities $(0.1959,0.1904,0.6137)$ and write degrees $1$--$5$ inside
the different cell types.

\subsubsection*{Partial-type local calibration}

The first layer runs over the first six cell types, whose sizes double from
$m=50$ to $1600$; for each cell type and data type, at received fractions
$r/m\approx1/4,1/2,3/4$, it compares the empirical mean
$\operatorname{mean}(L-L^2/r)$ against the closed form $\gamma_r(a)$ with
$L\sim\operatorname{Hypergeometric}(m,a,r)$, giving
$6\times3\times3=54$ non-degenerate check points ($10^6$ independent $L$
per point).

\begin{table}[!hbtp]
\centering
\caption{Partial-type local calibration of $\gamma_r(a)$ (54 points,
$10^6$ samples each).}
\label{tab:met-calib}
\small
\begin{tabular}{lrrrr}
\toprule
received fraction $r/m$ & points & max abs.\ relative error & mean signed relative error & max $|z|$ \\
\midrule
$\approx1/4$ & 18 & $0.218\%$ & $+0.0068\%$ & $1.48$ \\
$1/2$       & 18 & $0.309\%$ & $-0.0168\%$ & $3.09$ \\
$\approx3/4$ & 18 & $0.092\%$ & $+0.0149\%$ & $1.59$ \\
\bottomrule
\end{tabular}
\end{table}

The maximum absolute relative error over the 54 points is $0.309\%$ (Table~\ref{tab:met-calib}).  The
worst point is $(m,a,r)=(1600,1,800)$, with empirical mean $0.497832$ versus
closed form $0.499375$ and standardized error $|z|=3.09$.  The average
errors at the three received fractions are all near zero with inconsistent
signs and no systematic drift in $m$, $a$, or $r/m$, directly supporting the
finite-prefix factor $(r-1)/(m-1)$ of $\gamma_r(a)$.  These Monte Carlo
points cover the three received fractions within a type; $r=1$ and $r=m$
remain covered by the analytic boundary checks above.

\subsubsection*{End-to-end incremental estimation (known-fraction path)}

The second layer generates complete count arrays on a real MET mapper and
follows the known-fraction special case of
Proposition~\ref{prop:met-identifiability}: for each received partial type,
compute the within-type centered energy, normalize by
$\sum_j p_j\gamma_{tj}$ using the known data-type assignment probabilities,
and accumulate across types to obtain a scalar estimate $\dhat$.  Fixing
$d=16384$ and sweeping the difference fractions
$\{0,0.1,0.25,0.5,0.75,0.9,1\}$, estimates are taken at $r=2$ and at
$r/m\approx1/4,1/2,3/4,1$ along the first five cell types.  Table~\ref{tab:met-e2e} reports
representative prefix positions within the first three cell types; intervals
are min--max
over the seven sign compositions (20{,}000 trials per configuration;
$7\times20{,}000\times25=3{,}500{,}000$ prefix observations).

\begin{table}[!hbtp]
\centering
\caption{MET end-to-end incremental estimation (known-fraction path,
$d=16384$; min--max over seven sign compositions).}
\label{tab:met-e2e}
\small
\begin{tabular}{lrrrr}
\toprule
current type & $r/m$ of current type & cumulative prefix & mean($\dhat/d$) & RMSE($\dhat/d$) \\
\midrule
1 & $\approx1/4$ & 13  & $0.996$--$1.003$ & $0.405$--$0.411$ \\
1 & $1/2$       & 25  & $0.998$--$1.002$ & $0.288$--$0.293$ \\
1 & $\approx3/4$ & 38 & $0.998$--$1.002$ & $0.231$--$0.236$ \\
1 & $1$         & 50  & $0.999$--$1.001$ & $0.200$--$0.204$ \\
2 & $1/2$       & 100 & $0.999$--$1.001$ & $0.160$--$0.163$ \\
2 & $1$         & 150 & $1.000$--$1.001$ & $0.134$--$0.137$ \\
3 & $1/4$       & 200 & $1.000$--$1.001$ & $0.124$--$0.126$ \\
3 & $1$         & 350 & $1.000$--$1.001$ & $0.101$--$0.102$ \\
\bottomrule
\end{tabular}
\end{table}

Across all 20 prefixes with $r>2$, the maximum deviation of the mean from 1
over the seven sign compositions is $0.36\%$, with a maximum standardized
deviation of $1.49$ Monte Carlo standard errors and no systematic drift in
sign fraction or prefix.  As type 1 grows from about one quarter received to
complete, the RMSE falls from about $0.41$ to $0.20$; continuing into later
types, it drops to about $0.10$ at prefix 350 (Table~\ref{tab:met-e2e}).  These
results show that the estimate becomes progressively more precise as
additional cells within a type arrive; it need not wait for a complete-type
boundary.

\subsubsection*{The identifiability state boundary}

The scalar path above assumes known fractions $p_j$; when $p_j$ is unknown,
only the general matrix adapter of
Proposition~\ref{prop:met-identifiability} is available (take
$w^\top\Gamma=\one_J^\top$ and estimate $F_2$ by $w^\top E$), and whether
measurement is possible is governed by the identifiability criterion.  On the
same end-to-end prefix, we compute the current $\Gamma$'s rank and whether
$\one^\top$ lies in its row space step by step: with only type 1 received,
$\operatorname{rank}(\Gamma)=1$; with the first two types, rank 2; both
return \texttt{UNIDENTIFIABLE}.  After entering type 3 and receiving at
least two cells, the rank rises to 3 and the adapter turns
\texttt{IDENTIFIABLE}.  This state depends only on the prefix's mapping
parameters, not on the observed count values, so it stably reproduces the
\texttt{UNIDENTIFIABLE}$\rightarrow$\texttt{IDENTIFIABLE} transition defined by
Proposition~\ref{prop:met-identifiability}.

Passing the identifiability gate guarantees only that an unbiased estimate
exists, not that its variance is usable: the general matrix adapter can be
highly ill-conditioned just past the gate (its relative RMSE exceeds
100 upon entering type 3 and decreases only after later types provide
additional mapping directions; tier-by-tier numbers are in
Appendix~\ref{app:met-proofs}),
so weights or variance must be checked beyond the identifiability gate (Figure~\ref{fig:m1}).
This is also why the protocol defaults to the known-$p_j$ scalar path.  For
unbiasedness the two paths agree: the empirical mean deviations of all
identifiable points fall within $3.16$ Monte Carlo standard errors.

\begin{figure}[!hbtp]
\centering
\includegraphics[width=0.95\textwidth]{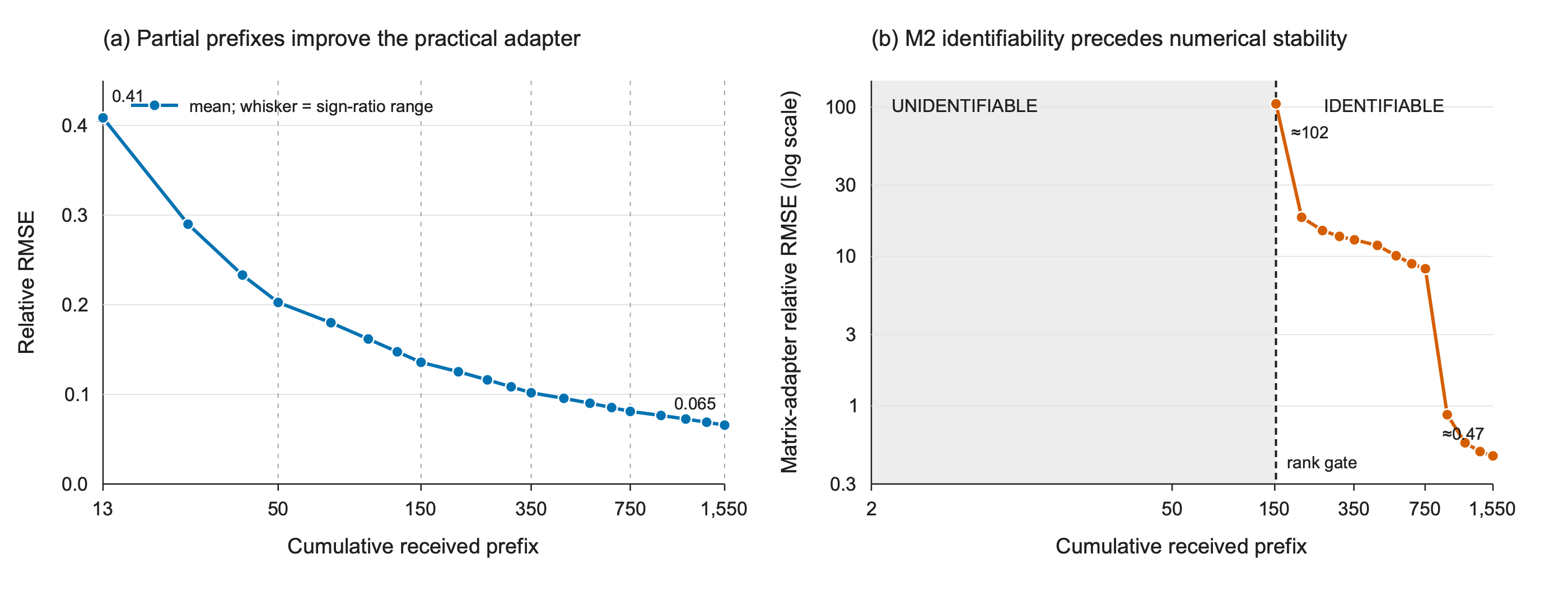}
\caption{MET partial-prefix precision and the M2 identifiability boundary.
Left: relative RMSE of the known-$p_j$ adapter over 20 prefixes with
$r>2$; points are means over the seven sign compositions with min--max
whiskers, and precision improves progressively as additional cells arrive,
both within a cell type and across type boundaries.  Right: the general
matrix adapter on a log y-axis; left of the dashed line the M2 row-space
condition fails, and after entering type 3 the estimate is identifiable but
still has very large variance, decreasing only after later types introduce
additional mapping directions.
Both x-axes are logarithmic to show early cell-level prefixes and later
type boundaries together.}
\label{fig:m1}
\end{figure}

\section{Proof of Corollary~\ref{cor:variant-failure} (Failure Conditioning
for Variants)}
\label{app:variant-failure}

\subsection*{Verification of the premises for the three variants}

Corollary~\ref{cor:variant-failure} requires the failure event $F$ to be a
measurable function of the mappings $\{a_x\}$ and the signs to be
independent of the mappings.  Both premises hold for the three variants.
For Irregular IBLTs, failure is determined by the 2-core of a hypergraph
with random edge degrees.  For Rateless and MET IBLTs, it is determined by
the 2-core of the hypergraph induced by the received prefix.  In every
case, the failure event depends only on the mappings and not on the signs.
The assumption of i.i.d.\ uniform $\pm1$ signs independent of the mappings
is the same as in \S\ref{sec:failure-cond}.  Together with the Plain
adapter, these observations establish the premises for all four adapters,
so the corollary applies directly to them.

\subsection*{Proof}

(i) Conditioned on all the mappings,
\[
  C^\top QC=\sum_xe_x+2\sum_{x<y}s_xs_y\,a_x^\top Qa_y.
\]
The cross-term coefficients are constants and $\E[s_xs_y]=0$, so
$\E[C^\top QC \given \{a_x\}]=\sum_xe_x$ holds for every mapping
realization.  Since $\mathbf1_F$ is a measurable function of the mappings,
$\E[\dhat\,\mathbf1_F]=\gamma^{-1}\E[\sum_xe_x\,\mathbf1_F]$, and dividing
by $p_F$ gives (i).

(ii) In this case $\sum_xe_x\equiv d\gamma$ almost surely; substitute into
(i).

(iii) By (i) and $\E[\sum_xe_x]=d\gamma$,
\[
  \E[\dhat \given F]-d
  =\frac{\Cov\bigl(\sum_xe_x,\mathbf1_F\bigr)}{\gamma\,p_F}.
\]
Combining Cauchy--Schwarz with $\Var(\mathbf1_F)=p_F(1-p_F)$ and, by
element independence, $\Var(\sum_xe_x)=d\sigma_E^2$ yields (iii).

(iv) Fix the received prefix, so that $w$, $\Gamma$ and the within-type
centering matrices $Q_{r_t}$ are deterministic.  Expanding each
$E_t=C_t^\top Q_{r_t}C_t$ as in (i) and forming their weighted combination
using $w$,
\[
  \dhat=w^\top E=\sum_xe_x
  +2\sum_{x<y}s_xs_y\sum_tw_t\,a_{x,t}^\top Q_{r_t}a_{y,t},
  \qquad
  e_x=\sum_tw_t\,a_{x,t}^\top Q_{r_t}a_{x,t}.
\]
The cross-term coefficients are again functions of the mappings alone and
$\E[s_xs_y]=0$, so $\E[\dhat \given \{a_x\}]=\sum_xe_x$ for every mapping
realization, which is (i) with $\gamma=1$.  For an element of data type $j$,
Proposition~\ref{prop:met-identifiability}'s per-type expectation gives
$\E[e_x]=\sum_tw_t\gamma_{tj}=(w^\top\Gamma)_j=1$, so
$\E[\sum_xe_x]=\sum_jd_j=d$: the weight condition $w^\top\Gamma=\one_J^\top$
is what makes the heterogeneous data types share one normalization.  If every
received cell type is complete ($r_t=m_t$), then
$a_{x,t}^\top Q_{r_t}a_{x,t}=a_{tj}-a_{tj}^2/m_t=\gamma_{m_t}(a_{tj})$ is
determined by the data type, hence $e_x=(w^\top\Gamma)_j=1$ almost surely and
the conditional mean is exactly $d$.  Otherwise the elements remain
independent but are no longer identically distributed across data types, so
$\Var(\sum_xe_x)=\sum_jd_j\sigma_{E,j}^2$, and the Cauchy--Schwarz step of
(iii) gives \eqref{eq:cor41-bound-met}.

\subsection*{Instantiation for the three variants}

\begin{itemize}[itemsep=0pt,topsep=2pt]
  \item \emph{Irregular:} $e_x=D_x-D_x^2/M$ is completely determined by
  the random degree $D_x$, so
  \[
    \sigma_E^2
    =\E[D^2]-\frac{2\,\E[D^3]}{M}+\frac{\E[D^4]}{M^2}-\gamma_{\mathrm{irr}}^2,
  \]
  depending only on the first four moments of the degree distribution.  For
  the $\Lambda$ design used in the validation of \S\ref{sec:irregular}
  ($\E[D^3]=749.775$, $\E[D^4]=13183.125$), at $M=1024$
  $\sigma_E/\gamma_{\mathrm{irr}}\approx1.05$, so the relative conditional
  deviation is at most $1.05\,d^{-1/2}\sqrt{(1-p_F)/p_F}$.
  \item \emph{Rateless:} the prefix self-energy varies with the cells
  visited by the random mapping path, and $\sigma_E^2$ can be computed explicitly symbol by symbol
  from Lemma~\ref{lem:rateless-hit}'s closed-form $q_i$ and
  Lemma~\ref{lem:rateless-pairwise}'s pairwise independence.
  \item \emph{MET:} this is the multi-type case (iv).  A prefix whose
  received cell types are all complete gives $e_x\equiv1$ and exact
  conditional unbiasedness; when partial types are present, each data type's
  $\sigma_{E,j}^2$ follows from the hypergeometric moments of
  Lemma~\ref{lem:met-partial} weighted by $w$, and
  \eqref{eq:cor41-bound-met} applies.  Both statements use the identifiability
  weights $w^\top\Gamma=\one_J^\top$ of
  Proposition~\ref{prop:met-identifiability}; on a prefix where those weights
  do not exist the adapter returns \texttt{UNIDENTIFIABLE} without producing
  an estimate, so failure conditioning does not arise.
\end{itemize}

\subsection*{Failure-conditioning check for the Irregular instance}

Take $d=1024$, the $\Lambda$ design of \S\ref{sec:irregular}, and ten
selected values of $M$ between 1024 and 1300 covering the transition region
from $p_F=1$ to $p_F=0.018$, with i.i.d.\ random signs and two fixed sign
compositions at 20{,}000 trials each (600,000 decoding trials in total).
Under i.i.d.\ signs, the conditional-mean deviation is visible in the
transition region and positive:
at $M=1150$ ($p_F=0.192$), mean$(\dhat/d \given F)=1.0127$; at $M=1175$
($p_F=0.062$), $1.0181$; both are between 11 and 14 Monte Carlo standard
errors (13.7 and 11.1 respectively),
providing clear empirical evidence of the mapping-selection effect described
in part (iii).  All
adequately sampled points with $0<p_F<1$ stay inside (iii)'s bound by a
wide margin (e.g., the bound is $0.0676$ at $M=1150$ and $0.1285$ at
$M=1175$).  At the deep-overload end ($M\le1075$), all 20{,}000 trials at
each configuration fail to decode, giving an empirical failure rate of
$\widehat p_F=1$.  The measured conditional-mean deviation is at most
$4.6\times10^{-4}$, or $1.2$ Monte Carlo standard errors, and is therefore at
the noise floor.  Because no successful decode was observed at these tiers,
the experiment cannot resolve how close the true $p_F$ is to $1$; we therefore
do not evaluate the plug-in bound using the empirical failure rate.  The
fixed one-sided sign comparison shows no directional
deviation of the scale seen for Plain under one-sided signs.  Under fixed
one-sided signs, the additional directional deviation is at most $0.5\%$
and only marginally significant.  The observed failure-conditioned
deviation for the Irregular adapter therefore appears to be dominated by
selection through the random self-energy rather than by the sign cross
terms.  This differs from the Plain $k$-regular case, where the self-energy
is deterministic and any conditional deviation under fixed signs comes from
the cross terms.  We treat this mechanism-level comparison as exploratory.

\section{Self-Sizing Protocol Supplement---Capacity Configuration
and Protocol Cost}
\label{app:protocol-details}

This appendix supplements the three protocol-level results of
\S\ref{sec:protocol}: the first subsection gives the tier-by-tier
configuration of first-round capacity, failure-conditioned lower quantile,
second-round multiplier, and success rate; the second examines parameter
sensitivity and clarifies the cost-accounting rules used in the
regret-versus-oracle comparison; and
the third explains Strata's recovery semantics and its structural difference
from a Plain IBLT under the same budget.

\subsection{First-round capacity tiers and capacity multipliers}

We first clarify the configuration semantics.  The protocol targets $99\%$
failure-conditioned lower-tail probability: throughout, $\delta=0.01$, and
the failed-only empirical $1\%$ quantile $q_{0.01}$ controls the risk of
severe underestimation.  Two layers of evidence must be distinguished:
Table~\ref{tab:c1} derives the recommended multipliers from $q_{0.01}$,
answering how the second-round capacity should be set; Table~\ref{tab:c2}
runs end-to-end simulations at several common multipliers, answering what
success rate this rule actually achieves in decoding.

The protocol must choose the first-round capacity $M_1$ in advance.  It
determines the first-round fixed bytes and estimation precision
$\operatorname{RSD}(\dhat)\approx\sqrt{2/(M_1-1)}$ (see \eqref{eq:rsd}); the
precision determines $q_\delta$, which in turn yields the second-round
capacity multiplier through $\alpha\ge\beta/q_\delta$.
Table~\ref{tab:c1} instantiates this parameter chain on five standard tiers.
Byte counts assume $32$\,B per cell, and the recommended multiplier is
computed with the calibrated $\beta=1.56$ (equivalently, the previous
$1.3$ operating point followed by the $1.2$ engineering margin) and rounded
up to two decimals.  For each tier,
$q_\delta$ takes the minimum empirical value over the tested configurations,
so one tier's recommendation stays safe across all the sign compositions and
loads in the table.

\begin{table}[!hbtp]
\centering
\caption{Configuration of the five first-round capacity tiers.  Each row
reports the first-round byte cost and estimation precision, the worst
empirical $q_\delta$, and the resulting recommended second-round multiplier.}
\label{tab:c1}
\small
\begin{tabular}{lrrrr}
\toprule
$M_1$ & first-round bytes & $\operatorname{RSD}(\dhat)$ & worst empirical $q_\delta$ & recommended $\alpha=\beta/q_\delta$ \\
\midrule
64   & 2\,KB  & 17.7\% & 0.6308 & 2.48 \\
256  & 8\,KB  & 8.9\%  & 0.8045 & 1.94 \\
512  & 16\,KB & 6.3\%  & 0.8598 & 1.82 \\
1024 & 32\,KB & 4.4\%  & 0.8997 & 1.74 \\
4096 & 128\,KB & 2.2\%  & 0.9482 & 1.65 \\
\bottomrule
\end{tabular}
\end{table}

The five tiers trade first-round fixed cost against second-round margin.
Growing $M_1$ from 64 to 256
costs 6\,KB more in the first round, raises $q_\delta$ from 0.6308 to
0.8045, and lowers the recommended multiplier from 2.48 to 1.94---the most
pronounced improvement in the table.  Continuing through 512, 1024, and 4096 cells
successively lowers the multiplier to 1.82, 1.74, and 1.65 while raising the
first-round fixed bytes to 16\,KB, 32\,KB, and 128\,KB.  Thus 64 cells suit
minimal fixed cost, 256 balances first-round bytes against second-round
margin, 512 provides higher precision, and the gains beyond 1024 are mainly
useful when the first-round precision budget permits the extra fixed cost.

Table~\ref{tab:c2} turns to protocol-level validation.  The experiment
sweeps $\alpha\in\{1.6,1.8,2.0\}$ on the same sign-balance grid and reports
the two success semantics distinguished in \S\ref{sec:success-calibration}:
overall two-attempt success and failed-only second-round success
$\Pr[ok_2 \given F]$, the latter corresponding directly to
Proposition~\ref{prop:capacity}'s conditional guarantee.

\begin{table}[!hbtp]
\centering
\caption{Second-round success rates over the $M_1\times\alpha$ evaluation
grid.  The $M_1=64$ stress tier is shown separately; the pooled overall row
covers deployment-grade tiers $M_1\ge256$.  The columns report the evaluated
multiplier values rather than the tier-specific recommendations of
Table~\ref{tab:c1}.}
\label{tab:c2}
\small
\begin{tabular}{lrrr}
\toprule
semantics & $\alpha=1.6$ & $\alpha=1.8$ & $\alpha=2.0$ \\
\midrule
$\Pr[ok_2 \given F]$, $M_1=64$ (stress)  & 86.589\% & 95.123\% & 98.168\% \\
overall two-attempt success (all trials; $M_1\ge256$) & 99.854\% & 99.957\% & 99.967\% \\
$\Pr[ok_2 \given F]$, $M_1=256$  & 99.336\% & 99.855\% & 99.896\% \\
$\Pr[ok_2 \given F]$, $M_1=512$  & 99.901\% & 99.935\% & 99.948\% \\
$\Pr[ok_2 \given F]$, $M_1=1024$ & 99.955\% & 99.969\% & 99.978\% \\
$\Pr[ok_2 \given F]$, $M_1=4096$ & 99.988\% & 99.995\% & 99.992\% \\
\bottomrule
\end{tabular}
\end{table}

The grid shows the effect of $M_1$ at each fixed value of $\alpha$.  The lowest
scanned point, $\alpha=1.6$, sits near the recommended range of the three
common tiers: below $M_1=256$'s recommended 1.94 and below $M_1=512$'s
and $1024$'s 1.82 and 1.74.  Accordingly, $M_1=256$
reaches $99.336\%$ at that point, while the latter two tiers already have
extra capacity margin and reach $99.901\%$ and $99.955\%$.  Raising
$\alpha$ to 1.8 brings $M_1=256$ to $99.855\%$.  Configurations derived
from a $99\%$ lower-tail target therefore achieve actual second-round
success rates close to $99.9\%$ at the measured operating points; the
design target stays at $99\%$, and the higher measured rates serve as
empirical validation.  The $M_1=4096$ row shows a small non-monotonic step of
$0.003$ percentage points between $\alpha=1.8$ and $2.0$; that tier retains
only
12--19 failed trials, so the difference is Monte Carlo fluctuation in the
saturation region.

\noindent\textit{Placement of the $d_{\mathrm{rec}}$ deduction.}  In
\eqref{eq:capacity} the recovered count $d_{\mathrm{rec}}$ is subtracted from
the amplified planning bound $\dhat/q_\delta$, not from $\dhat$ before
dividing by $q_\delta$.  Subtracting first would multiply the exactly known
$d_{\mathrm{rec}}$ by $1/q_\delta$ and could push $M_2$ below
$\lceil\beta(d-d_{\mathrm{rec}})\rceil$, breaking the sufficiency argument of
Proposition~\ref{prop:capacity}.  The $\max\{0,\cdot\}$ truncation activates
only when $\dhat/q_\delta<d_{\mathrm{rec}}$, i.e., only within the
underestimation event whose conditional probability is bounded by $\delta$;
outside that event the deduction is exact.  The quantile $q_\delta$
may come either from the distribution-free bound of
Proposition~\ref{prop:lower-tail} or from failed-only empirical quantiles,
and neither the formula nor the success bound changes.

\subsection{Parameter sensitivity of capacity-setting methods}
\label{app:param-sensitivity}

The main figure of \S\ref{sec:regret} fixes $M_1=64$.  To examine
sensitivity to this choice, we evaluate the full $M_1\times d$ grid,
identify where the ordering between self-sizing and Strata-first changes and
why, and give the expected byte ledger under a hypothetical workload
dominated by zero-difference requests.

The oracle, self-sizing, blind doubling, and Strata-first use the same byte
and round accounting: each cell occupies 32\,B, and capacities are rounded
up to multiples of 64 cells.

\subsubsection*{The $M_1\times d$ crossover phase diagram}

The crossover scan covers
\[
  M_1\in\{64,128,256,512,1024,1280\},
\]
with $d$ from 0 up to $10^5$: 72 configurations total.
Figure~\ref{fig:c1} reports the mean byte ratio of self-sizing+joint over
Strata-first.

\begin{figure}[!hbtp]
\centering
\includegraphics[width=0.8\textwidth]{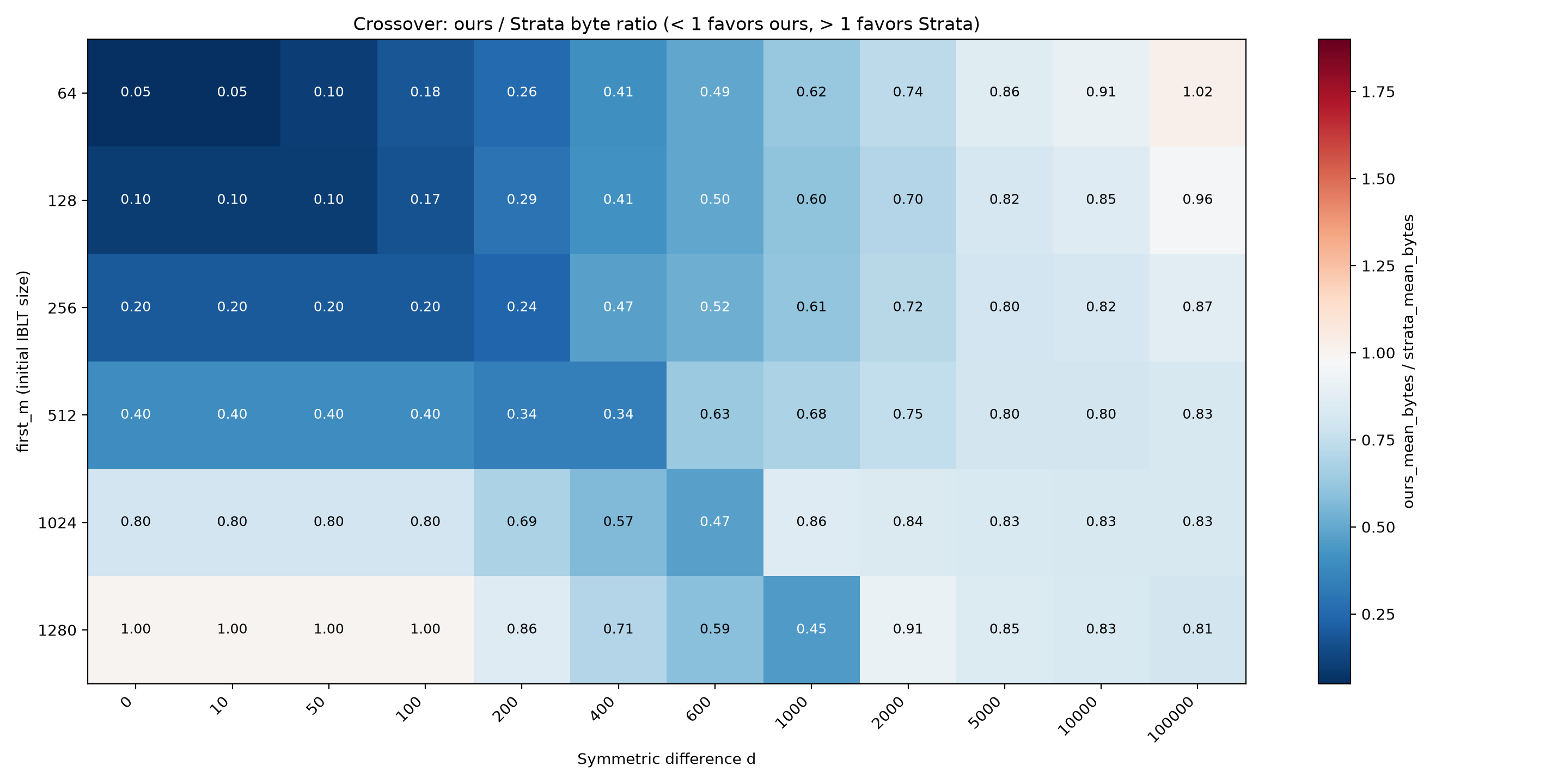}
\caption{Byte ratio self-sizing+joint / Strata-first over the $M_1\times d$
phase diagram.  71 of the 72 cells have ratio at most 1.}
\label{fig:c1}
\end{figure}

Of the 72 cells, 71 have ratio at most 1.  The sole exception is
$(M_1=64,d=10^5)$ with ratio 1.02: there the 2\,KB first-round IBLT sketch
has an estimation RSD of $17.7\%$ while Strata's standard configuration has
about
$11.0\%$; once $d$ is large enough that Strata's 40\,KB fixed summary has
been amortized, the precision difference causes a slight reversal.  At
$M_1=128$ the same endpoint returns to 0.96, and for $M_1\ge256$ the whole
line stays at most 0.87.

Strata's round advantage is concentrated in a narrow small-$d$ band: when
the bottom Strata layer can enumerate the differences directly, Strata-first
finishes in its estimation phase; for $M_1\ge256$, self-sizing's first-round
decodable range covers that window and the round advantage disappears.  The
corresponding round phase diagram is Figure~\ref{fig:c2}.

\begin{figure}[!hbtp]
\centering
\includegraphics[width=0.8\textwidth]{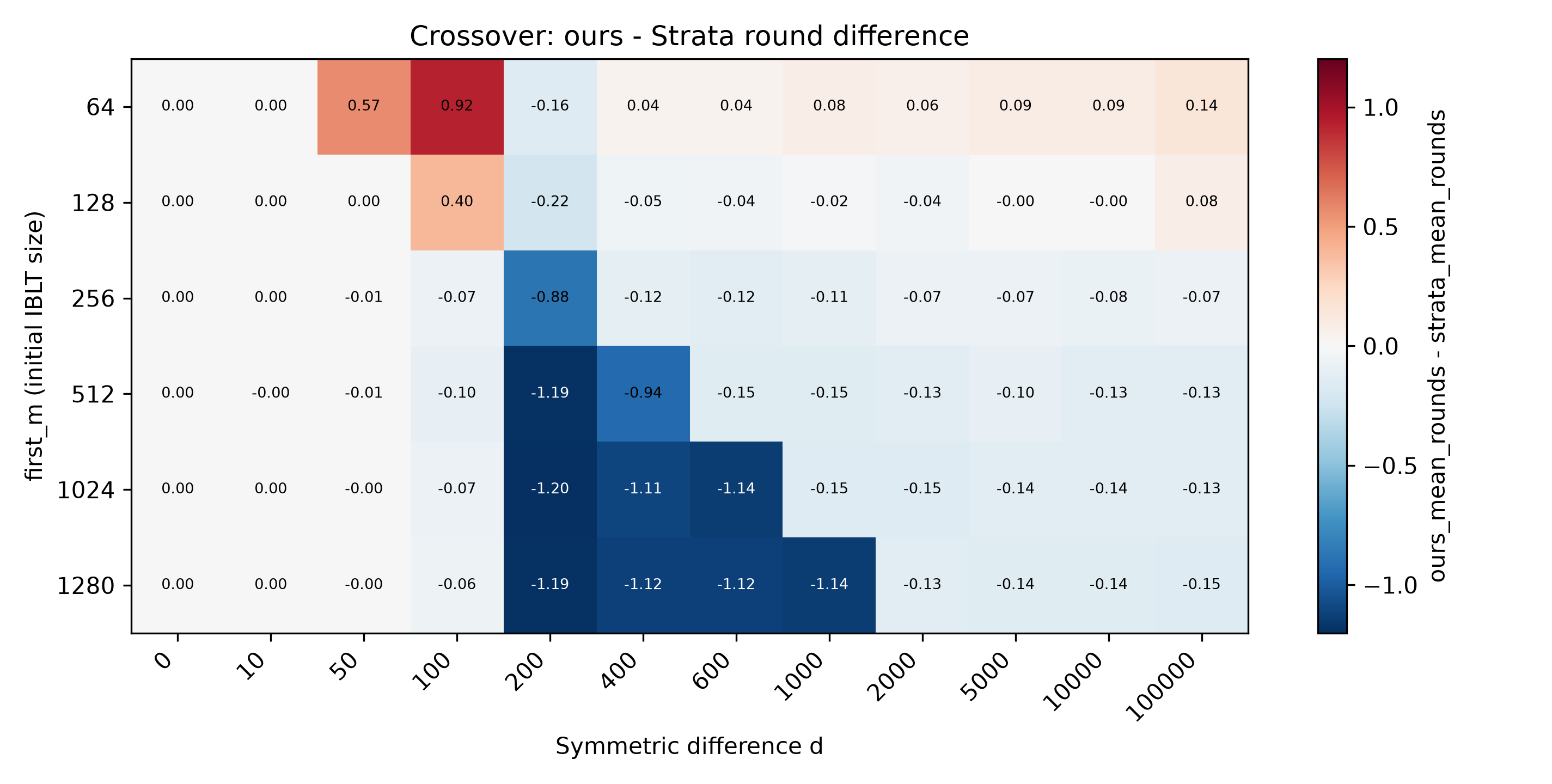}
\caption{Mean round difference between self-sizing+joint and Strata-first
over the $M_1\times d$ phase diagram (negative means self-sizing uses fewer
rounds).}
\label{fig:c2}
\end{figure}

\subsubsection*{Expected bytes when no-difference requests dominate}

Assume the two sides have no difference most of the time.  Take $d=0$
(no difference) for 95\% of all verification requests, fix the remaining
5\% at $d=10^4$, and use $M_1=64$.  Self-sizing then incurs a fixed
first-round cost of 2\,KB, whereas the standard Strata configuration uses a
40\,KB summary.

\begin{table}[!hbtp]
\centering
\caption{Expected bytes when 95\% of requests have $d=0$.}
\label{tab:c3}
\small
\begin{tabular}{lrr}
\toprule
method & expected bytes & relative to oracle \\
\midrule
oracle & 22{,}528 & 1.00 \\
self-sizing+joint & 29{,}268 & 1.30 \\
self-sizing (independent round two) & 33{,}875 & 1.50 \\
blind doubling & 54{,}272 & 2.41 \\
Strata-first & 67{,}807 & 3.01 \\
\bottomrule
\end{tabular}
\end{table}

Under this configuration, Strata-first's expected bytes are $2.32\times$
those of self-sizing+joint (Table~\ref{tab:c3}).

\subsubsection*{Deployment-tier regret at the operating point}

The main regret figure of \S\ref{sec:regret} deliberately fixes the compact
$M_1=64$ tier as a stress profile: with the under-provisioned $\alpha=1.4$,
roughly one fifth of first-round failures spill into a third round.  We
repeat the same byte/round accounting at the recommended deployment
operating point, $M_1=512$ with the multiplier $\alpha=1.82$ recommended in
Table~\ref{tab:c1} (the $1.824$ used in the cross-city KV deployment of
Appendix~\ref{app:e4}).  Table~\ref{tab:c6} reports the results.

\begin{table}[!hbtp]
\centering
\caption{Regret at the deployment operating point ($M_1=512$, $\alpha=1.82$;
1000 trials, $\theta_-=0.5$).  Self-sizing and its joint variant are
numerically identical here, because at large~$d$ the first-round residual
core holds negligible recoverable mass.}
\label{tab:c6}
\small
\begin{tabular}{lrrrr}
\toprule
$d$ & oracle bytes & self-sizing bytes & regret ($\times$oracle) & rounds \\
\midrule
$10^2$ & 16{,}384 & 16{,}384 & 1.00 & 1 \\
$10^3$ & 42{,}242 & 75{,}665 & 1.79 & 2 \\
$10^4$ & 411{,}648 & 602{,}884 & 1.46 & 2 \\
$10^5$ & 4{,}069{,}376 & 5{,}850{,}155 & 1.44 & 2 \\
\bottomrule
\end{tabular}
\end{table}

Byte regret stays $1.44$--$1.79\times$ over the oracle across the whole
range, slightly above the $M_1=64$ stress curves of
Figure~\ref{fig:e21} because the 16\,KB first-round sketch is a larger fixed
cost.  What the deployment point buys is the tail behavior: every trial
finishes in exactly two rounds, with the third-round retry tail of the
stress profile gone, consistent with the $99.935\%$ failed-only
second-round success at $\alpha=1.8$ reported in Table~\ref{tab:c2}.  The
two-round guarantee of the protocol therefore holds at the exact
configuration used in production, not only at the deliberately
under-provisioned stress tier.

\subsection{Strata's recovery capability and layered capacity}

The Strata Estimator~\cite{eppstein2011whatsdifference} consists of $L$
layers of small IBFs of $C$ cells each.  Each difference element falls into
a unique layer according to the trailing zeros of an independent hash; layer
$i$ receives an element with probability $2^{-(i+1)}$.  The receiver
subtracts and peels layer by layer from the top down; a successful layer
outputs true difference elements with their signs.  At the first layer that
cannot be fully peeled, the algorithm returns
\[
  2^{i+1}\times\text{(cumulative recoveries from previously successful layers)}.
\]
If all layers succeed, the algorithm returns the cumulative recoveries,
having already enumerated all of $A\triangle B$.  Strata-first can thus
finish in one round at small $d$, and elements recovered from the
successfully decoded preceding layers can be deducted exactly from the
remaining load of the subsequent main IBLT.

Full enumeration requires the most crowded layer 0 to succeed as well; the
number of elements it receives is $\operatorname{Binomial}(d,1/2)$.
Therefore Strata's full-solution scale is set by a single layer's capacity,
whereas allocating the same total cell budget to the Plain IBLT of
\S\ref{sec:count-measurement} places all cells in a single decoding domain
for joint peeling:
\[
  \text{Strata full-solution range}=O(C),
  \qquad
  \text{Plain IBLT full-solution range}=O(LC).
\]
The standard configuration takes $L=16$, $C=80$, 1280 cells total.
Adding layers mainly extends the dynamic range of geometric sampling
coverage; it does not merge the layers into one 1280-cell decoding domain.
Plain instead commits the whole budget to one joint decoding; if decoding
fails, the same pre-peeling counts remain available for estimation.

The two structures' estimation precision reflects the same budget
organization.  Strata stops at the first failed layer, and its point
estimate is mainly determined by the $O(C)$ samples recovered by previously
successful layers; the estimator of \S\ref{sec:count-measurement} (``Plain''
in the table below) uses the $M-1$ effective directions of the count vector
on the orthogonal complement of the constant direction, with large-$d$
relative standard deviation
\[
  \operatorname{RSD}(\dhat)\approx\sqrt{\frac{2}{M-1}}.
\]
At $d=1000$, the 1024-cell Plain estimator uses 20\% fewer cells than the
standard 1280-cell Strata configuration while achieving a substantially
lower RSD (Table~\ref{tab:c4}):

\begin{table}[!hbtp]
\centering
\caption{Point-estimate precision at $d=1000$, $\theta_-=0$.}
\label{tab:c4}
\small
\begin{tabular}{lrrrr}
\toprule
estimator & cells & mean($\dhat/d$) & RSD & 5\%--95\% quantiles \\
\midrule
Plain, $M=1024$ & 1024 & 0.9985 & 4.39\% & $0.926d$--$1.073d$ \\
Strata, $16\times80$ & 1280 & 0.9785 & 11.07\% & $0.808d$--$1.168d$ \\
\bottomrule
\end{tabular}
\end{table}

Taking exactly 1280 cells, the Plain estimator's exact variance formula
gives large-$d$ $\operatorname{RSD}\approx\sqrt{2/1279}=3.95\%$.  The
end-to-end results under the same 40\,KB first-round budget:

\begin{table}[!hbtp]
\centering
\caption{End-to-end cost under the same 1280-cell first-round budget at
$d=1000$, i.e., the same-budget comparison of \S\ref{sec:regret}.}
\label{tab:c5}
\small
\begin{tabular}{lrrr}
\toprule
first-round method & mean total bytes & mean rounds & $P(\mathrm{rounds}\ge3)$ \\
\midrule
Plain self-sizing, $M_1=1280$ & 41{,}435\,B & 1.010 & 0 \\
Strata-first, $16\times80$ & 91{,}386\,B & 2.155 & 15.5\% \\
\bottomrule
\end{tabular}
\end{table}

Strata's standard finite configuration provides no exact mean, variance, or
distribution formula for its point estimate.
\S\ref{sec:count-measurement} gives the exact mean and variance, a
fixed-dimension chi-square limit, and the failure-conditioned lower-tail
interface.

\section{Cross-Engine Implementation,\texorpdfstring{\\}{ }
Deployment Shape, and Evidence Scope}
\label{app:e1}

\subsection{Verification scope}

The cross-engine verification defines the scope of the numbers below.
Correctness checks occur in two places---the Oracle$\leftrightarrow$MySQL
production-shape replay and the Redis$\leftrightarrow$Pika cross-city
deployment---both compared item by item against ground truth.  Real-machine
performance numbers are reported only from runs that share a unified input,
unified network conditions, complete phase timing, and a consistent re-run
procedure, covering the Oracle$\leftrightarrow$MySQL production-shape
replay and the Redis/Pika cross-city deployment.

\subsection{Deployment and timing topology}

Two topologies carry the two experimental facets of this paper.  The
production-shape replay runs on one server with four containers: the Oracle
database, the Oracle-side sidecar and controller, the MySQL database, and
the MySQL-side sidecar, plus a network relay container that connects the two
sidecars through a controlled link; host hardware, container resource
limits, and the relay's RTT/bandwidth tiers are in \S\ref{app:e3-env}.  The
China Mobile cross-city deployment uses two cloud hosts in the same
cloud private network, with the sidecar co-located with its local
Redis/Pika, on a link of about 10\,Mbps and 28.8\,ms RTT.  The scan floor,
network emulation, and timing scope of the two topologies are independent,
and the main text's numbers are grouped by experiment.

\subsection{Cross-engine canonicalization}

Each record is canonicalized to a deterministic text form
(type-specific formatting, column concatenation with a fixed separator,
then \texttt{LOWER}) and hashed to a 56-bit MD5 fingerprint.  Bucket
placement, cell checksums, and fresh-round rehashing use independent hash
layers that do not affect the cross-engine alignment.  The full
canonicalization specification is in the supplementary material.

\section{Production Workload and Key-Rank Space Supplement}
\label{app:e2}

\subsection{Table-level sample filtering, time range, and statistical
semantics}

The table-level log covers the 90 days from 2026-04-29 to 2026-07-28.  After
removing entries that cannot be matched to an execution table, that have
both sides empty, that have no data to compare, or that have no primary key
for row-level location, 41{,}603 table-level runs that performed
row-level comparison remain.  Table-level time is measured in whole seconds,
so p50 and below sit at the timing resolution floor and read only as ``at
most a few seconds.''

\begin{table}[!hbtp]
\centering
\caption{Complete quantile ladder of non-zero $d_{\mathrm{ms}}$ (multiset
elements; a modified row contributes two elements).}
\label{tab:e2-quantiles}
\small
\begin{tabular}{lrr}
\toprule
non-zero $d_{\mathrm{ms}}$ quantile & value & share of cumulative difference \\
\midrule
p10   & 1          & 0.001\% \\
p25   & 4          & 0.003\% \\
p50   & 24         & 0.018\% \\
p75   & 385        & 0.175\% \\
p90   & 3{,}730    & 1.16\% \\
p95   & 30{,}424   & 3.80\% \\
p99   & 194{,}698  & 24.9\% \\
p99.9 & $1.61\times10^6$ & 47.4\% \\
max   & $4.41\times10^7$ & 100\% \\
\bottomrule
\end{tabular}
\end{table}

\begin{table}[!hbtp]
\centering
\caption{Complete quantile ladder of table size (max-side rows) and the
cumulative share of max-side rows accounted for by runs up to each quantile,
the same max-side semantics as Table~\ref{tab:prod-strat}.}
\label{tab:e2-size-quantiles}
\small
\begin{tabular}{lrr}
\toprule
table-size quantile & max-side rows & cumulative max-side rows \\
\midrule
p10    & 4                & 0.0\% \\
p25    & 30               & 0.0\% \\
p50    & 533              & 0.0\% \\
p60    & 1{,}718          & 0.0\% \\
p70    & 15{,}790         & 0.01\% \\
p75    & 49{,}641         & 0.03\% \\
p80    & 192{,}740        & 0.11\% \\
p85    & 504{,}680        & 0.35\% \\
p90    & 3{,}785{,}522    & 1.25\% \\
p95    & 10{,}714{,}983   & 5.56\% \\
p99    & 153{,}392{,}779  & 37.05\% \\
p99.9  & 607{,}488{,}678  & 91.03\% \\
max    & 1{,}155{,}898{,}765 & 100\% \\
\bottomrule
\end{tabular}
\end{table}

Of the non-zero difference cardinalities, 75.1\% of records satisfy
$d_{\mathrm{ms}}\le393$, decodable directly by a first-round IBLT of 512
cells (16\,KB); including $d=0$, about 89.1\% of runs finish within
one first-round interaction.  First-round decodability applies only to small
differences, however: these records account for just 0.175\% of the total
difference mass, which concentrates in the few records above p99
(Table~\ref{tab:e2-quantiles}).

In periodic reconciliation, a difference either disappears that run or
persists.  Over the 90-day log, consecutive-difference episodes appear 146
times: half appear once and are consistent by the next run; one third persist
from start to end of the window without repair; and the intermediate
category lasting dozens of runs never occurs.  Persisting differences carry
96.4\% of the positive-difference time, so tables dominating the cost take
the difference-recovery path almost every run; short-lived differences cost
almost no time.

Why the differences arise is not directly observable to the verification
layer: NineData operates the synchronization and verification service but
does not own the user data, so the sources below are inferred from
engineering experience.  The evidence is consistent with three recurring
sources.  First, active--active replication
under dual-datacenter writes can leave conflicts that the application itself
never detects; the persistence pattern above---many tables carrying nonzero
differences across consecutive runs without a repair action---is consistent
with users tolerating a residual inconsistency below a threshold they
consider acceptable.  Second, intermittent faults in the CDC or Kafka
pipeline---a partial worker outage under high table counts delaying or
losing events for a subset of partitions---produce exactly the banded
dirty-block pattern of \S\ref{sec:prod-rank}.  Third, residual cross-engine
normalization mismatches: the combinatorial surface of heterogeneous engine
pairs and version variants cannot be exhaustively covered, so small semantic
gaps inevitably remain, and the verification layer itself is what exposes
them.

The two ends of the table-size distribution correspond to two different cost
mechanisms.  Small tables (the cheapest 80\% of runs, 0.7\% of table-level
time) are dominated by fixed costs---connection establishment, job dispatch,
scheduling, and result logging---essentially independent of table size;
these are mostly master and reference tables carrying base information
(products, suppliers, organizations).  Large tables (the last two deciles,
99.2\% of table-level time) are dominated by the scan itself; these are the
real business data tables such as transaction and inventory tables.

Of the 41{,}603 records, 56.4\% have $d=0$; among records with $N\ge10^6$
that fraction is 40.2\%.  The longest 5\% of records account for 81.2\% of
cumulative table-level time; a single account alone accounts for 90.6\% of
cumulative table-level time, 93.0\% of table scale by the larger side, and
96.0\% of records with $N\ge10^8$.  Removing that account leaves the $d=0$
fraction at 54.5\%, so account concentration is an important boundary of
the cost-weighted interpretation, not sampling noise.  Difference
composition by mass is 42.3\% modified, 18.5\% source-only, and 39.2\%
target-only; by occurrence rate (the fraction of non-zero-$d$ records
containing at least one difference of each type), 36.4\%, 31.5\%, and
67.6\%.

\subsection{Ground truth and block-level logs of the three replay profiles}

The three profiles are the candidate tables P1/P2/P3 chosen from the
production profile (see \S\ref{sec:profiles});
Table~\ref{tab:e2-profiles} summarizes their ground
truth and block-level statistics, with per-column types and lengths in the
experiment design document's DDL listing.  The block-level statistics come
from one complete sharded run on 2026-07-29/30: the production task fixes a
chunk target of about 100{,}000 rows, concurrency 5, and a single-batch
lookup cap of 1{,}000 differences; the composite-key subqueries are mapped
back to logical chunks of about 100{,}000 rows by first-level boundaries.
The three logs' chunk chains are complete, all dirty records are pairable,
and each chunk has 99{,}073--100{,}130 rows, so the logs cover complete
tables rather than paged samples.

\begin{table}[!hbtp]
\centering
\caption{Ground truth and block-level statistics of the three replay
profiles, same semantics as Table~\ref{tab:profiles}.  ``Dirty'' means the
block has at least one subquery that triggered a detail lookup by
primary-key order.  After dividing the primary-key rank space into
100 equal buckets, ``non-empty rank buckets'' is the number of buckets the
dirty blocks fall into, ``top-20\% rank share'' is the fraction of dirty
blocks in the top 20\% of key-rank space, and the longest contiguous dirty
run is in blocks.}
\label{tab:e2-profiles}
\footnotesize
\setlength{\tabcolsep}{3pt}
\begin{tabular}{@{}p{1.1cm}p{1.3cm}p{1.8cm}p{1.8cm}p{1.5cm}p{1.3cm}p{1.4cm}p{1.8cm}p{1.2cm}@{}}
\toprule
profile & table ID & rows (src/tgt) & primary key & difference (s/m/t) & $d_{\mathrm{row}}$/$d_{\mathrm{ms}}$ & dirty blocks (rate) & rank bands (buckets / top-20\%) & longest run \\
\midrule
P1 & T1 & 609{,}671{,}577 / 609{,}668{,}227 & single column & 3{,}350 / 44{,}890 / 0 & 48{,}240 / 93{,}130 & 670 (10.99\%) & 24 / 59.6\% & 30 \\
P2 & T2 & 31{,}034{,}507 / 31{,}034{,}459 & 3-column composite & 48 / 672 / 0 & 720 / 1{,}392 & 24 (7.67\%) & 14 / 58.3\% & 9 \\
P3 & T3 & 154{,}356{,}834 / 154{,}356{,}830 & 3-column composite & 11 / 0 / 7 & 18 / 18 & 3 (0.20\%) & 3 / 100\% & 1 \\
\bottomrule
\end{tabular}
\end{table}

\section{Production-Shape Replay}
\label{app:e3}

\subsection{Experiment environment and timing semantics}
\label{app:e3-env}

The replay runs on one server: 2$\times$ Intel Xeon Platinum 8468, 96
logical CPUs, 2 NUMA nodes, about 2.0\,TiB memory.  The two databases are
Oracle 11g Enterprise Edition and PolarDB/MySQL 8.0 (measured MySQL
8.0.46), each in its own container with resources matching production: Oracle
with 128\,GiB memory and 40 CPUs, MySQL with 64\,GiB and 36 logical CPUs.
Each sidecar gets 8 CPUs and 8\,GiB; the scan concurrency is controlled by
the parameter W (W16 means 16 workers per side) and does not depend on
container CPU quotas.

Network conditions are imposed by a traffic-control container that shares
the sidecars' network namespace and applies the specified round-trip latency
and bandwidth limits using tc; application traffic does not pass through an
additional user-space relay.  Four conditions are applied:
direct connection, 22.5\,ms RTT/100\,Mbps, 22.5\,ms RTT/10\,Mbps, and
100\,ms RTT/10\,Mbps; each case is timed per phase.

End-to-end time is split into phases with the following meanings:

\begin{table}[!hbtp]
\centering
\caption{Phase definitions for end-to-end timing.}
\label{tab:e3-phases}
\small
\begin{tabular}{lp{2.0cm}p{7.6cm}}
\toprule
belongs to & phase & content \\
\midrule
IBLT & build & each side scans its local data, computes row fingerprints, writes the sketch \\
IBLT & rebucket & round two re-buckets the already-read fingerprints with a new seed; no rescan \\
IBLT & resolve & subtract the two sketches, peel, recover difference candidate keys \\
Merkle-style & boundary & generate the chunk boundaries shared by both sides \\
Merkle-style & checksum & compute and exchange per-chunk checksums, locate dirty blocks \\
Merkle-style & drill & fetch dirty-block details over the link, locate differences row by row \\
both & recheck & read back original records by the located candidate keys for final confirmation \\
\bottomrule
\end{tabular}
\end{table}

Resolve spans two sub-steps with very different costs.  Subtracting the two
sketches cell-by-cell and peeling are pure in-memory operations on the
controller, taking well under a millisecond even for tens of thousands of
candidates, so they are not separately visible in the reported wall-clock
(Table~\ref{tab:e3-iblt-phases}).  The recovered elements are fingerprint
\emph{ids} rather than primary keys, however; each side's sidecar maps the
candidate ids back to primary keys by a single pass over its local spill
cache, skipping non-candidate payloads.  This fingerprint-to-key lookup is
what dominates the reported resolve wall-clock; it reads only the local cache
and never touches the database or the network, which is why resolve stays
bandwidth-insensitive across the network tiers.

Every run is preceded by cache warm-up and followed by a correctness check.

\subsection{End-to-end matrix and phase composition: P1/P2/P3}
\label{app:e3-matrix-sec}

Figure~\ref{fig:s63} in the main text uses a unified \texttt{chunkSize=10000},
the worker tiers W16/W32, and three network tiers
(22.5\,ms/100\,Mbps, 22.5\,ms/10\,Mbps, 100\,ms/10\,Mbps),
with two repetitions per cell.

\begin{table}[!hbtp]
\centering
\caption{End-to-end times (seconds) of the three production profiles; each
cell lists the two runs (run 1 - run 2).  The E2E means reported in the main
text's Table~\ref{tab:e2e} and the phase tables average these two values.}
\label{tab:e3-matrix}
\footnotesize
\begin{tabular}{llrrrr}
\toprule
profile & network (RTT\,ms / Mbps) & IBLT W16 & IBLT W32 & Merkle-style W16 & Merkle-style W32 \\
\midrule
P1 & 22.5/100 & 790.3 - 797.0 & 600.5 - 672.7 & 1094.2 - 1094.3 & 614.6 - 690.2 \\
P1 & 22.5/10  & 800.1 - 792.4 & 614.4 - 700.4 & 1433.9 - 1436.0 & 962.5 - 1078.1 \\
P1 & 100/10   & 805.6 - 794.5 & 613.2 - 695.6 & 1436.2 - 1439.6 & 962.0 - 1060.3 \\
P2 & 22.5/100 & 39.1 - 38.7   & 37.4 - 37.1   & 40.8 - 37.0    & 40.0 - 38.9 \\
P2 & 22.5/10  & 39.7 - 40.2   & 36.8 - 37.9   & 46.0 - 47.5    & 49.7 - 47.8 \\
P2 & 100/10   & 38.9 - 38.7   & 37.1 - 36.9   & 46.9 - 46.3    & 48.7 - 48.6 \\
P3 & 22.5/100 & 158.1 - 153.6 & 160.1 - 166.5 & 158.9 - 173.8  & 155.2 - 156.0 \\
P3 & 22.5/10  & 150.0 - 149.0 & 162.7 - 158.7 & 137.3 - 137.7  & 160.5 - 161.1 \\
P3 & 100/10   & 152.3 - 151.2 & 162.9 - 158.1 & 136.4 - 138.9  & 160.8 - 159.6 \\
\bottomrule
\end{tabular}
\end{table}

IBLT varies little across the three network tiers.  The increase in
Merkle-style runtime for P1 at 10\,Mbps comes primarily from the dirty-range
drill-down phase.  The drill phase batches detail requests over the
link, so the request count does not translate into linear RTT amplification.

Tables~\ref{tab:e3-iblt-phases} and~\ref{tab:e3-merkle-phases} give
phase-level times for representative configurations.  The complete
18-cell stage tables across all network $\times$ worker configurations,
with data provenance, are in the supplementary material.  The largest
single cost of these runs is the per-row MD5 fingerprint computation.
Whether Merkle-style or IBLT, every local row must first be normalized and
hashed---IBLT construction and Merkle-style checksum generation share the same
dominant cost: scanning, normalizing, and fingerprinting every row.  They
differ in how those fingerprints are organized after this common step.  When
the table is large (P1, about 600M rows), this item dominates the end-to-end
time and both methods sit at hundreds of seconds (P1's IBLT build and Merkle-style's checksum are both about 594--760\,s).  This is the common floor for the
full-table verification paths evaluated here: every row must be read and
fingerprinted before the two sides can be compared.

The methods therefore differ only in the remaining phases, which are small in
share but are exactly where the two methods diverge.  On the IBLT side, the
build phase varies with the concurrency tier: at W16 the Oracle side is
slowest (P1 about 760\,s); at W32 the Oracle side drops to about 417\,s
while the MySQL side rises to about 595\,s, moving the bottleneck from
Oracle to MySQL.  The cause is the two ends' different use of CPU
concurrency: Oracle's per-row MD5 computation speeds up markedly with more
workers, while the same phase on MySQL barely improves (it rises slightly on
P1), so increasing concurrency only swaps the slow end from Oracle to MySQL;
P2 and P3 follow the same pattern.  The remaining phases---bucketing,
re-bucketing, peeling, and lookup recheck---sum to only a small share; P3's
18 differences decode directly in round one, with no rebucket.  For
Merkle-style localization, the phase most sensitive to the network and
difference distribution is drill-down (moving dirty-block details over the
link).  Drill varies with both bandwidth and difference scale: P1 has the
most data and the most dirty ranges, so its drill is by far the longest,
rising from about 40\,s at 100\,Mbps to about 389\,s at 10\,Mbps; P2's is
about 14\,s at 10\,Mbps; P3 has only 18 differences, so its drill never exceeds
2.4\,s.

This also explains the most counter-intuitive row of
Table~\ref{tab:e3-matrix}: at 10\,Mbps on P3, Merkle-style is faster than IBLT
(137.5\,s versus 149.5\,s).  Both methods' time is dominated by whole-table
fingerprint computation; P3 has only 18 differences; drill is near zero; and
Merkle-style runtime remains close to the checksum-dominated lower bound.
P1 is the opposite: many
differences and a wide dirty region, so at 10\,Mbps Merkle-style trails IBLT
by about 360\,s (1020.3\,s versus 657.4\,s), with the entire gap coming from
drill's dirty-range lookup volume.

\begin{table}[!hbtp]
\centering
\caption{IBLT phase-level times (seconds; mean of two repetitions; O/M are
the Oracle/MySQL ends; resolve and recheck are parallel wall-clock times).  Only
representative configurations are listed; P3's 18 differences decode
directly in round one, so rebucket does not occur and is 0.}
\label{tab:e3-iblt-phases}
\footnotesize
\setlength{\tabcolsep}{3pt}
\begin{tabular}{llrcccccc}
\toprule
profile & network & W & Build O/M & Partition O/M & Rebucket O/M & Resolve wall & Recheck wall & E2E \\
\midrule
P1 & 22.5/100 & 16 & 760.0/510.4 & 0.46/1.58 & 15.2/16.7 & 14.2 & 10.0 & 793.7 \\
P1 & 22.5/100 & 32 & 417.1/595.5 & 0.45/1.00 & 21.9/23.0 & 15.6 & 10.1 & 636.6 \\
P2 & 22.5/100 & 16 & 36.4/28.3  & 0.20/0.64 & 0.51/0.59  & 1.16 & 0.70 & 38.9 \\
P3 & 22.5/100 & 16 & 149.5/113.7 & 0.19/0.94 & 0/0        & 6.27 & 0.08 & 155.8 \\
P3 & 22.5/100 & 32 & 82.4/155.3 & 0.19/0.71 & 0/0        & 6.52 & 0.14 & 163.3 \\
\bottomrule
\end{tabular}
\end{table}

\begin{table}[!hbtp]
\centering
\caption{Merkle-style phase-level times (seconds; mean of two repetitions;
boundary is the sum of chunk/probe/sort; refine is the refinement
drill-down phase).  Only representative configurations are listed; P3's
drill stays below 2.4\,s, hugging the fixed-cost floor set by the checksum.}
\label{tab:e3-merkle-phases}
\footnotesize
\setlength{\tabcolsep}{3pt}
\begin{tabular}{llrcccccc}
\toprule
profile & network & W & Boundary & Checksum & Refine & Drill & Recheck & E2E \\
\midrule
P1 & 22.5/100 & 32 & 17.2  & 593.5 & 4.3  & 39.6  & 2.0 & 652.4 \\
P1 & 22.5/10  & 32 & 14.1  & 614.2 & 4.3  & 389.2 & 2.7 & 1020.3 \\
P2 & 22.5/10  & 16 & 1.09  & 32.00 & 0.01 & 13.56 & 0.10 & 46.8 \\
P3 & 22.5/10  & 16 & 4.01  & 131.68 & 0.02 & 1.78 & 0.00 & 137.5 \\
\bottomrule
\end{tabular}
\end{table}

\begin{figure}[!hbtp]
\centering
\includegraphics[width=\textwidth]{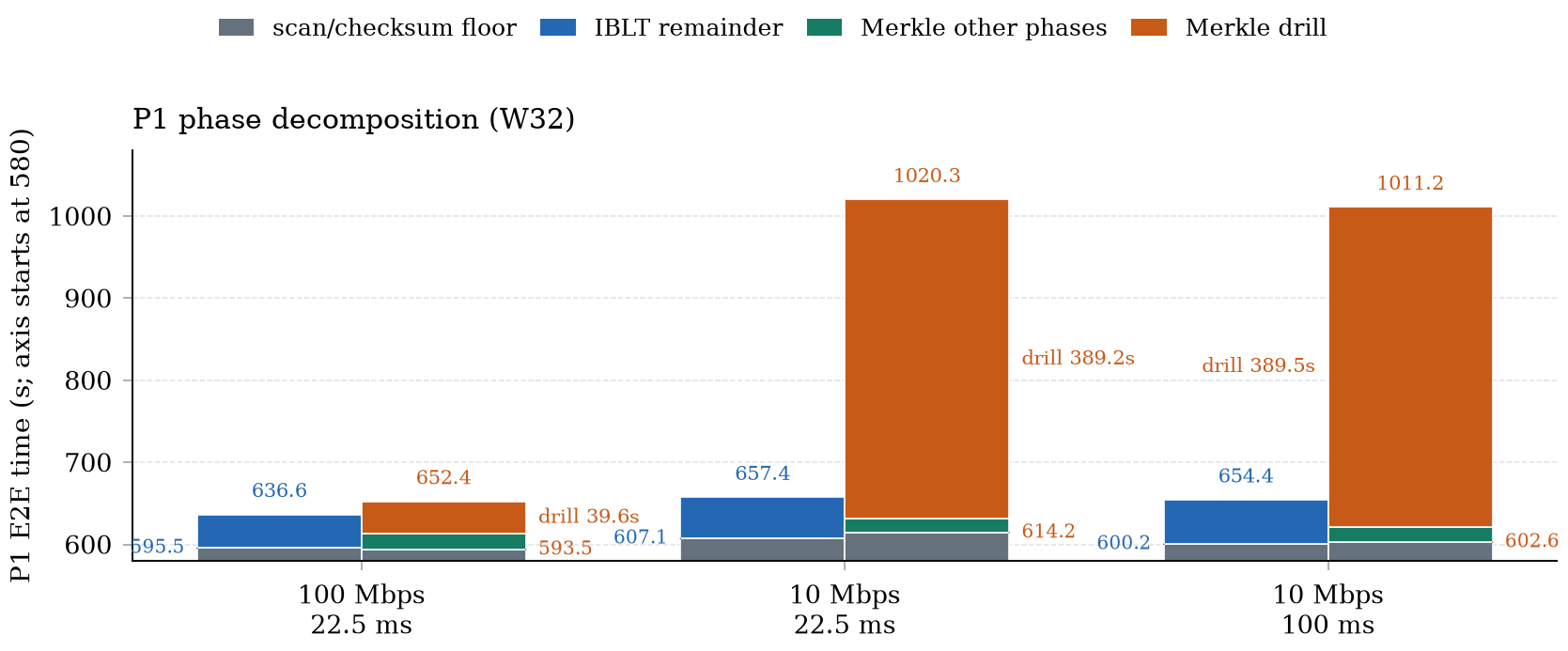}
\caption{P1's phase stack (W32): both methods first incur a nearly identical
full-table scan/checksum floor (with small numerical differences), and the
main divergence comes from the Merkle-style drill---rising from about 40\,s to
about 389\,s under limited bandwidth, with the remaining phases nearly
unchanged.}
\label{fig:e3-phases}
\end{figure}

\subsection{Worst case of the comparison path}

The controlled experiments stop the location axis at ``scattered''.  We next
evaluate a stress case in which differences are spread across nearly all
logical ranges, causing chunk-based localization to lose most of its
filtering benefit.  The precondition for a chunked summary to work is that a
few logical blocks carry all the differences; our counterexample is the
opposite: starting from P1, we add approximately one source-only row for
every 10{,}000 existing rows (60{,}968 additional source-only rows).  The chunked
summary almost loses its filtering power, and drill is equivalent to
re-shipping the whole table over the link.  The injected ground truth is
\texttt{minus=64,318}, \texttt{changed=44,890}.  The source-side difference
count increases from 48{,}240 to 109{,}208, while the total multiset
difference cardinality $d_{\mathrm{ms}}$ increases from 93{,}130 to
154{,}098.

\begin{table}[!hbtp]
\centering
\caption{P1-extra versus P1 end-to-end times (seconds; P1 is the mean of
two repetitions, P1-extra a single run, both passing the correctness gate).
The three 10\,Mbps Merkle-style tiers did not complete the drill phase and
are marked ``---'': the 22.5\,ms/10\,Mbps W32 and 100\,ms/10\,Mbps W32 runs
terminated with connection interruption, and the 100\,ms/10\,Mbps W16 run did
not complete within the observation window and was terminated.}
\label{tab:e3-extra}
\footnotesize
\begin{tabular}{llrrr}
\toprule
algorithm & network / W & P1 & P1-extra & increment \\
\midrule
IBLT  & 22.5/100 W16 & 793.7  & 810.2   & +2.1\% \\
IBLT  & 22.5/100 W32 & 636.6  & 692.6   & +8.8\% \\
IBLT  & 22.5/10 W16  & 796.3  & 817.1   & +2.6\% \\
IBLT  & 22.5/10 W32  & 657.4  & 691.5   & +5.2\% \\
IBLT  & 100/10 W16   & 800.1  & 854.5   & +6.8\% \\
IBLT  & 100/10 W32   & 654.4  & 700.7   & +7.1\% \\
Merkle-style & 22.5/100 W16 & 1094.2 & 4134.6  & +277.9\% \\
Merkle-style & 22.5/100 W32 & 652.4  & 3221.4  & +393.7\% \\
Merkle-style & 22.5/10 W16  & 1434.9 & 18502.5 & +1189.3\% \\
Merkle-style & 22.5/10 W32  & 1020.3 & ---     & --- \\
Merkle-style & 100/10 W16   & 1437.9 & ---     & --- \\
Merkle-style & 100/10 W32   & 1011.2 & ---     & --- \\
\bottomrule
\end{tabular}
\end{table}

The phase breakdown explains the observed difference between the two
methods.  Merkle-style's checksum grows only about 5\%--13\% and boundary
about 10\%--43\%, whereas drill-down time increases sharply---it ships every
judged-dirty range's details over the link (Table~\ref{tab:e3-extra}).  This
confirms the location conclusion of \S\ref{sec:prod-shape}: localization's
cost is set by ``how many blocks the differences fall into,'' not by the
difference cardinality.  On this input, three of the tested 10\,Mbps
configurations did not complete the drill-down phase at all.  At a fixed difference
cardinality, IBLT transfer cost is insensitive to how the differences are
distributed across key ranges, because its transfer volume is set by the
measured difference cardinality.  The stress case reinforces this point:
even though $d_{\mathrm{ms}}$ grows from 93{,}130 to 154{,}098, the
scan-dominated end-to-end time rises by only 2\%--9\%.

This comparison also has a protocol-level meaning.  Measurement happens
before decoding, so when the difference cardinality is extremely large, the
system can first read $\dhat$, decide by the criterion of
\S\ref{sec:boundaries} that the second-round capacity exceeds budget, and
return \texttt{FALLBACK} and move directly to full transfer or segmented
rebuild,
without transmitting a second-round sketch that already exceeds the
configured resource budget.  The evaluated localization path learns the
number of dirty ranges only after completing the summary and refinement
stages, by which time the cost is already paid---the comparison method's runtime of more than five
hours on P1-extra is exactly that.

\subsection{Ablation matrix: worker concurrency}

This matrix sweeps the worker concurrency tiers (W8/W16/W32/W64) to select
W for the formal horizontal matrix of \S\ref{app:e3-matrix-sec} and to check
the stability of
the numbers with three repetitions per cell.  It fixes P2 (31M rows, all 31
business value columns) and uses \texttt{chunkSize=100000} at the direct,
22.5\,ms, and 100\,ms RTT tiers (all at 100\,Mbps); each algorithm
$\times$ network $\times$ worker cell repeats three times, 72 valid runs
total.  The results support
three observations (Table~\ref{tab:e3-ablation}).  For IBLT, W16 gives the
lowest end-to-end time on P2 because the MySQL side becomes the bottleneck at
higher concurrency: the Oracle side's build drops markedly with more workers
while the MySQL side does not speed up and even rises slightly (the phase
mechanism is in \S\ref{app:e3-matrix-sec}).  W32 and W64 provide no further
improvement.  For Merkle-style localization, W32 and W64 perform similarly,
with W64 giving the lowest measured median.  The relative min--max spread of
the three repetitions stays below 15\% in every cell
(Figure~\ref{fig:e31}), and raising RTT from direct to 100\,ms does not
change the overall trend.

\begin{table}[!hbtp]
\centering
\caption{All end-to-end medians of the P2 72-run matrix (seconds;
parentheses give min--max over three repetitions).  Fixed
\texttt{chunkSize=100000}, three repetitions per cell.  W16 is the
best-performing IBLT setting for P2, while W64 gives the lowest Merkle-style
median and W32 is close.  The IBLT resolve phase takes approximately
3.3--3.5\,s, with the main time in both ends' build; Merkle-style's
concurrency gains come mainly from boundary/checksum.  All min--max ranges
fall within 15\%.}
\label{tab:e3-ablation}
\scriptsize
\setlength{\tabcolsep}{1pt}
\begin{tabular}{llrrrr}
\toprule
algorithm & network & W8 & W16 & W32 & W64 \\
\midrule
IBLT  & direct & 78.985 (78.484--80.041) & \textbf{47.738 (47.699--47.946)} & 49.122 (49.043--51.839) & 66.067 (65.885--75.070) \\
IBLT  & 22.5/100 & 78.756 (78.371--79.019) & \textbf{47.523 (47.137--48.649)} & 53.130 (52.007--56.663) & 65.423 (64.252--69.052) \\
IBLT  & 100/100 & 78.881 (78.587--79.753) & \textbf{48.036 (47.884--50.223)} & 52.798 (50.966--54.067) & 68.742 (62.902--68.776) \\
Merkle-style & direct & 100.541 (100.521--100.747) & 62.004 (57.913--63.073) & 51.317 (51.195--52.139) & \textbf{49.489 (49.145--49.957)} \\
Merkle-style & 22.5/100 & 102.854 (102.740--104.645) & 61.470 (61.423--61.563) & 55.506 (55.308--55.629) & \textbf{53.827 (53.729--53.989)} \\
Merkle-style & 100/100 & 102.383 (102.081--106.017) & 61.686 (61.376--65.645) & 53.767 (53.721--54.966) & \textbf{52.450 (52.230--52.497)} \\
\bottomrule
\end{tabular}
\end{table}

\begin{figure}[!hbtp]
\centering
\includegraphics[width=0.9\textwidth]{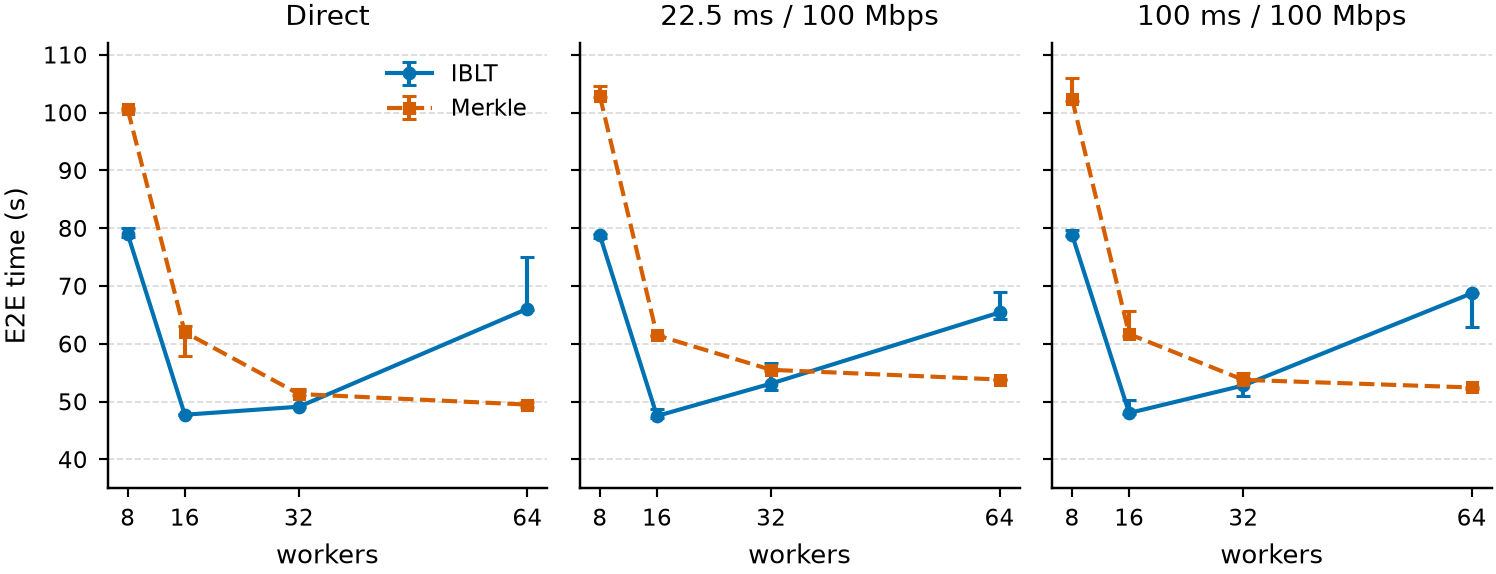}
\caption{End-to-end times of the P2 formal 72-run ablation matrix; three
repetitions per cell with min--max error bars.}
\label{fig:e31}
\end{figure}

\subsection{Capacity-setting actions and per-run estimates}

This subsection compares the second-round capacity computed from each
estimate with the true difference cardinality (Table~\ref{tab:e3-readings})
and verifies the resulting decoding outcome.  All three
profiles start from a fixed $M_1=512$ first round: P3's 18 differences
decode directly in round one; P1 and P2 fail round one and set the
second-round capacity once from the estimate ($M_2=\lceil1.52\times1.20\,\dhat\rceil$,
where $1.52$ is $\alpha=\beta/q_\delta$ and $1.20$ is the engineering margin).
Both then decode successfully in one shot, with \texttt{residual=0} and no
second expansion.

For correctness, the IBLT recovery records match the injected ground truth
item by item: P1 returns \texttt{plus\allowbreak=48240},
\texttt{minus\allowbreak=44890}, \texttt{residual\allowbreak=0}; P2
returns \texttt{plus\allowbreak=720}, \texttt{minus\allowbreak=672},
\texttt{residual\allowbreak=0}; P3's 11 source-only and 7 target-only
differences all pass truth verification.  The production logs do not
provide these tables' physical sizes; the locally rebuilt P1 was measured in
the deployment environment: its Oracle table data and primary key
index total 165.28\,GB, and its MySQL \texttt{.ibd} is about 192.58\,GB.

\begin{table}[!hbtp]
\centering
\caption{Per-run estimates and second-round capacities.  Measured RSD is the
sample standard deviation divided by the sample mean; theoretical RSD follows
Corollary~\ref{cor:rsd}.  Each row summarizes 12 runs.}
\label{tab:e3-readings}
\footnotesize
\setlength{\tabcolsep}{2pt}
\begin{tabular}{llrrrrr}
\toprule
profile & truth $d_{\mathrm{ms}}$ & $\dhat$ range & $\dhat$ sample mean & RSD (measured/theory) & $M_2$ range & $M_2/d_{\mathrm{ms}}$ \\
\midrule
P1 & 93{,}130 & 83{,}664.0--100{,}693.8 & 93{,}382.4 & 4.9\%/6.3\% & 152{,}604--183{,}666 & 1.64--1.97 \\
P2 & 1{,}392 & 1{,}247.8--1{,}464.4 & 1{,}353.8 & 4.6\%/6.3\% & 2{,}277--2{,}672 & 1.64--1.92 \\
P3 & 18 & 16.7--19.4 & 17.73 & ---/6.1\% & round one, $M_1=512$ & --- \\
P1-extra & 154{,}098 & 137{,}580.3--178{,}659.6 & 153{,}683.7 & 7.1\%/6.3\% & 250{,}947--325{,}876 & 1.63--2.11 \\
\bottomrule
\end{tabular}
\end{table}

With only 12 runs per profile, the sample standard deviation has a relative
standard error of about 20\%.  The measured RSDs of 4.9\% for P1 and 4.6\%
for P2 are therefore consistent with the theoretical value of 6.3\% at the
available sample size.

\subsection{Implementation notes and optimization of the comparison
baseline}
\label{app:e3-impl}

NineData's production implementation of Merkle-style localization incurs
substantial boundary-generation overhead, especially for composite primary
keys.  Composite-key boundaries are generated by \texttt{ORDER BY}: with a
fixed chunk size, the block count grows linearly with the row count (a
600M-row table needs 6{,}096 blocks), and generation is serial because block
$K$'s upper bound depends on the result for block $K-1$.  Together these
make the boundary stage dominate the whole run.  Composite keys, mixed
\texttt{and}/\texttt{or} predicates, optimizer-plan variability, and
different sorting semantics across Oracle, PolarDB, MySQL, ClickHouse, and
others make this ordered boundary generation difficult to optimize
consistently.  Tables without a primary or unique key have no sort boundary
at all, making localization directly unusable.  We therefore use an optimized
boundary-generation path in the replay so that the comparison reflects the
cost of localization after removing the production implementation's avoidable
serial overhead.

The production sharded logs directly corroborate this segment's cost (a
NineData production record from 2026-07-30).  The production table
corresponding to P2 (31M rows, composite key) spends nearly the first 30
minutes of its task on boundary (\texttt{ORDER BY}) before entering the
chunked checksum comparison, with a whole-table time of about 1 hour 56
minutes; another production table of about 95M rows with a composite key
spends about 36 minutes on boundary and about 2 hours 43 minutes in total.
By contrast, our optimized Merkle-style implementation on P1 (600M rows,
single-column key) has a boundary phase of only about 17\,s.

The local replay does not reproduce these native constraints, and the data
construction offers no ordering convenience: tables are inserted
concurrently out of order with primary and unique indexes rebuilt after
insertion, and primary keys are random string IDs rather than contiguous
numeric IDs.  The comparison baseline's boundary generation uses
\texttt{ROWID}/\texttt{NTILE} coarse sampling on the Oracle side, then
hands the same fences to MySQL for key-range queries; after optimization,
boundary takes about 5\,s---two orders of magnitude cheaper than the
production path on the same schema.

\section{Cross-City KV Field Deployment and Applicability}
\label{app:e4}

The cross-city deployment in \S\ref{sec:cross-city} examines three
questions: whether the measured estimate remains usable on a real link and
a hash-bucketed key space; whether self-sizing's end-to-end cost is still
dominated by the scan; and which properties of the KV key space determine
the crossover between the two methods.  This section gives the environment
semantics, per-scenario data, cost model, and applicability boundaries of
that verification.  It runs on Redis and Pika and complements the
relational evaluation of \S\ref{sec:prod-distribution}--\S\ref{sec:prod-shape}
with a contrasting key-space organization: relational databases chunk by
primary-key ranges with differences clustering along primary-key order,
while the KV environment hash-buckets and spreads differences uniformly
over the bucket space.  The deployment also evaluates a real cross-city
production link (whereas \S\ref{sec:prod-shape} emulates the link with
relays) and characterizes the observation window under online updates.

\subsection{Deployment environment and evidence semantics}

The verification runs on two cloud hosts in the same private network of
China Mobile, with each sidecar (the local agent of its side)
co-located with its local Redis/Pika instance.  The link is about 10\,Mbps
with 28.8\,ms RTT.  The baseline dataset contains approximately 4{,}890{,}077
keys on Redis 7.2.10.  About 99.999\% of the values use the Redis STRING
type, the mean record size is approximately 416\,B per key, and no keys have
expiration times.  In the fixed scenarios, both methods verify their
recovered differences item by item against the injected ground truth.  For
the G3 burst and G6 online-change runs, correctness is evaluated against the
differences observed within each method's scan window.

Scan concurrency is a local throughput tuning knob, not part of the method
comparison: Appendix~\ref{app:e3} already systematically ablates both methods'
concurrency sensitivity in a relational environment, so this verification
does not repeat concurrency tiers and both ends read data single-threaded.
Under a single thread both methods bear exactly the same scan cost, so the
end-to-end difference can be attributed to the methods themselves; IBLT's
benefit comes from cross-city payload and rounds, independent of scan
concurrency, and raising concurrency does not change the comparison's
direction.

Algorithm configuration is fixed: Merkle-style localization uses 4{,}096
hash buckets, and self-sizing IBLT uses first-round capacity $M_1=512$.

\begin{table}[!hbtp]
\centering
\caption{The six scenario groups G1--G6: construction and design purpose.
Injected differences remain constant during comparison.  G4 keeps the
B$\rightarrow$C lag by port-level link blocking and catches up with a PSYNC
increment afterwards; G5 uses different key prefixes to avoid circular
synchronization, and same-key value conflicts are not covered; G6 has no
fixed ground truth and only observes behavior under online change.
G4--G6 are single-shot scenarios while G2--G3 carry the quantitative
comparison; the correctness checks of all scenarios are in
Table~\ref{tab:e42}.}
\label{tab:e41}
\footnotesize
\setlength{\tabcolsep}{3pt}
\begin{tabular}{@{}p{0.7cm}p{5.0cm}p{3.8cm}p{5.2cm}@{}}
\toprule
ID & construction & differences & design purpose \\
\midrule
G1 & Redis$\rightarrow$Redis, $d=0$ and a single value difference
    & 0 / 1
    & fixed-cost baseline at $d=0$ and $d=1$ \\
G2 & Redis$\rightarrow$Redis, injected $d=100/500/937$
    & 100 / 500 / 937
    & cost comparison and estimate verification \\
G3 & Redis$\rightarrow$Pika, injected $d=100/500/937$ plus a burst of about $4\times10^4$
    & 100 / 500 / 937 / $\approx4\times10^4$
    & cross-engine behavior and the high-difference regime in which nearly all hash buckets become dirty \\
G4 & Redis A$\rightarrow$B$\rightarrow$C three-hop sync with B$\rightarrow$C lag held by port-level blocking; C missing $d=100/500/937$
    & 100 / 500 / 937
    & multi-hop lag coverage \\
G5 & 1{,}000 unique keys written on each side (different key prefixes to avoid circular sync)
    & 2{,}000 (bidirectional)
    & bidirectional difference shape \\
G6 & continuous LRU eviction at the target while redis-shake keeps syncing (data changing online)
    & observed
    & behavior under online change and the observation window; not an end-to-end performance point \\
\bottomrule
\end{tabular}
\end{table}

\subsection{Per-tier end-to-end data}

The two paths' scan costs provide the background for interpreting
Table~\ref{tab:e42}.  The Redis$\rightarrow$Redis scan floor (the shortest
time to fully read all keys) is about 28.1\,s (about 5.75\,$\mu$s/key);
the Redis$\rightarrow$Pika target is a disk \texttt{SCAN} whose
single-threaded scan floor is about 73--80\,s.  All six groups recover every
difference they should: the fixed scenarios match the injected
ground truth, and the G3 burst and G6 recover the differences observed
during the runs (Table~\ref{tab:e41}).  Table~\ref{tab:e42} summarizes the per-tier end-to-end
data of G1--G6: across the five tiers of G1--G2, the scan is 93.7\%--99.1\%
of IBLT's end-to-end time, and IBLT end-to-end times in the remaining
scenarios also remain close to the corresponding scan floor.

\begin{table}[!hbtp]
\centering
\caption{End-to-end cost of all G1--G6 scenarios.  ``---'' means the field
was not recorded for that scenario: G4--G6 do not record exact IBLT
cumulative transfer.  IBLT transfer bytes are the first- and second-round
cumulative (first round 16.4\,KB per side).  For runs without a fixed ground
truth (the G3 burst and G6), two values separated by ``/'' report the
differences observed by Merkle-style and IBLT localization, respectively;
they differ because the methods cover different scan windows.}
\label{tab:e42}
\footnotesize
\setlength{\tabcolsep}{2pt}
\begin{tabular}{@{}p{1.0cm}p{2.3cm}p{1.6cm}p{1.8cm}p{1.7cm}p{1.9cm}p{1.6cm}p{1.6cm}@{}}
\toprule
group & construction & differences & Merkle-style payload & Merkle-style end-to-end & IBLT capacity & IBLT end-to-end & IBLT cumulative transfer \\
\midrule
G1 & Redis$\rightarrow$Redis & 0 & 268.0\,KB & 28.3\,s & 512 (round one) & 29.5\,s & 32.8\,KB \\
G1 & Redis$\rightarrow$Redis & 1 & 398.9\,KB & 28.3\,s & 512 (round one) & 28.7\,s & 32.8\,KB \\
G2 & Redis$\rightarrow$Redis & 100  & 12.5\,MB  & 39.9\,s & 512 (round one) & 29.0\,s & 32.8\,KB \\
G2 & Redis$\rightarrow$Redis & 500  & 60.0\,MB  & 81.2\,s & 512$\rightarrow$975  & 30.0\,s & 95.2\,KB \\
G2 & Redis$\rightarrow$Redis & 937  & 105.3\,MB & 121.0\,s & 512$\rightarrow$1{,}701 & 29.6\,s & 141.7\,KB \\
G3 & Redis$\rightarrow$Pika & 100 & 12.8\,MB  & $\approx$82.6\,s & 512 & $\approx$73.3\,s & $\approx$32.8\,KB \\
G3 & Redis$\rightarrow$Pika & 500 & 59.0\,MB  & $\approx$123.2\,s & 512$\rightarrow$$\approx$955 & $\approx$73.5\,s & $\approx$94.0\,KB \\
G3 & Redis$\rightarrow$Pika & 937 & 106.0\,MB & $\approx$164.0\,s & 512$\rightarrow$$\approx$1{,}723 & $\approx$73.9\,s & $\approx$143.0\,KB \\
G3 & Redis$\rightarrow$Pika & 39{,}594\slash 39{,}502 & $\approx$518.7\,MB & $\approx$523.2\,s & 512$\rightarrow$72{,}357 & $\approx$80.1\,s & $\approx$4.66--5.02\,MB \\
G4 & cascade A$\rightarrow$B$\rightarrow$C & 100 & 12.8\,MB & 40.5\,s & 512 & 30.2\,s & --- \\
G4 & cascade A$\rightarrow$B$\rightarrow$C & 500 & 60.6\,MB & 81.9\,s & 512$\rightarrow$953 & 30.4\,s & --- \\
G4 & cascade A$\rightarrow$B$\rightarrow$C & 937 & 105.7\,MB & 122.5\,s & 512$\rightarrow$1{,}763 & 31.0\,s & --- \\
G5 & bidirectional 1{,}000 each & 2{,}000 & 197.6\,MB & 202.6\,s & 512$\rightarrow$3{,}677 & 30.6\,s & --- \\
G6 & redis-shake sync & 1{,}234\slash 1{,}163 & 134.9\,MB & $\approx$159.5\,s & 512$\rightarrow$2{,}066 & $\approx$34.5\,s & --- \\
\bottomrule
\end{tabular}
\end{table}

\subsection{Estimate deviation and second-round capacity setting}

This subsection verifies the \S\ref{sec:protocol} protocol's realization in
a real deployment: whether the estimate is accurate and whether the
resulting second-round capacity is sufficient for decoding.  Across the six runs,
the largest estimate deviation is $+8.1\%$, consistent with the theoretical
relative standard deviation of 6.3\% at $M_1=512$; the estimate and theory
agree, and the deviation is normal fluctuation from a limited sample.

\begin{table}[!hbtp]
\centering
\caption{Estimates and deviations of the six runs.  Truth is in the
multiset semantics $d_{\mathrm{ms}}$ of \S\ref{sec:implementation}: a
value change contributes its old and new fingerprints, so the $d=1$
scenario has $d_{\mathrm{ms}}=2$.}
\label{tab:e43}
\small
\begin{tabular}{llrr}
\toprule
run & truth $d_{\mathrm{ms}}$ & $\dhat$ & deviation \\
\midrule
$d=0$   & 0   & 0     & exact \\
$d=1$ (value change) & 2 & 2.0 & exact \\
$d=100$ & 100 & 108.1 & $+8.1\%$ \\
$d=500$ & 500 & 534.1 & $+6.8\%$ \\
$d=937$ & 937 & 932.1 & $-0.5\%$ \\
G6 & observed 1{,}163 & 1{,}132.4 & $-2.6\%$ \\
\bottomrule
\end{tabular}
\end{table}

The second-round capacity is set by $M_2=\lceil 1.824\,\dhat\rceil$ (1.824
is the product of the protocol's recommended multiplier 1.52 and the
engineering margin 1.2).  The three runs with recorded per-run estimates
match the formula exactly (534.1$\rightarrow$975, 932.1$\rightarrow$1{,}701,
1{,}132.4$\rightarrow$2{,}066) (Table~\ref{tab:e43}), and the other
second-round-triggering runs
also succeed in one shot; no run requires an additional capacity request
after the second round, and \texttt{FALLBACK} and \texttt{REJECT} never
trigger.

\subsection{Payload composition}

Each dirty bucket returns its full bucket details; the per-row payload is a
constant 106.0\,B (about 74\,B key plus 32\,B MD5 hex), a bucket holds on
average $N/B\approx1{,}194$ rows, so each dirty bucket is about 126.5\,KB;
adding the checksum vector of 268\,KB for $B=4096$, the total payload is
explained by
\[
  \mathrm{payload}\approx 268\ \mathrm{KB}+126.5\ \mathrm{KB}\times D,
\]
where $D$ denotes the number of dirty buckets,
with the four non-zero points deviating by at most 0.8\%, and the measured
effective throughput is about 1.211\,MB/s.  The model is calibrated to the
present deployment ($B=4096$, $\approx$10\,Mbps link); the per-bucket
payload depends on the mean bucket size $N/B$.  End-to-end time grows
approximately linearly with $d$ as well: from $d=1$ to $d=937$, Merkle-style
end-to-end time increases by approximately 93\,s, from 28.3\,s to 121.0\,s.

\FloatBarrier
\section*{Supplementary Material}
\addcontentsline{toc}{section}{Supplementary Material}
\label{app:supplementary}
The complete cross-engine canonicalization specification, the per-engine
hash availability matrix, the full phase-level stage tables of the
production-shape replay, and several analytic guarantees (the proof and
explicit quantiles of the finite-sample lower-tail bound, the finite-sample
refinement and proportional-growth direction for the failure-conditioned
shape, and the finite-sample
Kolmogorov bound) are collected below.

\setcounter{figure}{0}
\setcounter{table}{0}
\renewcommand{\thefigure}{S\arabic{figure}}
\renewcommand{\thetable}{S\arabic{table}}
\setcounter{theorem}{0}
\renewcommand{\thetheorem}{S\arabic{theorem}}
\renewcommand{\theHfigure}{supp.S\arabic{figure}}
\renewcommand{\theHtable}{supp.S\arabic{table}}
\renewcommand{\theHtheorem}{supp.S\arabic{theorem}}

The production-shape replay represents each record by the first 56 bits
(7 bytes, or 14 hexadecimal digits) of an MD5 digest.  The digest input is a
canonical string formed by concatenating the primary-key columns and the
business-value columns in their declared order, separated by
\texttt{CHAR(31)} (Oracle: \texttt{CHR(31)}).  The complete canonical string
is then converted to lowercase, so the replay treats case-only differences
as equal.  Per-type rules:

\begin{table}[!hbtp]
\centering
\caption{Cross-engine canonicalization and hashing rules.}
\label{tab:e1-rules}
\scriptsize
\begin{tabular}{@{}p{2.2cm}p{6.0cm}p{6.2cm}@{}}
\toprule
type/layer & current canonical or hash rule & role and boundary \\
\midrule
\texttt{NULL} & MySQL \texttt{COALESCE(expr,'<NULL>')}; Oracle \texttt{NVL(expr,'<NULL>')}
  & NULL uses a literal sentinel; the test data has no NULL, so the
  NULL/empty-string distinction is not yet covered by boundary samples. \\
\texttt{INT64} & both sides to decimal text; Oracle fixed \texttt{TO\_CHAR}
  number format, MySQL \texttt{CAST(... AS CHAR)}
  & normalizes integer text; negative boundaries not specifically tested. \\
\texttt{DECIMAL4} & \texttt{ROUND(q,4)*10000} first, then to integer text
  without decimals
  & P2's scale-4 money/decimal columns; negatives and rounding boundaries
  not specifically tested. \\
\texttt{DATETIME\_SEC} & \texttt{YYYY-MM-DD} \texttt{HH24:MI:SS} (MySQL
  \texttt{DATE\_FORMAT}, Oracle \texttt{TO\_CHAR})
  & drops milliseconds; timezone/boundary instants not independently
  covered. \\
\texttt{STRING} & MySQL \texttt{CAST(... AS CHAR)}; Oracle string expression
  & byte encoding and trailing spaces must match; Unicode/multibyte and
  CHAR padding not covered in this boundary suite. \\
column concat & \texttt{canonical(pk\_1)||CHR(31)||} \texttt{canonical(pk\_2)||...},
  then \texttt{LOWER} overall
  & field separator; escaping of the separator appearing in business
  strings not separately verified. \\
row fingerprint & MD5 of the full canonical, first 7 raw bytes (first 14
  hex), big-endian as a 56-bit \texttt{fp}
  & per-row alignment across databases, enters \texttt{fpXor}; bucket and
  cell checksums use separate independent hashes. \\
bucket split & \texttt{splitmix64(seed + fp)}, with $k=3$ positions sampled without
  replacement (exactly $k$ distinct positions)
  & IBLT bucket placement; does not change the row fingerprint. \\
cell checksum & \texttt{splitmix64(CHECKSUM\_SEED + fp)}
  & decides whether a $\mathrm{count}=\pm1$ cell is a pure cell; not used
  for cross-database canonical alignment. \\
fresh seed & round two derives a new bucket seed from \texttt{hashSeed}
  & fresh rehash changes only the bucket mapping, not \texttt{fp}. \\
\texttt{idXor} & primary-key canonical encoding compressed to 56 bits via
  Java \texttt{cheap8}/splitmix64
  & binds the payload to the same primary key on both sides; does not
  replace \texttt{fp}. \\
\bottomrule
\end{tabular}
\end{table}

The MD5 row fingerprint, bucket splitmix64, cell checksum, fresh seed, and
\texttt{idXor} thus serve five distinct roles in the implementation.  Primary and composite keys
are restricted to uppercase letters and digits at data-generation time so
that Oracle's binary order and MySQL's \texttt{general\_ci} agree in the
sort domain; the query path avoids \texttt{CONVERT}/\texttt{binary}
conversions that would break index use.  The Oracle-side dialect
expressions are version-dependent: \texttt{STANDARD\_HASH} is available on
Oracle 23ai, while 11g must use \texttt{DBMS\_CRYPTO.HASH(..., 2)} for the
same MD5 digest; this is dialect adaptation and does not change the
fingerprint semantics.  Each completed run compares the SQL ground truth with both the recovered
IBLT differences and the Merkle-style \texttt{changed}/\texttt{confirmedChanged}
results.  An IBLT run is considered fully decoded when \texttt{residual=0},
and the recovered positive and negative differences are verified item by
item.

The choice of hash function for row fingerprints is driven by two
cross-engine considerations.  CRC32 is faster when both ends' compatibility is
verified; MD5 is easier to unify across engines in input and output
semantics.  In a single-machine MySQL benchmark, $10^7$ hash evaluations take about
0.65\,s with CRC32 and 1.84\,s with MD5.  Computing and aggregating CRC32
values over a one-million-row relation takes about 1.2\,s.  These are hash-selection tests; the
Oracle--MySQL end-to-end timing is reported separately by the replay matrix.
The replay ultimately uses the first 56 bits of MD5, an engineering choice
driven by byte-level alignment and verifiability: identically named
functions do not by themselves guarantee identical input encoding, output
representation, byte order, or SQL conversion semantics across engines, so
candidates such as CRC32 and \texttt{ora\_hash} need pairwise equivalence
verification.  Candidate
selection is also constrained by sort semantics: Oracle's binary order and
MySQL's collation, \texttt{CHAR} trailing padding, charset encoding, and
case rules affect both boundary fences and row fingerprints, and the two
sides' hash functions align only when the normalized texts agree.  The
implementation therefore uses a fixed-width key domain of uppercase letters
and digits, fixes the canonical text before generating the Oracle-side
logical fences, and verifies sort and hash consistency separately.  The
performance scope of this choice is limited to the current schema and
verified inputs.

\subsection*{Hash function availability across engines}
\addcontentsline{toc}{subsection}{Hash function availability across engines}

Which hashing primitives are built in, and whether they are exposed under
the same name, both vary widely.  Table~\ref{tab:hash-avail} summarizes the
built-in availability of \texttt{MD5}, \texttt{CRC32}, \texttt{SHA1},
\texttt{SHA2(224)}, and \texttt{SHA2(256)}, together with each engine's
native (non-standard) hash function.  A checkmark means the function is
available in the current mainstream version.

\begin{table}[!hbtp]
\centering
\caption{Built-in hash function availability across database engines (the
four engines listed here; other engines are verified pairwise as needed).}
\label{tab:hash-avail}
\scriptsize
\setlength{\tabcolsep}{3pt}
\begin{tabular}{@{}lp{1.1cm}p{1.1cm}p{1.1cm}p{1.4cm}p{1.4cm}p{2.6cm}@{}}
\toprule
engine & \texttt{md5} & \texttt{crc32} & \texttt{sha1} & \texttt{sha2(224)} & \texttt{sha2(256)} & native hash \\
\midrule
MySQL & \checkmark & \checkmark & \checkmark & \checkmark & \checkmark & --- \\
PostgreSQL & \checkmark & --- & \checkmark & \checkmark\,($\ge$11) & \checkmark\,($\ge$11) & \texttt{hashtext} \\
Oracle & \checkmark & --- & \checkmark & \checkmark & \checkmark & \texttt{ora\_hash} \\
SQL Server & \checkmark & --- & \checkmark & --- & \checkmark & \texttt{checksum} \\
\bottomrule
\end{tabular}
\end{table}

\noindent This table gives two takeaways.  First, \texttt{MD5} is built into
all four engines while \texttt{CRC32} is exposed only by MySQL, so the
56-bit MD5 extraction is the common fallback for the engine configurations
verified in this work.  Second, the native hash functions differ across
engines (PostgreSQL's \texttt{hashtext}, Oracle's \texttt{ora\_hash}, SQL
Server's \texttt{checksum}), and ``same-name'' functions do not guarantee the
same algorithm, so any engine pair requires pairwise byte-level equivalence
verification.

\subsubsection*{Per-engine extraction of the 56-bit fingerprint}
\addcontentsline{toc}{subsubsection}{Per-engine extraction of the 56-bit fingerprint}

The row fingerprint takes the first 56 bits of the MD5 digest (the first 7
raw bytes), fitting a signed 64-bit integer for batched summation.  The
pairwise per-engine implementation of this extraction is carried out by
NineData in its reconciliation service; identical normalized inputs must
yield identical integers on both sides of a reconciliation.  Table~\ref{tab:hash-sql}
lists representative extraction expressions for four engines, showing how
the extraction is implemented in each engine.

\begin{table}[!hbtp]
\centering
\caption{Example 56-bit MD5 extraction expressions: NineData's per-engine
implementation in the reconciliation service; four engines are listed here.}
\label{tab:hash-sql}
\scriptsize
\setlength{\tabcolsep}{3pt}
\begin{tabular}{@{}>{\RaggedRight}p{2.3cm}>{\RaggedRight}p{9.6cm}@{}}
\toprule
engine & extraction expression \\
\midrule
MySQL & \texttt{\seqsplit{CONV(SUBSTR(MD5(CONCAT('abc')),1,14),16,10)}} \\
PostgreSQL / openGauss & \texttt{\seqsplit{CONCAT('x',SUBSTR(MD5('abc'),1,14))::BIT(56)::BIGINT}} \\
Oracle & \texttt{\seqsplit{TO\_NUMBER(SUBSTR(DBMS\_CRYPTO.HASH(RAWTOHEX('abc'),2),1,14),'XXXXXXXXXXXXXXXXX')}} \\
SQL Server & \texttt{\seqsplit{CONVERT(BIGINT,SUBSTRING(HASHBYTES('md5','abc'),1,7))}} \\
\bottomrule
\end{tabular}
\end{table}

The remaining engines are implemented pairwise by NineData following the
same idioms and are not expanded here.  Two constraints hold across engines:
some engines' raw hash output is endian-dependent and must be verified for
byte-level equivalence per pair; and any difference in the character
encoding of the canonical text changes the MD5 input, so charset must be
fixed before extraction.

\subsubsection*{Fingerprint width and collision risk}
\addcontentsline{toc}{subsubsection}{Fingerprint width and collision risk}

Truncation to 56 bits is a deliberate engineering choice with a quantifiable
correctness cost.  Modeling the truncated digest
as uniform on $2^{56}$ values, a table of $n$ rows carries about
$n(n-1)/2^{57}$ colliding pairs in expectation; on the largest replay table
(P1, about $6.1\times10^8$ rows) that is $2.6$ pairs, so a collision
somewhere in the table is likely rather than rare.  The cell checksum is
derived from the same \texttt{fp} and therefore does not separate two rows
that already collide.

What matters for reconciliation correctness is narrower.  Two distinct rows
sharing a fingerprint are harmless whenever both sides hold both of them:
each cancels in the sketch difference exactly as it should.  The estimate and
the decoding are affected only when a colliding pair straddles the
difference---one row present only on side $A$, the other only on side
$B$---in which case the two cancel and the pair is silently missed.  The
expected number of such pairs is about $d_Ad_B/2^{56}$, with $d_A,d_B$ the
per-side difference counts; across the three profiles this ranges from about
$7\times10^{-12}$ per run (P2) to about $3\times10^{-8}$ per run (P1), and the
downstream SQL recheck, which only re-examines reported candidates, does not
change it.  So the exposure is a missed difference, not a wrong one, at a
rate we can name.  A deployment that wants a stronger bound widens
\texttt{fp} to 64 or 128 bits and pays proportionally in transferred cell
width; nothing in the protocol, the estimator, or the canonicalization
depends on the width.

\subsubsection*{Semantic normalization selected per deployment}
\addcontentsline{toc}{subsubsection}{Semantic normalization selected per deployment}

The rules above define an \emph{exact} canonicalization: two rows agree if
and only if their canonical texts agree byte for byte, which is what the
replay and the cross-engine alignment check use, including the
\texttt{'<NULL>'} sentinel that keeps NULL distinct from the empty string.

Deployments frequently want something weaker, comparing rows up to a
business notion of equality.  The options we have seen requested are: time
columns compared without milliseconds; floating-point columns to two decimal
places; strings compared after trimming whitespace; NULL treated as the empty
string; and large fields compared on a prefix, or excluded.  Each of these is
a deployment-selected policy that deliberately declares certain unequal rows
equal, and each therefore introduces missed differences by construction; the
NULL-as-empty-string option in particular conflicts directly with the
sentinel rule of Table~\ref{tab:e1-rules}---that sentinel exists precisely to
keep NULL distinct from the empty string, while this option merges them---so
it is not the default behavior.  Which policy is
appropriate is an engineering decision about the data being reconciled, and
the measurement and protocol results of this paper are stated with respect to
whichever exact canonicalization the deployment fixes; they do not depend on
this choice.

\subsection*{Pairwise hash recommendation}
\addcontentsline{toc}{subsection}{Pairwise hash recommendation}

A faster native hash may be used only when both its cross-engine semantics
and its collision risk at the deployment's row and difference scale are
acceptable.  The evaluated replay uses the 56-bit MD5 fingerprint
(Table~\ref{tab:hash-sql}).  These recommendations serve only as a candidate
decision aid: production deployment still requires byte-level verification
on fixed samples per engine pair.

Tables~\ref{tab:e3-iblt-phases-full} and~\ref{tab:e3-merkle-phases-full}
give the complete phase-level breakdowns for all 18 experimental
configurations of the production-shape replay of Appendix~\ref{app:e3}.  The
paper lists only representative configurations; all 18 configurations are
reported here.

\subsection*{Data provenance and conventions}
\addcontentsline{toc}{subsection}{Data provenance and conventions}

\begin{itemize}
  \item P1 and P2 phase times are averaged over the two repetitions reported
    in the end-to-end matrix.
  \item P3 uses the same two-repetition averages, including the repeated
    22.5\,ms/10\,Mbps Merkle-style W16 configuration.
\end{itemize}

Times are in seconds.  Network columns report RTT in ms and bandwidth in
Mbps (22.5/100 means 22.5\,ms RTT, 100\,Mbps).  IBLT's O/M are the
Oracle/MySQL ends; for resolve and recheck, the reported value is the
elapsed wall-clock time while the two endpoints execute concurrently,
rather than the sum of their endpoint times.  Partition is the local step
that divides the retained fingerprint stream among workers before sketch
construction or rebucketing.  P3's 18 differences decode directly in round
one, so the rebucket phase does not occur and is 0.  Merkle's boundary is
the sum of chunk/probe/sort; Refine identifies the finer-grained dirty
ranges from the checksum hierarchy, and Drill fetches and compares the
row-level details of those ranges over the network.  E2E is measured
independently: the phase columns are instrumentation measurements and may
overlap because endpoint work and network operations execute concurrently,
so they should not be summed to reconstruct E2E.

\subsection*{IBLT phase-level times (all 18 configurations)}
\addcontentsline{toc}{subsection}{IBLT phase-level times (all 18 configurations)}

\begin{table}[!hbtp]
\centering
\caption{IBLT phase-level times for all 18 profile--network--worker
configurations (seconds; means over two repetitions).  O/M denote the Oracle
and MySQL endpoints.}
\label{tab:e3-iblt-phases-full}
\scriptsize
\setlength{\tabcolsep}{3pt}
\begin{tabular}{llrcccccc}
\toprule
profile & network & W & Build O/M & Partition O/M & Rebucket O/M & Resolve wall & Recheck wall & E2E \\
\midrule
P1 & 22.5/100 & 16 & 760.0/510.4 & 0.46/1.58 & 15.2/16.7 & 14.2 & 10.0 & 793.7 \\
P1 & 22.5/100 & 32 & 417.1/595.5 & 0.45/1.00 & 21.9/23.0 & 15.6 & 10.1 & 636.6 \\
P1 & 22.5/10  & 16 & 752.0/499.8 & 0.46/0.93 & 15.8/24.3 & 17.9 & 13.9 & 796.3 \\
P1 & 22.5/10  & 32 & 411.7/607.1 & 0.45/0.97 & 19.1/31.5 & 17.3 & 11.9 & 657.4 \\
P1 & 100/10   & 16 & 752.4/508.0 & 0.43/0.91 & 16.1/30.1 & 15.8 & 11.9 & 800.1 \\
P1 & 100/10   & 32 & 410.8/600.2 & 0.43/0.97 & 21.7/33.8 & 18.7 & 13.3 & 654.4 \\
P2 & 22.5/100 & 16 & 36.4/28.3 & 0.20/0.64 & 0.51/0.59 & 1.16 & 0.70 & 38.9 \\
P2 & 22.5/100 & 32 & 23.0/34.2 & 0.15/0.64 & 0.51/0.64 & 2.23 & 0.18 & 37.3 \\
P2 & 22.5/10  & 16 & 37.5/27.7 & 0.15/0.62 & 0.62/0.52 & 1.20 & 0.56 & 40.0 \\
P2 & 22.5/10  & 32 & 22.2/34.2 & 0.15/0.68 & 0.51/0.61 & 2.22 & 0.26 & 37.3 \\
P2 & 100/10   & 16 & 36.4/28.3 & 0.20/0.59 & 0.51/0.59 & 1.16 & 0.70 & 38.8 \\
P2 & 100/10   & 32 & 23.0/34.2 & 0.15/0.66 & 0.51/0.64 & 2.23 & 0.18 & 37.0 \\
P3 & 22.5/100 & 16 & 149.5/113.7 & 0.19/0.94 & 0/0 & 6.27 & 0.08 & 155.8 \\
P3 & 22.5/100 & 32 & 82.4/155.3 & 0.19/0.71 & 0/0 & 6.52 & 0.14 & 163.3 \\
P3 & 22.5/10  & 16 & 146.0/119.9 & 0.38/0.68 & 0/0 & 3.39 & 0.09 & 149.5 \\
P3 & 22.5/10  & 32 & 83.6/154.4 & 0.19/0.69 & 0/0 & 4.93 & 0.13 & 160.7 \\
P3 & 100/10   & 16 & 148.1/111.3 & 0.18/0.63 & 0/0 & 3.39 & 0.18 & 151.7 \\
P3 & 100/10   & 32 & 82.9/155.0 & 0.18/0.63 & 0/0 & 4.58 & 0.33 & 160.5 \\
\bottomrule
\end{tabular}
\end{table}

\subsection*{Merkle phase-level times (all 18 configurations)}
\addcontentsline{toc}{subsection}{Merkle phase-level times (all 18 configurations)}

\begin{table}[!hbtp]
\centering
\caption{Merkle-style phase-level times for all 18 profile--network--worker
configurations (seconds; means over two repetitions).}
\label{tab:e3-merkle-phases-full}
\scriptsize
\setlength{\tabcolsep}{3pt}
\begin{tabular}{llrcccccc}
\toprule
profile & network & W & Boundary & Checksum & Refine & Drill & Recheck & E2E \\
\midrule
P1 & 22.5/100 & 16 & 18.1  & 1031.2 & 4.3  & 43.0  & 1.9 & 1094.2 \\
P1 & 22.5/100 & 32 & 17.2  & 593.5  & 4.3  & 39.6  & 2.0 & 652.4 \\
P1 & 22.5/10  & 16 & 17.2  & 1026.1 & 4.9  & 389.0 & 2.6 & 1434.9 \\
P1 & 22.5/10  & 32 & 14.1  & 614.2  & 4.3  & 389.2 & 2.7 & 1020.3 \\
P1 & 100/10   & 16 & 16.4  & 1029.6 & 6.2  & 389.2 & 2.7 & 1437.9 \\
P1 & 100/10   & 32 & 16.4  & 602.6  & 4.4  & 389.5 & 2.7 & 1011.2 \\
P2 & 22.5/100 & 16 & 2.00  & 34.61  & 0.01 & 2.20  & 0.08 & 38.9 \\
P2 & 22.5/100 & 32 & 1.08  & 32.87  & 0.01 & 5.38  & 0.11 & 39.4 \\
P2 & 22.5/10  & 16 & 1.09  & 32.00  & 0.01 & 13.56 & 0.10 & 46.8 \\
P2 & 22.5/10  & 32 & 1.73  & 32.97  & 0.01 & 14.00 & 0.05 & 48.8 \\
P2 & 100/10   & 16 & 1.63  & 31.41  & 0.01 & 13.46 & 0.13 & 46.6 \\
P2 & 100/10   & 32 & 1.33  & 33.24  & 0.02 & 13.95 & 0.14 & 48.7 \\
P3 & 22.5/100 & 16 & 4.74  & 161.11 & 0.02 & 0.51  & 0.00 & 166.4 \\
P3 & 22.5/100 & 32 & 3.68  & 151.38 & 0.02 & 0.53  & 0.00 & 155.6 \\
P3 & 22.5/10  & 16 & 4.01  & 131.68 & 0.02 & 1.78  & 0.00 & 137.5 \\
P3 & 22.5/10  & 32 & 4.73  & 154.30 & 0.02 & 1.80  & 0.00 & 160.8 \\
P3 & 100/10   & 16 & 3.73  & 131.51 & 0.02 & 2.38  & 0.00 & 137.6 \\
P3 & 100/10   & 32 & 4.56  & 153.82 & 0.02 & 1.85  & 0.00 & 160.2 \\
\bottomrule
\end{tabular}
\end{table}

\subsection*{Proof of Proposition~\ref{prop:lower-tail} and explicit
failure-conditioned quantiles}

For any non-negative random variable $X$,
\[
  \E[X \given F]
  =\frac{\E[X\mathbf1_F]}{p_F}
  \le\frac{\E[X]}{p_F}.
\]
Take $X=(\dhat/d-1)^2$.  Theorems~\ref{thm:unbiased}--\ref{thm:variance}
give
\[
  \E\!\left[\left(\frac{\dhat}{d}-1\right)^2\right]
  =\Var\!\left(\frac{\dhat}{d}\right)
  =\frac{2(d-1)}{d(M-1)}.
\]
Hence
\[
  \E\!\left[\left.\left(\frac{\dhat}{d}-1\right)^2\right|F\right]
  \le
  \frac{2(d-1)}{d(M-1)p_F}.
\]
Applying Markov's inequality to the conditional measure: for any
$q\in(0,1)$, the event $\{\dhat/d\le q\}$ implies
$(\dhat/d-1)^2\ge(1-q)^2$, giving the fixed-sign lower-tail bound
\[
  \Pr\!\left[\dhat\le q\,d \given F\right]
  \le
  \frac{2(d-1)}{d(M-1)p_F(1-q)^2}.
\]

Under random signs, Theorem~\ref{thm:cond-mean} further gives
$\E[\dhat/d \given F]=1$, so
\[
  \Var\!\left(\left.\frac{\dhat}{d}\right|F\right)
  \le\frac{2(d-1)}{d(M-1)p_F}.
\]
Applying Cantelli's inequality to the conditional measure,
\[
  \Pr\!\left[\dhat\le q\,d \given F\right]
  \le\frac{v_F}{v_F+(1-q)^2},
  \qquad
  v_F=\Var\!\left(\left.\frac{\dhat}{d}\right|F\right),
\]
yielding a slightly tighter Cantelli lower quantile.

\subsubsection*{Explicit failure-conditioned quantiles (safety floor for
\S\ref{sec:quantile-interface})}

Inverting the tail bounds above into $q_\delta$ gives the analytic lower
quantiles that \S\ref{sec:quantile-interface} uses as a safety floor.  To
cap the underestimation probability at $\delta$, set the right-hand side of
the fixed-sign Markov bound to $\delta$, solve for $1-q$, and simplify using
$(d-1)/d\le1$:
\[
  \Pr\!\left[\dhat\le q\,d\;\Big|\;F\right]\le\delta
  \;\Longleftarrow\;
  1-q=\sqrt{\frac{2(d-1)}{d(M-1)p_F\delta}}\le\sqrt{\frac{2}{(M-1)p_F\delta}},
\]
so under any fixed signs,
\[
  q_\delta^{\mathrm{Markov}}
  =1-\sqrt{\frac{2}{(M-1)p_F\delta}}.
\]
Under random signs, use the Cantelli bound
$\Pr[\dhat\le q\,d \given F]\le v_F/(v_F+(1-q)^2)$, set it equal to
$\delta$, giving $(1-q)^2=v_F(1-\delta)/\delta$, and substitute
$v_F\le 2(d-1)/[d(M-1)p_F]\le 2/[(M-1)p_F]$:
\[
  q_\delta^{\mathrm{Cantelli}}
  =1-\sqrt{\frac{2(1-\delta)}{(M-1)p_F\delta}}.
\]
When only a lower bound on the failure probability is known, replace $p_F$
by $p_{F,\min}$.

Both inversions are admissible only when the resulting value lies in
$(0,1)$, which is also the range over which
Proposition~\ref{prop:lower-tail} is stated: the Markov form is positive
exactly when $(M-1)p_F\delta>2$, the Cantelli form when
$(M-1)p_F\delta>2(1-\delta)$.  When the expression is positive it certifies a
$q_\delta$ that Definition~\ref{def:qdelta} accepts, since
$\Pr[\dhat\ge q_\delta d \given F]\ge\Pr[\dhat>q_\delta d \given F]\ge1-\delta$,
so the boundary point costs nothing.  When it is non-positive, the moment
bound certifies no positive failure-conditioned lower quantile at all, and the
non-positive value must not be handed to the capacity formula
\eqref{eq:capacity}, which divides by $q_\delta$; the configuration then falls
back to a confidence-controlled failed-only empirical quantile, or the tier is
rejected.  This distinction matters at the tiers the protocol uses: at
$\delta=0.01$ and $p_F=1$ the Markov form needs $M>201$ and gives
$q_\delta\approx0.56$ at $M=1024$, while small sketches and small $p_F$ fall on
the non-positive side.

Both formulas are conservative because they cover every distribution with the
same moments, so they serve only as a theoretical floor that does not depend
on the empirical tail shape, and practical configuration uses the failed-only
empirical quantiles of \S\ref{sec:exp-validation}.

\subsection*{Finite-sample Kolmogorov bound}

Write
\[
  G_M(t)=\Pr[\chi^2_{M-1}\le(M-1)t].
\]

\begin{proposition}[Finite-sample Kolmogorov bound]
\label{prop:kolmogorov}
For any fixed sign pattern, any $M\ge2$, $1\le k<M$, and $d\ge1$,
\[
  \sup_t\left|
  \Pr[\dhat/d\le t]-G_M(t)
  \right|
  \le \frac{58(M-1)^{7/4}}{\sqrt d}.
\]
\end{proposition}

\begin{proof}
From $T=\lVert QC\rVert^2$ and $\dhat=T/\gamma$, for $t\ge0$,
\[
  \{\dhat/d\le t\}
  =
  \{\lVert QC/\sqrt d\rVert\le\sqrt{t\gamma}\},
\]
so the event is a Euclidean ball in the centered subspace
$H=\operatorname{range}(Q)\cong\mathbb R^{M-1}$; for $t<0$ both probabilities
are zero.

Write
\[
  \frac{QC}{\sqrt d}=\sum_{x=1}^d\xi_x,
  \qquad
  \xi_x=\frac{s_xQa_x}{\sqrt d}.
\]
The $\xi_x$ are mutually independent with zero mean, and the covariance sum
is $\sigma^2I_{M-1}$ on $H$.  Moreover,
\[
  \lVert s_xQa_x\rVert^2=a_x^\top Qa_x=\gamma
\]
is a deterministic constant.  Let $\zeta_x=\xi_x/\sigma$; then the covariance
sum is $I_{M-1}$, and using $\gamma/\sigma^2=M-1$,
\[
  \sum_x\E\lVert\zeta_x\rVert^3
  =\frac{(M-1)^{3/2}}{\sqrt d}.
\]
Rai\v{c}'s explicit Berry--Esseen bound for independent random vectors
and convex sets~\cite[Theorem~1.1]{raic2019multivariate}
gives
\[
  \sup_{A\ \mathrm{convex}}
  \left|
  \Pr\!\left[\sum_x\zeta_x\in A\right]-\Phi_{M-1}(A)
  \right|
  \le
  \bigl(42(M-1)^{1/4}+16\bigr)
  \frac{(M-1)^{3/2}}{\sqrt d}
  \le
  \frac{58(M-1)^{7/4}}{\sqrt d}.
\]
After standardization, the ball above has radius $\sqrt{t(M-1)}$, and its
Gaussian probability is exactly
\[
  \Pr[\chi^2_{M-1}\le t(M-1)]=G_M(t).
\]
This yields the stated Kolmogorov bound.
\end{proof}

In practice, this explicit constant is extremely conservative and serves
only as a finite-sample certificate, not as a basis for engineering
calibration: because Rai\v{c}'s theorem covers arbitrary distributions and
arbitrary convex sets, the bound is necessarily loose.  In actual runs and
Monte Carlo simulations (e.g.,
Table~\ref{tab:t1a} of \S\ref{sec:exp-validation}), the empirical error is
orders of magnitude smaller than this theoretical bound.

\subsubsection*{Finite-sample refinement of the conditional transfer}

Proposition~\ref{prop:cond-chi2} gives both the total-variation transfer and
the explicit success-probability bound
$1-p_F\le\min\{1,\varepsilon_{M,d}\}$.  Combining that transfer with
Proposition~\ref{prop:kolmogorov} through the triangle inequality bounds the
failed-only distribution directly against the chi-square reference:
\[
  \sup_t
  \left|
  \Pr[\dhat/d\le t \given F]-G_M(t)
  \right|
  \le \min\{1,\varepsilon_{M,d}\}
      +\frac{58(M-1)^{7/4}}{\sqrt d},
  \qquad
  G_M(t)=\Pr[\chi^2_{M-1}\le(M-1)t],
\]
the finite-sample counterpart of Eq.~\eqref{eq:cond-chi2}.  The first term
isolates the cost of failure selection.  For $k=3$ and $d/M=8$, it ranges
from $3.4\times10^{-8}$ at $M=64$ to $3.7\times10^{-6}$ at $M=4096$; at the
production-shaped load $d/M_1\approx182$, it is below $10^{-200}$ on the
standard tiers.  The Berry--Esseen term has a worst-case constant and is
numerically vacuous at these engineering parameters, so the displayed bound
is an asymptotic certificate rather than a calibration rule.  The much
closer transition-regime agreement in Figure~\ref{fig:fcond-qq} and
Table~\ref{tab:d3} remains empirical evidence, and production configuration
continues to use the failed-only empirical quantile.

\subsubsection*{Proportional-growth direction}

Proposition~\ref{prop:cond-chi2} fixes $M$ and lets failure become typical.
A complementary question lets $M,d\to\infty$ with $d/M\to\alpha\in(0,\infty)$,
so the failure probability may remain nondegenerate.  Under i.i.d.\ random
signs, conditioning on a mapping turns the centered estimator into a
second-order Rademacher chaos.  A quantitative de Jong argument can then be
organized around two mapping statistics: concentration of the total
quadratic energy and the maximum coordinate influence.  The intended
sufficient scale is that failure remain large relative to the exceptional
mapping set; the current second-moment route suggests $Mp_F\to\infty$.

This proportional-growth route would establish stability of the standardized
Gaussian law at the CLT scale.  It would not distinguish that law from the
nearby large-degree chi-square reference without sharper error control.  We
record the route here because it addresses transition loads where
Proposition~\ref{prop:cond-chi2} is silent; no protocol guarantee or empirical
calibration in this paper depends on it.

\end{document}